\documentclass[fleqn, final]{unmeethesis}
\usepackage{Definitions}
\pdfoutput=1
\usepackage[numbers, round, sort&compress]{natbib}

\setcitestyle{notesep={, },square, aysep={},yysep={,}}

\usepackage[margin=10pt,labelfont=bf]{caption}
\usepackage{pdfpages}

\definecolor{rob}{rgb}{0,0,0.75}

\usepackage[bookmarks = true, colorlinks = true, urlcolor=black, citecolor=black, linkcolor=black, hyperfootnotes=false]{hyperref}

\begin{document}

\bibliographystyle{apsrev}


\frontmatter

\title{Continuous Measurement and Stochastic Methods in Quantum Optical Systems}

\author{Robert Lawrence Cook}

\degreesubject{Ph.D., Physics}

\degree{Doctor of Philosophy \\ Physics}

\documenttype{Dissertation}

\previousdegrees{B.S., University of California  Santa Cruz, 2003}

\date{May 2013}

\maketitle

\makecopyright

\begin{dedication}
   Kylie, I promise we'll take a walk when this is all over.
\end{dedication}

\begin{acknowledgments}
   \vspace{0.4in}
   First of all I'd like to thank my most recent and final advisor Ivan Deutch.  When I started at UNM 10 years ago I had no idea what I wanted to study, only that a masters program seemed better than a job at the latest flying-Starbucks. It was your undergraduate quantum mechanics lectures that showed me how strange and rich the quantum world can be and they ultimately set me on the path to where I am today.  I will be forever grateful for your help and guidance though the bumpier parts of my graduate career.   I also have to thank Brad Chase.  Without you this dissertation would have taken a very different form.  Prior to reading the epic works of van Handel \emph{et al.} I would never have guessed that I'd become an advocate for mathematical formalism.  To Ben Baragiola I thank you for your friendship, enthusiasm and willingness to talk though a problem.  And to Heather Partner I will always be grateful for your support and camaraderie on the roller coaster ride that started at Los Alamos, ran through UNM and ended in Sandia.

   In my latest academic home of Room 30, I need to thank Carlos Riofr\'{i}o for your friendship, warmth and immediate inclusion into Deutsch group, Josh Combes for your shared enthusiasm for QSDEs, Leigh Norris for your kind hearted adoption of the luckiest goldfish on the planet, and Vaibhav Madhok for just being Vaibhav.  To the rest of Deutsch group - Bob Keating, Charlie Baldwin, and Krittika Goya - thanks for listening to me prattle on in group meeting about stochastic calculus and statistical estimation.  I hope I didn't bore you too much.  In the greater quantum information group I need to thank Professors Carl Caves and Andrew Landahl, current and former CQuIC students Jonas Anderson, Chris Cesare, Seth Merkel, Iris Reichenbach, Alexandre Tacla, Zhang Jiang, Matthias Lang, and Shashank Pandey.  I must also thank Vicky Bird for feeding us so well during arxiv review.  From my short tenure at Sandia national labs I need to thank Cort Johnson, Dan Stick, Todd Barrick, Dave Moehring, Francisco Benito, Peter Schwindt, Yuan-Yu Jau, Mike Mangan, Tom Hamilton, and Grant Biedermann for the help and support as I learned that cryogenic experiments are not for me.  I will never forget the time spent working with Roy Keyes, Tom Jones, Thomas Loyd, and Paul Martin.  While we may not have gotten a lot done we had a whole lot of fun doing it.  To Laura Zschaechner thanks for being a good friend and a shoulder to cry on.  And finally I'd like thank my parents and family for their love and support.
\end{acknowledgments}

\maketitleabstract

\begin{abstract}
  This dissertation studies the statistics and modeling of a quantum system probed by a coherent laser field.  We focus on an ensemble of qubits dispersively coupled to a traveling wave light field.  The first research topic explores the quantum measurement statistics of a quasi-monochromatic laser probe.  We identify the shortest timescale that successive measurements approximately commute.  Our model predicts that for a probe in the near infrared, noncommuting measurement effects are apparent for subpicosecond times.

The second dissertation topic attempts to find an approximation to a conditional master equation, which maps identical product states to identical product states.  Through a technique known as projection filtering, we find such a equation for an ensemble of qubits experiencing a diffusive measurement of a collective angular momentum projection, in addition to global rotations.  We then test the quality of the approximation through numerical simulations.  This measurement model is known to be entangling and without the rotations we find poor agreement between the exact and approximate predictions.  However, in the presence of strong randomized rotations, the approximation reproduces the exact expectation values to within 95\% accuracy.

The final topic applies the projection filter to the problem of state reconstruction.  We find an initial state estimate based on a single continuous measurement of an identically prepared atomic ensemble.  Given the ability to make a continuous collective measurement and simultaneously applying time varying controls, it is possible to find an accurate estimate given based upon a single measurement realization.  Previous experiments implementing this method found high fidelity estimates, but were ultimately limited by decoherence.  Here we explore the fundamental limits of this protocol by studying an idealized model for pure qubits, which is limited only by measurement backaction.  This ultimately makes the measurement statistics a nonlinear function of the initial state.  Via the projection filter, we find an efficiently computed approximation to the log-likelihood function.  Using the exact dynamics to produce simulated measurements, we then numerically search for a maximum likelihood estimate based on the approximate expression.  We ultimately find that our estimation technique nearly achieves an average fidelity bound set by an optimum POVM. 
\end{abstract}

\tableofcontents


\mainmatter

\chapter{ Introduction \label{chap:Intro} }

Within the past three decades, the ability to engineer individual quantum systems into highly nonclassical states has become a reality.   The fundamental technology that facilitated these revolutionary experiments is the coherent laser with its ability to address specific electronic transitions in matter.  A quasi-monochromatic laser can also introduce optical forces on position degrees of freedom.  While initially used to laser cool and trap atoms, coherent electronic superpositions  can also be transferred to external superpositions allowing for atom interferometers \citep{rasel_atom_1995} or highly nonclassical states in trapped ions \citep{haffner_scalable_2005,leibfried_creation_2005}.

In addition to providing control of an atomic system at a quantum level, the same laser systems can be used to measure the atomic state of the system.  The simplest of all detection methods is resonant fluorescence, where laser light resonant with a single transition is applied to an atom, which will scatter photons if that level is occupied.  However, if the internal state is in a superposition between the resonant level and an additional off-resonant `dark' state, the presence or absence of scattered light provides information about the internal state of the system to the experimenter \cite{itano_quantum_1993}.  

The quantum nature of the atom-light interaction carries over to off resonant applications.   In a low intensity regime, a free space laser with a carrier frequency significantly detuned from an atomic transition predominantly induces a state dependent energy shift without significantly exciting that transition.  The specific form of the interaction also depends upon the polarization of the exciting laser.  In a given parameter regime the resulting Hamiltonian dominates over the decoherence from absorption and subsequent emission, resulting in a controllable coupling between the atoms and the polarization of the probe laser \cite{deutsch_quantum_2010}.  This coupling affects both the atomic and polarization quantum states.  The fact that the laser is a traveling wave means that this is a fundamentally open quantum system and the output state of light carries with it some information about the atomic state.

Quantum mechanics is at its core a probabilistic theory where the wave function is a tool for computing the probability of observing experimental events.  Upon the receipt of a measurement outcome, an accurate description of the quantum system must reflect this new information.  This is true in an idealized projective measurement or in indirect measurements like those described above.  However there are significant differences between the method of state detection via resonance fluorescence and off-resonance polarization spectroscopy.

In fluorescence detection a vast majority of the scattered light is ultimately lost, either because only a small fraction of possible emission directions are observed or due to losses in the detection apparatus.  It takes a considerable experimental effort to measure as little as $5\%$ of the total scattered light from a single trapped ion \cite{vandevender_efficient_2010, streed_imaging_2011}.  A single ion will have to scatter a lot of light in order for an experimenter to be able to discriminate a bright state from a dark state with any reasonable confidence.  This means that after a relatively short time period it is likely that an ion prepared in a superposition of bright and dark states has scattered several photons that were lost to the experimenter.  Honest scientists would be forced to admit that while they were still uncertain as to the outcome of the measurement, they are quite certain that any coherence between the bright and dark states has been destroyed.

In the off-resonant scheme, almost all of the probe light can be collected, meaning that a clever experimentalist has access to nearly all of the information available.  After a short interaction time the measured state will in general change, but only in proportion to the amount of information gained.  The point is that armed with a complete measurement record it is possible to track the evolution of the state from an initial superposition to a final projected outcome.  This kind of measurement is known as a weak quantum nondemolition measurement (QND), and has been demonstrated in several different experiments involving ensembles of monatomic gasses in various parameter regimes.  One important consequence of this kind of measurement is that the projective outcome is a highly nonclassical state involving strong quantum coherence between all of the atoms in the ensemble.  While the experimental realities of photon scattering ultimately limit the system from reaching this eigenstate, an intermediate squeezed state has been observed on several occasions \cite{kuzmich_generation_2000, hald_spin_1999, schleier-smith_states_2010}, where the uncertainty of the measured observable is reduced below the standard quantum limit.

This dissertation is focused on the modeling of a quantum atomic system interacting with a quantized traveling-wave optical probe when that field is also measured continuously in time.  The most interesting quantum effects occur in the idealized case with no loss of light and a noise free measurement, which is the only case considered here.  The progression from an initial superposition state to a final measurement eigenstate is neither a time stationary process nor a linear transformation and so a model capable of tracking the full transition must be both time-adaptive and nonlinear.  Finding such a mathematical description is not a trivial exercise but one that has been extensively studied previously.

The ultimate objective is to apply this continuous measurement model to the problem of quantum state tomography.  Constructing an estimate for an arbitrary quantum state based upon experimental data is very resource intensive.  Specifying an arbitrary quantum state for a $d$-dimensional system requires at most $d^2 - 1$ parameters and for each parameter $N$ uncorrelated projective measurements generally gives an accuracy of $\sqrt{N}$.  Through an alternative protocol proposed by \citeauthor{silberfarb_quantum_2005}, these inefficiencies can be largely side-stepped, by applying a weak continuous measurement plus a well chosen dynamical control \emph{collectively} to an ensemble of identically prepared systems \citep{silberfarb_quantum_2005}.  If the control drives the system in such a way as to make the measurement informationally complete, then a single measurement record has encoded information about all $d^2 - 1$ parameters, albeit with a varying level of certainty and noise corruption.

In particular, we consider an atomic ensemble prepared in an identical tensor product state $\rho_{\text{tot}} = \rho_0^{\otimes n}$ that experiences a known Hamiltonian while simultaneously coupled to a traveling wave probe via a collective degree of freedom.  A continuous measurement of this probe then generates a measurement record that is strongly correlated with the evolution of the system.  With sufficient signal to noise, a statistical estimate of an unknown initial system state will in general have high fidelity with the true initial condition.  Using the weak measurement generated by an off-resonance probe, this state reconstruction procedure has been implemented in the laboratory, allowing reconstruction of the full hyperfine, $d = 16$, ground state manifold of a laser cooled neutral Cs atom ensemble \citep{smith_efficient_2006, riofrio_quantum_2011}.  However, these experiments were performed in a parameter regime where the amount of information lost to the environment dominated over any measurement induced backaction.  Chap. \ref{chap:QubitState} explores how this procedure performs in an opposite regime where decoherence is negligible and we retain a complete measurement record.   Arriving at this result requires several intermediate steps, particularly a detailed knowledge of how a classical statistical estimate is made and how that is applied to a quantum system.

The estimation of a possibly random signal from an observation corrupted by unwanted noise is known as \emph{filtering} and takes its origin from the work of Wiener \cite{Wiener_extrapolation_1949}, where the signal was assumed to be generated with time stationary statistics.   In a linear system with additive Gaussian noise, an optimum estimate to a continuous nonstationary signal was computed by \citeauthor{kalman_new_1961} \cite{kalman_new_1961} and is an indispensable tool in engineering and classical signal processing and control.  Not surprisingly an estimate to a nonlinear signal is significantly more challenging than in a linear system.  The nonlinear classical filter began with \citeauthor{stratonovich_conditional_1960} \citep{stratonovich_conditional_1960} and was later expressed in the useful language of It\={o} calculus by \citeauthor{kushner_differential_1964} \citep{kushner_differential_1964}.  Important contributions were made by \citeauthor{kallianpur_estimation_1968} \citep{kallianpur_estimation_1968} and \citeauthor{zakai_optimal_1969} \citep{zakai_optimal_1969}, steps that are particularly useful in formulating a quantum analog. Nonlinear filtering theory has an extremely wide range of applications including GPS based navigation, optimal stochastic control, financial portfolio optimization, audio and imaging noise removal and enhancement, speech recognition, weather prediction, and so on \cite{van_handel_filtering_2006}.  For each application there is a rich body of literature with a wealth of numerical and approximation methods allowing for practical and, in some cases, real-time implementations.

One of the fundamental tools that makes continuous-time classical estimation possible is stochastic calculus.  In the same way that a global function can be built up from an integral over local infinitesimals, a random signal can also be constructed from random increments.  In order for the filtering problem to be remotely tractable, it is necessary that the random increments originate from an uncorrelated noise source.  This means that the fundamental noise injected into an otherwise deterministic system is assumed to be uncorrelated \emph{white noise}.  Under this assumption the filtering equation is Markovian, meaning that the estimate updates only according to the latest measurement and its most recent value.  Classically, a white noise approximation is often well justified when the input noise is actually an aggregate effect from a large number of uncorrelated sources.  The canonical example is a particle experiencing Brownian motion.  Each impulse is a collision with a background molecule imparts a small amount of momentum to the larger particle. For times that are longer than the time between collisions, the net displacement is uncorrelated with previous intermediate times.  In this case not only are the collisions uncorrelated, the particle's displacement is Gaussian distributed with a variance that grows proportional to time.  Brownian motion is a classic example for a system influenced by Gaussian white noise, but due to the central limit theorem, these models are ubiquitous in systems with continuous trajectories.

Working with white noise directly adds another layer of sophistication to an already mathematically challenging topic \citep{holden_stochastic_2010}.  This is due to, among other problems, the fact that as white noise is defined to have completely uncorrelated fluctuations at any point in time, it is therefore discontinuous; continuity from one point to the next would imply correlations.  To build a mathematical framework that is both useful and provably consistent, the hard learned lesson is to frame the problem not in terms of the white noise itself but to instead use its integral, which is at least continuous\footnote{The calculus of randomized distributions characterizing white noise is referenced back to the integrated expressions anyway \citep{holden_stochastic_2010}.} \cite{oksendal_stochastic_2002}.  The most widely used representation for an integral over Gaussian white noise (in other words a mathematical model for Brownian motion) is the Wiener process.  Chap. \ref{chap:Math} reviews many of its defining properties, which we are able to leverage into a maximum likelihood estimate of an initial quantum state based upon a polarimetry measurement.

Beyond a successful model for integrated white noise, a classical filter needs to be able to manipulate these integrals in a full fledged calculus.  The subtlety of dealing with a randomized integral is that different limiting approximations lead to fundamentally different stochastic calculi.  The two most common forms of stochastic integration are the \emph{Stratonovich} integral and the \emph{It\=o} integral, with differing calculus rules and statistical properties \cite{oksendal_stochastic_2002}.  Both forms of integration are used here and a brief review is presented in Appendix \ref{app:SDEs}.  One drawback from initializing a model with a stochastic integral is that any predictions will dramatically depend on what kind of integral is used.  At one level the choice of integral is no different than a number of other approximations one makes in formulating a statistical model of a given physical system.

The lessons from developing a stochastic calculus for classical systems has also been applied to quantum models as well.  In the mid 80's \citeauthor{hudson_quantum_1984} developed an operator-valued quantum version of the classical It\={o} calculus \citep{hudson_quantum_1984}.  This quantum It\=o calculus is indispensable in modeling open quantum optical systems and has been applied to not only continuous quantum measurement \cite{belavkin_quantum_1992-1, bouten_introduction_2007}, but also quantum control (see \cite{gough_principles_2012} for an overview and introduction).  

The It\=o integral is based upon the assumption that the noise is completely uncorrelated so that the integral between times $[t_0, t_1)$ will be independent from the integral over times $[t_1, t_2)$ for any times, $0 \le t_0 < t_1 < t_2$, no matter how small the difference.  For an actual Brownian particle, this is only an approximation as the particle's velocity is correlated for times between atomic collisions.  The quantum It\=o integral also makes such an assumption but it does so by assuming that the operators representing equivalent integrals commute for nonoverlapping times, no matter how small.  Before immediately applying Hudson and Parthasarathy's formalism to the laser probed system, Chap. \ref{chap:QuantumLight} investigates for what times such an approximation applies, given that the resulting operators must also be consistent with a quasi-monochromatic description of a traveling wave light field.

The similarities between classical filtering and a quantum system subject to an indirect measurement should not be ignored.  In a classical setting one seeks an estimate of an unobserved system state consistent with a noisy measurement.  The fundamental goal in a quantum system is to predict the results of future measurements consistently and accurately with past measurements.  The unobserved atomic system is then the estimated quantity and the measured probe gives the noisy data.  To fully exploit this similarity and apply the techniques developed for classical systems, several fundamental questions about the nature of quantum statistics and classical probability must be addressed.

The development of a continuous-time filter for a quantum system was pioneered by the work of \citeauthor{belavkin_quantum_1980} starting in the early 80's \cite{belavkin_quantum_1980, belavkin_nondemolition_1988, belavkin_quantum_1992, belavkin_quantum_1992-1}.  These mathematically rigorous results developed and applied a deep relation between the algebraic and commutative properties of operators on a Hilbert space and the expression of classical stochastic processes.  Experimental observations in elementary quantum theory are represented by Hermitian operators and that upon making a measurement the random outcome corresponds to an eigenvalue of that operator.  The connection is then made though the following two insights.  The first insight is then to associate operators with random variables.  In classical probability all random variables are consistent, in the sense that all random variables will agree that the same underlying outcome of the system occurred, no matter what order they are queried.  For quantum systems it is a hard learned fact that only commuting operators will return consistent results.  Thus, the second insight is that in order to use a sequence of measurements for statistical inference, all of the measured operators must commute.  Additionally, any operator whose statistics we wish to infer must also commute with all measurements to date.  The utility of considering sets of commuting operators for the purposes of statistical inference is more well known in the physics community as the defining property of QND \cite{grangier_quantum_1998}.  Working within these limitations, the problems of noncommutativity is no longer an issue leading to a real and useful mapping between quantum measurements and classical probabilities.

A mapping between the quantum and the classical descriptions of probability can be more than just a guiding principle.  Through a formal isomorphism between commuting quantum operators and the language of classical filtering theory all of the above classical results can be easily applied.  The quantum filter developed by Belavkin is nothing more than a noncommuting analog of the classical Kushner–-Stratonovich equation of nonlinear filtering \cite{bouten_introduction_2007}.  Chap. \ref{chap:Math} reviews how a mapping between quantum operators and a formal classical probability model is made.  The purpose of this review is two fold.  The first is to provide the necessary background for a quantum filter.  The second is to shed light on how the algebraic language of classical probability theory can be applied to quantum systems thereby gaining new insights and intuitions into the quantum/classical divide.  When the quantum and classical coincide, nearly 50 years of engineering experience can either be immediately applied or adapted with some modifications.    Finally, this chapter shows how the quantum filter is equivalent to a generalized measurement by making a unitary extension to a larger dimensional Hilbert space.

Chap. \ref{chap:projection} applies one such method to the quantum system of an identical spin ensemble undergoing a polarimetry measurement in idealized conditions.  \citeauthor{brigo_differential_1998} applied the methods of differential geometry to simplify a classical filter \cite{brigo_differential_1998, brigo_approximate_1999}.  This method of making differential projections was adapted to a quantum system by \citeauthor{van_handel_quantum_2005} where they simplified a continuous quantum measurement of a strongly driven atom-cavity system into manifold states where the cavity has a Gaussian $Q$-function \cite{van_handel_quantum_2005}.  The method has been subsequently applied to other cavity QED systems \citep{mabuchi_derivation_2008, hopkins_reduced_2009, nielsen_quantum_2009}, collective spin systems in a linearized-Gaussian regime \citep{chase_parameter_2009}, and to find a low rank approximation to a general master equations in Lindblad form \cite{bris_low_2012}.  Chap. \ref{chap:projection} computes the orthogonal projection of an ensemble of $n$ qubits into the manifold of identical separable states of the form $\rho^{\otimes n}$ and numerically compares the accuracy of such an approximation to a complete evolution.

Using the projected filter as a computational tool, Chap. \ref{chap:QubitState} turns the problem of quantum state tomography essentially into a classical parameter estimation problem, where the classical parameters are the pointing angles of a spin coherent state constructed from the initial $n$ qubits.  The parameters are estimated by numerically computing a maximum likelihood estimate based upon a polarimetry measurement when the measurement statistics are strongly affected by quantum backaction.  By only including the conditional effects present in the projection filter we achieve an average reconstruction fidelity that nearly saturates an optimum bound given by any generalized measurement scheme \citep{massar_optimal_1995}.

\subsection{A note on quantum foundations}

Any work that addresses the quantum world and in particular quantum measurement eventually encounters some issue rooted in the foundations of quantum mechanics and the various interpretations one could assume.  This dissertation does not address quantum mechanical foundations in any meaningful way and attempts to remain agnostic about the reality of a quantum state or even the existence of a more fundamental theory.  Wherever possible we take a statistical perspective and implicitly assume that the simple models we construct may not be error free, in the sense that they do not include the whole the reality of a given experiment.

When considering quantum state tomography, we compare the conditional evolution of an ensemble of initial conditions and then select the state that maximizes a likelihood function.  In our numerical simulations, no member of that ensemble corresponds with arbitrary precision to the initial condition used to simulate the measurement record.  So in one sense the conditional state we calculate will always be incorrect.  However, in a field where the ontology of a quantum state is still debated, we take a conservative position and will not to presume to know that any conditional state is the true conditional state.  Instead we will only take the stance that what we calculate is a quantum state that best predicts any future measurement in a manner that is consistent with past results and the assumptions of the model.  We identify this state through the framework of quantum probability theory, a formalism that is less well known to physicists working in quantum information theory.  The final object that we calculate is ultimately no different from what is given by the usual stochastic Schr\"{o}dinger/master equations that are used in quantum optics.

The purpose of working with quantum probability theory is that it illustrates an immediate connection to classical probability and estimation theory.  In the classical setting, an estimator is a tool that is used to predict or estimate some quantity given a series of measurements.  The stochastic Schr\"{o}dinger equation is in a very real sense a quantum estimator.  We would rather not comment as to whether or not it is estimating the state of the system because it lacks knowledge of a theory extending beyond standard quantum mechanics or if it is the fundamental limit and there is no more information in existences.  In effect, we assume that the quantum state is simply a tool for making predictions about a quantum system.

\section{An executive summary}

This is a terse summary of the fundamental results of this dissertation, presented in the same order as the subsequent chapters.  This is not intended to be a gentle introduction to the material and assumes a strong familiarity with the background material.  We encourage an interested but nonexpert reader not to struggle too hard trying to comprehend this section and instead consult the main text and the associated appendices.

If any single global thesis can be applied to the entirety of this work it is that classical probabilistic methods are useful and that with some care they can be adapted to quantum systems.    The previous introduction discussed how to connect stochastic calculus and nonlinear filtering theory to a quantum system continuously probed by an optical field.   Chaps. \ref{chap:QuantumLight} and \ref{chap:Math} provide a physical and mathematical foundation for this connection while Chaps. \ref{chap:projection} and \ref{chap:QubitState} apply it to the specific problem of efficiently estimating an initial qubit state.  Chap. \ref{chap:Outtro} discusses possible directions this work could take.  In addition to this main matter, we include several appendices providing background material such as a review of paraxial optics (Appendix \ref{app:paraxialOptics}), stochastic differential equations (Appendix \ref{app:SDEs}), quantum stochastic differential equations (Appendix \ref{app:QSDEs}) and the quantum Wong-Zakai theorem (Appendix \ref{app:QuWongZakai}).

\subsection{Quantum optics and quantum stochastic differential equations }

Chap. \ref{chap:QuantumLight} shows how a second quantized picture of classical traveling wave packets reproduces the mathematical structure necessary for defining a formal quantum It\={o} stochastic calculus.  It also identifies the timescales for which a quasi-monochromatic field can be approximated as generating quantum white noise.  This is a regime that is independent from any system coupling or measurement apparatus and applies for a large family of states - include highly nonclassical states, such as multi-mode Fock states.

The specific model we consider is the second quantization of quasi-monochromatic wave packets \citep{deutsch_basis-independent_1991, garrison_quantum_2008, smith_photon_2007} where the single particle Hilbert space is the space of coherent state amplitudes for an associated classical field.  We assume a paraxial model where there is a factorization between a carrier plane wave $\exp(-i \omega_0 (t - z/c) )$, spatial mode function $\V{u}^\pf_T(x, y, z)$, and longitudinal envelope function $f(t - z/c)$.  The quasi-monochromatic approximation means that the longitudinal function must satisfy the inequality, $\abs{f(t)} \gg \frac{1}{\omega_0}\abs{\frac{\partial}{\partial t} f(t)} \gg \frac{1}{\omega_0^2} \abs{\frac{\partial^2}{\partial t^2} f(t) }$.

We ultimately seek creation and annihilation operators that are simultaneously quasi-monochromatic as well as consistent with a quantum white noise approximation.  To do so, we define $\ahat^\dag[\Vf(0)]$ to be the operator that creates a single quantum in a given spatial mode $\V{u}^\pf_T(x, y, z)$, with an envelope function $f$, referenced to some point along the optical axis.  The operator $\ahat[\Vf(t)]$ annihilates a quantum in a similar mode but one that has experienced free propagation for a time $t$.  We derive the unequal time commutation relation,
\begin{equation}
 \hspace{-10 pt} \Big[\,\ahat[\Vf(t_1)],\, \ahat^\dag[\Vf(t_2)]\, \Big]  \propto  e^{-i \omega_0(t_2 - t_1)} \left( f \star f\ (t_2 - t_1) - i\frac{1}{\omega_0} \frac{df}{dt}\star f\ (t_2 - t_1) \right)
\end{equation}
where $g \star f$ is the cross-correlation function of $g$ and $f$ and the proportionality constant is simply a scaling factor that can be absorbed into the definition of $f$.  Physically it is entirely reasonable that if a classical envelope is no longer temporally correlated then the associated field operators should commute.  To the best of our knowledge this is a new result in the characterization of quantized fields.

The canonical definition for quantum white noise is that there exist the creation and annihilation operators $[\ahat(t), \ahat^\dag(t')] \propto \delta(t - t')$.  Therefore in order for quasi-monochromatic light to be consistent with a white noise approximation, not only does $f \star f\ (t_2 - t_1) \rightarrow \delta( t_2 - t_1 )$ in a suitable limit but $\frac{1}{\omega_0} \frac{df}{dt}\star f\ (t_2 - t_1) \rightarrow 0$.  

From an approximation to white noise, (in a rotating frame) we then use the limiting white noise operators and a recent theorem by \citeauthor{gough_quantum_2006} \cite{gough_quantum_2006}, reviewed in Appendix \ref{app:QuWongZakai}, to consider the dispersive Faraday interaction, in an idealized regime where the possibility for multiple scattering events is nonnegligible.  The limiting object is a quantum stochastic It\={o} equation for the propagator that describes the unitary evolution between the field and the atomic system.  Using well know results in quantum stochastics we write down the equivalent master equation in Lindblad form.

\subsection{Classical and quantum probability theory}

Chap. \ref{chap:Math} is a mathematical review of well known results from classical and quantum probability theory, which serves as a foundation for the novel work in later chapters.  This review is conducted with an emphasis for physicists and attempts to explain and justify the concepts while omitting the proofs.  The end goal is have an understanding of how the conditional master equation results from a mapping between sets of commuting operators and a classical probability space.  This is a critical point as the driving noise in the conditional master equation is \emph{not} a Wiener process, but is instead the random outcomes of a continuous quantum limited measurement.  The resulting classical stochastic process $\set{y_t}_{t \ge 0}$ is only a Brownian motion when the measurements are (\emph{i}) of a field quadrature in the vacuum state and (\emph{ii}) there is no system coupling to that quadrature, \emph{i.e.} the measurement has no system information.  

The second objective of this chapter is to emphasize the general power of this technique and to discuss how the language of classical probability theory can be used to identify semiclassical subspaces embedded in a quantum system.  In order to do so in a relatively self-contained manner we review the basic elements in the triple $(\Omega, \mathcal{F}, \mathbbm{P})$ forming a classical probability space.  The infinite dimensional example we focus on is the sample paths for a Brownian motion and explicitly describe the relevant $\sigma$-algebra. 
In order to introduce the quantum conditional expectation, we first review the classical conditional expectation and more generally how expectation values are computed in the measure theoretic framework.  We then introduce the concept of time-adapted processes and martingales as both are crucial in the quantum case.  The review of classical probability theory concludes by discussing the Wiener process and the Wiener measure over the space of continuous functions.

From a firm description of classical probability theory we then discuss the quantum analog.  We explicitly show how one identifies a classical probability space from the set of mutually commuting observables by taking $\Omega$ as the set of possible eigenvalues, $\mathcal{F}$ as the $\sigma$-algebra generated by those eigenvalues, and the probability measure $\mathbbm{P}$ as the quantum expectation under the state $\rho$ of the associated projectors.  From this semiclassical description we then introduce the noncommutative analog were one omits a sample space of compatible outcomes, identifies the $\sigma$-algebra with a $*$-algebra of operators (or a von Neumann algebra in the infinite dimensional case), and the probability measure with a valid quantum state $\rho$. We then explain the power of generating sub-$*$-algebras from sets of operators focusing on the important object of the commutant.  We specifically explain how the commutant is the largest space of operators that we can condition on a sequence of commuting observations and how it contains noncommuting elements.  Armed with that description we identify the properties of the quantum conditional expectation, and provide an explicit construction for how it is in correspondence with the generalized measurements found in quantum information theory.   From the discussion of the quantum conditional expectation we then state the resulting the quantum filter as it is generated from the observation process $\set{ Y_t = U_t^\dag (A_t + A_t^\dag) U_t }_{t \ge 0}$ under vacuum expectation.

While the quantum filter is an elegant expression for a conditional operator, it rarely closes to a finite set of quantum stochastic differential equations.  Rather it is more useful to work with an effectively semiclassical equation, the conditional master equation.  Here we use the term semiclassical in a sense that does not imply a suboptimal approximation but rather to indicate that the quantum measurement process $\set{ Y_t }_{t \ge 0}$ (a family of operators) is treated as a classical stochastic process $\set{ y_t }_{ t \ge 0 }$ (a family of classical random variables defined on the probability space $( \Omega, \mathcal{F}, \mathbbm{P} )$ ).  The probability measure $\mathbbm{P}$ matches the statistics of $\set{ y_t }_{ t \ge 0 }$ to the quantum measurement statistics $\set{ Y_t }_{t \ge 0}$.  While this mapping ``demotes'' the measurement operators to a classical process, it still treats the system quantum mechanically, by propagating a density operator $\set{ \rho_t }_{t \ge 0}$.  Generally the statistics of $\set{ y_t }_{ t \ge 0 }$ will depend upon a quantum system expectation value, and so this is semiclassical and not a fully classical probability model.  This system density operator matches the quantum conditional expectation by enforcing the equality
\begin{equation}
  \pi_t(X)|_{\set{Y_t = y_t}} = \Tr(\rho_t X)
\end{equation}
for every system operator $X$ and time $t$.

The quantum filter is derived in terms of a quantum It\={o} equation and so the resulting semiclassical conditional master equation is a matrix-valued classical It\={o} equation.  In Chap. \ref{chap:projection} we are required to express it in terms of a Stratonovich integral and so we derive the associated correction factor here.  The chapter closes by finding a conditional Schr\"{o}dinger equation that corresponds to the more general master equation in the case of pure states.  This equation is useful for numerical simulation as propagating a complex vector is more efficient than a complex matrix.

\subsection{Projection filtering for qubit ensembles}

Chap. \ref{chap:projection} derives an approximate form of the conditional master equation for an ensemble of $n$ qubits under the assumption that the state will remain nearly an identical separable state.  The approximation is made though a technique known as \emph{projection filtering}, developed to reduce the dimension of a classical filtering equation by formulating the space of solutions as a Riemannian manifold and then making an orthogonal projection onto a lower dimensional manifold.  The lower dimensional manifold that we wish to project onto is the space of density matrices that can be written as $\varrho = \rho^{\otimes n}$ for some valid single qubit state $\rho$.  The appeal of the projection filtering technique is that it is algorithmic in nature, in that after identifying the desired manifold and making a choice of metric, finding the optimal projection is reduced to a problem of matrix algebra. Due to the simplicity of the qubit we are able to solve for this projection analytically.

A third of this chapter reviews the fundamentals of differential geometry, focusing on the mapping between qubit states and the Bloch ball.  We refer to the set of valid quantum states for a $d < \infty$ dimensional quantum system as $\man{S}(d)$ and the three-dimensional unit ball as $\man{B}$.  The metric we use is the trace inner product $\inprod{A}{B}_\varrho = \Tr(A^\dag B)$ for $A, B \in T_\varrho\,\man{S}(d)$.  From the standard mapping between points in the Bloch ball and qubit states $\rho\ :\ \man{B} \subset \R^3 \rightarrow \man{S}(2)$
\begin{equation}
  \rho(\Vx) = \half \left(\ident + x^i \sigma_i \right)
\end{equation}
we identify a basis $\set{D_i \define \half \sigma_i}$ for the tangent space $T_{\rho(\Vx)}\man{S}(2)$.  The resulting trace inner product induces an Euclidean metric on $\man{B}$,
\begin{equation}
  \inprod{D_i}{D_j}_\rho = \tfrac{1}{4} \Tr(\sigma_i \sigma_j) = \half \delta_{ij}.
\end{equation}
The manifold we ultimately want to consider is the set of density operators
\begin{equation}
  \man{P} \define \set{\rho(\Vx)^{\otimes n}\ :\ \Vx \in \man{B} } \subset \man{S}(2^n).
\end{equation}
Any derivative we define for $\varrho \in \man{P}$ must distribute over the tensor product structure, and so we identify the tangent space
\begin{equation}
  T_{\varrho(\Vx)}\man{P} = \lspan{ D_i(\Vx) = \sum_{\ell =1}^n \rho(\Vx)^{\otimes \ell -1} \otimes \half \sigma_i \otimes \rho(\Vx)^{\otimes n-\ell } }.
\end{equation}
For $n \ne 1$, the metric on the Bloch ball induced from the trace inner product is no longer Eucildean.  Instead it is given by the matrix
\begin{equation}
  \begin{split}
    g_{ij}(\Vx) &= \Tr(D_{i}(\Vx)\, D_{j}(\Vx)\,)\\
        &= \frac{n}{2^n} \left(1 + \abs{\Vx}^2 \right)^{n-1}\, \delta_{ij}  +
            \frac{n(n-1)}{2^n} \left(1 + \abs{\Vx}^2 \right)^{n-2} \, x^k x^\ell\, \delta_{k i} \delta_{\ell j}.
  \end{split}
\end{equation}
This metric is however isotropic, which can be seen by converting to spherical coordinates.  The resulting line element is
\begin{equation}
  ds^2 = \frac{n}{2^n} \left(1 + r^2 \right)^{n-1} \left( \frac{1 + n r^2}{1 + r^2}\ dr^2 + r^2\ d\theta^2 + r^2 \sin^2 \theta\ d \phi^2 \right).
\end{equation}
With this non-Euclidean metric, we then wish to apply the projection map $\Pi_{\man{P}}\ :\ T_{\varrho} \man{S}(2^n) \rightarrow T_{\varrho} \man{P}$, defined as
\begin{equation}
  \Pi_\man{P}(X) = g^{ij}(\Vx)\inprod{D_j(\Vx)}{ X }_{\varrho} \, D_i(\Vx)
\end{equation}
to each terms in the conditional master equation.

A general unconditioned master equation written in Lindblad form is,
\begin{equation}\label{chIntro:eq:generalMaster}
    \tfrac{d}{dt} \varrho = -i [ H,\, \varrho] + \mathcal{D}[L](\varrho)
\end{equation}
for some Hamiltonian $H$ and jump operator $L$.  As the master equation describes a valid quantum evolution, the righthand side of this equation must describe a vector in the tangent space $T_\varrho \man{S}(2^n)$.  Applying the projector $\Pi_{\man{P}}$ to the general master equation results in a new master equation, describing the evolution of a modified state $\varrho|_\man{P}$,
\begin{equation}
    \frac{d}{dt} \varrho|_\man{P} = g^{ij}(\Vx) \Big( \inprod{D_j(\Vx)}{ -i[H,\, \varrho(\Vx) ] }_{\varrho}  + \inprod{D_j(\Vx)}{ \mathcal{D}[L](\varrho) }_{\varrho}  \Big ) \, D_i(\Vx).
\end{equation}
This new equation is guaranteed to both produce a valid quantum evolution as well as constrain the state to remain in $\man{P}$.
Performing this projection in the case of a conditional master equation is essentially no different, with one caveat due to the subtle nature of stochastic integrations.

Converting the derivative $\frac{d}{dt} \varrho|_\man{P}$ into a differential form $d\rho|_\man{P}$ has no ambiguity in interpretation as the differential
\begin{equation}
    d \varrho|_\man{P} = a^i(\Vx)\, D_i(\Vx)\, dt
\end{equation}
describes a valid mapping between the tangent space $T_t \R^+$ and the tangent space $T_{\varrho(\Vx)} \man{P}$.  However, interpreting a general stochastic differential in terms of a differential form is problematic because the explicit path-wise derivative generally does not exist.  Even if one were to solve the stochastic differential equation
\begin{equation}
    d \varrho|_\man{P} = B(\Vx_t)\, dw_t
\end{equation}
for $B(\Vx_t) \in T_{\varrho(\Vx_t)}\man{P}$ there is no \emph{a priori} reason to assume that the resulting solution will remain in $\man{P}$.  In developing the projection filtering technique, \citeauthor{brigo_differential_1998} found that a solution to a general It\={o} equation decidedly does not satisfy this property \cite{brigo_differential_1998}.  The problem is that the drift induced by the second order nature of the It\={o} rule causes the solution to leave $\man{P}$ even when the integrand is in the proper tangent space.  The saving grace is that the orthogonal projection method does constrain the solution when the original equation is written in Stratonovich form.  The bottom line is that in order to project the conditional master equation into the tangent space $T_\varrho \man{P}$ it must first be written as a  Stratonovich integral.  Chap. \ref{chap:Math} calculates the proper conversion factor, generating the Ito correction map $\mathcal{I}_c[L](\varrho)$ for a general measurement operator $L$.

The conditional master equation is given by the Stratonovich equation
\begin{equation}
        d \varrho = -i [ H_\tot,\, \varrho] dt + \mathcal{D}[L_\tot](\varrho) dt + \mathcal{I}_c[L_\tot] (\varrho) dt +  \mathcal{H}[L_\tot](\varrho)\circ dv_t
\end{equation}
where the maps $\mathcal{D}[L_\tot](\cdot)$, $\mathcal{H}[L](\cdot)$, and $\mathcal{I}_c[L](\cdot)$ are given in eqs. (\ref{chProj:eq:dissipator}, \ref{chProj:eq:conditioning}, and \ref{chProj:eq:ItoCorrectionMapIntro}) respectively.
The subscript $_\tot$ used here is used to specify that these operators act on the joint Hilbert space over all $n$ qubits.

The projections of each term are computed relatively generally, but under the assumption that the operators $H_\tot$ and $L_\tot$ act identically and independently on each qubit in the ensemble.  This means that for the single qubit operators $H$ and $L$ the joint operators are equal to
\begin{equation}
  H_\tot = \sum_{\ell =1}^n \ident^{\otimes \ell -1} \otimes H \otimes \ident^{\otimes n-\ell }
\end{equation}
and
\begin{equation}
  L_\tot = \sum_{\ell =1}^n \ident^{\otimes \ell -1} \otimes L \otimes \ident^{\otimes n-\ell }.
\end{equation}

From the general expressions we also specialize to the examples of $L = \sqrt{\kappa}\, \half \sigma_z$ and $H = \half (f^1(t) \sigma_x + f^2(t) \sigma_y + f^3(t) \sigma_z)$ for a constant rate $\kappa$ and deterministic real valued control fields $f^i(t)$.  This specialized example corresponds to an idealized model of a dispersive measurement of a collective angular momenta and a time varying but uniform magnetic field.  The final expression we calculate for this example and call \emph{the projection filter} is a system of coupled It\={o} stochastic differential equations that correspond to the single particle Bloch vector components, $\Vx_t$,
\begin{equation}
\begin{split}
    dx_t =&\, a^1(\Vx, t)\, dt - \sqrt{\kappa}\, x\, z\, dv_t,\\
    dy_t =&\, a^2(\Vx, t)\, dt - \sqrt{\kappa}\, y\, z\, dv_t,\\
    dz_t =&\, a^3(\Vx, t)\, dt + \sqrt{\kappa}(1 - z^2)\, dv_t.
\end{split}
\end{equation}
The deterministic integrands $a^i(\Vx, t)$ are
\begin{equation}
  \begin{split}
    a^1(\Vx,t) =&\, f^2(t)\,z - f^3(t)\, y - \half \kappa\, x + \kappa\, \gamma(r)\, x\, z^2,\\
    a^2(\Vx,t) =&\, f^3(t)\,x - f^1(t)\, z - \half \kappa\, y + \kappa\, \gamma(r) \, y \,z^2,\\
    a^3(\Vx,t) =&\, f^1(t)\,y  -f^2(t)\, x - \kappa\, (n-1)\big(\tfrac{1 - r^2 }{1 + r^2}\big)\, z - \kappa\, \gamma(r)\, z^3
  \end{split}
\end{equation}
with the function
\begin{equation}
    \gamma(r) \define (1 - r^2) \left( \frac{n\,(n+1)}{2\,(1 + n\, r^2)} -  \frac{1}{1 + r^2} \right).
\end{equation}
Finally the stochastic increment $dv_t$ is the innovation process, calculated from the measurement process $y_t$ (no relation to the Bloch vector component) with a differential
\begin{equation}
    dv_t =  dy_t - n \sqrt{\kappa}\, z_t\, dt.
\end{equation}

The nonlinear function $\gamma(r)$ has two important zeros that simplify the projection filter dramatically.  The first is that when $n=1$, $\gamma(r) = 0$ for every value of $r$.  Furthermore it is easy to compute that when evaluating the projection filter for $n=1$, the equations are identical to a set of conditional Bloch vector equations one obtains directly from the conditional master equation. In other words, the projected space is the whole manifold of solutions, $\man{P} = \man{S}(2)$.  The second zero occurs for $\gamma(r = 1) = 0$ for any $n$.  The $n$ and $r$ dependent terms in $a^3(\Vx,t)$ also evaluate to zero for $r = 1$, meaning that for any $n$ the projected pure state evolution is essentially identical to the evolution of a single qubit state.  The only remaining $n$ dependence is that the innovation process requires the expected measurement outcome to be scaled by a factor of $n$.

There are three elements that makes this projection filter tractable for obtaining an analytic expression.  The first is the isotropic nature of trace inner product metric, dramatically simplifying the calculation.  The second is that the Pauli matrices form a simple basis for $2 \times 2$ complex matrices and have equal eigenvalues.  The third is the identical and independent assumption for the joint operators.  This allows for terms that would in general result in ensemble averages to be given by identical single particle values.

The final element of Chap. \ref{chap:projection} is a series of numerical experiments testing the accuracy and performance of the projection filter against exact simulations for initial pure spin coherent states.  The rate $\kappa$ sets the measurement timescale and so all times in the simulations are compared to this rate, effectively setting it to 1.  Each simulation ran for a fixed time $t = 0.2\, \kappa^{-1}$.  As the quality of the projection filter should explicitly depend upon $n$, these simulations test the qubit numbers $n = \set{1, 25, 50, 75, 100}$.  The average performance data included a sample of $\nu = 100$ isotropically sampled initial qubit states with a single noise realization for each initial state.

In addition to testing the performance as a function of $n$, it also tests two different control functions $f^i(t)$.  The first is for $f^i(t) = 0$ for all $t$ and $i$.  This corresponds to a QND measurement of the $z$ projection of the total angular momentum formed by the qubit ensemble, $J_z$, and is known to produce spin squeezing.   We find that an initial state involving $50$ qubits prepared in $+J_x$ eigenstate produced $\approx 10$ dB of squeezing in one measurement duration.  A squeezed state is inherently not a product state, and so serves as a worst case scenario for the projection filter and acts as a lower bound on its performance.

The zero field measurement is compared to the case of a strong randomized control sequence.  Chap. \ref{chap:QubitState} uses the projection filter in an algorithm to reconstruct the initial condition of a SCS from a continuous measurement of $J_z$, characterized by the rate $\kappa$.  In order to obtain information about observables other than $J_z$, an external control Hamiltonian must be applied.  For reasons discussed in Sec. \ref{chQubit:sec:ControlLaw}, this takes form of a sequence of global $\pi/2$ rotations, where each rotation is about an axis $\V{n}$ independently sampled from a uniform distribution.  Fully characterizing the control amplitude $\V{f}(t)$ requires specifying the amplitude and duration of each pulse, as a larger Larmor frequency is needed to enact the same rotation in a shorter time.  For simplicity, we will fix $\V{f}(t)$ to have a constant magnitude and so for a pulse duration $\tau$ the control field is then given by,
\begin{equation}\label{chIntro:eq:Controls}
\V{f}(t) = \frac{\pi}{2\, \tau} \sum_{m =1} \indicate{[{m-1}, m)}\hspace{-4pt}\left(\,t/\tau\right)\, \V{n}_m
\end{equation}
where $\indicate{[a, b)}(t)$ is the indicator function for the interval $[a, b)$ and  $\set{ \V{n}_m} $ are \emph{i.i.d.} unit vectors drawn from a isotropic distribution.

The accuracy of the projection filter was tested by comparing how well it is capable of reproducing the conditional expectation values of the collective angular momentum components $J_i$ as well as the squared overlap between the exact state and the equivalent spin coherent state that is made from an ensemble of $n$ identical pure qubits.  The time-dependent results are presented in Figures \ref{chProj:fig:ProjAverageError} and \ref{chProj:fig:ProjAvgerageFidelity}.   The RMS error in the conditional expectation values were generally independent of the number of qubits, likely due to the fact that for pure states the projection filter dynamics are essentially independent of $n$.   With the randomized controls the RMS error was $\lesssim 5\%$ of the total spin length in all 3 expectation values.  In the absence of a control field, there was a general linear increase in the $J_x$ and $J_y$ errors also reaching the $5\%$ level, while there was a noticeable increase in the $J_z$ error with a final value in the $ \sim 5 - 10 \%$ range.  The poorer performance is attributable to the effect the spin squeezing has on the mean values.

The squared overlap between the exact state and the equivalent spin coherent state exhibits a strong dependence upon both the number of qubits and the control fields.  In the uncontrolled case this metric monotonic decreased for all $n > 1$ dropping to 0.75 for $n = 25$ and 0.48 for $n = 100$.  This is in stark contrast to the simulations including the randomized controls.  While the resulting average fidelity was noticeably poorer for large $n$, the minimum value was $> 0.8$ for all $n$.  While the trend was to have poorer fidelities at longer times, the decrease was not monotonic implying that the control wave form could be optimized to maximize the average overlap with the spin coherent state and thereby minimizing the information lost by performing the the projection.

We hypothesize that the state remains closer to a product state because the randomized controls tends to mix both the squeezed and anti-squeezed components leading to a near zero average.  Not only does the mean spin rotate, but the orientation of the squeezing ellipse also rotates.  As the rotation axes are chosen from a uniform distribution, the squeezed component is just as likely as the anti-squeezed component to be oriented along the measurement axis.  At any given time, the uncertainty in the $J_z$ component is equally likely to be above or below the uncertainty of an equivalent spin coherent state.  Therefore it is difficult for any significant squeezing to develop, and thereby keep the exact state closer to the product state description.

\subsection{Qubit state reconstruction}
Chap. \ref{chap:QubitState} describes how to use the quantum filtering formalism in order to construct a tomographic estimate for an unknown initial quantum state from an ensemble of identical copies experiencing a joint continuous measurement.  We make a maximum likelihood estimate (MLE) of the initial state, based upon the statistics of a \emph{single} continuous measurement realization.  The purpose of this work is to extend previous results \cite{silberfarb_quantum_2005, riofrio_quantum_2011, smith_efficient_2006}, which used a continuous measurement for quantum state tomography, into a regime where the quantum backaction significantly effect the measurement statistics.  In an idealized numerical study, we find that such an estimate can nearly saturate an optimum reconstruction bound.  Much is known about the fundamental quantum limits of reconstructing pure qubit states from a finite number of measurements. \citeauthor{massar_optimal_1995} showed that given $n$ copies of a pure qubit state, it is possible to find a generalized measurement that returns the highest average fidelity between the estimated state and the correct initial state \cite{massar_optimal_1995}.  The average is made not just over measurement outcomes but also over an unbiased set of possible input states.  The optimum average fidelity bound is simply given $\langle\mathcal{F}\rangle_\text{opt}  = (n+1)/(n+2).$

We consider here an idealized model of an ensemble of $n$ qubits identically coupled to a single traveling wave quantum light field via a linearized Faraday interaction.  Under certain approximations discussed in Sec. \ref{chQuLight:sec:faraday}, a measurement of the orthogonal quadrature contains information about the collective angular momentum variable $J_z$, with a coupling rate $\kappa$.  The ensemble is assumed to be prepared in a pure spin coherent state characterized by the unknown polar angles $(\theta,\, \phi)$.  However the initial qubit state is not a QND variable, meaning that $\rho(\theta,\phi)^{\otimes n}$ does not commute with the fundamental interaction Hamiltonian.  The implication of this is that it is impossible to find a consistent method for inverting the forward time dynamic to arrive at a conditional expression for the initial state.

To circumvent this problem we instead map the quantum state estimation problem to a parameter estimation problem to find a MLE of $(\theta, \phi)$.  A single continuous measurement realization of a noncommuting output quadrature results in a stochastic process $\set{y_t}_{t \ge 0}$ that contains information about the atomic ensemble.  Because of this information, its statistics parametrically depend upon the unknown angles.  While a MLE based upon a single data point would perform quite poorly, we find a conditional estimate based upon the entire \emph{trajectory} performs quite well, when the measurement is informationally complete.  To ensure informational completeness, a known time-varying control Hamiltonian is applied to the system, thereby mixing all spin projections with the measurement axis.  \citeauthor{riofrio_quantum_2011} found that an efficient and unbiased control policy is to choose a set of operations capable of generating any single particle state and then randomly varying the magnitude of each control in a piecewise constant way \cite{riofrio_quantum_2011}.  Here we use a similar control policy by including in the modal a uniform magnetic field with a constant field strength that rotates the collective angular momentum vector by $\pi/2$ in a period $\tau$ about randomly chosen rotation axes.  In a fixed final time we find $40$ rotations provides enough information to obtain high fidelity estimates.

In the semiclassical probability space induced by a measurement realization, the appropriate probability measure $\mathbbm{P}$ has a parametric dependence on the initial angles $(\theta, \phi)$.  Identifying this dependence is best seen by considering not the conditional statistics of $\set{y_t}_{t \ge 0}$ but instead the calculated innovation process $\set{v_t}_{t \ge 0}$.  For the measurement model considered here this process is given by
\begin{equation}
  v_t = y_t - 2 \sqrt{\kappa} \int_0^t ds\, \Tr(J_z\, \rho_s(\theta, \phi)\, ),
\end{equation}
where $\rho_s(\theta, \phi)$ is the system density operator calculated via the conditional master equation assuming that initial condition is given by the angles $(\theta, \phi)$.  The innovation process $v_t$ is shown in Sec. \ref{chMath:sec:Innovation} to have the statistics of a Wiener process if $\Tr(J_z\, \rho_s(\theta, \phi)\, )$ corresponds to the exact quantum conditional expectation of the Heisenberg picture operator $U^\dag_t J_z U_t$.   If this correspondence cannot be made because $\rho_s(\theta, \phi)$ used an incorrect initial condition, then $v_t$ will not have statistics of a Wiener process for every measurement realization.  Here we use this fact to find the MLE for $(\theta, \phi)$.  We seek the initial condition that makes the innovation process most likely to be a Wiener process.   As the conditional master equation is a nonlinear equation, we resort to a mixture of numerical and analytical methods for finding an approximation to the true likelihood function.

The Wiener measure gives the probability for a Wiener process sampled at times $\set{t_i : i = 1, \dots, n}$ will be within associated intervals $I_i = (a_i, b_i)$ and is given by the integral
\begin{equation}
  P( \set{v_{t_i} \in I_i} ) = \int_{a_1}^{b_1} dv_1 \dots \int_{a_n}^{b_n} dv_n \prod_{i = 1}^n \left( \frac{1}{\sqrt{2 \pi \Delta t_i}} \exp \left( -\frac{(v_i - v_{i -1})^2}{2 \Delta t_i} \right) \right).
\end{equation}
Because this is a Gaussian probability density, the MLE coincides with the least squares estimate.  For an equally spaced mesh of finite time intervals both the normalization factor and the denominator of the exponent are irrelevant for the purposes of computing a MLE.  Therefore maximizing the likelihood function is equivalent to minimizing the quadratic variation,
\begin{equation}\label{chIntro:eq:QVcost}
     \mathrm{QV}( v_t ) \define \sum_{i=1}^n\, \left( \Delta y_i - 2 \sqrt{\kappa} \Delta t \Tr(J_z\, \rho_{t_{i-1}}(\theta, \phi)\, ) \, \right)^2
\end{equation}
The minimization of this function with respect $(\theta, \phi)$ is computed numerically as we are unable to find an analytic solution to the conditional master equation and the dependence upon the initial condition is nonlinear.

Every evaluation of this function requires a full numerical integration of the conditional master equation.  Numerically integrating an exact conditional pure state is a computationally intensive task.  Here we test spin ensemble involving $25 \le n \le 100$ qubits, which require order $n$ complex number to fully describe the relevant conditional dynamics.   While it is computationally feasible to integrate the exact equation to generate a simulated measurement record, for every measurement record we used a total of $500$ evaluations of the quadratic variation cost function.  We found it infeasible to use an exact expression for computing $\mathrm{QV}( v_t )$ and instead sought a reasonably accurate approximation.  The approximation we use is the projection filter developed in Chap. \ref{chap:projection}.  Under identical conditions to the dynamics used here, the projection filter is able to match the expectation value $\Tr(J_z\, \rho_{t})$ to within 95\% accuracy of the exact value while only propagating 3 real numbers for any number of qubits.  By minimizing with respect to the approximate filter, the limiting computational element became generating the simulated measurement record.  While in a higher dimensional space one could approach the problem with a gradient descent algorithm, we find it more efficient to simply make a dense Monte Carlo sampling of the entire Bloch sphere\footnote{Due an issue involving numerical stability we start with a uniform sampling of mixed states and then make a subsequently smaller sample of pure states.  See Sec. \ref{chQubit:sec:Stability}.} and then select the most likely sample.

In order to understand what role backaction plays in limiting the reconstruction fidelity, we compare the performance of the projection filter estimate to one that ignores completely the conditional dynamics.  Instead of propagating a conditional state, this estimate only considers the Hamiltonian dynamics generated by the magnetic field rotations.  In other words by solving the Heisenberg equation of motion,
\begin{equation}
  \frac{d}{dt} \sigma_z(t) = +i [\half f^j(t) \sigma_j(t), \sigma_z(t)]
\end{equation}
for the controls $f^i(t)$ given in Eq. (\ref{chIntro:eq:Controls}), the expectation value $ \sqrt{\kappa}\, n \Tr(\sigma_z(t) \rho(\theta, \phi) )$ reproduces the expected signal, ignoring the backaction.  The purpose for computing this in the Heisenberg picture is so that we are able to solve for the dynamical observables once, and then apply that solution for any initial condition.

The results of the numerical experiments are given in Fig. \ref{chQubit:fig:FidelPlot}.  The estimate based upon the projection filter nearly achieves the optimum $(n+1)/(n+2)$ fidelity bound, averaged over $\nu = 1000$ trials for $n = (25, 40,  55, 70, 85, 100)$ qubits.  The difference between the optimum bound and the numerical averages never exceeded 0.21\%, a deviation that is likely statistically significant but not attributable to any fundamental Monte Carlo sampling errors.  In comparison the backaction-free estimator performed significantly poorer especially for higher qubit numbers.  This suggest that including the conditional dynamics is indeed important in this idealized scenario.

The cause of the discrepancy between the projection filter estimate and the backaction-free estimate is likely due to a bias that develops when all measurement effects are ignored.  This can be see in Figures \ref{chQubit:fig:FilterBias} and \ref{chQubit:fig:SchBias}.  These figures plot the average reconstruction fidelity as a function of the measurement duration.  The filtering based estimate shows a monotonic rise in the average fidelity which then saturates at a level only slightly below the optimum bound.  As $n$ increases this saturation occurs at earlier times.  In contrast, the estimate based only on Hamiltonian evolution does not have monotonic increase in reconstruction fidelity.  For $n = 55, 70, 85, 100$ the average fidelity reaches a maximum and then has significant \emph{decrease} as more data are collected.  For $n = 25, 40$ it is possible that a decrease might also have occurred if the simulation continued for longer times.

When backaction is ignored, the assumption of pure unitary evolution implies that no coherence is lost during the course of the measurement.  If at time $t$ the randomized controls managed to return to the original orientation then the backaction-free estimator would ``weight'' the data received at that time just as much as the data obtained at time $t = 0$.  In comparison, the filtering based estimate knows that while the rotations may have canceled, the expected signal at time $t$ is not what it was at time $t = 0$, precisely because of the conditional effects.  By not including this information the unitary estimate is biased away from the optimum estimate.

\chapter{ Quantum Optics and Quantum Stochastic Differential Equations\label{chap:QuantumLight} }

The objective of this chapter is to identify how the formalism of quantum stochastic differential equations is implemented in the context of quantum optics.  This is done by first showing how the second quantization of classical quasi-monochromatic traveling wave packets gives the natural structure necessary for defining the quantum It\={o} integral.  We then show under what conditions a wave packet operator can be treated as generating a localized field operator, which is necessary for a Markov approximation.  From that localized structure, we then review how this defines a quantum Wiener process and relates to the quantum white noise formalism usually presented in quantum optics.  With a wave packet description of quantum white noise, we then review how a system coupling to these operators generates a quantum stochastic differential equation for the propagator.  Generating this equation is intimately related to the operator ordering of the field operators, which is also related to defining different kinds of stochastic equations.  Here we review this fact and how the propagator is derived.  Finally, we apply this result to the Faraday Hamiltonian and discuss the interaction in the limits of both strong and weak number coupling.

\section{Quantum Stochastic Process in Optical Fields }

The classical stochastic process is most generally defined as a family of random variables $\set{x_t}_{t\ge 0}$ indexed by time $t$.  A \emph{quantum stochastic process} is then a family of operators $\set{X_t}_{t\ge 0}$ also indexed by time.  This definition allows for a slightly more general structure than simply an operator is dependent upon time.  A common example of a time-dependent quantum operator is a Heisenberg picture operator, $X(t)$, acting on some Hilbert space $\Hilbert$ with its dynamics given by a unitary transformation.  Conversely, a quantum stochastic process implies something more general, where the spectrum of $X_t$ could be time-dependent and even the Hilbert space upon which it acts nontrivially could be continuously changing in time.

A concrete and pertinent example is a continuous wave laser that is switched on at time $t_0 = 0$ and then switched off at some later time $t > 0$.  Consider a stationary observer located a distance $d = c\, \tau$ away from the laser who begins counting photons with a perfectly efficient detector at time $t_0$.  By a time $0 < s < \tau$ it is clear that this observer will have not observed any of the laser light and so, in absence of any corrupting background, the probability for observing anything must be zero.  The point of this example is that in the time interval $[0, s]$ any observations must be modeled as projectors acting on a part of Fock space that is independent of the part displaced by the laser.

It is perfectly reasonable to use a Schr\"{o}dinger picture description where the observer makes projective measurements on a volume of the electromagnetic field and the free-field Hamiltonian acts in such a way as to propagate the state of the field fixed from the laser position to the detector.  In Sec. \ref{chMath:sec:QSDEs}, a mapping between a set of commuting observables and a classical probability space is developed, where there is a one-to-one correspondence between classical random variables and commuting operators.  To utilize these tools of classical probability theory, it is most natural to work in a Heisenberg picture where the states remain fixed and the unitary evolution is applied to the observables of interest.  To develop a Heisenberg picture formulation for a continuous measurement, we require a mathematical structure that can cope with the fact that as time progresses a stationary observer will measure an ever increasing set of operators, and these operators act upon different parts of the field's Hilbert space.  This family of operators is the quintessential definition of a quantum stochastic process.

In classical stochastic calculus, the Wiener process is the fundamental random process from which the It\={o} integral is constructed and from there other processes are defined.  In the quantum setting, we require equally fundamental operator-valued noise processes from which we will construct other processes.  But as we are seeking a description of a continuous optical measurement, those processes should arise from the quantized electromagnetic field.  The next section reviews the canonical quantization of the free electromagnetic field to identify the operator nature of the quantized field.

Throughout this chapter we will be discussing both classical and quantized elements of the electromagnetic field.  In order to make this distinction, the classical vector fields for the vector potential, electric field, \emph{etc.} will be denoted as $\V{\mathcal{A}}(\Vx, t)$, $\V{\mathcal{E}}(\Vx, t)$ and their quantized operator expressions as $\V{A}(\Vx, t)$, $\V{E}(\Vx, t)$.  We will quantize the free space electromagnetic field following the classic text by \citeauthor{cohen-tannoudji_photons_1989}, and use SI units \cite{cohen-tannoudji_photons_1989}.  For reference, the spatial Fourier transform of a function $f(\Vx, t)$ is defined as
\begin{equation}
  \F{f}(\Vk, t) \define \int_{\R^3} \frac{d^3x}{\sqrt{(2 \pi)^3}}\ e^{-i \Vk \cdot \Vx}\, f(\Vx,t)
\end{equation}
and the inverse transform is
\begin{equation}
f(\Vx, t) \define \int_{\R^3} \frac{d^3k}{\sqrt{(2 \pi)^3}}\ e^{+ i \Vk \cdot \Vx}\, \F{f}(\Vk,t).
\end{equation}

\subsection{Free space quantization \label{chQuLight:sec:freeSpaceQuantization} }

When rigorously quantizing the free space electromagnetic field, one begins by defining a scalar Lagrangian functional, with respect to variations in the vector potential, $\boldsymbol{\mathcal{A}}(\Vx, t)$, whose minimization reproduces Maxwell's equations \cite{cohen-tannoudji_photons_1989}.  The field Hamiltonian is
\begin{equation}\label{chQuLight:eq:classicHharmonicOscillator}
    \mathcal{H}_f =  \int_{\R^3} d^3x\left( \frac{\abs{\boldsymbol{\Pi}}^2}{2 \varepsilon_0} + \frac{1}{2}\varepsilon_0 c^2 \abs{\grad \times \boldsymbol{\mathcal{A}}}^2\right)
\end{equation}
with the conjugate variable to the vector potential $\boldsymbol{\Pi}$ being
\begin{equation}
    \boldsymbol{\Pi} =  \varepsilon_0 \frac{\partial}{\partial t} \boldsymbol{\mathcal{A}} = - \varepsilon_0 \boldsymbol{\mathcal{E}}.
\end{equation}
(The final equality with the electric field is made by assuming that there are no free charge.)
From this classical Hamiltonian, the connection with quantum mechanics is made by noting this is the Hamiltonian of a continuous set of harmonic oscillators with canonical variables $\set{\boldsymbol{\mathcal{A}}, \boldsymbol{\Pi}}$.  Quantization then promotes these variables to canonically commuting operators.

By choosing to work in the Coulomb gauge, it is easily shown that if  $\widetilde{ \boldsymbol{\mathcal{A}} }(\Vk,t)$ is the Fourier transform of $\boldsymbol{\mathcal{A}}$ then
\begin{equation}
    \Vk \cdot \widetilde{\boldsymbol{\mathcal{A}}} = 0.
\end{equation}
This constraint results in two free polarization components (labeled by $s \in \set{1,2}$) for each Fourier component, defining the vectors $\Ve_q(\Vk)$ which satisfy the properties
\begin{subequations}
  \begin{align}
    \Vk \cdot \Ve_q(\Vk) &= 0,\\
    \Ve^*_q(\Vk) \cdot \Ve_{q'}(\Vk) &= \delta_{q q'}, \quad \text{and}\\
    \sum_q e_{qi}^*(\Vk)\, e_{qj}(\Vk)  &= \delta_{ij} - k_i k_j / \abs{\Vk}^2\label{chQuLight:eq:transVectorComplete}
  \end{align}
\end{subequations}
where $i,j$ refer to the Cartesian components.
In terms of the real space operators the components of the quantized fields $\VA(\Vx)$ and $\V{E}(\Vx)$ satisfy the commutation relations,
\begin{align}\label{chQuLight:eq:fieldCommutation}
    [A_i(\Vx), A_j(\Vx')] = [E_i(\Vx), E_j(\Vx')] = 0\\
    [A_i(\Vx), - \varepsilon_0 E_j (\Vx')] = i \hbar \,\delta^{T}_{ij}(\Vx - \Vx')
\end{align}
where $\delta^{T}_{ij}(\Vx - \Vx')$ is the transverse delta function, defined as
\begin{equation}
    \delta^{T}_{ij}(\Vx - \Vx') \define\int\, \frac{d^3k}{(2 \pi)^3}\, e^{+ i \Vk\cdot(\Vx - \Vx')} \left(\delta_{ij} - \frac{k_i k_j}{\abs{\Vk}^2} \right).
\end{equation}

In reciprocal space, the vector potential is most suitably expressed in terms of the annihilation ($\ahat_q(\Vk)$) and creation ($\ahat^\dag_q(\Vk)$) operators associated with the polarization vectors $\Ve_q(\Vk)$.  They obey the commutation relations,
\begin{align}\label{chQuLight:eq:canonicalCommutationAnnihilation}
    [\ahat_q(\Vk), \ahat_{q'}(\Vk')] &= [\ahat^\dag_q(\Vk), \ahat^\dag_{q'}(\Vk')] = 0 \quad \text{and}\\
    [\ahat_q(\Vk), \ahat^\dag_{q'}(\Vk')] &= \delta_{q q'}\,\delta (\Vk - \Vk').
\end{align}
The Schr\"{o}dinger pictures operators for the vector potential and electric field are then given by
\begin{align}
    \VA(\Vx) &= \sum_q  \int  \frac{d^3k}{(2 \pi)^{3/2}}\sqrt{\frac{ \hbar }{2 \varepsilon_0 c \abs{\Vk}}} e^{i \Vk\cdot\Vx} \, \Ve^*_q(\Vk)\, \ahat_q(\Vk)  + h.c., \label{chQuLight:eq:quantumA}\\
    \V{E}(\Vx) &= i \sum_q  \int  \frac{d^3k}{(2 \pi)^{3/2}} \sqrt{\frac{\hbar c \abs{\Vk} }{2 \varepsilon_0 }} e^{i \Vk\cdot\Vx} \Ve^*_q(\Vk)\, \ahat_q(\Vk)  + h.c.,\quad\text{and} \label{chQuLight:eq:quantumE}\\
    \V{B}(\Vx) &= i \sum_q  \int  \frac{d^3k}{(2 \pi)^{3/2}}\sqrt{\frac{ \hbar }{2 \varepsilon_0 c \abs{\Vk}}} e^{i \Vk\cdot\Vx} \, \Vk \times\, \Ve^*_q(\Vk)\, \ahat_q(\Vk)  + h.c.\label{chQuLight:eq:quantumB}
\end{align}
Substituting these expressions into the Hamiltonian results in the simplified form
\begin{equation}\label{chQuLight:eq:quantumHsymmetric}
    H_f = \half \hbar\, c \sum_q \int d^3k\ \abs{\Vk} \left(\,\ahat_q(\Vk) \ahat^\dag_q(\Vk) + \ahat^\dag_q(\Vk) \ahat_q(\Vk)\,\right).
\end{equation}
We will follow the standard practice of discarding any vacuum energy contributions and write
\begin{equation}\label{chQuLight:eq:quantumH}\,
    H_f = \hbar \, c \sum_q \int d^3k\ \abs{\Vk}\, \ahat^\dag_q(\Vk) \ahat_q(\Vk).
\end{equation}

\section{Wave Packets, Fock Space and Stochastic\\ Processes }
For a single simple harmonic oscillator, the state space is spanned by a complete basis of states labeling the number of quanta in the oscillator.  In the free-space EM field, we have instead a continuous distribution of oscillators, each representing a plane wave Fourier component with one of two orthogonal polarization states.  This continuous nature means that the field operators are unbounded in two ways: for a given $\Vk$ we can have a countably infinite number of quanta and a single quantum can have an unbounded amount of energy if we allow pure plane wave states with arbitral large wave numbers $\abs{\Vk}$.  The solution to these problems is to consider only states of light for which these operators return finite quantities.

Notice that in Eqs. (\ref{chQuLight:eq:quantumA} - \ref{chQuLight:eq:quantumH}) the plane wave operators $\ahat_q(\Vk)$ and $\ahat^\dag_q(\Vk)$ act as operator-valued integral kernels  where they are combined with various weighting functions to form the physically relevant operators.  The fact that they have the singular commutation relation $[ \ahat_q(\Vk),\,\ahat^\dag_{q'}(\Vk')] = \delta_{q q'}\,\delta(\Vk - \Vk')$ implies that they are only well defined in the context of an integral, where the Dirac delta function is well behaved.  The point is that the domain of the operators $\VE$, $\VB$ and $H_f$ that return finite eigenvalues should not be considered as a set of distinct plane wave oscillators, but instead in terms of continuous functions defined over ranges of Fourier components.  We will refer to these distributions as \emph{wave packets}, in that by constructing a properly weighted distribution over plane waves one arrives with a localized pulse or packet of waves that propagates in some direction.  Rather than initially discussing wave packets in terms of single quanta, it is easier to first define wave packet states in terms of semiclassical states of light that generalize the coherent state of a single mode harmonic oscillator.

\subsection{Wave packets \label{chQuLight:sec:wavePackets} }

A semiclassical wave packet identifies those states of light that reproduce coherent classical radiation when one takes expectation values of the quantized operators $\VA(\Vx)$, $\VE(\Vx)$, \emph{etc.}  This relationship has been identified by many authors, \emph{e.g.} \citet{deutsch_basis-independent_1991, garrison_quantum_2008, smith_photon_2007}.  We review these results here, focusing on the physical interpretation for the wave packet distributions.

The single-mode coherent state, $\psi = \ket{\alpha}$, is characterized by the complex amplitude $\alpha$, where the mean photon number is given by $\abs{\alpha}^2$ and is an eigenstate of the annihilation operator $\ahat\, \ket{\alpha} = \alpha\, \ket{\alpha}$.  We have seen that in the canonical quantization of the free field, each plane wave and transverse polarization vector has its own annihilation operator $\ahat_q(\Vk)$ and so for a corresponding coherent state of light requires a complex vector valued function $\Vg(\Vk)$. The coherent state $\psi[\Vg]$ satisfies the equation
\begin{equation}
    \ahat_q(\Vk)\, \psi[\Vg] = g_q(\Vk)\, \psi[\Vg].
\end{equation}
By hypothesizing the existence of the states $\psi[\Vg]$ we would like to see how the coherent amplitude function $\Vg(\Vk)$ relates to physical quantities in expectation.  By taking the expectation value of the (vacuum energy removed) Hamiltonian Eq. (\ref{chQuLight:eq:quantumH}) we can easily see that
\begin{equation}\label{chQuLight:eq:wavePacketExpectH}
    \expect{H_f}_{\psi[\Vg]} = \hbar c\, \sum_q \int d^3 k\, \abs{\Vk}\, \abs{g_q(\Vk)}^2.
\end{equation}
An equally trivial calculation results in
\begin{equation}
    \expect{\VA(\Vx)}_{\psi[\Vg]} = \sum_q  \int \frac{d^3k}{(2 \pi)^{3/2}}\sqrt{\frac{\hbar }{2 \varepsilon_0 c \abs{\Vk}}} e^{i \Vk\cdot\Vx} \,  g_q(\Vk)\, \Ve^*_q(\Vk) + c.c. \label{chQuLight:eq:classicalA}
\end{equation}
and
\begin{equation}
    \expect{\VE(\Vx)}_{\psi[\Vg]} = i \sum_q  \int  \frac{d^3k}{(2 \pi)^{3/2}} \sqrt{\frac{\hbar c \abs{\Vk} }{2 \varepsilon_0 }} e^{i \Vk\cdot\Vx}\, g_q(\Vk)\, \Ve^*_q(\Vk)  + c.c. \label{chQuLight:eq:classicalE}
\end{equation}
It is possible to invert these two equations and so express $\Vg(\Vk)$ in terms of the spatial Fourier transform of a classical vector potential, $\VF{\mathcal{A}}(\Vk, t)$.
Performing this inversion we find that,
\begin{equation}\label{chQuLight:eq:normalVariables}
    \Vg(\Vk) = \sqrt{\frac{\varepsilon_0 }{2 \hbar c \abs{\Vk} }} \left(c \abs{\Vk}\,\VF{\mathcal{A}}(\Vk, 0) + i \left.\frac{\partial}{\partial t} \VF{\mathcal{A}}(\Vk, t)\right|_{t = 0} \right).
\end{equation}
When quantizing the field,  Eq. (\ref{chQuLight:eq:normalVariables}) and its adjoint are nothing more than the ``normal variables'' that are in classical correspondence to $\ahat_q(\Vk)$ and $\ahat_q^\dag(\Vk)$ \cite{cohen-tannoudji_photons_1989}.

It is common in optics to relate the physical classical fields $\boldsymbol{\mathcal{A}}$, $\boldsymbol{\mathcal{E}}$, and $\boldsymbol{\mathcal{B}}$ to a unitless mode function.  In terms of the vector potential this results in the \emph{ansatz},
\begin{equation}\label{chQuLight:eq:classicalAMode}
    \boldsymbol{\mathcal{A}}(\Vx, t) = \mathcal{A}_0 \left(\, \V{u}^\pf(\Vx, t) +  \V{u}^{\nf}(\Vx, t)\, \right)
\end{equation}
where $\mathcal{A}_0$ is a real constant and $\V{u}^\pf(\Vx, t)$ is a complex unit-less mode function.  The fact that the vector potential is required to be real, we have the relation that
\begin{equation}\label{chquLight:eq:frequencyComponents}
    \V{u}^{\nf}(\Vx, t) =  \V{u}^{\pf\, *}(\Vx, t).
\end{equation}
As $\V{u}^\pf(\Vx, t)$ is unitless, its integral
\begin{equation}
    \mathrm{v} = \int d^3x\ |\V{u}^\pf(\Vx, t)|^2
\end{equation}
has units of volume and is referred to as the mode volume of the field.
By taking the spatial Fourier transform of Eq. (\ref{chQuLight:eq:classicalAMode}) we have
\begin{equation} \label{chQuLight:eq:classicalAFourier}
    \widetilde{\boldsymbol{\mathcal{A}}}(\Vk, t) =  \mathcal{A}_0 \left(\VF{u}^\pf(\Vk, t) + \VF{u}^{\nf}(\Vk, t) \right)
\end{equation}
The purpose of separating between $\V{u}^\pf(\Vx, t)$ and $\V{u}^\nf(\Vx, t)$ is to allow for the separation between positive and negative frequency components, respectively.  For a free field then,
\begin{equation} \label{chQuLight:eq:modePfFourier}
    \V{u}^\pf(\Vk, t) = \V{u}^\pf(\Vk, 0)\, e^{- i c \abs{\Vk}\, t}.
\end{equation}
The Fourier space version of Eq. \ref{chQuLight:sec:freeSpaceQuantization} is
\begin{equation}
   \V{u}^\nf(\Vk, t) = \VF{u}^{\pf\, *}(- \Vk, t).
\end{equation}
To simplify the expression for $\Vg(\Vk)$ as given in Eq. (\ref{chQuLight:eq:normalVariables}), we need to compute the time derivative of $\VF{\mathcal{A}}(\Vk, t)$. Substituting Eq. (\ref{chQuLight:eq:modePfFourier}) into Eq. (\ref{chQuLight:eq:classicalAFourier}) and computing the derivative we have
\begin{equation} \label{chQuLight:eq:classicalAFourierDerivative}
    \frac{\partial}{\partial t} \VF{\mathcal{A}} (\Vk, t) = - i c \abs{\Vk} \mathcal{A}_0 \left( \VF{u}^\pf(\Vk, t) - \VF{u}^{\pf\, *}(-\Vk, t) \right).
\end{equation}
Substituting this expression into Eq. (\ref{chQuLight:eq:normalVariables}) leads to
\begin{equation}\label{chQuLight:eq:wavePacketSimple1}
    \Vg(\Vk, t) =  \mathcal{A}_0\, \sqrt{\frac{2 \varepsilon_0 c \abs{\Vk} }{ \hbar }}\, \VF{u}^\pf(\Vk, t).
\end{equation}
Rather than including the vector potential constant $\mathcal{A}_0$, which usually contains information about the overall intensity of the field, it is useful to relate it back to the magnitude of $\Vg$.  We first define the characteristic wave number $k_1$ as
\begin{equation}
    k_1 \define  \int d^3 k\ \abs{\Vk}\, \frac{\abs{\VF{u}^\pf(\Vk, 0)}^2}{\mathrm{v}}.
\end{equation}
By considering $\mathrm{v}^{-1}\,\abs{\VF{u}^\pf(\Vk, 0)}^2$ to be a normalized distribution in reciprocal space, then $k_1$ is the average magnitude.  With this definition
\begin{equation}
    \norm{\Vg}^2 = \mathcal{A}_0^2\, \frac{2\, \varepsilon_0\, c\, k_1\, \mathrm{v}}{\hbar}.
\end{equation}
Inverting this relationship results in
\begin{equation}\label{chQuLight:eq:wavePacketSimple}
    \Vg(\Vk, t) = \norm{\Vg}\, \sqrt{\frac{ \abs{\Vk} }{k_1} }\, \frac{\VF{u}^\pf(\Vk, 0)}{\sqrt{\mathrm{v}}} \, e^{-i c \abs{\Vk} t}.
\end{equation}

It is worth noting that the units of Eq. (\ref{chQuLight:eq:wavePacketSimple}) is of root volume and that $\norm{\Vg}$ now acts as a unitless scaling factor.  This final formula shows the fundamental relationship between a distribution over coherent state amplitudes $\Vg(\Vk)$ and the positive frequency Fourier components of the mode function $\VF{u}^\pf(\Vk, 0)$.  While in one sense this has simply been an algebraic exercise (expressing one distribution over spatial frequencies in terms of another) the real utility of this expression is that the mode function $\V{u}^\pf(\Vx, t)$ has  practical implications as it describes the spatial and temporal properties of a propagating laser beam.

Finally, we express the expected energy in a wave packet state in terms of the envelope function. Simply substituting Eq. (\ref{chQuLight:eq:wavePacketSimple}) into Eq. (\ref{chQuLight:eq:wavePacketExpectH}) results in,
\begin{equation}
    \expect{H_f}_{\psi[\Vg]} = \hbar c\, \norm{\Vg}^2  \int d^3 k\ \frac{\abs{\Vk}^2 }{k_1}\, \frac{\abs{ \VF{u}^\pf(\Vk, 0)}^2}{\mathrm{v}}.
\end{equation}
Similarly to defining the mean wave vector $k_1$ we can define a two-norm wave vector $k_2$,
\begin{equation}
  k_{2} = \left(\int d^3 k\ \abs{\Vk}^2 \, \frac{\abs{ \VF{u}^\pf(\Vk, 0)}^2}{\mathrm{v}} \right)^{\half}
\end{equation}
so that
\begin{equation}
    \expect{H_f}_{\psi[\Vg]} = \hbar c\, \norm{\Vg}^2 \frac{(k_2)^2}{k_1}.
\end{equation}
If, however, $\abs{ \VF{u}^\pf(\Vk, 0)}^2$ is a sharply peaked function centered at some large vector $\Vk_0$, then we have that $\abs{\Vk_0} \approx k_1 \approx k_2$.  In this case the average energy is then
\begin{equation}
\expect{H_f}_{\psi[\Vg]} \approx \hbar \omega_0\, \norm{\Vg}^2
\end{equation}
where $\omega_0 = c \abs{\Vk_0}$.

\subsection{Weyl operators \label{chQuLight:sec:WeylOperators} }

Assuming the existence of the semiclassical states is only a first step, but real utility comes from finding the family of operator that generate these states.  In the context of the simple harmonic oscillators, the coherent state with amplitude $\alpha$ is generated by the unitary displacement operator
\begin{equation}\label{chQuLight:eq:displacmentSHO}
  D_{\text{sho}}(\alpha) = \exp\left(\alpha\, \ahat^\dag - \alpha^\ast\, \ahat\right) \quad\text{with }\quad
  \ket{\alpha} = D_{\text{sho}}(\alpha)\,\ket{0}.
\end{equation}
Writing (\ref{chQuLight:eq:displacmentSHO}) in terms of its generator $\Upsilon(\alpha)$
\begin{equation}
\begin{split}
    D_{\text{sho}}(\alpha) =&\, \exp\left(- i \Upsilon(\alpha)\right)
\end{split}
\end{equation}
we find that
\begin{equation}
    \Upsilon(\alpha) = i \left(\alpha\, \ahat^\dag - \alpha^\ast\, \ahat \right).
\end{equation}
Note that as $g_q(\Vk)$ is a distribution of coherent amplitudes over all plane wave modes, we make the correspondence
\begin{equation}
    \alpha^\ast \ahat \ \rightarrow \  g^\ast_q(\Vk)\, \ahat_q(\Vk).
\end{equation}
But as this is a pointwise weighting over each plane wave, we define the total field operators $\ahat[\Vg]$ and $\ahat^\dag[\Vg]$ to be
\begin{equation}
    \ahat[\Vg] \define \sum_q \int d^3 k\ g^\ast_q(\Vk)\, \ahat_q(\Vk)
\end{equation}
and
\begin{equation}
    \ahat^\dag[\Vg] \define \sum_q \int d^3 k\ g_q(\Vk)\, \ahat^\dag_q(\Vk).
\end{equation}
This are sometimes called smeared creation and annihilation operators as they have been spread over a range of $\Vk$ values.  By applying the commutation relations (\ref{chQuLight:eq:canonicalCommutationAnnihilation}), it is easy to see that
\begin{equation}\label{chQuLight:eq:HeisenbergCommutation}
    \left[ \ahat[\Vf],\, \ahat^\dag[\Vg]\right] =  \int d^3 k\ \Vf^*(\Vk)\cdot \Vg(\Vk).
\end{equation}
An important property that we will use is that by the linearity of the integral over $d^3k$ we have that, for complex coefficients $c_1$ and $c_2$
\begin{equation}
  c_1\, \ahat^\dag[\Vf] + c_2\, \ahat^\dag[\Vg]  = \ahat^\dag[c_1 \Vf + c_2 \Vg]
\end{equation}
and
\begin{equation}
  c_1\, \ahat[\Vf] + c_2\, \ahat[\Vg]  = \ahat[c^\ast_1 \Vf + c^\ast_2 \Vg].
\end{equation}
In other words $\ahat^\dag[\cdot]$ is linear in its argument and $\ahat[\cdot]$ is \emph{anti-linear}.  The continuous analog of the displacement operator, called a Weyl operator, is
\begin{equation}
  \weyl[\Vg] \define \exp\left(\ahat^\dag[\Vg] - \ahat[\Vg]\right)
\end{equation}
and the coherent state $\psi[\Vg]$ is given by
\begin{equation}
  \psi[\Vg] = \weyl[\Vg]\, \ket{\vac}.
\end{equation}
Applying the Zassenhaus formula to the Weyl operator shows that
\begin{equation}
  \weyl[\Vg] = \exp(\ahat^\dag[\Vg])\, \exp(\ahat[\Vg])\, \exp\left(-\half \int d^3k\ \abs{\Vg(\Vk) }^2 \right).
\end{equation}
Note that because $\ahat[\Vg]\, \ket{\vac} = 0$ for any $\Vg$, this implies that
\begin{equation}
  \psi[\Vg] = \exp\left(-\half \int d^3k\ \abs{\Vg(\Vk)}^2 \right)\, \exp(\ahat^\dag[\Vg]) \ket{\vac}.
\end{equation}
When proving limits involving sequences of coherent states, it is often more convenient to work with unnormalized state vectors.  Therefore, it is common to define an exponential vector
\begin{equation}
  \e[\Vg] \define \exp(\ahat^\dag[\Vg])\, \ket{\vac}.
\end{equation}
One particularly useful relationship that we will end up applying repeatedly is that
\begin{equation}
  \ahat[\Vf]\,\psi[\Vg] = \int d^3k\ \Vf^*(\Vk)\cdot\Vg(\Vk)\ \psi[\Vg],
\end{equation}
\emph{i.e.}, $\psi[\Vg]$ is an eigenstate of \emph{any} smeared annihilation operator $\ahat[\Vf]$, regardless of the smearing function $\Vf$.  If, however, the functions $\Vf$ and $\Vg$ are orthogonal, then that eigenvalue could very well be zero.

\subsection{Fock space\label{chQuLight:sec:FockSpace}}
A number of useful relations can be derived involving the Weyl displacement operators and the exponential vectors.  Before doing so it is necessary to introduce some of the formal and algebraic properties of second quantization and Fock spaces.  Note that if we take $\Vf$ and $\Vg$ to be any square integrable complex functions, then the right hand side of (\ref{chQuLight:eq:HeisenbergCommutation}) forms an inner product on a Hilbert space of wave packets.  We will denote the inner product as,
\begin{equation}
    \inprod{\Vg}{\Vf} \define \sum_q \int d^3 k\  g^\ast_q(\Vk) f_q(\Vk)
\end{equation}
and the Hilbert space of wave packets as
\begin{equation}
    \hilbert \define \set{ \Vg(\Vk)\ :\ \inprod{\Vg}{\Vg} < \infty}.
\end{equation}

A Fock space $\Fock$ is a total Hilbert space describing a unknown and possibly unbounded number of particles that are each represented by states in a single particle Hilbert space $\hilbert$.  If we are given a single particle from $\hilbert$, then the full Hilbert space of two such particles is the tensor product of two such Hilbert spaces.  Likewise for three particles, there would be three fold product.  We will notate the joint space of n particle as $\hilbert^{\otimes n} = \hilbert \otimes \hilbert\otimes \dots \otimes \hilbert$ where there are $n$ such products.  In terms of a total space with an indeterminant number of particles, each subspace that contains $n$ particle will be mutually orthogonal.  Thus the total space is the direct sum over each subspace.  If we take the space for zero particles to be the complex numbers, ($\hilbert^{\otimes 0} = \mathbbm{C}$), then the full Fock space is given by
\begin{equation}
    \Fock_{\text{full}}(\hilbert) = \bigoplus_{n = 0}^{\infty}\, \hilbert^{\otimes n}.
\end{equation}
The reason for the notation $\Fock_{\text{full}}(\hilbert)$ is that if $\hilbert$ is a Hilbert space of bosonic particles than only states that are symmetric under particle exchange will apply.  We denote the symmetric subspace of $\hilbert^{\otimes n}$ to be $\hilbert^{\otimes_s\, n}$.  So the symmetric Fock space is given by
\begin{equation}
    \Fock_{\text{sym}}(\hilbert) = \bigoplus_{n = 0}^{\infty}\, \hilbert^{\otimes_s n}.
\end{equation}
We are strictly interested in bosonic particles, so throughout this document we when refer to $\Fock(\cdot)$ we are referring to the symmetric Fock space.

For a single simple harmonic oscillator, the coherent state $\ket{\alpha}$ is expressed in terms of the number states $\ket{n}$ as
\begin{equation}
  \ket{\alpha} = \sum_{n = 0}^\infty \frac{\alpha^n}{\sqrt{n!}}\, e^{-\abs{\alpha}^2/2}\, \ket{n}.
\end{equation}
The equivalent expression for the wave packet state $\psi[\Vf]$ is
\begin{equation}
  \psi[\Vf] = e^{-\norm{\Vf}^2/2} \bigoplus_{n = 0}^\infty \frac{\Vf^{\otimes n}}{\sqrt{n!}},
\end{equation}
From the relation that $\inprod{\Vf^{\otimes n}}{\Vg^{\otimes n}} = \inprod{\Vf}{\Vg}^n$, we have
\begin{equation}
  \braket{\psi[\Vf]}{\psi[\Vg]} = \exp\left(  - \half (\,\norm{\Vf}^2 + \norm{\Vg}^2\,)  + \inprod{\Vf}{\Vg} \right)
\end{equation}
or equivalently
\begin{equation}
  \braket{\e[\Vf]}{\e[\Vg]} = e^{ \inprod{\Vf}{\Vg} }.
\end{equation}

A number of useful properties involving the Weyl displacement operators and the exponential vectors are the following:
\begin{itemize}
  \item The Weyl operators obey the composition law
    \begin{equation}\label{chQuLight:eq:weylComposition}
        \weyl[\Vg]\weyl[\Vf] = \exp\left(-\half ( \inprod{\Vg}{\Vf} - \inprod{\Vf}{\Vg} \right) \weyl[\Vg + \Vf].
    \end{equation}
  \item The action of the Weyl operator on an exponential vector is
    \begin{equation}\label{chQuLight:eq:weylExponential}
        \weyl[\Vg]\, \e[\Vf] = e^{- \inprod{\Vg}{\Vf} - \norm{\Vg}^2/2}\,\e [\Vf+ \Vg].
    \end{equation}
  \item The linear span of all the exponential vectors (and equivalently the coherent states) is dense in the symmetric Fock space $\Fock(\hilbert)$, meaning that any state in $\Fock(\hilbert)$ can be represented by a limiting sequence of a linear combination of exponential vectors \cite{barchielli_continual_2006}.
  \item Written in terms of the single particle inner product, the exponential vector $\e[\Vg]$ is an eigenvector of the annihilation operator $\ahat[\Vf]$ with,
    \begin{equation}
      \ahat[\Vf]\,\e[\Vg] = \inprod{\Vf}{\Vg}\ \e[\Vg].
    \end{equation}
\end{itemize}

\subsection{A basis independent expression for the wave packet inner product\label{chQuLight:sec:basisIndependent}}
An alternative to expressing $\norm{\Vg(\Vk, t)}^2$ in terms of the characteristic parameters, $\mathrm{v}$, $k_1$, \emph{etc.} is to relate it to a basis independent expression involving the physical (classical) fields $\V{\mathcal{E}}$ and $\V{\mathcal{A}}$.  From Eq. (\ref{chQuLight:eq:wavePacketSimple1}) we can see that
\begin{equation}
    \Vg(\Vk, t) =  \, \sqrt{\frac{2 \varepsilon_0 c \abs{\Vk} }{ \hbar }}\, \VF{\mathcal{A}}^\pf(\Vk, t)
\end{equation}
and that
\begin{equation}\label{chQuLight:eq:wavePacketNormA}
    \norm{\Vg(\Vk, t)}^2 =  \frac{2 \varepsilon_0}{\hbar}\int d^3 k\ c \abs{\Vk} \, \VF{\mathcal{A}}^{\pf\, *}(\Vk, t) \cdot \VF{\mathcal{A}}^{\pf}(\Vk, t).
\end{equation}
The presence of the factor $c \abs{\Vk}$ makes this expression inherently tied to the $\Vk$ basis and not immediately expressible in terms real space quantities.  However by recognizing that in the Coulomb gauge $\V{\mathcal{E}} = - \frac{\partial}{\partial t}\V{\mathcal{A}}$ and $\VF{\mathcal{A}}^{\pf}(\Vk, t) = \VF{\mathcal{A}}^{\pf}(\Vk, 0) e^{-i c \abs{\Vk} t}$ we have the equality
\begin{equation}
    c \abs{\Vk} \VF{\mathcal{A}}^{\pf}(\Vk, t) = - i\, \VF{\mathcal{E}}^{\pf}(\Vk, t).
\end{equation}
Substituting this relation into Eq. (\ref{chQuLight:eq:wavePacketNormA}),
\begin{equation}
    \norm{\Vg(\Vk, t)}^2 =  \frac{i 2 \varepsilon_0}{\hbar} \int d^3 k\ \VF{\mathcal{E}}^{\pf\, *}(\Vk, t) \cdot \VF{\mathcal{A}}^{\pf}(\Vk, t).
\end{equation}
This expression is basis independent, in the sense that we can take the inverse transforms to arrive at
\begin{equation} \label{chQuLight:eq:wavePacketFieldNorm}
    \norm{\Vg}^2 =  \frac{i 2 \varepsilon_0}{\hbar} \int d^3 x\ \V{\mathcal{E}}^{\pf\, \ast}(\Vx, t) \cdot \V{\mathcal{A}}^{\pf}(\Vx, t).
\end{equation}

In \citep{smith_photon_2007}, \citeauthor{smith_photon_2007} derive a Dirac quantization scheme for a photon wave function, equivalent to the more standard expressions reviewed in Sec. \ref{chQuLight:sec:freeSpaceQuantization}.  In that work they assume that for each polarization vector $q$ there exists a countable set of complete scalar orthonormal wave packets $\set{ g_{j\, q}(\Vk) }$, which therefore satisfy the properties
\begin{equation}
    \begin{split}
     \sum_j\, g_{j\, q}(\Vk)^* g_{j\, q}(\Vk') &= \delta(\Vk - \Vk')\quad \text{and}\\
      \int d^3k\ g_{j\, q}(\Vk)^* g_{j'\, q}(\Vk) &= \delta_{j\, j'}.
    \end{split}
\end{equation}
They then observe that the classical electric fields $\set{\mathcal{E}^{\pf}_{j\, q}(\Vx, t)}$  in correspondence to these wave packets, via Eq. (\ref{chQuLight:eq:classicalE}), are no-longer orthogonal in a real space overlap integral precisely because of the weighting factor of $\sqrt{\abs{\Vk}}$,
\begin{equation}
      \int d^3x\ \mathcal{E}^{\pf\,\ast}_{j\, q}(\Vx, t)\ \mathcal{E}^{\pf}_{j'\, q'}(\Vx, t) \ne \delta_{j\, j'}\,\delta_{q, q'}.
\end{equation}
They also observe that if instead one considers the overlap with the vector potential then the orthogonality is preserved, due to the cancelation of the factors of $\sqrt{\abs{\Vk}}$.  This is precisely the statement that if $\inprod{\Vg_j(\Vk, t)}{\Vg_{j'}(\Vk, t)} = \delta_{j\, j'}$, then
\begin{equation} \label{chQuLight:eq:wavePacketFieldInprod}
    \begin{split}
      \inprod{\Vg_j(\Vk, t)}{\Vg_{j'}(\Vk, t)} &=  \frac{i 2 \varepsilon_0}{\hbar} \int d^3 x\ \V{\mathcal{E}}^{\pf\, \ast}_{j}(\Vx, t) \cdot \V{\mathcal{A}}^{\pf}_{j'}(\Vx, t)\\
      &= \frac{- i 2 \varepsilon_0}{\hbar} \int d^3 x\ \V{\mathcal{A}}^{\pf\, \ast}_{j}(\Vx, t) \cdot \V{\mathcal{E}}^{\pf}_{j'}(\Vx, t) = \delta_{j, j'}.
    \end{split}
\end{equation}
In quantizing a photon's wave function \citeauthor{smith_photon_2007} consider $\V{\mathcal{E}}^{\pf}(\Vx, t)$ to be the fundamental single particle wave functions.  Secondly they observe that in order to preserve orthogonality in the real space inner product then the dual vectors are not $\V{\mathcal{E}}^{\pf\, \ast}(\Vx, t)$ but are instead proportional to $\V{\mathcal{A}}^{\pf\, \ast}(\Vx, t)$.

In this work we will continue to view the single particle vectors to be the wave packet functions $\Vg$ and not the associated classical electric field.  This is for two reasons. Firstly, it is mathematically convenient that the vector dual to the wave packet $\Vg$ is simply its complex conjugate.  The second reason is that we continue to treat $\Vg$ as an analogy with the simple harmonic oscillator's coherent state and the vector potential $\mathcal{A}$ and the electric field $\mathcal{E}$ are in correspondence with $X$ and $P$ quadratures.

\subsection{Fock space and stochastic srocesses\label{chQuLight:sec:FockSpaceStochastProcesses}}
 By structuring the Hilbert space of the free EM field as a Fock space defined over a single particle Hilbert space $\hilbert$, we can now define a quantum stochastic processes and a quantum stochastic calculus.  Consider again the example of a coherent laser pulse propagating towards a photon counter.  A quantum description of a traveling wave laser pulse is a coherent wave packet state $\psi[\Vg]$ where $\Vg(\Vk,t)$ is related to the classical field by Eq.(\ref{chQuLight:eq:wavePacketSimple1}).  Imagine a perfect space fixed detector that is capable of returning a voltage directly proportional to the total energy in a given classical wave packet $\Vg(\Vk,t)$.  Furthermore, imagine that this detector is activated between the times $[t_0, t_1]$, and after this interval the voltage is read.  If the ``entirety'' of an incident wave packet $\Vg$, could be absorbed in that time, then the detector should be modeled as making a projective measurement on the part of Fock space containing $\psi[\Vg]$.  Depending upon the details of the detector, it could likely have recorded pulses that were similar enough to $\Vg$, either in magnitude, temporal profile or carrier frequency.  For instance, a 100\% efficient detector with a linear response should be able to measure pulses with $2\ \mu$W of average laser power just as well as a pulse with $200$ mW of power.  By modeling a physical measurement as a Hermitian operator acting upon some Fock space, we need to define the set of possible wave packets the detector could have completely measured.   Fig. \ref{chQuLight:fig:ParaxialMeasurementModel} shows a schematic where a paraxial laser pulse is focused upon a gated photo-detector.

\begin{figure}[bht!]
	\begin{center}
		\includegraphics[width=1\hsize]{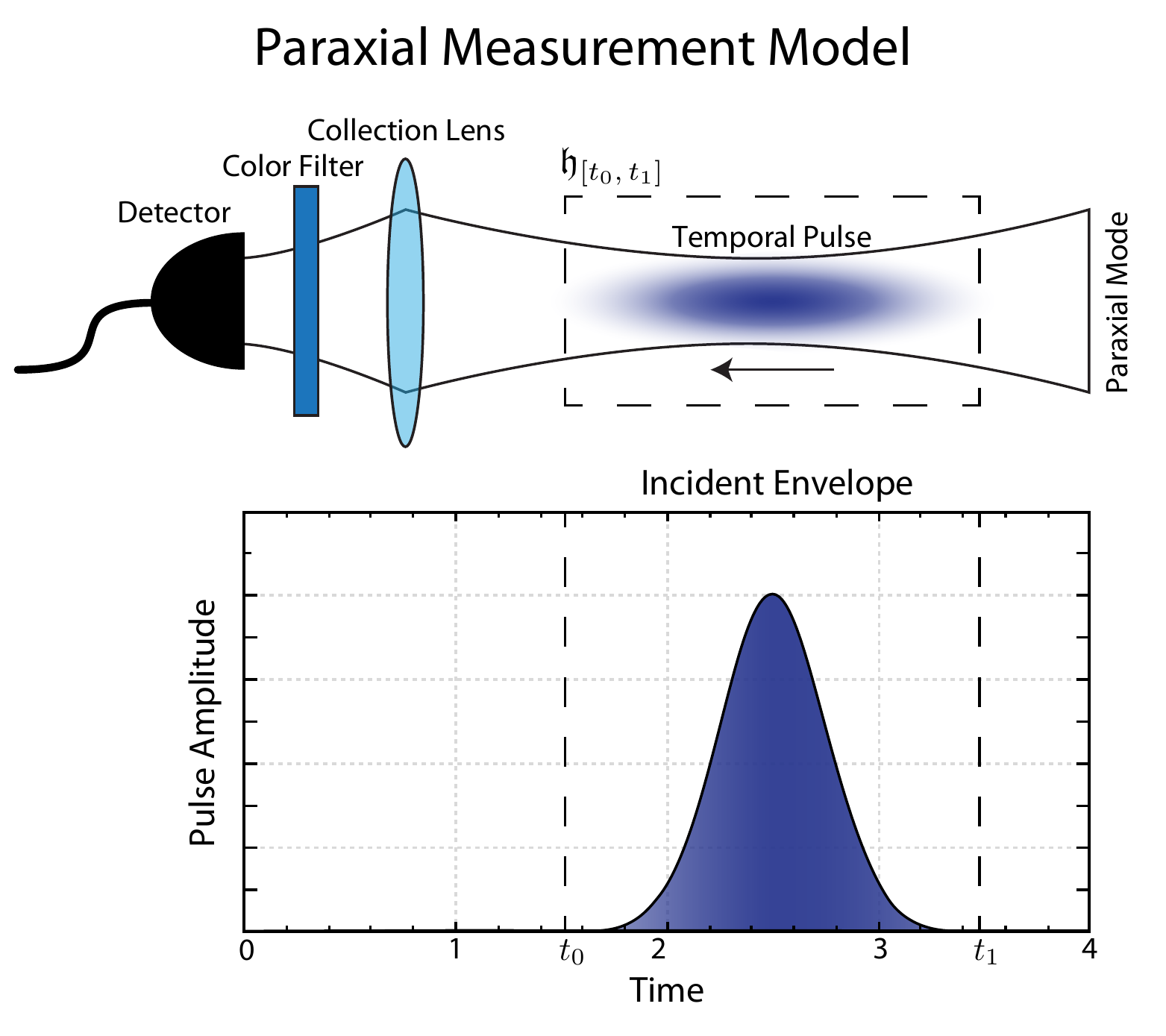}
	\end{center}
   		 \caption{A Model of a Paraxial Measurement.  A photo-detector is positioned relative to an optical system which defines a paraxial beam with a characteristic wavelength and mode profile.  Here the optical system is defined simply by a focusing lens and color filter with the beam schematically indicated by lines of constant intensity.   The detector is activated between times $t_0$ and $t_1$ which corresponds, at time $t = 0$ to a pulse localized in space in the region $\Delta z = c (t_1 - t_0)$.  For a perfect detector with a linear response, the integrated output current will be proportional to the total pulse energy and can be modeled as making a projective measurement on the sub-Fock space of pulses localized between these times. \label{chQuLight:fig:ParaxialMeasurementModel} }
\end{figure}

In the second quantization formalism we can give a mathematical chain from a classical wave packets to field operators.  In the time interval $[t_0, t_1]$ a fixed detector could projectively measure some set $\set{\Vg}$ of incident wave packets and the linear span of these wave packets forms a subspace $\hilbert_{[t_0, t_1]} \subset \hilbert$.  In turn $\hilbert_{[t_0, t_1]}$ defines a subspace $\Fock(\hilbert_{[t_0, t_1]}) \subset \Fock(\hilbert)$. Furthermore there exist operators $O_{[t_0, t_1]}$ that act nontrivially on coherent states $\psi[\Vg] \in \Fock(\hilbert_{[t_0, t_1]})$ but as the identity for any $\psi[\Vg^\perp]$ for $\Vg^\perp \notin \hilbert_{[t_0, t]}$.  An operator $X_{[t_0, t_1]}$ identified by this procedure then defines a quantum stochastic process, by considering the family of operators $\set{X_{[t_0, t]}\, :\, t_0 \le t < \infty}$.  The requirement that $X_{[t_0, t]}$ acts trivially on coherent states $\psi[\Vg^\perp]$ defines a process that is \emph{time-adapted}, in direct analogy with a time-adapted classical stochastic process (see Sec. \ref{chMath:sec:processes} for the definition of a time-adapted classical stochastic process).  Turning this qualitative procedure into a mathematically sound object requires explicitly constructing $\hilbert_{[t_0, t_1]}$, which means defining what it means to measure the entirety of a wave packet in a finite time interval.  While na\"{i}vely this may seem trivial, in practice it intersects with the problem of defining a localizable photon in quantum field theory.  We will illustrate why this is an issue next.

\subsection{Localized wave packets and stochastic processes\label{chQuLight:sec:localization}}
While the canonical quantization of the free field is most easily performed in the Fourier domain, the mathematical structure of the second quantized Fock space $\Fock(\hilbert)$ is generally basis independent.  The operators $\ahat[\Vg]$ and $\ahat^\dag[\Vf]$ can be related to the coherent states $\psi[\Vh]$ without any reference to the fact that the wave packets are originally defined with respect to $\Vk$.  Any unitary transformation of $\Vg$ is an equally valid expression of the wave packet state in that the Hilbert space of wave packets $\hilbert = \set{\Vg :\ \inprod{\Vg}{\Vg} < \infty}$ is basis independent.  The only element that depends upon $\Vg$ being defined in the Fourier domain is its relationship to the spatial profile of the mode function $u^\pf(\Vx, t)$.  But as we have defined the wave packets in the Fourier domain, it is not immediately apparent what effect the constraint $\norm{\Vg}^2 < \infty$ has on $u^\pf(\Vx, t)$. One drastic result of this constraint is that it prohibits one from defining fields that are strictly localized in space \cite{garrison_quantum_2008,bialynicki-birula_exponential_1998}.

To see why this is true, consider a one-dimensional case where we wish to define square wave pulse of duration $L$, with a carrier frequency $\omega_0 = c\, k_0$.  The mode function for such a pulse is
\begin{equation}
  u_{\text{loc}}^{\pf}(z, t) = \indicate{[0, L]}\!(z - c t)\, \exp( +i k_0 (z - c t) ).
\end{equation}
Taking the spatial Fourier transform shows that
\begin{equation}
  \widetilde{u}_{\text{loc}}^{\pf}(k, t) = \frac{L}{\sqrt{2 \pi}}\, \operatorname{sinc}\big( \half L (k - k_0) \big)\, \exp\big( -i \half(k - k_0)\, L - i c k\, t \big)
\end{equation}
and from Eq. (\ref{chQuLight:eq:wavePacketSimple1}), we then have $ g_{\,\text{loc}}(k, t) \propto \sqrt{k}\, \widetilde{u}_{\text{loc}}^{\pf}(k, t)$.  The $\sqrt{k}$ factor makes all the difference, as if we try and calculate the norm we find
\begin{equation}
  \norm{ g_{\,\text{loc}}(k, t) }^2 \propto \int_{-\infty}^\infty dk\ \abs{k}\,  \text{sinc}^2( \half L (k - k_0) )  = \infty.
\end{equation}
The failure of this calculation stems from the fact that the indicator function $\indicate{[a, b]}$ is a discontinuous function and this discontinuity presents itself in the Fourier domain by this divergence.  This example implies that there is a nonlocalizable property of photon wave packets.  This nonlocal property of a photon wave function has been studied by many authors, with various definitions for a photon's wave function, see \cite{smith_photon_2007} for a pertinent discussion. While this example does not show that any localized wave packet suffers from this or a similar problem, this is indeed the case.  In \cite{bialynicki-birula_exponential_1998}, \citeauthor{bialynicki-birula_exponential_1998} proves that the energy density of a photon can be localized no better than an exponential function, $\exp( - f(r) )$, where $f(r)$ grows slightly slower than a linear function in $r$.

At this point, we must make some approximation thereby admitting localized states of light.  Ultimately this means that we will relax the measurement window to include functions localized with in exponentially damped tails.  However, we gain physical insight by considering a temporal rescaling so that on a timescale that is long compared to an optical period, a smoothly varying function can appear to be a localized discontinuous function.  We will also show that though this rescaling, one obtains all the familiar approximations in quantum optics, namely the Markov, quasi-monochromatic, and rotating wave approximations.  It also provides a gateway for defining quantum white noise and equivalently the necessary conditions for applying a quantum Wong-Zakai theorem to arrive at a physical realization of quantum stochastic calculus.

\section{Paraxial Envelopes and Measurable Pulses\label{chQuLight:sec:measureableSpace}}
Before constructing temporally localized wave packet we first must address the spatial/temporal decomposition indicated in Fig. \ref{chQuLight:fig:ParaxialMeasurementModel}.  An excellent mathematical description for the focused collection of light by a series of thin lenses is to model the system paraxially where a coherent plane wave is propagating along the optical axis but is spatially and temporally modulated by a slowly varying envelope function.  This envelope function describes how the plane wave is localized to the optical axis as well as how the phase fronts are distorted by the optical system \cite{siegman_lasers_1986,hecht_optics_2002}.  Appendix \ref{app:paraxialOptics} reviews the derivation of the paraxial wave equation, which describes the propagation of a slowly-varying-envelope, as well as computes its spatial Fourier transform.

A paraxial and quasi-monochromatic wave is characterized by the complex function
\begin{equation} \label{chQuLight:eq:paraxialFactoredSolution}
    \V{\mathcal{U}}^{\pf}(\Vx,t) = f(t_r )\, \V{u}^\pf_T(\Vx_T, z)\, e^{- i \omega_0\, t_r }.
\end{equation}
Here $\Ve_z$ is the axis of propagation, $\Vx_T$ is the remaining transverse coordinates and $t_r = t - z/c$ is the retarded time.  The paraxial mode function $\V{u}^\pf_T(\Vx_T, z)$ describes how the carrier wave (with angular frequency $\omega_0$, and wave number $k_0$) is modulated as it propagates along the optical axis.  Note that this is a time independent quantity.  Any nontrivial time dependence is given by the temporal envelope function $f(t_r)$, which decouples from the paraxial mode $\V{u}^\pf_T$.  The only requirement is that $f(t_r)$ be slowly varying,  $\abs{\frac{d}{d t} f} \ll \omega_0\,\abs{f}$, ensuring the full solution is quasi-monochromatic.  The problem of finding the space of measurable wave packets now translates into finding the temporal envelopes $\set{f(t_r)}$ that ``fit'' in the measurement window $[t_0, t_1]$.

Identifying the appropriate spatial mode function $\V{u}_T^\pf(\Vx_T, z)$ for a given optical system, is simply a problem of classical optics and is well modeled by a Hermite-Gaussian mode function \cite{hecht_optics_2002}.  Here we are only concerned with the fact that such a function exists and is well defined and has a given ``transverse area''.   In free space conservation of energy requires that the total power passing though a plane transverse to the optical axis be conserved, which manifests thought the property that for any paraxial mode we have
\begin{equation}
\int_{\R^2} d^2x_T\ \abs{\VF{u}_T^\pf(\Vk_T, z)}^2 \define \sigma_T
\end{equation}
and that $\sigma_T$ is \emph{independent} of $z$ for any finite $z$.  While the distribution of energy in the transverse plane can vary due to diffraction, the total power passing though an infinite transverse plane will be conserved.  This transverse area can be combined with the square integrated temporal duration
\begin{equation}
  \tau\define \int dt\ \abs{f(t)}^2,
\end{equation}
to construct the total mode volume
\begin{equation}
  \mathrm{v} = c\, \tau\, \sigma_T.
\end{equation}

We have already shown how a wave packet $\Vg$ is related to the spatial Fourier transform of a classical vector potential, $\VF{\mathcal{A}}^\pf(\Vk, t)$, and how that can be expressed in terms of a unitless mode function $\V{u}^\pf(\Vx, t)$.  The spatial Fourier transform of the paraxial mode function is given by
\begin{equation}
    \VF{\mathcal{U}}^{\pf}(\Vk, t) = c\, \F{f}\left(\omega(\Vk) - \omega_0\right)\, \VF{u}^{\pf}_T(\Vk_T, 0)\,e^{-i \omega(\Vk) t}.
\end{equation}
where $\F{f}$ is the temporal Fourier transform of the pulse envelope, $\VF{u}^{\pf}_T(\Vk_T, 0)$ is the spatial transform of the mode function with respect to the transverse coordinates $\Vx_T$ (evaluated at $z = 0$) and $\omega(\Vk)$ is the approximate frequency
\begin{equation}
  \omega(\Vk) \define c \abs{\Vk} \approx c\, \left(\frac{\abs{\Vk_T}^2}{2 k_0} + k_z\right).
\end{equation}
Eq. (\ref{chQuLight:eq:wavePacketSimple1}) relates $\Vg(\Vk, t)$ to $\VF{\mathcal{A}}^\pf(\Vk, t)$ and so
\begin{equation}
  \Vg(\Vk, t) =  \mathcal{A}_0\, \sqrt{\frac{2 \varepsilon_0 \omega(\Vk)}{ \hbar }} c \, \F{f}\left(\omega(\Vk) - \omega_0\right)\, \VF{u}_T^\pf(\Vk_T)\,e^{-i \omega(\Vk) t}.
\end{equation}

In Sec. \ref{chQuLight:sec:wavePackets}, we eliminated the constant $\mathcal{A}_0$ in favor of an expression in terms of characteristic parameters, namely the mode volume $\mathrm{v}$ and the characteristic wave number $k_1$.  Here we abandon $k_1$ in favor of a frequency $ \omega_1 = c k_1$, which is equal to
\begin{equation}
    \omega_1 = \frac{1}{\mathrm{v}}\int d^3k\ \omega(\Vk)\, \abs{c\, \F{f}\left(\omega(\Vk) - \omega_0\right)\, \VF{u}_T^\pf(\Vk_T)}^2.
\end{equation}
Calculating this integral is easiest with the change of variables\\ $\set{\Vk_T, k_z} \rightarrow  \set{\Vk_T, \omega(\Vk)}$, we find that
\begin{equation}
    \omega_1 = \int \frac{d^2 k_T}{\sigma_T} \abs{\VF{u}_T^\pf(\Vk_T)}^2\, \int \frac{d\omega(\Vk)}{\tau}\, \omega(\Vk)\, \abs{\F{f}\left(\omega(\Vk) - \omega_0\right)}^2.
\end{equation}
Implicit in the paraxial approximation, is the requirement that $\abs{\VF{u}_T^\pf(\Vk_T)}^2 \rightarrow 0$ as $\abs{\Vk_T}^2 \rightarrow \infty$.  This fall off implies that we can treat the factor $c\, \abs{\Vk_T}^2/\, 2 k_0$ in $\omega(\Vk)$ as a finite and independent offset to the $d \omega(\Vk)$ integral and not consider how $\omega(\Vk)$ converges as $k_z \rightarrow -\infty$ with $\abs{\Vk_T}^2 \rightarrow \infty$.  We can then make another change of variables $\omega(\Vk) \rightarrow \nu + \omega_0$ so that
\begin{equation}
    \omega_1 = \int \frac{d^2 k_T}{\sigma_T} \abs{\VF{u}_T^\pf(\Vk_T)}^2\, \int \frac{d \nu}{\tau}\, (\omega_0 + \nu)\, \abs{\F{f}(\nu)}^2.
\end{equation}
when $f(t)$ is real-valued, it is simple to show that $\abs{\F{f}(\nu)}^2$ is an even function and therefore mean zero.  In that case we have
\begin{equation}
        \omega_1 = \omega_0 \int \frac{d \nu}{\tau}\ \abs{\F{f}(\nu)}^2 + \int \frac{d \nu}{\tau}\ \nu\, \abs{\F{f}(\nu)}^2 = \omega_0
\end{equation}
as one would intuitively expect.  In the case of a complex valued $f(t)$, $\abs{\F{f}(\nu)}^2$ will not in general be mean zero.  In this more general case,
\begin{equation}
    \omega_1 = \omega_0\left(1  +  \int \frac{d \nu}{\tau}\ \frac{\nu}{\omega_0}\, \abs{\F{f}(\nu)}^2 \right).
\end{equation}
For the one-dimensional localized pulse, Sec. \ref{chQuLight:sec:localization} demonstrated that the second integral is infinite.  By integral expressions for the Fourier transforms and then integrating by parts, we can show that
\begin{equation}\label{chQuLight:eq:characteristicFrequencyCorrection}
    \int d \nu\ \frac{\nu}{\omega_0}\, \abs{\F{f}(\nu)}^2 =  - i \int d t\ \frac{1}{\omega_0}\left( \frac{d}{d t} f^\ast(t) \right)\, f(t).
\end{equation}
By the slowly varying envelope approximation, we require that $\abs{\frac{d}{d t} f(t)} \ll \omega_0 \abs{f(t)}$, and in order for the quasi-monochromatic regime to hold this integral must be a small correction factor.
Therefore, we find that
\begin{equation}\label{chQuLight:eq:wavePacketParaxial}
    \Vg(\Vk, t) \approx \norm{\Vg} \sqrt{\frac{\omega(\Vk)\,c }{\omega_0\, \tau\, \sigma_T}} \, \F{f}\left(\omega(\Vk) - \omega_0\right)\, \VF{u}_T^\pf(\Vk_T,0)\,e^{-i \omega(\Vk) t}.
\end{equation}

\subsection{Paraxial wave packets in the time domain \label{chQuLight:sec:wavePacketsTimeDomain}}
Even in a quasi-monochromatic regime, the wave packet $\Vg$ is still tied the Fourier basis due to the factor of $\sqrt{\omega(\Vk)}$.  We just showed that when $f(t)$ is real-valued or very slowly varying then this factor plays no role in calculating $\norm{\Vg}^2$.  Here we would like express the inner product between wave packets that share the same spatial mode in terms of real space coordinates in order observe what effect this ``nonlocal'' factor has on their temporal distinguishably.  If we are able to make the approximation that $\omega(\Vk) \approx \omega_0$ for a family of wave packets then there will be a simple unitary relationship between wave packets the real and Fourier domains.  The goal of this section is to identify this family.

Consider the two wave packets $\Vg_1(\Vk, t)$ and $\Vg_2(\Vk, t)$ that share the same paraxial mode function $\V{u}_T^\pf(\Vx_T, z)$ and carrier frequencies, but have differing temporal profiles $f_1(t_r)$ and $f_2(t_r)$.  For simplicity we will assume that the corrective factor of Eq. (\ref{chQuLight:eq:characteristicFrequencyCorrection}) is small and so Eq. (\ref{chQuLight:eq:wavePacketParaxial}) is valid for each wave packet.  If we then calculate the \emph{unequal} time inner product, we have that
\begin{multline}
    \inprod{\Vg_1(\Vk, t_1)}{\Vg_2(\Vk, t_2)} = \frac{\norm{\Vg_1}\norm{\Vg_2}}{\sqrt{\tau_1\, \tau_2}}\\ \int d^3 k\ \frac{\omega(\Vk)\, c}{\omega_0\, \sigma_T} \, \abs{\VF{u}_T^\pf(\Vk_T,0)}^2\,  \F{f}_1^\ast\left(\omega(\Vk) - \omega_0\right)\,\F{f}_2\left(\omega(\Vk) - \omega_0\right) e^{-i \omega(\Vk)(t_2 - t_1)}.
\end{multline}
By again making the change of variables $\set{\Vk_T, k_z} \rightarrow \set{\Vk_T, \nu}$ with $\nu = \omega(\Vk) - \omega_0$ we are able to integrate out the transverse degrees of freedom to arrive at
\begin{equation}
    \inprod{\Vg_1(\Vk, t_1)}{\Vg_2(\Vk, t_2)} = \frac{\norm{\Vg_1}\norm{\Vg_2}}{\sqrt{\tau_1\, \tau_2}} \int d \nu \, \frac{ (\omega_0 + \nu)}{ \omega_0}\,\F{f}_1^\ast(\nu)\,\F{f}_2(\nu) e^{-i (\omega_0 + \nu) (t_2 - t_1)}.
\end{equation}
In analogy with the convolution theorem, it is easy to show that
\begin{equation}
    \int d \nu  \F{f}_1^\ast(\nu)\,\F{f}_2(\nu) e^{-i \nu t} = \int ds\,  f^\ast_1(s)\, f_2(s + t) = f_1 \star f_2\,(t)
\end{equation}
where $f_1 \star f_2\,(t)$ is the cross-correlation function between $f_1$ and $f_2$ evaluated at time $t$.  Furthermore by repeating the integration by parts transformation from Eq. (\ref{chQuLight:eq:characteristicFrequencyCorrection}) we have that
\begin{equation}
    \int d \nu\ \nu\,  \F{f}_1^\ast(\nu)\,\F{f}_2(\nu) e^{-i \nu t} = - i \frac{df_1}{dt} \star f_2\ (t).
\end{equation}
Combining these two facts,
\begin{multline}\label{chQuLight:eq:paraxialInnerProduct}
    \inprod{\Vg_1(\Vk, t_1)}{\Vg_2(\Vk, t_2)} = \frac{\norm{\Vg_1}\norm{\Vg_2}}{\sqrt{\tau_1\, \tau_2}}e^{-i \omega_0(t_2 - t_1)}\\ \left(f_1 \star f_2\ (t_2 - t_1) - i\frac{1}{\omega_0} \frac{df_1}{dt}\star f_2\ (t_2 - t_1) \right)
\end{multline}
While previously we were able to show that for zero delay and a real-valued envelope the second term would be identically zero, this is clearly not the case for different wave packets.  However, due to the slowly varying envelope approximation we know that this must be a small correction.  Ignoring this correction results in
\begin{equation}\label{chQuLight:eq:zeroOrderCrossCorrelation}
    \inprod{\Vg_1(\Vk, t_1)}{\Vg_2(\Vk, t_2)} \approx \frac{\norm{\Vg_1}\norm{\Vg_2}}{\sqrt{\tau_1\, \tau_2}}e^{-i \omega_0(t_2 - t_1)}\ f_1 \star f_2\, (t_2 - t_1).
\end{equation}
Eq. (\ref{chQuLight:eq:zeroOrderCrossCorrelation}) shows that the overlap between the two wave packets is proportional to the cross correlation function of the temporal envelopes.  Physically this is a extremely satisfying result, as if we have the two (paraxial) field operators $\ahat[\Vg_1(\Vk, t_1)]$ and $\ahat^\dag[\Vg_2(\Vk, t_2)]$ then
\begin{equation}
  \left[\,\ahat[\Vg_1(\Vk, t_1)],\,\ahat^\dag[\Vg_2(\Vk, t_2)]\,\right] \propto f_1 \star f_2\, (t_2 - t_1)
\end{equation}
meaning that field operators for uncorrelated temporal envelopes commute!
Furthermore if we can construct a wave packet $\V{\varphi}(\Vk, t)$ whose temporal envelope $\varphi(t_r)$ is (approximately) delta correlated in time then,
\begin{equation}
  \left[\,\ahat[\V{\varphi}(\Vk, t)],\,\ahat^\dag[\V{\varphi}(\Vk, t')]\,\right] \propto \delta(t'-t).
\end{equation}
This is significant because this is the defining feature of quantum white noise, which is discussed in Sec. \ref{chQuLight:sec:QuantumWhiteNoise}.  Before doing so, we will apply the results of this section the defining $\hilbert_{[t_1, t_2]}$.

\subsection{The measurable subspace\label{chQuLight:sec:measureableConstruction}}
In the ideal situation, the measurable wave packet are the wave packets defined on the paraxial mode $\V{u}_T^{\pf}(\Vx_T, z)$ with a envelope functions $f(t)$ such that $f(s) = 0$ for all $ s \notin [t_0, t_1]$.  Unfortunately because of the problem of localization no such physical wave packets exist.   If we allow for discontinuous functions, then for any function $g(t)$,
\begin{equation}
  g_{[t_0, t_1]}(t) \define \indicate{[t_0, t_1]}\!(t)\, g(t)
\end{equation}
is clearly zero for any $ t \ne [t_0, t_1]$ and therefore would be an element of $\hilbert_{[t_0, t_1]}$.  In order to define approximately localized temporal envelopes we need an approximate form of the indicator function $\indicate{[t_0, t_1]}\!(t)$, \emph{i.e.} a smooth cut off function.  A common choice for such a function is to convolve $\indicate{[t_0, t_1]}\!(t)$ with a smooth positive normalized distribution function $\varphi^{(\sigma)}(t)$, where $\sigma$ represents the degree of localization.  For a concrete example, if $\varphi^{(\sigma)}(t)$ is a mean-zero normalized Gaussian with variance $\sigma^2$ then,
\begin{equation}\label{chQuLight:eq:cutoffFnc}
\begin{split}
    \indicate{[t_0, t_1]}^{(\sigma)}\!(t) &\define \varphi^{(\sigma)} \ast \indicate{[t_0, t_1]}\!(t)\\
    &= \int ds\ \tfrac{1}{\sqrt{2 \pi\, \sigma^2}}\exp(- \tfrac{(t-s)^2}{2 \sigma^2}) \indicate{[t_0, t_1]}\!(s)\\
    &= \half\left(\operatorname{erf}(\tfrac{t - t_0}{\sqrt{2} \sigma}) - \operatorname{erf}(\tfrac{t - t_1}{\sqrt{2} \sigma}) \right).
\end{split}
\end{equation}
Note that because $\lim_{\sigma \rightarrow 0} \varphi^{(\sigma)}(t) = \delta(t)$, we also have $\lim_{\sigma \rightarrow 0} \indicate{[t_0, t_1]}^{(\sigma)}\!(t) = \indicate{[t_0, t_1]}\!(t)$.  Again, to maintain a quasi-monochromatic field f(t) must be slowly varying.  This statement is quantified by the relation $\frac{1}{\omega_0} \abs{\frac{\partial f}{\partial t} }\ll \abs{f}$, and in terms of $\sigma$ this means
\begin{equation}
\frac{1}{\sigma\, \omega_0} \ll 1.
\end{equation}

In quantum optical systems a carrier frequency of $\omega_0 = 2\pi\, \times\, 370\text{ THz}$ is not uncommon, and has a corresponding wavelength of $\lambda_0 = 810 \text{ nm}$.  A common use for a laser at this wavelength is to generate optical pules as short as 5 fs in duration \cite{diels_ultrashort_2006}.  If we take this to be a minimum but physically realizable timescale than we would have $(\sigma\, \omega_0)^{-1} \sim 0.08$.  While a wave packet of this duration is still relatively slowly-varying, its likely that the correction to $\omega_1$ in Eq. (\ref{chQuLight:eq:characteristicFrequencyCorrection}) could be a nonnegligible contribution, as well as other higher order effects.  A convenient limit would be to set $\sigma$ such that $(\sigma\, \omega_0)^{-1} \sim 10^{-3}$, meaning that for the near infrared wavelengths $\sigma \sim 0.1\text{ ps}$.
\begin{figure}[hbt!]
	\begin{center}
		\includegraphics[width=1\hsize]{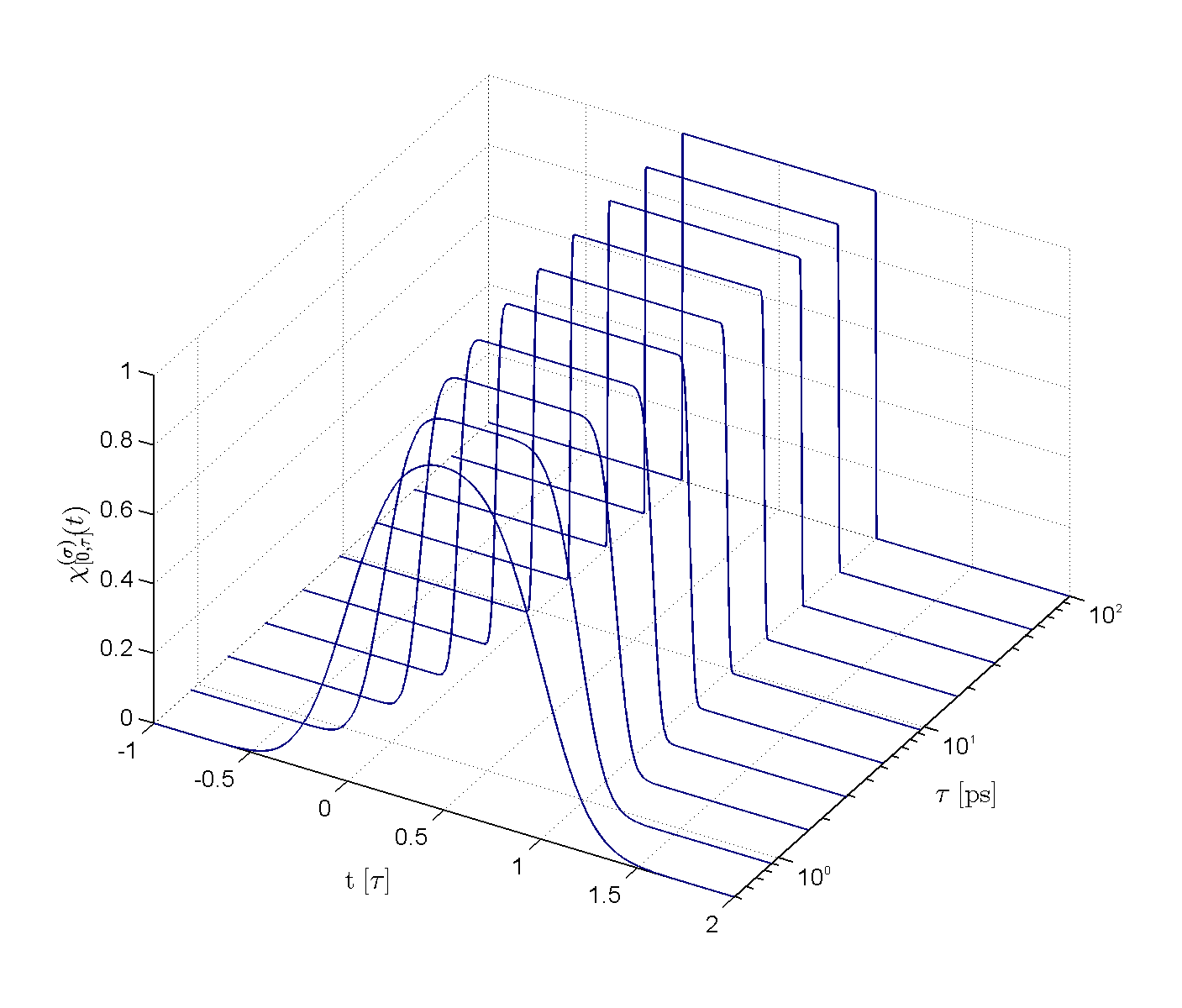}
   		 \caption{Approximations for a localized pulse.   Shown here are a series of slowly varying temporal envelopes.  Each envelope is a unit pulse centered at zero, with a variable duration $\tau$, convolved with a Gaussian smoother with $\sigma = 0.1$ ps.  The pulse duration $\tau$ ranges between $0.5$ ps to $100$ ps.  Each pulse is plotted verses $t$ in units of $\tau$. \label{chQuLight:fig:PulseScaling} }
		\end{center}
\end{figure}
While this sets a physically realistic smoothing variance, it does not say when $\indicate{[t_0, t_1]}^{(\sigma)}\!(t)$ is a good approximation to an actual indicator function, as this requires a comparison between the smoothing and its overall duration.  Fig. \ref{chQuLight:fig:PulseScaling} illustrates this distinction by plotting $\indicate{[0, \tau]}^{(\sigma)}\!(t)$ with $\sigma = 0.1\text{ ps}$ and a series of durations, $0.5\text{ ps}\le\,\tau\,\le 100 \text{ ps}$.  For visual comparison each indicator is plotted in scaled units of $\tau$.  Simple inspection shows that for intervals on the scale of $\tau \gtrsim 10 \text{ ps}$ a smoothing variance of $0.1 \text{ ps}$ makes an excellent approximation to the truly discontinuous function.  Note that as $t$ extends beyond the interval $[0, \tau]$, $\indicate{[0, \tau]}^{(\sigma)}\!(t)$ decays like an error function scaled by $\sigma$ and this extent is independent of $\tau$ for $\tau \gg \sigma$. So that for $\tau  = 10$ ps and $\tau = 0.1$ ps, $\indicate{[0, \tau]}^{(\sigma)}\!(\tau + 3 \sigma) \approx 10^{-3}$.

With this function in hand we can now identify a space of wave packets that are able to be projectively measured approximately in the time window of $[t_0, t_1]$.  For any valid envelope function $f(t)$, we can define a localized version
\begin{equation}
    f_{[t_0, t_1]}^{(\sigma)}(t) \define \indicate{[t_0, t_1]}^{(\sigma)}\!(t)\, f(t).
\end{equation}
Clearly these functions are good temporal envelopes and approximately fit in $\hilbert_{[t_0, t_1]}$.  Note that we can actually increase the space of valid wave packets by observing that with an appropriate $\sigma$, $\indicate{[t_0, t_1]}^{(\sigma)}\!(t)$ is itself a valid temporal envelope.   Therefore for any function $f(t) < \infty$ whose support is contained in the interval $[t_0, t_1]$, we can define a smooth version via convolution
\begin{equation}
  f^{(\sigma)}(t) \define \varphi^{(\sigma)}\ast f\, (t).
\end{equation}
And so any wave packet whose temporal envelope is defined in this way is approximately measurable in the time interval $[t_0, t_1]$.  This can be formalized by defining the set of functions,
\begin{equation}
  \mathscr{S}^{(\sigma)}_{[t_0, t_1]} \define \set{ \varphi^{(\sigma)}\ast f\, (t)\ :\ \int dt\, f(t) < \infty ,\ \operatorname{supp}(f) \subseteq [t_0, t_1] }
\end{equation}
and the set of wave packets
\begin{equation}
  \mathfrak{s}^{(\sigma)}_{[t_0, t_1]} = \lspan{\Vg = \sqrt{\frac{\omega(\Vk)\,c }{\omega_0} }\, \frac{\F{f}\left(\omega(\Vk) - \omega_0\right)}{\sqrt{\abs{t_1 - t_0}}}\ \frac{\VF{u}_T^\pf(\Vk_T,0)}{\sqrt{\sigma_T}}\,e^{-i \omega(\Vk) t}\ : \ f  \in \mathscr{S}^{(\sigma)}_{[t_0, t_1]} }.
\end{equation}
We then have the limit
\begin{equation}
  \mathfrak{h}_{[t_0, t_1]} = \lim_{\substack{\sigma \rightarrow 0\\ \omega_0\sigma \rightarrow \infty}} \mathfrak{s}^{(\sigma)}_{[t_0, t_1]}.
\end{equation}

\section{The one-dimensional limit\label{chQuLight:sec:oneDimensionalLimit}}

In the previous sections we showed that if there are two quasi-monochromatic wave packets $\Vg_1$ and $\Vg_2$ defined on the same paraxial mode but with differing temporal envelops $f_1$ and $f_2$, then the commutator between $\ahat[\Vg_1]$ and $\ahat^\dag[\Vg_2]$ is proportional to the cross-correlation function $f_1\star f_2$.  Furthermore the proportionality is independent of the details of the paraxial mode.  This suggests that moving to a simplified, one-dimensional model is both appropriate and fruitful.  This section makes this connection and relates it to the standard representations of quantum white noise.

The end of Sec. \ref{chQuLight:sec:wavePacketsTimeDomain} suggested defining a wave packet $\V{\varphi}(\Vk, t)$ whose temporal envelope, $\varphi(t_r)$ is delta correlated in time.  In defining the smoothed set of functions $\mathfrak{s}^{(\sigma)}_{[t_0, t_1]}$, Sec. \ref{chQuLight:sec:measureableConstruction} took any integrable function defined on the interval $[t_0, t_1]$, $f_{[t_0, t_1]}(t)$, and convolved it with a Gaussian distribution $\varphi^{(\sigma)}(t)$ to obtain an envelope consistent with the quasi-monochromatic approximation.  To move to a one-dimensional model, we will factor out $f_{[t_0, t_1]}(t)$ from field operator and define an operator-valued density $\ahat[\V{\varphi}^{(\sigma)}(t)]$. We will shortly show that this is approximately delta commuting in time.

Deriving this factorization is not difficult, and begins by first noting that as the Fourier transform of a convolution is proportional to the product of the Fourier transforms, we have that
\begin{equation}
  \F{f}^{(\sigma)}(\nu) = \sqrt{2 \pi}\, \F{f}(\nu)\, \F{\varphi}^{(\sigma)}(\nu) = \left( \int ds\, f(s) e^{+i \nu\, s}\right)\, \F{\varphi}^{(\sigma)}(\nu).
\end{equation}
Substituting this expression into the definition of $\Vg^{(\sigma)}(\Vk, t)$, as written in Eq. (\ref{chQuLight:eq:wavePacketParaxial}), we have,
\begin{equation}
    \Vg^{(\sigma)}(\Vk, t) = \frac{\norm{\Vg}}{\sqrt{\tau}}\, \int ds\ f(s) e^{-i \omega_0\, s}\ \V{\varphi}^{(\sigma)}(\Vk, t - s)
\end{equation}
where
\begin{equation}\label{chQuLight:eq:SmoothingWavePacket}
  \V{\varphi}^{(\sigma)}(\Vk, t) \define \sqrt{\frac{\omega(\Vk)\,c }{\omega_0\, \sigma_T}} \, \F{\varphi}^{(\sigma)}\left(\omega(\Vk) - \omega_0\right)\, \VF{u}_T^\pf(\Vk_T,0)\,e^{-i \omega(\Vk) t}.
\end{equation}
By the anti-linear nature of the creation operator $\ahat[\Vg]$ we are able to bring the integral over $s$ out of the operator to write,
\begin{equation}\label{chQuLight:eq:SmoothWavePacketOperator}
  \ahat[\Vg^{(\sigma)}(\Vk, t)] = \frac{\norm{\Vg}}{\sqrt{\tau}}\, \int ds\ f^*(s)\, e^{+i \omega_0 s}\, \ahat[\V{\varphi}^{(\sigma)}(t - s)].
\end{equation}
Note that $\ahat[\Vg^{(\sigma)}(\Vk, t)]$ and $\ahat[\V{\varphi}^{(\sigma)}(t - s)]$ have different units, as the former is unitless while the latter has units of $1/\sqrt{\mathrm{time}}$, ultimately arising from the fact that $\varphi^{(\sigma)}(t)$ is a density over time and therefore has units.  Eq. (\ref{chQuLight:eq:SmoothingWavePacket}) has the following physical implications.  First is that the integral is the point-wise weighting of an annihilation operator by a complex amplitude, completely akin to original analogy of one-dimensional simple harmonic oscillator $\alpha^*\, \ahat \leftrightarrow f^*(t)\, \ahat(t)$.  The second implication is that the complex weighting function in general matches the phase of the carrier wave, resulting the explicit appearance of the $e^{+i \omega_0 s}$ factor.  The third implication is that both the function $f(s)/\sqrt{\tau}$ and the operator $\ahat[\V{\varphi}^{(\sigma)}(t - s)]$ have the same units, which are the same as white noise.  The final implication comes from the fact that because we have the commutator
\begin{equation}\label{chQuLight:eq:smoothingWavePacketCommutator}
  \left[\ahat[\V{\varphi}^{(\sigma)}(t)],\, \ahat^\dag[\V{\varphi}^{(\sigma)}(t')]\right] = e^{-i \omega_0(t' - t)} \,\varphi^{(\sigma)} \star \varphi^{(\sigma)}(t'-t) + O\left((\sigma\, \omega_0)^{-1} \right)
\end{equation}
and the limit $\lim_{\sigma \rightarrow 0}  \varphi^{(\sigma)}\star\varphi^{(\sigma)}(t'-t) = \delta(t'-t)$, then
\begin{equation}  \label{chQuLight:eq:smoothWavePacketCommutatorLimit}
\lim_{\substack{ \sigma\rightarrow 0 \\ \omega_0\sigma\rightarrow \infty } }  \left[\ahat[\Vg^{(\sigma)}_1(\Vk, t)],\,\ahat^\dag[\Vg^{(\sigma)}_2(\Vk, t)]\right] = \frac{\norm{\Vg_1}\norm{\Vg_2}}{\sqrt{\tau_1\,\tau_2}} \int ds\ f_1^*(s)\, f_2(s).
\end{equation}
This implies that if we have two square integrable functions, $h_1(t)$ and $h_2(t)$ then these functions can count as members of a single particle Hilbert space, $\hilbert' = \mathcal{L}^2(\R)$.  Furthermore we can define a Fock space $\Fock(\hilbert')$ and ultimately the field operators $\ahat[h_1]$ and $\ahat[h_2]$.  If
\begin{equation}
  h_1(t) \cong \frac{\norm{\Vg}}{\sqrt{\tau}} f_1(t),
\end{equation}
then we can draw the formal equivalence
\begin{equation}
  \ahat[h_1] \cong \lim_{\substack{ \sigma\rightarrow 0 \\ \omega_0\sigma\rightarrow \infty } }\, \ahat[\Vg_1^{(\sigma)}].
\end{equation}

While discussing the statistical aspects of quantum light can be interesting in its own right, the real fun is when that light is coupled to another quantum system.  Sec. \ref{chQuLight:sec:QuantumWhiteNoise} shows how when a system couples though an interaction Hamiltonian to operators similar to $\ahat[\V{\varphi}^{(\sigma)}(t)]$ and $\ahat^\dag[\V{\varphi}^{(\sigma)}(t)]$, the limiting object can be written in terms of a quantum stochastic process on the joint Hilbert space $\Hilbert_{sys}\otimes\Fock(\hilbert')$.  Furthermore it discusses how this relates to the standard expressions in quantum optics involving simpler models of quantum white noise.  Before including the system however, we will show what is gained by taking this discontinuous limit and how it is useful for defining quantum stochastic processes.

\section{Quantum Wiener processes and the \\ continuous-time decomposition \label{chQuLight:sec:QuWienerProcesses} }

Quantum stochastic integrals were first defined mathematically by Hudson and Par-thasarathy in \citeyear{hudson_quantum_1984}.  There they formulated a quantum version of a It\={o}-type stochastic integral where the fundamental differentials, in correspondence to the classical Wiener process and other jump processes, are operators acting on a bosonic Fock space \citep{hudson_quantum_1984}.  Independently \citeauthor{gardiner_input_1985} formulated a physical description of quantum white noise operators where creation and annihilation operators are associated with excitations in a bosonic heat bath, which are then used as driving noise sources in a quantum Langevin equation \citep{gardiner_input_1985}.  This second formulation is the most well known in the quantum optics community (see, \emph{e.g.}, the well written reference \citep{gardiner_quantum_2004}) but is less amenable for directly applying the filtering techniques of classical probability theory.  The picture of a heat bath does not immediately induce a picture of a traveling flow of information from a probe system to a detector.  Rather it instills a picture of a system immersed in stationary and chaotic environment and it is unclear what it means quantum mechanically to ``measure the bath''.  While one certainly could, and often does, construct a large scale flow in the bath running from the system to an independent observer such a construction ultimately resembles a wave packet description.

If one instead explicitly includes time into a description of the environment, as \citeauthor{hudson_quantum_1984} do, then statistical properties necessary for defining a quantum Wiener process and a quantum It\={o} integral, namely the ability to construct time-adapted processes, is a direct consequence.  We will shortly review how this is done, but first note that in contrast to relying on the system to dictate how the bath is modeled, this represents a more axiomatic approach in that the statistical properties of the bath are postulated independently from the system.  Now clearly the physics of the entire system-probe-measurement combination will dictate whether or not this is an appropriate model.  The purpose of this section is to show what is gained in this formulation.

Like the \citeauthor{gardiner_input_1985} formulation, this formulation begins by assuming a bosonic Fock space, however here it assumes that it is a second quantization of a single particle Hilbert space,
\begin{equation}
  \hilbert \define \mathcal{L}^2(\R^+) \otimes \hilbert'
\end{equation}
where $\mathcal{L}^2 (\R^+)$ represents the Hilbert space of square integrable functions defined on the positive real line (representing time) and $\hilbert'$ is an auxiliary Hilbert space.  Almost all formulations immediately assume that $\hilbert'$ is a finite $d$-dimensional system and so every $\Vg(t)$ is effectively a complex-vector-valued function, \emph{i.e.} $\Vg\, :\ \R^{+} \rightarrow \Cn{d}$.  In this case the inner product between two single particle vectors is
\begin{equation}
  \inprod{\Vf}{\Vg} = \int_0^\infty dt\ \Vf^*(t)\cdot \Vg(t) < \infty.
\end{equation}
From this single particle Hilbert space the symmetric Fock space $\Fock(\hilbert)$, exponential vectors $\e[\Vg]$, and Weyl displacement operators $\weyl[\Vg]$ are identical to their definitions in Sec. \ref{chQuLight:sec:FockSpace}.  And most importantly, we can define the annihilation and creation operators
\begin{equation}
  A^{i}_t \define \ahat[\indicate{[0, t]} \Ve_i]
\end{equation}
and
\begin{equation}
  A^{i\, \dag}_t \define \ahat^\dag[\indicate{[0, t]} \Ve_i]
\end{equation}
that have the commutation relation
\begin{equation}
  \left[A^{i}_s,\, A^{j\, \dag}_{t}\,\right] = \int_0^\infty ds'\ \indicate{[0, s]}\!(s') \indicate{[0, t]}\!(s')\ \Ve_i \cdot \Ve_{j} = \delta_{i, j} \min(s, t).
\end{equation}
These processes serve as two of the building blocks of quantum stochastic calculus and are in analogy with an n-dimensional Wiener process.  To show this last analogy, we must first discuss how we specifically include time in $\hilbert$.

\subsection{The continuous-time tensor decomposition\label{chQuLight:sec:continuousTimeTensor}}
In basis quantum mechanics there is an intimate connection between statistical independence and a tensor product structure.  When a complete system is described by the tensor product of two Hilbert spaces \emph{and} the total state is a product state then both systems can be considered statistically independent.  Specifically in this case two operators from the individual Hilbert spaces $X_1$ and $X_2$ are statistically independent in the sense that
\begin{equation}
    \langle X_1 \otimes X_2\rangle_\rho = \langle X_1\rangle_{\rho_1} \langle X_2\rangle_{\rho_2}.
\end{equation}
In Sec. \ref{chQuLight:sec:FockSpaceStochastProcesses} we introduced the notion of a time-adapted quantum stochastic process, where a quantum operator $\operatorname{O}_{[0, t]}$ was time-adapted if it acted as the identity on any coherent $\psi[\Vg^\perp]$ where $\Vg^\perp$ is \emph{excluded} from the time interval $[0,t]$;  $\Vg^\perp(s) = 0$ for $0 \le s\le t$.  We defined two mutually orthogonal spaces of wave packets $\hilbert_{[0, t]}$ and $\hilbert^\perp_{[0,t]}$ which in turn have their associated Fock spaces $\Fock(\hilbert_{[0, t]}])$ and $\Fock(\hilbert^\perp_{[0, t]})$.  The classical definition of a time-adapted stochastic process is that the process is statistically independent of all events in the future. In light of the connection between statistical independence on the one hand and a tensor product structure on the other, it seems reasonable to have
\begin{equation}
  \Fock(\hilbert) \cong \Fock(\hilbert_{[0, t]})\otimes\Fock(\hilbert^\perp_{[0, t]})
\end{equation}
where $\cong$ indicates a unitary equivalence.  But if $\hilbert_{[0, t]}$ represents all wave packets localized to $[0, t]$ then it also seems reasonable to conclude that $\hilbert^\perp_{[0, t]} = \hilbert_{(t, \infty)}$.

In this section we will show that this tensor product decomposition is not only possible for any single time $t$ but it is also possible for any sequence of $n$ ordered times $\set{t_n\, :\,0 < t_1 < \dots < t_n < \infty} $.  This is called the \emph{continuous-time tensor decomposition} and is the relation that
\begin{equation}\label{chQuLight:eq:continousTimeTensorProduct}
  \Fock(\hilbert)\cong \Fock(\hilbert_{[0, t_1)})\otimes\Fock(\hilbert_{[t_1, t_2)})\otimes \dots \otimes \Fock(\hilbert_{[t_n, \infty)}).
\end{equation}
The proof of this statement is outlined in lemma \ref{chQuLight:lemma:tensorDecomposition} (which is essentially proposition 19.6 of \citep{parthasarathy_introduction_1992}).

Now if Eq. (\ref{chQuLight:eq:continousTimeTensorProduct}) is true, then for any partitioning of time, no matter how small, this Fock space decomposes into a tensor product between the various partitions.  Furthermore if we have operators $\set{\operatorname{O}_{[t_i, t_{i +1})}}$ which are each adapted to the interval $[t_i, t_{i +1})$ then we have that the expectation values factorize,
\begin{equation}
    \Big\langle\,\prod_i\,\operatorname{O}_{[t_i, t_{i +1})}\, \Big\rangle_{\psi[\Vg]} = \prod_i\,\expect{\operatorname{O}_{[t_i, t_{i +1})}}_{\psi[\Vg_{[t_{i}, t_{i+1})}]}.
\end{equation}
In other words, if both the operators and the state respects the continuous-time decomposition then those operators will be statistically independent for independent times.

Why is this important? Well, Sec. \ref{chMath:sec:Wiener} reviews the basic properties of a Wiener process and shows how its defining feature is that its restrictions to independent time increments are statistically independent and that each are mean zero Gaussian random variables of variance $t_{i +1} - t_i$.   Therefore, any quantum analog of a Wiener process must also respect the continuous-time decomposition.  The fact that the classical Wiener process satisfies the Markov and martingale properties is a direct consequence of this independence \cite{van_handel_stochastic_2007}. Sec. \ref{chMath:sec:processes} reviews the definition of these two properties and how they relate to taking conditional expectation values.  Additionally, the It\={o} definition of a stochastic integral, (see Appendix \ref{app:SDEs} ) is defined in such a way so that the integral $\int x_t\, dw_t$ is also a martingale and that if $x_t$ is Markovian than so is the integral.  For a quantum stochastic integral to also have these desirable properties, a necessary criteria is that a process $\set{X_t}_{t \ge 0}$ must be statistically independent of all future events.  In the next section we return to the operators $ A^{i}_t$ and $A^{i\, \dag}_t$ and show how they can be used to construct a quantum Wiener process, but first we include a proof of the continuous-time decomposition.

\begin{lemma} \label{chQuLight:lemma:tensorDecomposition}
Given the single particle Hilbert space $\hilbert = \mathcal{L}^2(\R^+)\otimes \Cn{d}$ and an ordered sequence of times  $\set{t_n} = \set{t_i \in \R^+\ :\  0 < t_1 < \dots < t_n < \infty}$, the symmetric Fock space satisfies unitarily equivalence
$ \Fock(\hilbert)\cong \Fock(\hilbert_{[0, t_1]})\otimes\Fock(\hilbert_{[t_1, t_2)})\otimes \dots \otimes \Fock(\hilbert_{[t_n, \infty)})$ where $\hilbert_{[t_i, t_{i+1}]} = \mathcal{L}^2([t_i, t_{i+1}])\otimes \Cn{d}$.
\end{lemma}

\begin{proof}[Sketch of Proof.]
For any vector $\Vg \in \hilbert$ we can define the projection of $\Vg$ onto a time interval via
\begin{equation}
    \Vg_{[t_0, t_1]}(t) \define \,\indicate{[t_0, t_1]}\!(t)\ \Vg(t)
\end{equation}
and that for all $n$ times we have
\begin{equation}
  \Vg(t) = \Vg_{[0, t_1)}(t) + \Vg_{[t_1, t_2)}(t) + \dots + \Vg_{[t_n, \infty)}(t).
\end{equation}
As this is true for any element $\hilbert$, we have the natural decomposition,
\begin{equation}\label{chQuLight:eq:continousTimeSingleHilbert}
  \hilbert \cong \Hilbert_{[0, t_1)} \oplus \Hilbert_{[t_1, t_2)} \oplus \dots \oplus \Hilbert_{[t_n, \infty)}.
\end{equation}
where $\Hilbert_{[t_i, t_{i+1})}$ is the space of square integrable vector valued functions of dimension $d$ defined on the interval $[t_i, t_{i+1})$.
Because of this decomposition, proving that $\Fock(\hilbert)$ satisfies the tensor decomposition now means proving unitary equivalence
\begin{equation}
  \Fock\Big(\,\bigoplus_{i = 1}^n \hilbert_{[t_i, t_{i+1})}\,\Big) \cong \bigotimes_{i= 1}^n \Fock(\hilbert_{[t_i, t_{i+1})}).
\end{equation}
This is easily shown by first noting that
\begin{equation}
    \inprod{\Vg_{[t_i, t_{i+1})}}{\Vf_{[t_j, t_{j+1})}} = \delta_{i,j}\, \inprod{\Vg_{[t_i, t_{i+1})}}{\Vf_{[t_i, t_{i+1})}}.
\end{equation}
This however implies that for the exponential vectors $\e[\Vg]$ and $\e[\Vf]$
\begin{equation}
    \braket{\e[\Vg]}{\e[\Vf]} = \exp(\inprod{\Vg}{\Vf}) = \prod_{j = 1}^n \exp\left(\inprod{\Vg_{[t_i, t_{i+1})}}{\Vf_{[t_i, t_{i+1})}}\right).
\end{equation}

If we define the transformation $V:\, \Fock(\hilbert) \rightarrow \Fock(\hilbert_{[0, t_1]})\otimes \dots \otimes \Fock(\hilbert_{[t_n, \infty)})$  such that
\begin{equation}
    V\, \e[\Vg] = \e[\Vg_{[0, t_{1})}]\otimes\dots\otimes\e[\Vg_{[t_n, \infty)}],
\end{equation}
The inner product between two transformed vectors is then
\begin{equation}
    \braket{V\,\e[\Vg]}{V\,\e[\Vf]} = \prod_{j = 1}^n e^{\inprod{\Vg_{[t_i, t_{i+1})}}{\Vf_{[t_i, t_{i+1})}}} = e^{\inprod{\Vg}{\Vf}}.
\end{equation}
This shows that $V$ is a unitary transformation between the exponential vectors.  However because the exponential vectors are dense in the symmetric Fock space, $V$ linearly extends to any vector in $\Fock(\hilbert)$.  \qedhere
\end{proof}

\subsection{The quantum Wiener process \label{chQuLight:sec:QuWienerCharacteristic} }
In \citep{bouten_introduction_2007}, \citeauthor{bouten_introduction_2007} given an elegant derivation of how the quadratures
\begin{equation} \label{chQuLight:eq:QtPtQuardature}
  Q^{i}_t \define A^{i}_t + A^{i\,\dag}_t  \quad \text{and} \quad P^{i}_t \define i\left(\, A^{i\,\dag}_t -  A^{i}_t\, \right),
\end{equation}
have the statistics of a Wiener processes, when the field is in the vacuum state.   For the sake of completeness we reproduce this derivation here.

In Sec. \ref{chQuLight:sec:WeylOperators}, we introduced $\ahat[\Vg]$ and $\ahat^\dag[\Vg]$ though the generators of the coherent state $\psi[\Vg]$ through the relation, $\psi[\Vg] = \weyl[\Vg]\ket{\vac} = \exp(\ahat^\dag[\Vg] - \ahat[\Vg])\, \ket{\vac}$.  The argument of this displacement operator defines a Hermitian generator
\begin{equation}
  \Upsilon[\Vg] \define i \left(\ahat^\dag[\Vg] - \ahat[\Vg]\right)
\end{equation}
so that $\exp(\ahat^\dag[\Vg] - \ahat[\Vg]) = \exp(-i \Upsilon[\Vg])$.

For a classical random variable $\operatorname{x}$ then
\begin{equation}
    \varphi_{\operatorname{x}}(\kappa)\define \mathbbm{E}\Big(\exp(i \kappa \,\operatorname{x})\, \Big )
\end{equation}
is the characteristic function for that random variable and therefore characterizes its statistics.  The Weyl operator $\exp(-i \Upsilon[\Vg])$ is nearly equivalent to the characteristic function, up to the constant $\kappa$ and a minus sign.  However though the anti-linear property of $\ahat[\Vg]$ we have $\ahat[\lambda\,\Vg] = \lambda^\ast \ahat[\Vg]$, but if $\lambda = - \kappa$, (real $\kappa$) then $\exp(-i \Upsilon[-\kappa \Vg]) = \exp(+i \kappa \Upsilon[ \Vg ] )$.   Converting this operator into a true characteristic function simply means taking an expectation value with respect to the field state.  If the state is in a coherent state, $\psi[\Vf]$, then
\begin{equation}
    \varphi_{\Upsilon[\Vg]}(\kappa) \define \expect{\exp(+i \kappa \Upsilon[ \Vg] )}_{\psi[\Vf]},
\end{equation}
which characterizes the statistics of the operator $\Upsilon[\Vg]$.  In terms of the Weyl displacement operators this means that
\begin{equation}
    \varphi_{\Upsilon[\Vg]}(\kappa) = \braOket{\psi[\Vf]}{\weyl[-\kappa\, \Vg]}{\psi[\Vf]} = e^{-\norm{\Vf}^2} \braOket{\e[\Vf]}{\weyl[-\kappa\, \Vg]}{\e[\Vf]}.
\end{equation}
Eq. (\ref{chQuLight:eq:weylExponential}) relates the action of the Weyl operator to the exponential vector showing that this simplifies to
\begin{equation}
\begin{split}
    \varphi_{\Upsilon[\Vg]}(\kappa)& = \exp(-\norm{\Vf}^2 + \kappa \inprod{\Vg}{\Vf} - \kappa^2\,\norm{\Vg}^2/2 ) \braket{\e[\Vf]}{\e[-\kappa\, \Vg + \Vf]}\\
    & = \exp(-\norm{\Vf}^2 + \kappa \inprod{\Vg}{\Vf} - \kappa^2\,\norm{\Vg}^2/2 -\kappa \inprod{\Vf}{\Vg} + \norm{\Vf}^2)\\
    & = \exp( i\, \kappa\, 2\, \operatorname{Im}\, \inprod{\Vg}{\Vf} - \kappa^2\,\norm{\Vg}^2/2 ).
\end{split}
\end{equation}
The final line is recognizable as the characteristic function of a Gaussian random variable of mean  $2\, \operatorname{Im}\, \inprod{\Vg}{\Vf}$ and variance $\norm{\Vg}^2$.
Note that when $\Vg = \indicate{[0, t)}\hspace{-4 pt}(t)\, \Ve_j$,
\begin{equation}
  \Upsilon[ \indicate{[0, t)}\hspace{-4 pt}(t)\, \Ve_j ] = i \left(\, A^{i\,\dag}_t -  A^{i}_t\, \right) = P^{i}_t
\end{equation}
and
\begin{equation}
  \Upsilon[ - i \indicate{[0, t)}\hspace{-4 pt}(t)\, \Ve_j ] =   A^{i\,\dag}_t  + A^{i}_t  = Q^{i}_t.
\end{equation}
In either case, $\norm{\Vg}^2 = t$ and so that both operators have variance $t$, regardless of the coherent amplitude $\Vf$ of the underlying state.

When the state of the field is in vacuum ($\Vf = 0$) then both quadratures are mean zero, Gaussian random variables whose variance is given by $t$.  Lemma \ref{chQuLight:lemma:tensorDecomposition} also shows that the operator $\Upsilon[ \indicate{[s, t)}\hspace{-4 pt}(t)\, \Ve_j ] = P^{i}_t - P^{i}_s$ is a generator of displacements in the Fock space $\Fock(\hilbert_{[s, t)})$ and therefore commutes with any generator for states in $\Fock(\hilbert_{[0, s)})$.  Clearly the vacuum respects the continuous tensor product decomposition and therefore the quantum stochastic processes $\set{Q^{i}_t}_{t \ge 0}$ and $\set{P^{i}_t}_{t \ge 0}$ have, in vacuum expectation, the statistics of Wiener processes.

\subsection{The units of quantum noise}
An extremely observant reader might be concerned about the quadrature definitions of $Q_t$ and $P_t$ given in Eq. (\ref{chQuLight:eq:QtPtQuardature}).  The issue lies in the units and physical interpretation of the single particle wave functions.  In the strictest sense of second quantization, the normalized vector $\Vf \in  \mathcal{L}^2(\R^+)\otimes \Cn{d}$ are single particle wave functions whose square represents the probability density for observing the particle at some point in its domain.  In order for $\Vf$ to be a square normalized density, $\int dt \abs{\Vf}^2 = 1$, means that $\Vf$ must have units of $\frac{1}{\sqrt{\text{time}}}$.  Consequently the operators $\ahat[\Vf]$ and $\ahat^\dag[\Vf]$ must be unitless as their commutator is subsequently unitless.  However as written, the quadratures $Q^j_t$ and $P^j_t$ have the commutation relation
\begin{equation}
  \left[ Q^j_t,\, P^j_t \right] = 2 i\, t,
\end{equation}
which clearly has units of time.  The solution to this distinction is realize that when defining $A_t$ we should really be considering the field operators relative to some characteristic rate $\gamma$.  Through the linearity of $\ahat[\Vf]$, we clearly have
\begin{equation}
  \ahat[\sqrt{\gamma} \indicate{[0, t]} ] = \sqrt{\gamma}\, A_t.
\end{equation}
The whole point of defining $A_t$ in this ways is that regardless to the magnitude of $\gamma\, t$, the scaled quadrature $\sqrt{\gamma}\, Q_t$ will still have the statistics of a Brownian motion, simply with the diffusion rate $\gamma$.

The objective of Sec. \ref{chQuLight:sec:measureableSpace} was to identify on what scales we could treat a quasi-monochromatic field to be statistically independent for independent increments of time.  Fig. \ref{chQuLight:fig:PulseScaling} showed the scaling of a smoothed characteristic function $\indicate{[0, \tau]}^{(\sigma)}\hspace{-4 pt}(t)$ for a fixed smoothing variance and a variable duration $\tau$.  The act of smoothing limited the derivative to be at most on the order of $1/\sigma$ and for $\tau \sim 10^3\, \sigma$ this had little effect on the visual appearance of the smoothed function.  Note that the actual correction term to inner product in Eq. (\ref{chQuLight:eq:paraxialInnerProduct}) compared the rate of change of the temporal envelope to the carrier frequency $\omega_0$ and the introduction of the smoothing distribution is to simply limit this derivative.  Assuming that $\gamma$ represents the rate of diffusion then for times $\tau \sim 1/\gamma$ any pulse should appear to have a discontinuous derivative on this scale.  In other words, in order to treat an optical field as generating a quantum Wiener process we must have
\begin{equation}
  \gamma \ll \sigma^{-1} \ll \omega_0.
\end{equation}
In any realistic application, the physics of the system typically sets values for $\gamma$ and $\omega_0$.  In an atomic physics context the carrier frequency is usually a dipole allowed optical transition, leading to $\omega_0 \sim 2 \pi \times 100\text{ THz}$.  For such a transition the measurement timescale is on the order of the lifetime of the excited state, $t_{\text{decay}} \sim 10\, \text{ ns}$, meaning that typically, $\gamma \sim 2 \pi \times 10\text{ MHz}$.  This leaves 7 orders of magnitude between these two scales.  In many atomic systems this is actually an upper bound on the measurement rate.  Typically one will consider off resonant light leading to a significantly slower diffusion rate.  In term of this dissertation, Chap. \ref{chap:QubitState} applies a quantum stochastic treatment to an idealized model of the Faraday interaction, where $\gamma$ is reduced by a factor of one over this frequency difference.  This will be discussed in detail in Sec. \ref{chQuLight:sec:faraday}.  Before considering this specific model we will review how to transition from a smooth deterministic Schr\"odinger equation, to one involving quantum stochastic integral with respect to $A_t$ and $A_t^\dag$.

\section{Systems Interacting with Quantum Noise \label{chQuLight:sec:QuantumWhiteNoise} }

Much of this chapter has alluded to coupling a quantum system of interest to a traveling wave field and describing the resulting evolution in terms of a quantum stochastic process.  In the quantum optics literature a system coupled to delta commuting field operators have been discussed since the work of \citeauthor{gardiner_input_1985}, if not before \cite{gardiner_input_1985}.  The field operators are typically defined as the Fourier transform of one-dimensional operators quantizing a continuous spectrum of harmonic oscillators.  That is, from the operators $\ahat(\omega)$ and $\ahat^\dag(\omega)$,  $[\ahat(\omega), \ahat^\dag(\omega')] = \delta(\omega - \omega')$,
\begin{equation}\label{chQuLight:eq:gardinerQuWhiteNoise}
  \ahat(t) \define \frac{1}{\sqrt{2 \pi} }  \int_{\infty}^{\infty} d\omega\ e^{+i \omega t} \, \ahat(\omega).
\end{equation}
It then follows that $[\ahat(t),\ahat^\dag(t')] = \delta(t'-t)$.  From these operators one typically formulates the noncommuting quadratures, $\hat{q}(t) =  \ahat^\dag(t) + \ahat(t)$, and  $\hat{p}(t) = i\, \ahat^\dag(t) - i \, \ahat(t)$.  If the state of the field is specified to be in vacuum, then $\expect{\hat{q}(t)}_\vac = 0$ and $\expect{\hat{q}(t)\hat{q}(t')}_\vac = \delta(t - t')$.  In other words, in vacuum expectation $\hat{q}(t)$ has the statistics of white noise and the same is true for $\hat{p}(t)$ and any rotated combination of the two.  This is why the field operator $\ahat(t)$ is typically given the designation as quantum white noise.

A system is introduced to the problem typically through a linear interaction Hamiltonian where after making a couple of approximations vary much in line with our assumed separation of timescale, the interaction Hamiltonian reads \cite{gardiner_quantum_2004}
\begin{equation}\label{chMath:eq:gardinerHamiltonian}
  H_{int}(\lambda, t) = i \hbar \sqrt{\gamma} \left( \ahat^{\dag}(\lambda, t)\, \hat{c} - \ahat(\lambda, t)\, \hat{c}^\dag \right)
\end{equation}
where $\hat{c}$ is a generic system operator and $\ahat(\lambda, t)$ is an operator that limits to $\ahat(t)$ as $\lambda \rightarrow 0$ \footnote{Actually \citeauthor{gardiner_input_1985} consider the limit in the frequency domain and so they consider a bandwidth $\theta$ that approaches infinity.  For our purposes it is more convenient to consider $\lambda \rightarrow 0$, but the spirit is the same.}.  To arriving at $H_{int}(\lambda, t)$, a transformation to an interaction picture was made and all time dependence was associated with either $\ahat(\lambda, t)$ or its adjoint.  In this interaction picture, the joint state of the system and field evolves under a unitary propagator $U(\lambda,t)$, which satisfies the equation
\begin{equation}
  \frac{d}{dt} U(\lambda,t) = -\tfrac{i}{\hbar} H_{int}(\lambda,t)\, U(\lambda,t).
\end{equation}
It is well known that for a general time-dependent Hamiltonian the resulting propagator is given by a time ordered exponential
\begin{equation}\label{chQuLight:eq:UtTimeOrdered}
  U(\lambda, t) = \vec{\mathcal{T}} \exp \left(-\frac{i}{\hbar}\int_0^t ds\, H_{int}(\lambda, s) \right).
\end{equation}
The ultimate goal is to interpret the operator $\lim_{\lambda \rightarrow 0} U(\lambda, t) \define U_t$ as a solution to an equivalent quantum stochastic differential equation.  Appendix \ref{app:QSDEs} reviews the mathematical background needed to fully discuss these kinds of equations in the language of quantum stochastic processes acting in terms of a second quantized Fock space.

In the textbook formulation of quantum noise, \citeauthor{gardiner_quantum_2004} define a form of quantum It\=o calculus that is explicitly tide to the statistical properties of the state of the field and considers only a small family of Gaussian states \cite{gardiner_quantum_2004}.  Conversely, the quantum It\={o} calculus defined by \citeauthor{hudson_quantum_1984} has an It\={o} rule that is \emph{independent} of the field state and the proof of convergence holds over a large domain of possibly correlated system field states \cite{hudson_quantum_1984}.  The chief distinction between the two formulations is that the quantum optics derivation is based in a specific model, one that does not initially assume the structure necessary for the more abstract version.

From the point of view of a physicist trying to model a quantum system starting from an interaction Hamiltonian and the Schr\"odinger equation, it is not at all clear how and when that fits into the abstract quantum It\={o} calculus.  While \citeauthor{hudson_quantum_1984} gave criteria for when a quantum stochastic differential equation describes a unitary process, they did not specify how one should arrive at such a process from a Schr\"odinger equation and an approximating principle.  This is exactly what the quantum optics derivation provides, however without making the final connection to the state independent It\={o} calculus.  In \citeyear{accardi_weak_1990}, \citeauthor{accardi_weak_1990} proved this connection where they showed how a unitary propagator generated from a linear Hamiltonian similar to Eq. (\ref{chMath:eq:gardinerHamiltonian}) converged to an It\={o} integral as specified by \citeauthor{hudson_quantum_1984}.  Additionally they showed that the proof holds in ``the weak sense of matrix elements.''  What this means is the following.

Suppose we are given two coherent states with smoothed wave packets $\psi[\Vg_1^{(\sigma)}]$ and $\psi[\Vg_2^{(\sigma)}]$ and that when $\sigma \rightarrow 0$ we have the equivalent discontinuous coherent states $\psi[\Vg_1']$ and $\psi[\Vg_2']$.  Furthermore suppose we are given a stochastic processes $\set{X_{t}^{(\sigma)}}_{t \ge 0}$ that contains both system and field operators.  Then $X_{t}^{(\sigma)} \rightarrow X_t$ in the weak sense of matrix elements if, for arbitrary system state vectors $\phi_1$ and $\phi_2$,
\begin{equation}
  \lim_{\substack{\sigma \rightarrow 0\\ \sigma\omega_0 \rightarrow \infty} } \braOket{\phi_1\otimes \psi[\Vg_1^{(\sigma)}]}{X_{t}^{(\sigma)}}{\phi_2\otimes \psi[\Vg_2^{(\sigma)}]} = \braOket{\phi_1\otimes \psi[\Vg'_1]}{X_{t}}{\phi_2\otimes \psi[\Vg_2']}.
\end{equation}
It's worth noting \citeauthor{accardi_weak_1990} shows that the limit holds for the bear propagator $U(\lambda, t)$ as well as for the Heisenberg evolution of a system operator, \ie $U^\dag(\lambda, t) X U(\lambda, t)$.

We mention this here because as long as the total system-field state $\rho_\tot$ can be represented in terms of the matrix elements
\begin{equation*}
    \lim_{\substack{\sigma \rightarrow 0\\ \sigma\omega_0 \rightarrow \infty} } \braOket{\phi_1\otimes \psi[\Vg_1^{(\sigma)}]}{\rho_\tot}{\phi_2\otimes \psi[\Vg_2^{(\sigma)}]},
\end{equation*}
then the quantum stochastic representation is appropriate.  Clearly these states are much more complex than simply mean-zero Gaussian states as they can represent entangled states between the system and the quasi-monochromatic field as well as nonclassical field states such as superpositions between modes and even single photon states.  However not every field state is included, as we must allow for discontinuous yet quasi-monochromatic matrix elements\footnote{For defining a reasonable calculus it is also required that the wave packet amplitudes take on a large but finite maximum value.}.
In other words, the total state $\rho_\tot$ must be compatible with the approximations that make the stochastic representation possible to begin with.

Here we are also interested in moving beyond an interaction Hamiltonian that is linear in $\ahat^\dag(\lambda,t)$ and $\ahat(\lambda,t)$.  This is because the Faraday interaction is fundamentally quadratic in the field operators as it describes the scattering of light in one polarization state to another, see Sec. \ref{chQuLight:sec:faraday}.  Fortunately, in 2006 \citeauthor{gough_quantum_2006} extended the results of \citeauthor{accardi_weak_1990} to include scattering or conservation interactions.  In classical stochastic calculus the conversion between a smooth ordinary differential equation and a stochastic differential equation is the Wong-Zakai theorem.  Not surprisingly, the conversion between a smooth Schr\"odinger equation and quantum stochastic differential equation for the propagator is called the quantum Wong-Zakai theorem in \cite{gough_quantum_2006}.

The specific Hamiltonian that the quantum Wong-Zakai theorem considers and the one that we will use here is the interaction Hamiltonian, (sum over repeated indices with $i,j = 1\dots d$)
\begin{equation} \label{chQuLight:eq:QuWongZakaiHamiltonian}
  H_{int}(\lambda, t) = \hbar \left(  E_{ij}\, \ahat^\dag_i(\lambda, t)\, \ahat_j(\lambda, t) + E_{i0}\, \ahat_i^\dag (\lambda, t) +   E_{0j}\, \ahat_j(\lambda, t) + E_{00} \right)
\end{equation}
where $\set{E_{\alpha \beta} :\  \alpha, \beta = 0, \dots, d}$ are bounded operators acting on a system Hilbert space $\Hilbert_{sys}$.
Each term in $H_{int}(\lambda, t)$ physically represent the following:
\begin{itemize}
  \item $E_{00}$ is an operator acting solely on system degrees of freedom, independent of the bosonic modes, with units of frequency, \emph{e.g.} it could be what remains of the free system Hamiltonian after transforming to an interaction picture.
  \item $E_{i0}$ is a system operator that accompanies the creation of an excitation in the $i^{th}$ bosonic mode centered at time $t$.  A canonical example would be an operator proportional to an atomic lowering operator, with units of $1/\sqrt{time}$.
  \item $E_{0j}$ is a complementary process where, at time $t$, an excitation in the $j^{th}$ mode is removed.
  \item $E_{ij}$ is a unitless system operator weighting an instantaneous scattering of quanta from the $j^{th}$ mode to the $i^{th}$.  When $i = j$ this can be interpreted as a system coupling to the number of quanta in that mode at time $t$.
\end{itemize}
Note that as this Hamiltonian is required to be self-adjoint, the system operators must satisfy the constraint $E_{\alpha\, \beta} = E_{\beta\, \alpha}^\dag$.

Adding the quadratic term has an interesting and slightly unexpected effect on the physics.  Again the term $E_{ij}\, \ahat^\dag_i(\lambda, t)\, \ahat_j(\lambda, t)$ represents the instantaneous transfer of a photon from mode $j$ to mode $i$.  However for $\lambda > 0$, $\ahat^\dag_i(\lambda, t)$ has temporal extent, meaning that its possible for the system to interact again with the scattered quanta.  If the magnitude of $E_{ij}$ is relatively small then the possibility of re-interaction maybe relatively small, but we have yet to impose any such constraint.  Fortunately the whole problem of converting from the equation $\dot{U}(\lambda,t) = -\tfrac{i}{\hbar} H_{int}(\lambda,t)\, U(\lambda,t)$ to a quantum stochastic differential equation, including the possibility of multiple scattering events was solved by Gough.  Some of the details of this derivation is reviewed in Appendix \ref{app:QuWongZakai}.

Intimately related to this conversion is how the formal quantum It\=o integral is related to the operator ordering of the constituent field operators.  There exists a fundamental connection between the rules of quantum It\={o} calculus (see Appendix \ref{app:QSDEs}) and whether an iterated integral containing a sequence of field operators are either time or normally ordered.  This connection was also formalized by Gough in \cite{gough_quantum_2006}, which Appendix \ref{app:QuWongZakai} reviews. The bottom line is that in order for the limit of the time ordered exponential in Eq. (\ref{chQuLight:eq:UtTimeOrdered}) to be interpreted as a solution to an equivalent quantum It\={o} stochastic integral it must be put into \emph{normal order} with all of the annihilation operators to the right of the creation operators. The effects of the multiple scattering events become mathematically apparent when converting from the time ordered solution to a normally ordered form.

Before we are able to fully write down the resulting propagator we must address two important issues.  The first is to concretely link the operator $\ahat_i(\lambda, t)$ to the wave packet theory introduced in this chapter as well consider a kind of field operator wholly different from what we have considered up to this point.

\subsection{Quantum white noise in paraxial wave packets\label{chQuLight:sec:paraxialQuNoise}}

There are several different derivations that lead to bosonic operators $\ahat(\lambda, t)$ and $\ahat^\dag(\lambda, t)$ that result in calling the object $\lim_{\lambda \rightarrow 0} \ahat(\lambda, t)$ quantum white noise.  For instance \citeauthor{gardiner_quantum_2004} use a wide-bandwidth limit where they assume that the interaction Hamiltonian $H_{int}$ is initially specified in the frequency domain and that the system couples preferentially frequencies centered at a large transition frequency $\omega_0$.  They then assume that the coupling between the system and the operators $\ahat(\omega)$ is nearly flat in a frequency band centered at $\omega_0$.  When this flat coupling band is sufficiently wide, the effect is for the system to be interacting with an operator representing a white spectrum and is therefore delta correlated \cite{gardiner_quantum_2004}.  \citeauthor{accardi_quantum_2002} take a different approach by considering a weak-coupling/long-time limit.  The weak coupling implies that on short ``optical'' timescales the system field interaction can be considered perturbatively but that on longer ``mesoscopic'' times the aggregate effect is nontrivial and the field fluctuations develop a diffusive characteristic.  Though a subtle re-scaling of time, a field operator $\ahat(\lambda, t)$ emerges and is delta commuting as $\lambda \rightarrow 0$ \cite{accardi_quantum_2002}.  Rather than fixing ourselves to a specific system-field interaction, this work has focused instead on integrating the language of second quantization and stochastic processes with a realistic description of classical optics meaning that neither model fully fits our needs.    Instead we hope to find a description that is not tide to a specific model but capture the spirit of each.

Consider the time ordered exponential
\begin{equation}
 U(\lambda, t) = \vec{\mathcal{T}} \exp \left(-\frac{i}{\hbar}\int_0^t ds\, H_{int}(\lambda, s) \right)
\end{equation}
where $H_{int}(\lambda, t)$, is given in Eq. (\ref{chQuLight:eq:QuWongZakaiHamiltonian}).  Expanding the exponential to just two terms shows us that
\begin{equation}\label{chQuLight:eq:timeOrderedUtexpanded}
\begin{split}
  U(\lambda, t) =& \ident - \tfrac{i}{\hbar}\int_0^t ds H_{int}(\lambda, s) + \dots\\
  =& \ident -i  E_{ij}\,\int_0^t ds\ \ahat^\dag_i(\lambda, s)\, \ahat_j(\lambda, s) -i  E_{i0}\int_0^t ds\ \ahat_i^\dag (\lambda, s) \\
  &\quad -i E_{0j}\int_0^t ds\ \ahat_j(\lambda, s) -i  E_{00} \int_0^t ds + \dots.
\end{split}
\end{equation}
Now consider the creation term, $E_{i0}\, \ahat_i^\dag(\lambda, s)$.  Physically, this means that localized in an time interval near time $s$, create an excitation in the $i^{th}$ field mode and while you're at it, apply the system operator $E_{i0}$.  Suppose the joint system was in a pure product state $\ket{\psi_i}\ket{\vac}$ and that $\ket{\psi_i}$ happens to be an eigenstate of $E_{i0}$ with eigenvalue $h_i$.  Then the time integral over this term acting on this state gives
\begin{equation}
  -i  E_{i0} \int_0^t ds\ \ahat_i^\dag (\lambda, s) \ket{\psi_i}\ket{\vac} =  -i  h_i \int_0^t ds\ \ahat_i^\dag (\lambda, s) \ket{\psi_i} \ket{\vac}.
\end{equation}
This integral looks \emph{almost} like a smoothed creation operator for a wave-packet with temporal envelope function $h_i(t) = -i\, h_i \indicate{[0, t]}\hspace{-4pt}(s)$ acting on vacuum.  Specifically, Eq. (\ref{chQuLight:eq:SmoothWavePacketOperator}) gives the expression for the one-dimensional smooth wave packet $\ahat[\Vg^{(\sigma)}(\Vk, t)]$.  Taking the adjoint of that equation results in
\begin{equation}
  \ahat^\dag[\Vg^{(\sigma)}(\Vk, t)] = \frac{\norm{\Vg}}{\sqrt{\tau}}\, \int_0^\infty ds\ f(s)\, e^{-i \omega_0 s}\, \ahat^\dag[\V{\varphi}^{(\sigma)}(t - s)].
\end{equation}
The major distinction between this operator and the integral $-i  h_i \int_0^t ds\ \ahat_i^\dag (\lambda, s)$ is the existence of the carrier phase $e^{-i \omega_0 s}$.  This means that in order to interpret the integral $-i  E_{i0} \int_0^t ds\ \ahat_i^\dag (\lambda, s)$ as creating a system dependent extended single photon state, the interaction must be inherently phase modulated at the carrier frequency $\omega_0$.   But this is exactly the same statement as \citeauthor{gardiner_input_1985} when they assume that the system interacts with the bath at a large characteristic frequency.    Therefore without specifying a detailed interaction model we can say that
\begin{equation}
  \ahat_i^\dag (\lambda, s) \cong e^{-i \omega_0 s}\, \ahat^\dag[\V{\varphi}^{(\sigma)}_i( - s)]
\end{equation}
and
\begin{equation}
  \ahat_j (\lambda, s) \cong e^{+i \omega_0 s}\, \ahat[\V{\varphi}^{(\sigma)}_j( - s)]
\end{equation}
for some system characteristic frequency $\omega_0$ and smoothing wave packets $\V{\varphi}^{(\sigma)}_i(\Vk, s)$ and  $\V{\varphi}^{(\sigma)}_j(\Vk, s)$.  The presence of the time reversal might be a little puzzling at first, but this is simply due to the fact that this wave packet was defined with respect to a convolution, which always time reverses one of the two functions.   One possible mapping to include a set of $d$ distinct modes is to assume the model considers a set of $d$ paraxial spatial mode functions $\V{u}_{i}^{(+)}(\Vx_T, z)$, which satisfy the orthogonality relation, $\int d^2 x_T\, \V{u}^*_{i}(\Vx_T, z)\cdot \V{u}_{j}(\Vx_T, z) = \delta_{ij}\, \sigma_T$.

To complete this discussion we should identify what the parameter $\lambda$ means in the wave packet context.   We are able to explicitly compute the unequal time commutator from Eq. (\ref{chQuLight:eq:paraxialInnerProduct}) and remembering that for the smoothing kernel $\norm{\Vg}/\sqrt{\tau} = 1$.  This results in
 \begin{equation}\label{chQuLight:eq:paraxialQuWhiteNoiseCommutator}
    \begin{split}
      \left[\ahat_i(\lambda,t),\,\ahat_j^\dag(\lambda,t') \right] &=e^{+i \omega_0 (t-t')} \inprod{\V{\varphi}^{(\sigma)}_i(-t)}{\V{\varphi}^{(\sigma)}_j(-t')}\\
      &= \delta_{ij}\,\left( \varphi^{(\sigma)} \star \varphi^{(\sigma)}\ (t - t') - i\frac{1}{\omega_0} \frac{d\varphi^{(\sigma)} }{dt}\star \varphi^{(\sigma)}\ (t - t') \right).
    \end{split}
\end{equation}
$\lambda$ is simply a parameter representing the formal limit that as $\lambda \rightarrow 0$, $\sigma \rightarrow 0$ and $(\sigma\, \omega_0)^{-1} \rightarrow 0$.

\subsection{The scattering process\label{chQuLight:sec:Lambda}}
Up until now, the only kind of field operator we have considered is a creation operator associated with a given single particle state $\Vg$.  While these operators are vitally important, it does leave out the possibility of a whole other class of field operators.  Eq. (\ref{chQuLight:eq:timeOrderedUtexpanded}) expanded the time order exponential for $U_t$ to first order, which included the integral $\int_0^t ds\ \ahat^\dag_i(\lambda, s)\, \ahat_j(\lambda, s)$.  We will now show that this operator is quite different from the product of two smeared wave packet operators, particularly in the limit $\lambda \rightarrow 0$.

Consider the exponential vectors $\e[\Vf(\lambda)]$ and $\e[\Vh(\lambda)]$, $\Vf,\Vh \in  \mathcal{L}^2(\R+)\otimes\Cn{d}$, defined as
\begin{equation}
  \ket{\e[\Vf(\lambda)]} \define \exp\left(\int_0^\infty dt\ f_i(t) \ahat^\dag_i(\lambda, t) \right) \ket{\vac}.
\end{equation}
If $\int_0^t ds\ \ahat^\dag_i(\lambda, s)\, \ahat_j(\lambda, s)$ where some how equivalent to the product $\ahat^\dag[\Vg] \ahat[\Vg]$  as $\lambda \rightarrow 0$, for some wave packet $\Vg$ we would have the eigenvalue relationship
\begin{equation}
  \lim_{\lambda \rightarrow 0} \braOket{\e[\Vf(\lambda)]}{\ahat^\dag[\Vg(\lambda)]\ahat[\Vg(\lambda)]}{\e[\Vh(\lambda)]} = \inprod{\Vf}{\Vg}\inprod{\Vg}{\Vh}\ \braket{\e[\Vf]}{\e[\Vh]}.
\end{equation}
We can explicitly show that this is not the case.  To simplify the notation, we'll define the function
\begin{equation}
  c_{ij}(\lambda, t-s) \define [\ahat_i(\lambda, t),\, \ahat^\dag_j(\lambda, s)],
\end{equation}
which has the property that
\begin{equation}
  \lim_{\lambda \rightarrow 0} c_{ij}(\lambda, t-s) = \delta_{ij}\, \delta(t - s).
\end{equation}
Explicit calculation then shows that
\begin{multline}
 \hspace{-15 pt} \braOket{\e[\Vf(\lambda)]}{\int_0^t ds\ \ahat^\dag_i(\lambda, s)\,\ahat_j(\lambda, s) }{\e[\Vh(\lambda)]} =\\
  \hspace{-10 pt} \int_0^t ds \int_{0}^{\infty} ds_1\, f^*_\ell(s_1) c^*_{i \ell }(\lambda, t - s)\, \int_{0}^{\infty} ds_2\,  c_{j k}(\lambda, s - s_2) h_k(s_2)\ \braket{\e[\Vf(\lambda)]}{\e[\Vh(\lambda)]}.
\end{multline}
Then in the discontinuous limit,
\begin{multline}
  \lim_{\lambda \rightarrow 0} \braOket{\e[\Vf(\lambda)]}{\int_0^t ds\ \ahat^\dag_i(\lambda, s)\,\ahat_j(\lambda, s) }{\e[\Vh(\lambda)]} \\
   = \int_0^t ds\, f^*_i(s)\, h_j(s)\ \braket{\e[\Vf]}{\e[\Vh]}.
\end{multline}
In terms of an inner product on single particle wave vectors this is actually equal to,
\begin{multline}
  \lim_{\lambda \rightarrow 0} \braOket{\e[\Vf(\lambda)]}{\int_0^t ds\ \ahat^\dag_i(\lambda, s)\,\ahat_j(\lambda, s) }{\e[\Vh(\lambda)]} \\
   = \inprod{\Vf}{\indicate{[0, t]} \Ve_i \Ve_j \cdot \Vh}\ \braket{\e[\Vf]}{\e[\Vh]}.
\end{multline}

This is clearly not an eigenvalue relationship involving a product of wave packets.  In actuality it is a second quantization of an operator acting on the wave packets themselves \cite{barchielli_continual_2006,parthasarathy_introduction_1992}. The specific operator here is multiplication by the indicator function $\indicate{[0,t]}\hspace{-4pt}(s)$ and the dot product into the dyad of basis vectors $\Ve_i \Ve_j$.   In the language of the Hudson and Parthasarathy formulation of QSDEs, this operator is a scattering or conservation process, and is notated as $\Lambda^{ij}_t$.  While these processes can be derived without reference to a limiting integral, our purposes we simply take this to be a definition
\begin{equation}
  \Lambda^{ij}_t \define \lim_{\lambda \rightarrow 0} \int_0^t ds\ \ahat^\dag_i(\lambda, s)\,\ahat_j(\lambda, s) .
\end{equation}
Appendix \ref{app:QSDEs} reviews the notation and manipulation of QSDEs in terms of the quantum It\={o} differentials $dA^{i}_t$, $dA^{j\,\dag}_t$ and $d \Lambda^{ij}_t$.

\subsection{The limiting stochastic propagator\label{chQuLight:sec:quWongZakaiLimit}}

With a firm connection between the interaction Hamiltonian with scattering terms and the wave packet theory, we are now able to express the limiting Hamiltonian in terms of an It\={o} form quantum stochastic differential equation.  Appendix \ref{app:QSDEs} shows that the most general quantum stochastic integral usually considered defines a process,
\begin{equation}
  U_t = U_0 + \int_0^t d\Lambda^{ij}_s\, F^{ij}_s + \int_0^t d A^{i\, \dag}_s\, F^{i0}_s + \int_0^t d A^{j}_s\, F^{0j}_s + \int_0^t ds\, F^{00}_s.
\end{equation}
In Sec. \ref{appQSDEs:sec:dUt} it also shows what constraint must be placed on the operators $F^{\alpha \beta}_s$ in order for $U_t$ to be a unitary process.  As any unitary can be written as $\exp(-i A_t)$ for some generator $A_t$, instead of working with $F^{\alpha \beta}_s$ it is more convenient to define the operators $G^{\alpha\beta}_s$ so that
\begin{equation}
  F^{\alpha\beta}_s = G^{\alpha \beta}_s U_s.
\end{equation}
The unitary constraints written in terms of $G^{\alpha \beta}_s$ is given in Eq. (\ref{appQSDEs:eq:dUtConstraints}).

The bottom line result of the Quantum Wong-Zakai theorem is that the operators $G^{\alpha \beta}_s$ are expressible in terms of the system operators $E_{\alpha \beta}$ defining $H_{int}(\lambda, t)$ in Eq. (\ref{chQuLight:eq:QuWongZakaiHamiltonian}) a matrix of constants $\kappa_{ij}$.  As $E_{\alpha \beta}$ are assumed to be time independent this results in $G^{\alpha \beta}_s = G^{\alpha \beta}_0$ and so we will omit the time index and demote the superscripts to subscripts.  While the constants $\kappa_{ij}$ are in general complex the simplest of all cases is when $\kappa_{ij} = \half \delta_{ij}$.  Not only is this the simplest of cases it is also well motivated for our problem and so we will use it here, see Appendix \ref{appQuWZT:sec:gaugeFreedom}.    With these simplifications, the propagator $U_t$ is expressible as a quantum stochastic It\={o} integral, which solves the recursive QSDE,
\begin{equation}
  dU_t =  G_{ij}\, U_t\ d\Lambda^{ij}_t + G_{i0}\, U_t\ d A^{i\, \dag}_t + G_{0j}\,U_t\ d A^{j}_t + G_{00}\, U_t\ dt.
\end{equation}
The limiting coefficients $G_{\alpha \beta}$ are
\begin{equation}\label{chQuLight:eq:LimitdUCoeffs}
    G_{\alpha \beta} = -i E_{\alpha\beta} - \half E_{\alpha i} \left(\frac{1}{\ident + i\, \half \boldsymbol{\mathbbm{E}} }\right)_{ij}\,  E_{j \beta},
\end{equation}
where $i$ and $j$ start from 1 and we defined $\boldsymbol{\mathbbm{E}}$ as the matrix of operators $E_{ij}$.  The appearance of this matrix in the denominator is precisely due to the possibility of having multiple scattering events.  A Neumann series is the operator-valued generalization of a geometric series, so that for an operator $A$, $\sum_{n = 0}^\infty A^n = (1 - A)^{-1}$,  which is well defined whenever $1 - A$ is invertible.  The equivalent It\={o} coefficients, generates a Neumann series of operators where in this case $A = -i \half \boldsymbol{\mathbbm{E}}$ and $A^n$ represents a quantum scattering between the modes $n$ times.   The limiting coefficient then involves the $i,j$ component of the operator/matrix inverse $(\ident + i \half \boldsymbol{\mathbbm{E}})^{-1}$.    For an intuitive physical picture, the coefficients $G_{\alpha \beta}$ can be interpreted in the following way.

Each coefficient $G_{\alpha \beta}$ can be roughly thought of a right-to-left acting transformation occurring on the system, dependent upon on how it couples though the field.  The original, direct couplings $E_{\alpha\beta}$ are still present as shown in the first term in Eq. (\ref{chQuLight:eq:LimitdUCoeffs}).  In addition to the direct coupling, there are the effects of coupling thought the various modes.  As an example the second part of the $G_{i0}$ coefficient shows a photon can be created in the $i^{th}$ mode not just by just direct excitation, represented by the $-i E_{i0}$ term, but also by first exciting $j^{th}$ mode, and then scattering any number of times and then finally being emitted into the $i^{th}$.  The same goes for $G_{00}$ and $G_{ij}$ except these either leave the field unchanged or transfer a quantum from mode $j$ to mode $i$.

\subsection{A simple 1D example}
Nearly the simplest of all nontrivial examples of the quantum Wong-Zakai theorem is when $d = 1$, and
\begin{equation}
  E_{11} = 0, \quad E_{10} = i\, \sqrt{\gamma}\, D,\quad E_{01} = -i\,\sqrt{\gamma}\, D^\dag,\quad\text{and }\quad E_{00} = H_{sys}
\end{equation}
where $D$ and $H_{sys}$ are system operators.  In other words, these are the coefficients for a total Hamiltonian
\begin{equation}
  H_{int}(\lambda, t) = \hbar \left( i \sqrt{\gamma}\, D\, a^\dag(\lambda, t) - i \sqrt{\gamma}\, D^\dag\, a(\lambda, t) + H_{sys} \right).
\end{equation}
This is an extremely common model in quantum optics where an atomic dipole operator $D$ couples with the rate $\gamma$ to a quantized quasi-monochromatic electric field, with ``white noise'' creation operator $a^\dag(\lambda, t)$.  $H_{sys}$ is the remaining system operator which includes any residual detuning of the field mode from the system transition frequency or any externally applied controls.  The quantum Wong-Zakai theorem states that this pre-limit Hamiltonian generates a propagator with the coefficients
\begin{equation}
  G_{11} = 0, \quad G_{10} = \sqrt{\gamma}\, D,\quad G_{01} = - \sqrt{\gamma}\, D^\dag,\quad\text{and }\quad G_{00} = -i H_{sys} - \half \gamma\, D^\dag D.
\end{equation}
This results in the propagator $U_t$ satisfying the QSDE
\begin{equation}\label{chQuLight:eq:1Dpropagator}
  d U_t = \left( \sqrt{\gamma}\, D\, dA^{\dag}_t - \sqrt{\gamma}\, D^\dag\, dA_t -i H_{sys}\, dt - \half \gamma\, D^\dag D\, dt\right) U_t.
\end{equation}
Next section we will apply the quantum Wong-Zakai theorem to the much more interesting case of the Faraday interaction where $d = 2$ and $E_{ij} \ne 0$.

\section{The Faraday Interaction\label{chQuLight:sec:faraday} }
The Faraday interaction is physically based on an optical field propagating in a polarizable medium.  In classical optics it is used as a magneto-optical effect, where the polarization of a linearly polarized probe is rotated by an amount proportional to the component of the magnetic field parallel to the direction of propagation.  At a macroscopic level, it is modeled terms of the energy shift of a polarizable particle induced by an oscillating electric field.  Such an energy shift can be easily implemented and precisely controlled by applying a quasi-monochromatic laser to a rarefied monatomic gas.  The laser's effect on a single atom can be considered a perturbation to its atomic ground states, if the laser drive is in a ``low saturation'' regime.  Specifically the laser's carrier frequency must be close to, but significantly off resonance from, a ground state transition.  Furthermore the intensity must be small enough so that the total number of excited atoms will be negligibly small.  Not surprisingly, the details of the ground state atomic structure effects both the magnitude and direction of the induced polarization and often leads to effects beyond a simple linear rotation of the probe polarization.  The derivation of this kind of interaction is elegantly presented by \citeauthor{deutsch_quantum_2010} in \cite{deutsch_quantum_2010}. Here we will consider the simplest of all settings where the atomic ground state is given by a spin 1/2 particle.  Such a ground state is experimentally realizably, if the atom has a single valance electron and negligible hyperfine structure, or if two valence electrons form a spin singlet ground state and the nucleus has a total spin $I = 1/2$.

In either the classical or quantum mechanical setting, the polarizability Hamiltonian for an atom located at position $\Vr_{a}$ is given by
\begin{equation}
  H_{polar} = - \VE^{(-)}(\Vr_{a})\cdot\overleftrightarrow{\V{\alpha}}\cdot{\VE}^{(+)}(\Vr_{a})
\end{equation}
where $\overleftrightarrow{\V{\alpha}}$ is the polarizability tensor.  Quantum mechanically, $\overleftrightarrow{\V{\alpha}}$ is an operator acting solely on the atomic ground states.  It is worth noting that in any real system there will be additional decoherence due to spontaneous emission, which in this treatment we will ignore.  Shortly, we will note that the strength of the coherent interaction is proportional to $\Gamma/\Delta$ where $\Gamma$ is the excited state decay rate and $\Delta$ is the probe detuning from the atomic resonance.  One can additionally show that the incoherent photon scattering is proportional to $\Gamma/\Delta^2$.  When $\Gamma/\Delta$ is small, $\Gamma/\Delta^2$ is smaller and so the decoherence is often ignored.

The connection to the quantum Wong-Zakai theorem is that $H_{polar}$ is quadratic in the field operators.  As $\VE^{(+)}(\Vr_{a}) \propto  \ahat^\dag$ and $\VE^{(-)}(\Vr_{a}) \propto \ahat$, the components of $\overleftrightarrow{\V{\alpha}}$ will be identifiable with the operators $E_{ij}$.  By ignoring spontaneous emission we are also able to restrict our attention to a single quasi-monochromatic paraxial mode.  While in principle the atom couples to any electric field at its location, including the quantized vacuum, we will be applying a coherent displacement in a definite mode $\V{u}^{(+)}(\Vx_T, z)$, with an envelope function $f(t)$.  As we have seen this envelope function is expressible in terms a convolution with the smoothing function $\varphi^{(\sigma)}(t)$, see Sec \ref{chQuLight:sec:paraxialQuNoise}, and so the relevant field operators we will be considering are $\ahat(\lambda, t)$ and its adjoint.   In other words, we are simply using this as an example for applying all of the theoretical machinery developed in this chapter.

Before finally writing down $H_{int}(\lambda, t)$, we will make one more extremely useful but only marginally justifiable approximation.  Here we assume that the spatial distribution of the atoms will be irrelevant and all atoms in the ensemble can be treated as existing at the same location in space.  From the point of view of the slowly varying envelope, this is a reasonable assumption if the dimension of the gas along the direction of propagation is on the order of $c\, \sigma$.  If the longitudinal extent of the atoms becomes significant then $H_{int}$ would have to treat atoms at the beginning of the gas differently from the atoms at the end.  In addition to having a spatial-temporal dependence, any realistic paraxial beam will have some intensity variation in both $\Vx_T$ and $z$.  If we take $\V{u}^{(+)}(\Vx_T, z)$ to be a standard Hermite-Gaussian beam, the transverse and longitudinal intensity can be treated as approximately constant if the gas has a transverse area that is small when compared to the beams characteristic area $\sigma_T$.  However, if a beam of a fixed input power has a large transverse area then it will have a relatively low intensity at any give point relative to a beam with a smaller $\sigma_T$.  Ultimately this means that the more uniform the probe is, the weaker the over all interaction will be.

With all of the above caveats and assumptions the Faraday interaction Hamiltonian is
\begin{equation}\label{chQuLight:eq:vacuumFaraday}
    H_{int}(\lambda, t) = \hbar \frac{\chi_0}{3} J_z\,\left( \ahat_r^{\dag}(\lambda, t) \ahat_r(\lambda, t) - \ahat_l^{\dag}(\lambda, t) \ahat_l(\lambda, t) \right)
\end{equation}
with the operators and constants defined though the following.  For a spin 1/2 ground state, the single atom polarizability $\overleftrightarrow{\V{\alpha}}$ is diagonal in the circular polarization basis, $\Ve_r$ and $\Ve_l$.  Up to an irreverent global energy shift $\overleftrightarrow{\V{\alpha}} \propto \sigma_z ( \Ve_r^* \Ve_r - \Ve_l^* \Ve_l)$ where $\sigma_z$ is the pauli $z$ operator and the proportionality constant depends upon the specifics of the atomic physics. By assuming that all of the atoms exist at the same location in space, computing the Hamiltonian for the whole ensemble reduces to computing a sum over the individual polarizabilities which further reduces to summing over all of the $\sigma_z$ operators.   It is well known that with $N$ spin 1/2 particles, a collective pseudo-spin $\V{J}$ can be defined whose components ($i = x, y, z$) are
\begin{equation}
    J_{i} = \sum_{n = 1}^N \half \sigma^{(n)}_i.
\end{equation}
The details of finding the proper proportionality to express the total interaction as Eq. (\ref{chQuLight:eq:vacuumFaraday}) is given in \cite{deutsch_quantum_2010} but the dimensionless constant $\chi_0$ has an exceedingly simple form with,
\begin{equation}
  \chi_0 = \frac{\sigma_0}{\sigma_T}\, \frac{\Gamma}{2 \Delta}
\end{equation}
where $\sigma_0$ is the resonant scattering cross-section for the given transition.  This constant can be viewed as giving the probability that a single atom absorbs and remits a photon into the paraxial beam.  In order to make all the above approximations valid (\emph{e.g.} assuming that the intensity is near constant and that spontaneous emission is negligible) requires that both $\sigma_0 \ll \sigma_T$ as well as $\Gamma \ll \Delta$, meaning that $\chi_0 \ll 1$.

When proving that the convergence of $U(\lambda, t)$ in the sense of matrix elements $\lim_{\lambda \rightarrow 0} \braOket{\e[\Vf(\lambda)]}{U(\lambda, t)}{\e[\Vh(\lambda)]}$, Gough used the useful relationship that
\begin{equation}
  \lim_{\lambda \rightarrow 0} \braOket{\e[\Vf(\lambda)]}{U(\lambda, t)}{\e[\Vh(\lambda)]} = \lim_{\lambda \rightarrow 0} \braOket{\vac}{\tilde{U}(\lambda, t)}{\vac}
\end{equation}
where $\tilde{U}$ is the propagator, except with the replacements
\begin{equation}
    \ahat_i(\lambda, t) \rightarrow \ahat_i(\lambda, t) + h_i(t) \quad \text{and}\quad \ahat^\dag_i(\lambda, t) \rightarrow \ahat_i^\dag(\lambda, t) + f^*_i(t).
\end{equation}
As we are applying this limit when the field has a coherent displacement, it is sufficient to work with these displaced versions and pretend that the field is in the vacuum.   The heart of Faraday interaction is the rotation of a linearly polarized input and so we will assume that our displacement is linearly polarized and that the equivalent amplitude function $\Vf(t) \in \mathcal{L}^2\otimes\Cn{2}$ represents the displacement in the electric field.  With these two constraints we have that
\begin{equation}
    f_r(t) = f^*_l(t) = \tfrac{i}{\sqrt{2}}\, f_0(t)
\end{equation}
where $f_0(t)$ is real-valued. Typically experiments involving the faraday interaction, the driving displacement is operated in switched on to a constant value  for the duration of an experiment.  In this case $\max(f_0) = \sqrt{N_L/\tau} $ where $N_L$ is the average number of photons in a pulse of duration $\tau$.

In the case of this displacement, we have the effective vacuum Hamiltonian
\begin{multline}
    H_{int}(\lambda, t) = \hbar \frac{\chi_0}{3} J_z\,\Big( \big(\ahat_r^{\dag}(\lambda, t)  + f^*_r(t) \big)\big(\ahat_r(\lambda, t)  + f_r(t) \big) \\
        - \big(\ahat_l^{\dag}(\lambda, t)  + f^*_l(t) \big)\big(\ahat_l(\lambda, t)  + f_l(t) \big) \Big).
\end{multline}
It is useful to define the function,
\begin{equation}
    \kappa(t) = \left(\frac{\chi_0}{3} f_0(t)\right)^2.
\end{equation}
When $f_0(t)$ is held at a constant level we have a characteristic rate $\kappa = (\chi_0/3)^2\, N_L/\tau$.
From this definition we can expand out the interaction
\begin{equation}\label{chQuLight:eq:displacedFaraday}
\begin{split}
    H_{int}(\lambda, t) &= \hbar \tfrac{\chi_0}{3} J_z\,\left( \ahat_r^{\dag}(\lambda, t) \ahat_r(\lambda, t) - \ahat_l^{\dag}(\lambda, t) \ahat_l(\lambda, t) \right) \\
     &\ + i\,\hbar \sqrt{\tfrac{\kappa(t)}{2}}\, J_z \big( \ahat_r^{\dag}(\lambda, t) + \ahat_l^{\dag}(\lambda, t) - \ahat_r(\lambda, t) - \ahat_l(\lambda, t) \big).
\end{split}
\end{equation}
Sec. \ref{chQuLight:sec:QuantumWhiteNoise} discusses the operators $E_{ij}$ as representing the scattering of field quanta in a system dependent way.  In addition, the limiting coefficients show that when $E_{ij}$ takes on relatively large values, there is a possibility for multiple scattering events.  In the practical approximations that resulted in $\chi_0 \ll 1$, the probability for multiple scattering is relatively small, unless $J_z$ takes on obscenely large values.  Note that when $f_0(t) \gg 1$, there can still be a significant interaction, as this means that the second linear term dominates the interaction.   Also note that we have  $\frac{1}{\sqrt{2}}\left(\ahat_r(\lambda, t) + \ahat_l(\lambda, t)\right) = \ahat_h(\lambda, t)$, \emph{i.e.} an annihilation operator for horizontal polarization.  Dropping quadratic terms in favor of the terms with a large displacement, we now have an approximately linear interaction in a single horizontally polarized mode.
\begin{equation}
    H_{int}(\lambda, t) \approx  i\,\hbar \sqrt{\kappa(t)}\, J_z \big( \ahat_h^{\dag}(\lambda, t)  - \ahat_h(\lambda, t) \big).
\end{equation}
From the 1-D example this means that we have $G_{10} = \sqrt{\kappa(t)}\, J_z$ and so if we add a time-dependent system control Hamiltonian $H_c(t)$ to this expression we have a propagator
\begin{equation}\label{chQuLight:eq:dUtPolarimetry}
    d U_t = \left( \sqrt{\kappa(t)}\, J_z\, dA^{h\,\dag}_t - \sqrt{\kappa(t)}\, J_z\, dA^{h}_t  - \half \kappa(t)\, J_z^2\, dt -i H_c(t)\, dt \right) U_t.
\end{equation}
This is the propagator that we will be considering in Chaps. \ref{chap:projection} and \ref{chap:QubitState}.

It is worth noting that the Faraday interaction has been applied to several different continuous measurement models in the QSDE formalism with varying levels of initial assumptions \cite{van_handel_modelling_2005, bouten_scattering_2007}.  In \cite{van_handel_modelling_2005}, the free field was assumed to be well modeled by a QSDE and a Faraday like interaction was derived via adiabatically eliminating an excited state as well as an artificial cavity mode which left many questions unanswered as the derivation was made strictly through a single mode picture and did not address the fundamental two mode structure of the Faraday interaction.  In contrast \cite{bouten_scattering_2007} considered a scattering interaction, however there they simply took the fundamental scattering interaction before the displacement and substituted the scattering processes $d\Lambda^{rr}_t$ and $d \Lambda^{ll}_t$ for the white noise operators.  What that model failed to consider was the effect of normally ordering the field operators in obtaining the proper It\={o} correction.  In the language of the quantum Wong-Zakai theorem, the propagator initially considered in \cite{bouten_scattering_2007} should have been interpreted as a quantum Stratonovich equation and not a Quantum It\={o} equation.

\subsection{The quadratic Faraday interaction}
It would be a shame to discuss the full solution to the quantum Wong-Zakai theorem and not give an example that retains the scattering interaction.  The Faraday interaction is a prime candidate for this, in a world where we have a weak drive $f_0(t) \sim 1$ but it is possible to see some kind of effect.  From Eq. (\ref{chQuLight:eq:displacedFaraday}) we can identify
\begin{equation}
\begin{split}
  E_{rr} &= -E_{ll} = \tfrac{\chi_0}{3} J_z,\\
  E_{r0} &= E_{l0} = i \sqrt{\tfrac{\kappa(t)}{2}} J_z, \\
  E_{0r} &= E_{0l} = - i \sqrt{\tfrac{\kappa(t)}{2}} J_z, \quad \text{and}\\
  E_{00} &= 0.
\end{split}
\end{equation}
We can substitute these operators into the coefficients $G_{\alpha \beta}$ in Eq. (\ref{chQuLight:eq:LimitdUCoeffs}).  After some algebraic simplifications we find that
\begin{equation}\label{chQuLight:eq:FaradaydUCoeffs}
\begin{split}
  G_{rr} &=  G_{ll}^\dag = \frac{-i \tfrac{\chi_0}{3} J_z}{1 + i \tfrac{\chi_0}{6} J_z},\\
  G_{r0} &= -G_{0r} =  \frac{ \sqrt{ \tfrac{\kappa(t)}{2}}\,J_z}{1 + i \tfrac{\chi_0}{6} J_z}\\
  G_{l0} &= -G_{0l} =  \frac{\sqrt{ \tfrac{\kappa(t)}{2}}\, J_z}{1 - i \tfrac{\chi_0}{6} J_z}, \quad \text{and}\\
  G_{00} &= -\frac{\tfrac{\kappa(t)}{2}\, J_z^2}{1 +  \left(\tfrac{\chi_0}{6} J_z\right)^2}.
\end{split}
\end{equation}
App. \ref{app:QSDEs} reviews the usual formulation of the propagator $dU_t$, in terms of the operators $S_{ij}$, $L_i$ and $H$.  Some more simple algebra shows that
\begin{equation}\label{chQuLight:eq:FaradaydUSLH}
\begin{split}
  S_{rr} &=  S_{ll}^\dag = \frac{1 -i \tfrac{\chi_0}{6} J_z}{1 + i \tfrac{\chi_0}{6} J_z},\\
  S_{rl}&= S_{lr} = 0,\\
  L_{r} &= L_{l}^\dag = \frac{ \sqrt{ \tfrac{\kappa(t)}{2}}\,J_z}{1 + i \tfrac{\chi_0}{6} J_z}, \quad \text{and}\\
  H &= 0.
\end{split}
\end{equation}
In order for $dU_t$ to be a unitary process, $S_{ij}$ must form a unitary matrix of operators, \emph{i.e.} $S_{ij}^\dag S_{jk} = S_{ij} S_{jk}^\dag = \delta_{ik}$, which is clearly satisfied in this case.  So finally the Faraday interaction generates the propagator $U_t$, which solves the QSDE
\begin{multline} \label{chQuLight:eq:dUtFaraday}
 \hspace{-10pt}dU_t = \Big( (S_{rr} - 1)\, d\Lambda^{rr}_t + (S_{rr}^\dag - 1)\, d\Lambda^{ll}_t +  L_{r}\,( dA^{r\, \dag}_t - dA^{r}_t) + L_{r}^\dag\,( dA^{l\, \dag}_t - dA^{l}_t) - L_{r}^\dag L_r \Big) U_t
\end{multline}
with the initial value $U_0 = \ident$.  In the small $\chi_0$ limit we were able to write the propagator in terms the linear polarized field operators $A_t^{h}$ and $A_t^{h\,\dag}$.   This is not the case here as the right and left polarization states have different atomic coupling operators.  This might result in creating some system dependent elipticity to the probe laser, however more analysis is clearly needed.

The power of writing the propagator in terms of the $(S, L, H)$ parameters is that a large number of results have already been computed for general coefficients which can simply be applied here.  As an example suppose we wish to compute the unconditioned master equation of the atomic system assuming that the displaced field is in vacuum, \emph{i.e.} other than the coherent drive laser.  Then the master equation is given in Lindblad form with jump operators $L_{r}$ and $L_l$.  Specifically the system density operator $\rho(t)$ is the solution to
\begin{equation}
\begin{split}
  \frac{d \rho}{d t} &= \mathcal{D}[L_r](\rho) + \mathcal{D}[L_l](\rho)\\
  & = L_r\, \rho\, L_r^\dag - \half L_r^\dag L_r\, \rho - \half  \rho\, L_r^\dag L_r  + L_l\, \rho\, L_l^\dag- \half L_l^\dag L_l\, \rho - \half \rho\, L_l^\dag L_l \\
  & = \half\,\kappa(t) \left( \tfrac{ J_z}{1 + i \tfrac{\chi_0}{6} J_z}\, \rho\, \tfrac{ J_z}{1 - i \tfrac{\chi_0}{6} J_z}  + \tfrac{ J_z}{1 - i \tfrac{\chi_0}{6} J_z}\, \rho\, \tfrac{ J_z}{1 + i \tfrac{\chi_0}{6} J_z}  - \tfrac{ J_z^2}{1 +  \left(\tfrac{\chi_0}{6} J_z\right)^2}\, \rho - \rho \, \tfrac{J_z^2}{1 +  \left(\tfrac{\chi_0}{6} J_z\right)^2} \right).
\end{split}
\end{equation}
Here we can see yet again that when $\chi_0 \rightarrow 0$ we recover the standard dissipative master equation with measurement operator $\kappa(t)\, J_z$.  Also note that when the system is prepared in either an eigenstate or a mixture of eigenstates of $J_z$ then it does not evolve in time.

\chapter{Classical and Quantum Probability Theory\label{chap:Math}}

This chapter serves two purposes.  The first and primary intention is to present a number of known results from classical and quantum probability theory, which will serve as a foundation for the novel work in later chapters. Those results rely on a detailed knowledge of the (classical) statistical properties of the Wiener process and so we review them here.  Additionally, we need to know how to extract a classical Wiener process from a fundamentally quantum system.  The procedure of identifying a stochastic process embedded in a quantum system is one useful application of a more general mapping between quantum systems and classical probability theory.  The second purpose of this chapter is to emphasize the power of this technique and to discuss how the language of classical probability theory can be used to identify certain symmetries that might exist in a quantum system.  In order to do this coherently, we also review some of the basic elements of classical probability theory.

By working in the language of classical probability theory, the tools of nearly 80 years of classical mathematical analysis can be applied to quantum problems, with one important example being a continuous-time quantum filter.  We will not rederive it here, merely discuss its origins, limitations and various formulations.  A particulary important form is the conditional master equation, an equation that is in some sense semiclassical and can be viewed as being generated by the quantum-to-classical mapping.  Here we use the term semiclassical in the sense that the measurement record is modeled as a real valued classical stochastic process whose statistics are given by a quantum system expectation (see Sec. \ref{chMath:sec:CME}).  Chaps. \ref{chap:projection} and \ref{chap:QubitState} work exclusively with this equation, albeit in three different variations.  The final topic of this chapter is to show how these various forms are derived.


\section{Classical Probability Theory \label{chMath:sec:ClassicalProbability} }
A physics Ph.D. program does not generally include a course in measure theory or axiomatic probability theory.  Most physics problems only consider a handful of discrete or real-valued random variables and so applying a full measure theoretic context is unnecessary.  However in some instances, working only with a probability density function becomes either intractable or conceptually problematic.  One example is when one is attempting to understand the behavior of a random function defined over continuous time.  In principle, this requires describing an uncountable number of random variables, one for each possible time, where the density function at a given time could be highly correlated with past (and maybe even future) times.

Furthermore, when adding the possibility of statistical inference to the picture, defining individual density functions becomes even more convoluted.  Consider trying to estimate the history of the random variable $x_t$ based upon a continuous observation of a nonlinear function of $x$, \emph{e.g.} $f(x_t) = \sin(x_t)$.  Writing down a joint and marginal density functions for $x_t$ and $f(x_t)$ is not particularly straightforward, as they are clearly distinct objects but are hardly independent.  In the long run, a much more efficient way of doing business is to decouple the notions of random events and their associated probabilities from the specifics of any one random variable.   By finding a way to associate $x_t$, $f(x_t)$, and maybe even a third random variable $y$ to the same underlying structure of events, we can then calculate the probability associated with those events, independent of the specifics of $x$, $y$ or $f(x)$.  The way this decoupling is made is by invoking some of the structure found in measure theory.

An axiomatized probability model contains three elements, usually written as the triple $(\Omega, \mathcal{F}, \mathbbm{P})$, with $\Omega$ being a sample space, $\mathcal{F}$ a $\sigma$-algebra of events, and $\mathbbm{P}$ is a probability measure over those events \cite{van_handel_stochastic_2007,schuss_theory_2009, oksendal_stochastic_2002}.   We will now discuss each element including specific examples.  Ultimately, we are interested in describing diffusive measurements and so we will focus on the example of Brownian motion.  Brownian motion is the canonical example for a system experiencing unforced diffusion and the Wiener process is the most widely used mathematical model for such a system.  Chap. \ref{chap:QuantumLight} already encountered an instance of a Wiener process, in the vacuum statistics of the quadrature $A_t + A^{\dag}_t$.

The first element of a probability space, $\Omega$, is called the \emph{sample space} and describes the set of all possible outcomes of the model.  In a system with a discrete number of outcomes, a flip of a coin or a roll of a die, then $\Omega$ is simply the set of all possible outcomes.  For the coin $\Omega = \set{\text{heads}, \text{tails} }$ and the die $\Omega = \set{1, 2, 3, 4, 5, 6}$.  In addition to these discrete examples, the sample space could also be uncountably infinite.   For a Brownian particle, moving in $d$ dimensions, The sample space is the space of all possible trajectories.  As a particle's trajectory must be a continuous real valued function, $\Omega$ is then then the set of all continuous functions of time \cite{schuss_theory_2009}
\begin{equation}
  \Omega = \set{ \omega(t) :  \R^+ \rightarrow \R^d,  \quad \omega \text{ continuous} }.
\end{equation}

The next element of the probability model, $\mathcal{F}$, is a $\sigma$-algebra over the sample space.  This represents any ``sensible'' question we can ask about the various outcomes.  Each object in the algebra represents such a question and is called an \emph{event}.  In this formalism, probabilities are computed not from the individual outcomes in $\Omega$, but instead from the events in $\mathcal{F}$.  The reason for this distinction is to exclude pathological cases that arise when working with uncountable sets and is the same reason measure theory was developed.  When the sample space $\Omega$ is uncountably infinite, one can find highly pathological sets that can be used to obtain paradoxical results.  For instance, by choosing just 5 disjoint subsets from the unit ball, one can construct, simply through translations and rotations, two independent and identical copies of that ball \cite{tao_introduction_2011}.  It would be problematic for a probability model to consider these kinds of sets, as one could then double the probability for picking a point in the unit ball simply by doubling the ball.  Identifying the elements of $\mathcal{F}$ with \emph{sensible} questions means that we are excluding these kinds of pathologies.

In the discrete case the sensible questions are things like, ``Did the die land with an even number?'', ``Did it land showing the number 6?'', or even ``Did it land showing any number 1 though 6?''.  Mathematically, these questions represent \emph{sets} of the underlying outcomes.  These correspond to the sets $\set{2, 4, 6}$, $\set{6}$, and $\set{1, 2, 3, 4, 5, 6} = \Omega$ respectively.  The $\sigma$-algebra $\mathcal{F}$ is the set of these sets, representing any possible question -event- we can ask about the system.   For a finite and discrete number of outcomes, $\mathcal{F}$ is usually the power set, in that it is the set of all possible sets one can make out of $\Omega$.  Operationally speaking, a $\sigma$-algebra has the following definition \cite{van_handel_stochastic_2007}.  A $\sigma$-algebra $\mathcal{F}$ is a collection of sets of $\Omega$ satisfying the following three properties\footnote{The ``$\sigma$'' in $\sigma$-algebra is to mean ``countable'' \cite{schuss_theory_2009}. },
\begin{enumerate}
  \item If a countable number of sets $\set{A_n}_{n \in \mathbbm{N}} \in \mathcal{F}$ then $\cup_{n} A_n \in \mathcal{F}$.
  \item If $A$ is a set in $\mathcal{F}$ than its complement, $A^c$, is also in $\mathcal{F}$.
  \item $\mathcal{F}$ must contain the space $\Omega$, and therefore by the second property, its complement the empty set $\set{\emptyset}$.
\end{enumerate}

For the case of Brownian motion, $\mathcal{F}$ is the $\sigma$-algebra of all ``cylinder sets'', which are defined in the following way \cite{schuss_theory_2009}.  In real valued random variables probabilities are given in terms of intervals.  The probability that a random variable $x$, with probability density $p(x)$, has a value in the interval $[a, b]$ is given by the integral $\int_a^b dx\, p(x)$.  Here the event is the interval $[a, b]$ and is an element of the \emph{Borel $\sigma$-algebra}, $\mathcal{B}$.  This is essentially the set of all intervals, open and closed over, the real line.

However, at any given time, a $d$-dimensional Brownian motion will take on values in $\R^d$.  In order to ask if a trajectory landed in some interval, $I$, we must also specify an associated time, $t$, for that measurement.  A basic cylinder set is then specified by both a time and an interval.  The actual set, $C(t; I)$, is the set of all Brownian trajectories that are in $I$ at time $t$,
\begin{equation}
  C(t;I) = \set{\omega \in \Omega\ :\ \omega(t) \in I }.
\end{equation}
A trivial example is the set $C(t, \R^d) = \Omega$, \emph{i.e.} all continuous trajectories will have a value in $\R^d$ at any time $t$.  A nontrivial example in one dimension is to ask for the set of all trajectories that are between $a =$ -10 $\mu$m and $b =$ 5 $\mu$m at time $t =$ 5 ms.

In addition, questions that involve multiple times are also sensible.  It is perfectly reasonable to ask, ``what one-dimensional trajectories are in $I_1 = (a_1, b_1)$ at time $t_1$ and in $I_2 = [a_2, b_2)$ at time $t_2 > t_1$?''  This is also a cylinder set, $C(t_1, t_2; I_1, I_2)$.  An image that might be helpful is to imagine that the cylinder set $C(t_1, t_2,\dots t_n; I_1, I_2, \dots I_n)$ defines the set of trajectories that successfully navigates the ``slalom'' defined by these intervals and these times.

The $\sigma$-algebra we will use for analyzing Brownian motion is the $\sigma$-algebra generated by the cylinder sets defined by all countable sequences of times and all open sets of $\R^d$ at those times \cite{schuss_theory_2009}.  Note that the issues of discussing a uncountably infinite number of random variables is avoided by defining the cylinder sets for a countable number of times.  In fact asking questions about an uncountable number of events is ultimately identified as ``unreasonable'' as it allows for the introduction of pathological possibilities.  Here we are only interested in describing the continuous sample paths of a Brownian particle and means that we can safely consider a countable number of events, \emph{e.g.} times defined by a sequence of rational numbers.

The final element of a probability space is the probability measure, $\mathbbm{P}$.  It defines the probability for observing the events in $\mathcal{F}$.  Mathematically $\mathbbm{P}$ is a function that takes \emph{sets of $\Omega$}, (elements of $\mathcal{F}$) and maps them to real numbers between zero and one, $\mathbbm{P}:\mathcal{F} \rightarrow [0, 1]$.  In order for a valid measure to be a probability measure we must have:
\begin{enumerate}
  \item The probability of something happening be one, $\mathbbm{P}(\Omega) = 1$, and the probability of nothing happening be zero, $\mathbbm{P}(\emptyset) = 0$.
  \item The probability of the union of a countable number of disjoint events in $\mathcal{F}$ must be additive,
  \begin{equation}
    \mathbbm{P}(\cup_{n} A_n) = \sum_{n}\mathbbm{P}( A_n) \quad \text{ if } A_n \cap A_m = \emptyset \text{ for } A_n, A_m \in \mathcal{F} \text{ and } n\ne m.
  \end{equation}
\end{enumerate}
The requirement that a probability measure be countably additive is simply a statement that if $A$ is independent of $B$ then the probability to observe $A$ or $B$ is the sum of the two probabilities.

Shortly we will discuss what the probability of observing a given cylindrical set, if the trajectories in those sets represent unforced Brownian motion.

\subsection{Stochastic processes and random variables \label{chMath:sec:processes} }
From a well constructed probability space we now need to see how random variables fit into the measure theoretic context.  Much more can be said on this topic than we can include here, so an interested reader is encouraged to consult \cite{williams_probability_1991,van_handel_stochastic_2007,schuss_theory_2009,oksendal_stochastic_2002}.  Chap. \ref{chap:QubitState} requires a reasonable understanding of the statistical properties of a one-dimensional Brownian motion, the Wiener process, and so we will focus on that example here.

Abstractly, a random variable $f$ is a function that maps elements of $\Omega$ to another space, usually the real numbers.  Placing a \$50 bet that a coin toss will land heads is an example of a random variable.  Another example of a random variable the indicator function $\indicate{A}\hspace{-4pt}(\omega)$ for any event $A \in \mathcal{F}$.  Chap. \ref{chap:QuantumLight} already found many uses for an indicator function, which in a probabilistic context, is a random variable defined as
\begin{equation}
  \indicate{A}\hspace{-4pt}(\omega) = \left\{\begin{array}{cc}
                                               1 & $if $ \omega \in A \\
                                               0 & $if $ \omega \notin A
                                             \end{array}
   \right. .
\end{equation}
Such a random variable is deceptively simple but is also extremely useful.  One of its primary uses is that they relate set operations in $\mathcal{F}$ to algebraic operations on random variables.  It is easy to show that the random variable $x(\omega) = \indicate{A}\hspace{-4pt}(\omega) + \indicate{B}\hspace{-4pt}(\omega)$ is equal to $1$ whenever $\omega \in A \cup B$.  Also the random variable $y(\omega) = \indicate{A}\hspace{-4pt}(\omega)\, \indicate{B}\hspace{-4pt}(\omega)$ is equal to $1$ only when $\omega \in A \cap B$, meaning
\begin{align}
\indicate{A\cup B }\hspace{-4pt}(\omega)&=  \indicate{A}\hspace{-4pt}(\omega) + \indicate{B}\hspace{-4pt}(\omega)\\
\indicate{A\cap B }\hspace{-4pt}(\omega)&=  \indicate{A}\hspace{-4pt}(\omega)\, \indicate{B}\hspace{-4pt}(\omega).
\end{align}

For the case of 1D diffusion, one of the most important random variable is parameterized by time and simply returns the value of trajectory at that time.  For all times $t \ge 0$, we define the function $x_t : \Omega \rightarrow \R$ such that
\begin{equation} \label{chMath:eq:brownianRandomVariable}
  x_t(\omega) = \omega(t).
\end{equation}
This definition might seem a bit pedantic, but note that the trivial random variable $y(\omega) = \omega$ is not real-valued, $\omega$ describes the entire trajectory not just at any one specific time.  From $x_t(\omega)$ a whole host of other random variables can be defined through functional composition.  A man could place a \$50 bet on whether or not a diffusive particle will be greater than +5 $\mu$m from its starting point by time $t = 1$ ms.  That bet is the composition $b( x_{1 \text{ms}}(\omega))$ where $b$ maps real number to $\pm$50.

The random variable $x_t$ in Eq. (\ref{chMath:eq:brownianRandomVariable}) only gives a snapshot of a trajectory at that time.  In order to describe the trajectory dynamically in time as a random variable, there is the notion of a stochastic process.  Most generally, a stochastic process is a family of random variables $\set{x_t}_{t \in I}$ indexed by some parameter $t$, almost always representing time.  Typically we will take time to start at $0$ and either let it continue off towards infinity or, when convenient, stop at some finite time.  When discussing the concept of the process we will use the notation $\set{x_t}_{t \ge 0}$ and $x_t$ is the random variable given at that time.    Before discussing a couple important types of processes, we should know how to compute the probability for a random variable to evaluate to a range of values.

The previous section showed that a probability measure $\mathbbm{P}$ acts on elements of $\mathcal{F}$ and returns probabilities.  To compute the probability of a \$50 dollar bet $b$ to win, we need to identify the set of events that the function $b: \Omega \rightarrow {\pm 50}$ evaluates to $50$.  Because $b$ is a function acting on $\Omega$, we can also consider its inverse map $b^{-1}$.   If a given random variable can take on a continuum of values, we can still run into pesky problems of having uncountable numbers of things.  The solution to this is to again only consider sensible sets of outcomes for any given random variable.  For the random variable $x_t$ from Eq. (\ref{chMath:eq:brownianRandomVariable}) we will have to take its inverse map $x_t^{-1}$ to act only on elements of the Borel $\sigma$-algebra, $\mathcal{B}$.  When we ask for probabilities of observing certain values of a random variable $x$, we must ask for the probability of observing \emph{sets} or intervals in the range of $x$.

For every ``reasonable'' interval that $x$ maps to, there must be a corresponding element $A \in \mathcal{F}$ in order for us to be able to calculate the probability of that underlying event.  Such a random variable is called \emph{measurable}.  If a random variable $x$ is not measurable, then there is little we can say about it when its outcomes lead to unreasonable questions.  In other words, a nonmeasurable random variable has an inverse that generates sets not in $\mathcal{F}$.  Faced with this possibility we can either ignore such questions and pray they never occur or redefine the probability space in order to make these sets measurable.    A nontrivial example of this problem is suppose we had a random variable $y_t$ that returned the value $1$ whenever the sample path exhibited a discontinuous jump in the time interval $[0, t)$ and zero otherwise.  The question, ``what is the probability of $y_t$ returning 1?'', corresponds to $\mathbbm{P}\big( y_t^{-1}(1) \big)$.  If our probability space is constructed \emph{only} of continuous functions, then we technically can't answer this question as the pre-image $y_t^{-1}(1)$ ask for the set of functions that have a discontinuity for times $0 \le s < t$, which is \emph{not} an element of $\mathcal{F}$.

By defining random variables as measurable functions, we can easily relate the statistics of multiple random variables to each other though their inverse maps.  Consider the stochastic process $\set{x_t}_{t\ge 0}$ defined by Eq. (\ref{chMath:eq:brownianRandomVariable}).  Then $x_{t_1}$ and $x_{t_2}$ are two random variables taking on values in the real number line.  Suppose we wish to calculate the probability of observing $x_{t_1}$ in the interval $(a, b)$ and $x_{t_2}$ in the interval $(c, d]$.  Individually, we have
\begin{equation}
  x^{-1}_{t_1} \big( (a, b)\, \big)= C(t_1;\, (a,b)\,) \ \in \mathcal{F}
\end{equation}
 and
\begin{equation}
  x^{-1}_{t_2}\big( (c, d]\,  \big) = C(t_2;\, (c, d]\, ) \ \in \mathcal{F}
\end{equation}
The joint probability of these two events is simply the probability of the intersection of these two sets,
\begin{equation}
  \mathbbm{P}\Big(\,  C(t_1;\, (a,b)\,) \cap C(t_2;\, (c, d]\, )\,\Big) = \mathbbm{P}\Big(\,  C\big(t_1, t_2; \, (a,b),\, (c, d]\, \big)\, \Big)
\end{equation}

\subsection{Expectation values, the conditional expectation, and measurability \label{chMath:sec:classicalConditioning}}
The most fundamental operation one performs with random variables is computing their expectation values.  If the random variable $z(\omega)$ takes on a finite number of values, $\set{z^{(i)}\ : \ i = 1, \dots, n}$, then calculating the expectation value for $z$ is no different than in the nonmeasure theoretic context
\begin{equation}\label{chMath:eq:discreteExpect}
  \mathbbm{E}(z) \define \sum_{i = 1}^n z^{(i)}\ \mathbbm{P}\left( z = z^{(i)} \right).
\end{equation}
The expectation value of $z$ is the average of all its outcomes, weighted by how likely they are to occur.   Note that writing $\mathbbm{P}\left( z = z^{(i)} \right)$ is shorthand for finding the event, $A_i \define z^{-1}\left(  z^{(i)} \right)$ with
\begin{equation}
\mathbbm{P}(z = z^{(i)}) \define \mathbbm{P}(A_i) = \mathbbm{P}\big(\,\set{\omega \in \Omega\ :\ z(\omega) = z^{(i)}}\big).
\end{equation}

In addition to these simple random variables, we need to formulate expectation values for random variables that can take a continuum of values.  This is done by defining a measure theoretic version of a standard Riemann integral, called the \emph{Lebesgue integral}.  One path for this construction is to make an approximation for $x_t$ that takes on a finite number of values.  The expectation value of such a discrete approximation is easily computed though Eq. (\ref{chMath:eq:discreteExpect}).  Then by taking a suitable limit where the number of values become continuous we can calculate the proper expectation value.  At this point this procedure is a bit vague, but as we have not specified a measure for Brownian motion, it is difficult to be more specific.  When discussing the Wiener process, we will be able to be more clear.

Probability theory gets a lot more lively when instead of considering simple expectation values we consider conditional quantities.  When working with simple random variables finding a conditional expectation values is no harder in a probability space than in a standard context.  No matter how sophisticated the framework, Bayes' rule still applies, in that for two events $A_1$ and $A_2$ we have
\begin{equation}
  \mathbbm{P}\left(A_1 \cap A_2\right) = \mathbbm{P}\left(A_1 | A_2\right) \mathbbm{P}(A_2).
\end{equation}
Whenever $\mathbbm{P}(A_2) \ne 0$  we can invert to find the conditional probability of $A_1$ given $A_2$,
\begin{equation}
  \mathbbm{P}\left( A_1 | A_2\right) = \frac{\mathbbm{P}\left(A_1 \cap A_2\right)}{\mathbbm{P}(A_2)}.
\end{equation}
We emphasize that we are calculating the probability of an event $A_1$ occurring, conditional on the event $A_2$.  While $A_2$ could correspond to the pre-image of a single value of a simple random variable, it could also correspond to another random variable taking on a range of values or even the intersection between the outcomes of two random variables.  All of these different possibilities correspond to the same underlying event and thus carry the same information.

Turning this conditional probability into a conditional expectation value is simply a matter of weighting the outcomes of one random variable by the conditional probabilities. We will illustrate this in a quick example.  Consider the two simple random $z(\omega)$ and $y(\omega)$ with values $\set{z^{(1)}, \dots, z^{(n)}}$ and $\set{y^{(1)}, \dots, y^{(m)}}$.   For each variable and each outcome find the corresponding events, $A_i = z^{-1}(z^{(i)})$ and $B_j = y^{-1}(y^{(j)})$.  Then the conditional expectation value of $z$ given that $y = y^{(j)}$ is
\begin{equation}
\begin{split}
  \mathbbm{E}\left( z \middle| y = y^{(j)} \right) &= \sum_{i = 1}^n z^{(i)}\ \mathbbm{P}\left( z = z^{(i)} | y = y^{(j)} \right) = \sum_{i = 1}^n z^{(i)}\ \frac{\mathbbm{P}\left( A_i \cap B_j \right)}{\mathbbm{P}(B_j)}.
\end{split}
\end{equation}
If one is having to compute this conditional expectation value by hand, hopefully the events $A_i$ and $B_j$ are relatively simple and that computing the probability of their union is relatively straight forward.  But even if this is not the case, by adding the structure that events are sets and probabilities are measures on events, computing conditional quantities does not require defining a new probability structure like a joint density function for each random variable we want to consider.

Sometimes it is convenient to write the conditional expectation not directly in terms of the events $A_i$ and $B_i$, but instead in terms of a regular expectation value and an indicator function.  This method will be useful when both $z$ and $y$ are not simple random variables, also when we want to abstract away $y$ and instead think about conditioning on just some abstract event $B$.  A nice property of indicator functions is that because they only the value 1 on a single event, we can write
\begin{equation}
  \mathbbm{E}(\indicate{A_i})  =  \mathbbm{P}\left( A_i \right)
\end{equation}
and in particular,
\begin{equation}
  \mathbbm{P}\left( A_i \cap B_j\right) = \mathbbm{E}(\indicate{A_i\cap B_j}) = \mathbbm{E}(\indicate{A_i}\indicate{B_j}).
\end{equation}
As the expectation value is a linear operation we have that
\begin{equation}\label{chMath:eq:condtionalExpectValue}
\begin{split}
  \mathbbm{E}\left( z \middle| B_j \right) &= \sum_{i = 1}^n z^{(i)}\ \frac{\mathbbm{E}(\indicate{A_i}\indicate{B_j}) }{\mathbbm{E}(\indicate{B_j})} = \frac{\mathbbm{E}\left(\,\left( \sum_i z^{(i)} \indicate{A_i} \right)\, \indicate{B_j}\right) }{\mathbbm{E}(\indicate{B_j})} = \frac{\mathbbm{E}\left( z \indicate{B_j}\right) }{\mathbbm{E}(\indicate{B_j})}.
\end{split}
\end{equation}
The last equality here is particulary useful as it holds even when $z$ isn't a simple random variable. 
When continuous-time and continuous-valued random variables are involved, it should not be surprising that explicitly computing conditional quantities by finding the underlying sets and computing their union is often impracticable.  While working with the indicator function $\indicate{{B_j}}$ makes some substantial simplifications, often an explicit computation is still impractical. 
Instead, taking an indirect root is often fruitful and one of the prime tools for doing so is what is called the \emph{conditional expectation}.

Moving from a conditional expectation value, to a conditional expectation, is only one step more complicated.  If one is prepared for computing the conditional expectation value for every outcome of $y$, \emph{i.e.} having computed all of the sets $B_j$ and know which have zero probability, you can write down a \emph{random variable} that in some sense computes all of the conditional expectation values at once.  The conditional expectation of $z$ on $y$ written as $\mathbbm{E}(z|y)$, is a random variable that takes on the value $\mathbbm{E}\left( z \middle| y = y^{(j)} \right)$ whenever $y = y^{(j)}$.   This works through the following.  For every outcome $y^{(j)}$ we have an underlying event $B_j$.  Whenever $\mathbbm{P}(B_j) \ne 0$ we can define the random variable,
\begin{equation}\label{chMath:eq:condtionalExpectation}
\begin{split}
  \mathbbm{E}\left( z \middle| y \right) \define \sum_{j} \mathbbm{E}\left( z \middle| y = y^{(j)} \right) \ \indicate{B_j}\hspace{-4pt}(\omega). 
\end{split}
\end{equation}
Whenever $\mathbbm{P}(B_j) = 0$ we can give the conditional expectation value any value we wish, safe in the knowledge that the probability for obtaining such an arbitrary values is zero.  A lot can be said for this object, but one of the most important is that it can be viewed as a reasonable estimate for $z$ given information about $y$.  More specifically suppose you wanted to find a least-mean-squared estimate for what value $z$ would return when the outcome you receive, $\omega$, results in $y(\omega) = y^{(i)}$.  Generally speaking, there are many different $\omega$ for any value of $y^{(i)}$ also the set of $\omega$ that gives this value, may return multiple values for $z$.  It turns out that if you want to make any estimate for any random variable, $z$, conditioned on the events given by another, $y$, then the estimate for $z$ must have the form
\begin{equation} \label{chMath:eq:Ymeasureable}
\begin{split}
  \hat{z}(\omega) = \sum_{j = 1}^m a_j\, \indicate{B_j}\hspace{-4pt}(\omega) 
\end{split}
\end{equation}
where $B_j$ are the events generated by the various outcomes for $y$ and $a_j$ are constants.  It should not be too surprising that by finding the constants that correspond to the least-mean-square estimate for $z$ given $y$ turn out to be the conditional expectation values $a_j = \mathbbm{E}\left( z \middle| y = y^{(j)} \right)$.

Any random variable that can be written as Eq. (\ref{chMath:eq:Ymeasureable}) is called a $y$-measurable random variable.  Note that $y$ itself and almost any function of $y$ can written like this.  The sets $B_j$ can be used to construct a $\sigma$-algebra by taking all (countable) unions and complements, then the random variable $\hat{z}$ is measurable with respect to that $\sigma$-algebra.  The $\sigma$-algebra formed by these sets, written as $\mathcal{Y}$ or sometimes $\sigma\set{y}$, is called the sigma algebra generated by $y$.  It is not difficult to show that if you have multiple random variables $x$, $y$, \emph{etc.} you can form a $\sigma$-algebra, $\sigma\set{x, y, \dots}$, generated by all of those as well simply by taking unions and complements of all of their various events.

All of the above concepts, joint probabilities, the conditional expectation, generating $\sigma$-algebras from random variables, \emph{etc.}, can be extended to continuous random variables by taking appropriate limits.  More often than not, one tries to avoid taking an actual limit, and instead looks for a random variable that is $y$-measurable and satisfies the basic properties of the condition expectation.  For general reference, here are some of those properties:
\begin{enumerate}
  \item The conditional expectation is linear,
        \begin{equation}
          \mathbbm{E}\left( \alpha\, x + \beta\, z  \middle| y \right) = \alpha \mathbbm{E}\left( x \middle| y \right) + \beta \mathbbm{E}\left( z \middle| y \right)
        \end{equation}
        for constants $\alpha$ and $\beta$.
  \item  The conditional expectation is consistent with usual expectation values,
        \begin{equation}
          \mathbbm{E}\big(\, \mathbbm{E}\left( x \middle| y \right)\, \big) = \mathbbm{E}(\, x \, ).
        \end{equation}
  \item  The conditional expectation of any $y$-measurable random variable, $x$, is itself
        \begin{equation}
          \mathbbm{E}\left( x\, \middle| y \right) = x.
        \end{equation}
  \item  If $z$ and $y$ are independent than the conditional expectation is just the expectation value for $z$ (times the ``identity'') \begin{equation}
          \mathbbm{E}\left( z \middle| y \right) = \mathbbm{E}(z)\, \indicate{\Omega}\hspace{-4pt}(\omega).
        \end{equation}
  \item For any $y$-measurable random variable $y'$, the conditional expectation value satisfies the property that
        \begin{equation}
          \mathbbm{E}\big(\, \mathbbm{E}\left( z \middle| y \right)\, y' \big) = \mathbbm{E}(\,z\, y' \, ).
        \end{equation}
\end{enumerate}
while this last property could be inferred from the second and third, we include it here because it is often what is used to ``guess'' what the conditional expectation is, without going though a bare bones construction.

Much, much more can be said about the conditional expectation.  The derivation of classical and quantum filtering theory is based upon computing a conditional expectation of some unobserved processed, $\set{x_t}$  based upon measurements of a correlated process $\set{y_t}$.  Unfortunately, we will be unable to go into the detail here but interested reader is encouraged to seek out a number of good references on the subject, some of which are \cite{van_handel_stochastic_2007,oksendal_stochastic_2002,bouten_introduction_2007,chase_parameter_2009}.

\subsection{Special processes - time-adaption and martingales}

Having discussed random variables and stochastic processes in terms of a classical probability space $(\Omega, \mathcal{F}, \mathbbm{P})$, we now need to introduce a couple of important and useful processes.  Chap. \ref{chap:QuantumLight} already introduced the concept of a time-adapted process.  An stochastic process $\set{x_t}_{t \ge 0}$ is time-adapted when it depends only on events defined in the present or past and not on the future.  Having introduced the concept of a measurable function and a $\sigma$-algebra over the cylindrical sets $C(t_1, t_2, \dots; I_1, I_2, \dots)$ we can easily give a precise meaning to time-adapted process.  A stochastic processes $\set{x_t}_{t \ge 0}$ is time-adapted when each random variable $x_t$ is measurable with respect to the $\sigma$-algebra, $\mathcal{F}_t$, generated from the cylindrical sets $C(t_1, \dots , t_n \le t; I_1, \dots, I_n)$.  (Sometimes a more general definition avoids using these specific cylinders and simply uses an indexed sequence of $\sigma$-algebras $\mathcal{F}_0 \subset \mathcal{F}_s \subset \mathcal{F}_t \subset \mathcal{F}$ called a \emph{filtration} \cite{van_handel_stochastic_2007}.  A process adapted to this filtration is called $\mathcal{F}_t$-adapted.)  In the context of statistical estimation, working with time-adapted processes is essential, as these are the processes that are independent of any future events.  Within the bounds of time-adapted processes there are an additional type of stochastic processes that have special and simplifying characteristics in terms of their conditional statistics, called \emph{martingales}.

A martingale is an important kind of stochastic process which plays a crucial role in classical probability theory.  More importantly for our purposes,  they will play a crucial role in significantly simplifying a quantum conditional master equation as we will see in Sec. \ref{chMath:sec:Innovation}.  They are used to represent fair betting games where no amount of past information is helpful in predicting future events.  The defining property is that the conditional expectation of any future value of the process is simply given by its current value.  You expect to leave a fair casino with the same amount of money as you had when you entered\footnote{In contrast to a real casino.}.  In essence a martingale is a random process where the conditional mean of any future increment is zero, \cite{williams_probability_1991, oksendal_stochastic_2002, van_handel_stochastic_2007}.  To illustrate this property, consider taking a fair coin and flip it $N$ times.  A typical sequence $\omega$ may be something like,
\begin{equation}
\omega = \set{H,\, T,\, H,\, H,\, H,\, H,\, T,\, T,\, H,\, \dots}.
\end{equation}
For any sequence $\omega$ we can create a random variable $x_n$ which is equal to the number of heads minus the number of tails seen in the first $n$ flips.  So that for this sequence
\begin{equation}
x = \set{1,\, 0,\, 1,\, 2,\, 3,\, 4,\, 3,\, 2,\, 3,\,\dots}.
\end{equation}
In the case of a fair coin, $x_n$ is a martingale.

To see why, note that in each flip there is equal probability of the coin landing heads or landing tails.  So that for any $n$ we have the expectation value,
\begin{equation}
\mathbbm{E}(x_n) = 0.
\end{equation}
In this specific realization, after the first four flips $x_4$ is not 0, but is in fact 2. However, because any future flip are independent of the past we should not expect to see any more heads than tails.  This means that conditioned upon the first four outcomes we should not expect for $x_{n \ge 4}$ to be 0, but instead it should average around 2.  In other words
\begin{equation}
\mathbbm{E}(x_m - x_n | \set{x_1, x_2\,\dots x_n}) = 0 \qquad \text{for all } m \ge n.
\end{equation}
This is the fundamental property of a martingale which is usually written as,
\begin{equation}
\mathbbm{E}(x_m | \set{x_1, x_2\,\dots x_n} ) = x_n \qquad \text{for all } m \ge n.
\end{equation}

Now imagine that that the coin was not in fact fair.  Say that the probability for heads was $P_H = 2/5$ and the probability for tails was $P_T = 3/5$.   Then in this case $x_n$ is not a martingale, as we would instead expect $x_n$ to trend negative.  Or in other words $\mathbbm{E}(x_m | \set{x_1, x_2\,\dots x_n} ) < x_n$ for $m >n$.  But while $x_n$ may not be a martingale, it is sometimes possible to construct from $x$ another random variable that is a martingale.  This kind of process is called a semi-martingale and defines the class of processes that are capable of being used in an It\={o} integral.  The fact that this is possible will play a crucial step in finding a maximum likelihood estimate in Chap. \ref{chap:QubitState}.

\subsection{The Wiener process \label{chMath:sec:Wiener} }
One of the most important processes one can consider is the \emph{Wiener process}, a mathematical model for Brownian motion named for Norbert Wiener.  In addition to elegantly describing diffusive motion, the Wiener process is used to model nearly any systems interacting with white noise.  Chap. \ref{chap:QuantumLight} already found an application outside of diffusive motion, that is the statistics of the quadrature $Q_t$ and $P_t$ under vacuum expectation.  This section reviews of the properties of the Wiener process, including how it relates to a classical probability space introduced in Sec. \ref{chMath:sec:ClassicalProbability}.

The defining characteristics of a Wiener process are two fold:
\begin{enumerate}
  \item A Wiener process makes a continuous trajectory in time, with probability one.
  \item A Wiener process has increments that are independent, mean zero, Gaussian distributed random variables with a variance given by the increment's time duration.
\end{enumerate}
The first property is obvious while the second is a bit more involved and requires some explanation.  Consider the stochastic process $\set{w_t}_{t \ge 0}$.  For times $0 <  s  < t $,  we can define the random variables, $a = w_s - w_0$ and $b = w_t - w_s$.  If $\set{w_t}_{t \ge 0}$ is a Wiener process then $a$ and $b$ are statistically independent and $a$ is a mean zero, Gaussian random variable with variance $s$ and $b$ is also a mean zero Gaussian, with variance $t-s$.  Sec. \ref{chMath:sec:processes} found that if a random process has statistically independent increments and each increment is mean zero then its a martingale, and so the Wiener process is also a martingale.  There are many more interesting and sometimes nonintuitive properties that one can calculate for a Wiener process, see \cite[Chap. 2]{schuss_theory_2009}.   Some interesting facts are (\emph{i}) a Wiener process is nondifferentiable with probability one and (\emph{ii}) if at some time $t$ a Wiener process takes on the value $w(t) = a$ then it will take on that value and infinite number of times in every interval $[t, t + \Delta t]$.

The statistical properties of a Wiener process are deceptively simple, and yet exceedingly rich.  The second defining property allows us to find a connection between this simple statement and a nontrivial probability measure over the space of continuous functions.  While we just used the times $0 <  s  < t $ to demonstrate what this property means there is nothing stopping us from using a countable sequence of times $ 0 < t_1 < t_2 < \dots\,$ .  We then know that for a Wiener process the random variables $\Delta w_i \define w_{t_i} - w_{t_{i-1}}$ are all independent mean zero Gaussian random variables with variances $\Delta t_i \define {t_i} - {t_{i-1}}$.  For each time $t_i$ we can calculate the probability that the Wiener process lies in the interval $I_i = (a_i, b_i)$.  This ultimately turns out to be
\begin{equation} \label{chMath:eq:wienerPathIntegral}
  P( \set{w_{t_i} \in I_i} ) = \int_{a_1}^{b_1} d\w_1 \int_{a_2}^{b_2} d\w_2 \dots \prod_{i} \left( \frac{1}{\sqrt{2 \pi \Delta t_i}} e^{-\frac{(\w_i - \w_{i -1})^2}{2 \Delta t_i} } \right),
\end{equation}
which is known as \emph{Wiener's discrete path integral} \cite{schuss_theory_2009}.  Notice, that by picking the sequence of times $ 0 < t_1 < t_2 < \dots\,$ and the intervals $I_i = (a_i, b_i)$, we just defined a cylindrical set, $C( t_1, t_2, \dots\, ; I_1, I_2, \dots)$.  Eq. \ref{chMath:eq:wienerPathIntegral} is a probability for observing a continuous trajectory to lie within this set, and therefore we can use these integrals to define a probability measure on the space of continuous functions $\omega : \R^+ \rightarrow \R$.  Not surprisingly, this is called the $\emph{Wiener measure}$
\begin{equation} \label{chMath:eq:wienerMeasure}
  \mathbbm{P}\Big(C( t_1, t_2, \dots\, ; I_1, I_2, \dots) \Big) =  P( \set{w_{t_i} \in I_i} ).
\end{equation}
It is worth noting that under this measure, all trajectories which do not have $\omega(0) = 0$ are given zero probability, \emph{i.e.}
\begin{equation}
  \mathbbm{P}\Big(\set{\omega \in \Omega\, :\ \omega(0) = 0 } \Big) = 1.
\end{equation}
This brings us back to the probability space $(\Omega, \mathcal{F}, \mathbbm{P})$ defined over the continuous functions $\omega(t)$ with the $\sigma$-algebra $\mathcal{F}$ over the cylindrical sets.  If the probability measure $\mathbbm{P}$ is the Wiener measure as defined above than the stochastic process, $\set{ w_t(\omega) = \omega(t)\, :\ 0 \le t < \infty}$  is Wiener processes.

\section{Quantum Probability Theory \label{chMath:sec:QSDEs} }


These same concepts can also be applied to quantum theory either directly or with some modification.  
The mathematics of quantum stochastic calculus and noncommutative probability theory is a broad and detailed subject, one that is beyond our scope.  Reasonable introductions with an emphasis on filtering can be found in \citep{van_handel_feedback_2005,bouten_introduction_2007} and with more detailed treatments in \citep{parthasarathy_introduction_1992,hudson_quantum_1984, barchielli_continual_2006}.  However a certain amount of review is necessary in order to address the physical implications of the formalism.  Before discussing the truly quantum nature of noncommutative probability theory, we will discuss its similarities with the classical theory.

\subsection{Embedding the quantum into the classical\label{chmath:sec:quantumIntoClassical}}
This section reviews how to constructs a classical probability space from a set of mutually commuting quantum observables.  The purpose for this review is two fold.  First, the quantum filtering problem relies upon this kind of mapping.  The continuous measurements we will be making is described by a set of mutually commuting operators which is increasing in time.  The eigenvalues that we receive will be viewed as a (read mapped to) classical random variables on a classical probability space.  The second reason for this review is to emphasizes its limitations.  In Chap. \ref{chap:QubitState} the quantum filter is used to estimate an unknown initial state of a qubit.  A natural tool in classical systems is the \emph{smoother} which is an estimate for an unobserved system at some past time, given measurements up to some current time.  However, na\"{i}vely applying this classical technique violates a necessary condition that allows for the classical mapping.

In classical probability theory we found that random variables could be viewed as functions mapping elements of the sample space to real numbers.  At its most practical level, quantum theory is used to predict the outcomes of experiments where the measured observables are represented as Hermitian operators acting upon some underlying Hilbert space.  The first step in bringing classical probability theory to the quantum is to formulate an analogy between classical random variables and Hermitian operators.
\begin{center}
\begin{tabular}{ccc}
  classical & $\leftrightarrow$ & quantum \\
  $x(\omega)$ & $\leftrightarrow$ & $X$ \\
\end{tabular}
\end{center}
This is a natural analogy, as the basic operation in classical probability is to calculate the expectation values of random variables.

The next association is that in the classical theory we have the probability measure $\mathbbm{P}$ to calculate expectation values while in the quantum we have the system state $\rho$.  This analogy is best illustrated in a discrete example where the classical sample space is the set of a finite number of $d$ realizations, $\Omega = \set{\omega_1, \dots \omega_d}$.  The $\sigma$-algebra $\mathcal{F}$ for this space is then the power set of $\Omega$ and the probability measure $\mathbbm{P}$ is completely described by the probabilities of the singleton events $p_i = \mathbbm{P}(\{\omega_i\})$.  The classical expectation value of a simple random variable $x(\omega)$ in this space is then
\begin{equation}\label{chMath:eq:classicalExpect}
    \mathbbm{E}(x) = \sum_{i = 1}^d x(\omega_i)\, \mathbbm{P}(\{\omega_i\}).
\end{equation}
In the quantum case a system described by a Hilbert space $\Hilbert$ of dimension $d$ is equipped with a positive trace one density matrix $\rho$.  Expectation values of operators $X$ acting on $\Hilbert$ are of course calculated as
\begin{equation}\label{chMath:eq:quantumExpect}
    \mathbbm{E}(X) = \Tr( \rho X)
\end{equation}
which will sometimes also be notated as $\expect{X}$.  Thus we have the correspondence
\begin{center}
\begin{tabular}{ccc}
  classical & $\leftrightarrow$ & quantum \\
  $\mathbbm{E}(x)$ & $\leftrightarrow$ & $\Tr( \rho X)$. \\
\end{tabular}
\end{center}
Rather than this loose analogy, a formal equivalence is possible where certain aspects of quantum theory can be embedded into a classical probability space.

While the classical probability space has a fixed, albeit sometimes abstract set of realizations $\Omega$, identifying such a set in quantum mechanics is problematic.  In the spirit of deterministic classical physics, the sample space $\Omega$ most often represents locally realistic fates of the system.  The probability of observing certain events is given by the probability measure $\mathbbm{P}$ which act on subsets of $\Omega$.  The utility of probability theory is that we have an event $A = \set{\omega \in \Omega \ |\ a \le x(\omega) \le b }$ which has a concrete meaning for multiple random variables, not just $x$ but $f(x)$.  In addition there could certainly be another random variable $y$ such that when it take on the values $c \le y \le d$ whenever $x$ is observed in the interval $[a, b]$, and so both correspond to the same underlying event $A$.   It is clear from Bell's theorem that this locally realistic interpretation of $\Omega$ is not consistent with quantum mechanics.

A less ambitious task is to find a classical probability representation that is capable of describing the joint statistics of \emph{compatible observations}.  Compatibility of two observables $X$ and $Y$ means that $[X, \, Y] = 0$ and more importantly they share a set of eigenvectors $\set{\ket{e_i}}$.  For each eigenfunction we have the projector $\proj_i = \ketbra{e_i}{e_i}$ and the operators $X$ and $Y$ have the spectral decompositions
\begin{equation}
    X = \sum_i x_i \, \ketbra{e_i}{e_i} \quad\text{and} \quad Y = \sum_i y_i \, \ketbra{e_i}{e_i}.
\end{equation}
Note that the eigenvalues $x_i$ and $y_i$ need not be distinct as they could have degenerate subspaces.

In a $d$-dimensional system there are at most $d$ distinct, mutually orthogonal projectors associated with a set of commutating operators $\mathcal{C} = \set{X,\, Y,\, Z,\dots}$.  If we associated projectors $\set{\proj_\lambda = \ketbra{\lambda}{\lambda}}$ so that
\begin{equation}
  X = \sum_{\lambda = 1}^d x_{\lambda} \, \proj_\lambda
\end{equation}
and
\begin{equation}
  \mathbbm{E}(X) = \sum_{\lambda = 1}^d x_{\lambda} \, \Tr(\, \rho\, \proj_\lambda ).
\end{equation}
The mapping between discrete quantum mechanics and classical probability is to associate the set of labels for projectors with the sample space in a classical probability space.  Then we have $\Omega = \set{\lambda_i\ : \ i = 1 \dots d}$ and $\mathcal{F}$ is simply the power set of $\Omega$.  From this assignment, the probability measure is simply the quantum expectation value of the associated projectors.  For example, probability for the event $\set{\lambda_1,\, \lambda_2}$ is
\begin{equation}
  \mathbbm{P}(\set{\lambda_1,\, \lambda_2}) = \Tr(\rho\,  \proj_{\set{\lambda_1,\, \lambda_2}} ) = \Tr(\rho\, \ketbra{\lambda_1}{\lambda_1} ) + \Tr(\rho\, \ketbra{\lambda_2}{\lambda_2} ).
\end{equation}
In classical probability the simplest of simple random variables are the indicator functions $\indicate{F}$ (for every set $F \in \mathcal{F}$) which correspond to the projectors in the  $ X \mapsto x(\omega)$  mapping.   This procedure is formalized in Theorem 2.4 of \citep{bouten_introduction_2007} and is summarized in Table \ref{chMath:table:quantumClassicalMapping}.
\begin{table}[thb]
    \begin{center}
    \begin{tabular}{c|c}
      Classical & Quantum \\ \hline
      $\Omega $  &   $\set{\lambda_1,\, \lambda_2,\, \dots}$  \\
      $\mathcal{F} $ & $\big\{ \set{\lambda_1},\, \set{\lambda_1,\,\lambda_2},\, \dots \big\}$ \\
      $\mathbbm{P}( \set{\lambda_i} ) $ & $\Tr(\rho\, \proj_{\lambda_i}\, )$ \\
      $x(\omega)$ &  $X$ \\
    \end{tabular}
    \end{center}
    \caption{The spectral mapping between a set of commuting observables and a classical probability space.\label{chMath:table:quantumClassicalMapping}}
\end{table}

The above discussion extends also to the case of infinite dimensional Hilbert spaces and operators with continuous spectra \cite{bouten_introduction_2007}.  For any ``normal'' operator\footnote{A normal operators is one that commutes with its adjoint and so has a spectral decomposition. It can be written in terms of commuting Hermitian and anti-Hermitian parts.} $A$, which may take on all values in $\R$, there exist a spectral decomposition for $A$ such that
\begin{equation}
  A = \int_\R a \proj(da)
\end{equation}
where $\proj(da)$ is the \emph{spectral measure}, also called a projection valued measure, associated with $A$ taking on values in interval $da$.  Often explicitly constructing $\proj(da)$ is a little tricky, especially if one must first identify any vectors $\psi$ in Hilbert space for which $\phi \define A\, \psi$ is not well behaved.  ( Hermitian operators whose eigenvalues span the entire real line, so called unbounded operators, exhibit these kinds of problems.    A trick for dealing with this case is to compute the spectral measure for the bounded operator $T = (A + i \ident)^{-1}$. This works because any function $f(A)$ commutes with $A$, therefore they share the same projectors.  When $T$ takes on the complex value $\lambda$ there is then the corresponding value for $A$, $a = \lambda^{-1} - i$.)  Once armed with a spectral measure for $A$ we can then find an equivalent classical probability model, whose sample space are labels for the possible values of $A$.

Regardless of whether or not the associate operators are unbounded, we emphasize this spectral mapping is only applicable to a subspace of operators which all commute with the underlying projectors.  While this may seem to indicate that the mapping is severely limited, in practice it is extremely useful for describing ancilla assisted measurements.  If one is interested in computing conditional expectation values for operators that commute with the projectors defining the classical space, then the quantum-to-classical mapping is still applicable.

\subsection{Quantum probability \label{chMath:sec:QuProb}}
The spectral mapping to a classical probability space lacked a representation that is independent of the specific choice of projectors.  Furthermore, Bell's theorem shows that there are no locally realistic sample spaces consistent with quantum mechanics.  The first step in discussing quantum mechanics in a probability theoretic framework is to omit the sample space $\Omega$ \cite{maassen_quantum_2010}.  At the level of making practical calculations, the sample space provided an underlying structure for associating random variables with the probability measure.  By observing a random outcome, $x(\omega) = a$ we were able to see what event this corresponds to and then calculate the probability for that event given the measure $\mathbbm{P}$.  In other words, we identify the set of possible realizations that are compatible with this observation then evaluate the probability for this event.

In quantum theory the underlying Hilbert space of the system, $\Hilbert$, provides this necessary structure.  By making the association between Hermitian operators and the results of experiments we already have the necessary mapping between random variables and probabilities.  In the above spectral mapping between quantum to classical, we associated events with collections of possible eigenvalues and so even in the infinite dimensional case, the probability of observing an event is given by the expectation value of the corresponding union of projectors.  In the fully quantum case, we need to consider all possible projections, not just those projections that commute.  The mathematical object that is guaranteed to contain all possible projections is a $\ast-$algebra (read ``star''-algebra) of operators\footnote{The $\ast$ in the name comes from the mathematical convention of using $^*$ to represent an operator adjoint rather than $^\dag$.}.  Therefore the correspondence between the $\sigma-$algebra $\mathcal{F}$ in classical probability space is a $\ast-$algebra $\mathcal{A}$ of operators on $\Hilbert$.

A $*-$algebra of operators acting on a Hilbert space, $\Hilbert$ is defined as the set of operators $\mathcal{A}$ so that
\begin{enumerate}
  \item $\mathcal{A}$ contains all complex linear combinations of its elements.  For all $A,\,B \in \mathcal{A}$  we also have  $C = c_1 A + c_2 B \in \mathcal{A}$ for any complex coefficients $c_1$ and $c_2$.
  \item $\mathcal{A}$ contains all adjoints of its elements.  $A \in \mathcal{A}$ implies that  $A^\dag \in \mathcal{A}$.
  \item $\mathcal{A}$ contains all products of its elements.  $A,B \in \mathcal{A}$ implies $AB \in \mathcal{A}$.
  \item $\mathcal{A}$ contains the identity $\ident$.
\end{enumerate}
In the finite dimensional case where $\Hilbert = \mathbbm{C}^n$ the largest $\ast-$algebra acting on $\Hilbert$ is simply $\mathcal{M}^n$, the space of all complex $n\times n$ matrices.  However the reason for introducing this algebraic structure is not just for a love of mathematical formalism.  In the same way that a set of classical random variables generate a $\sigma-$algebra, $\sigma\{x, y, z, \dots\}$ (see Sec. \ref{chMath:sec:classicalConditioning}) a set of operators generate a $\ast-$algebra.  For example, the above spectral mapping means that there is a $\ast-$algebera of operators generated by the set commuting of projectors $\set{\proj_\lambda}$.  The fact that a $\ast-$algebra generated from a set of commuting operators still commute means that a \emph{commutative $\ast-$algebra} is the set of operators that define the events in a classical probability space.

When $\Hilbert$ is infinite dimensional, defining a suitable $*-$algebra (or sub-$*-$algebra) becomes much more tricky.  This is particularly true when trying to show that any limiting sequence of operators in the $*-$algebra is still in the algebra.  As one might imagine, taking limits of unbounded operators becomes problematic as a sequence of operators might converge when acting on one set of vectors, but diverge when acting on another.  The details of how one solves these issues is beyond our scope.   Suffice it to say, the solution is to first consider only \emph{bounded operators} (while keeping the $T = (A + i \ident)^{-1}$ trick in mind) and then include the limits of all the sequences of operators in the algebra that converge on a class of well defined states.  The technical name for such an algebra is a \emph{von Neumann algebra} \cite{bouten_introduction_2007}.  Generally though, one hardly ever needs to apply this kind of construction directly.

One additional concept that is useful, particularly for discussing a quantum conditional expectation, is that of a \emph{commutant}.  Suppose you are given a set of operators $\mathcal{S}$ and you want to know what are the set of operators that commute with $\mathcal{S}$.  That set is call the commutant and is notated as $\mathcal{S}'$.  To see why this idea is import consider the following.  Assuming that you are able to measure the operators $\set{Y_1, Y_2, \dots, Y_n}$ and that all of these operators commute with each other.  We just showed how to form a commutative von Neumann algebra from this set, but one might wonder if that's the whole story in a quantum to classical mapping.  The answer turns out to be no.  The wiggle room is from the fact that the number of distinct eigenvectors made from the operators $\set{Y_1, Y_2, \dots, Y_n}$ may not be equal to the dimension of the underlying Hilbert space.  In that case, you can find operators $A$ and $B$ where $[A, B] \ne 0$ and still have $[A, Y_i] = [B, Y_i] = 0$.

One example is in a two qubit system.  Suppose you only measure $\sigma_z$ on one qubit but leave the other one alone.  Then the projectors  $\ketbra{+1}{+1}\otimes \ident$ and $\ketbra{-1}{-1}\otimes \ident$ form the singleton events in the the classical probability model.  But clearly any operator on the second qubit commutes with these projectors and so there is more in this system than is wholly representable in a classical probability system.    But because the second system does commute with these projectors it is possible to form a quantum conditional expectation of the system upon the first.   The commutant gives you the set of all possible operators that can be mapped onto a classical probability space though a quantum conditional expectation.  We will briefly discuss this mapping next.

\subsection{The quantum conditional expectation}

In a classical probability space, if we are given a random variable $y$, or more generally a set of random variables $\set{y_s\ :\ 0 \le s \le t  }$, the distinct outcomes of those  variables form a set of events.  From these events we are able to take unions and complements to make a $\sigma$-algebra, $\mathcal{Y} = \sigma\set{y_s\ :\ 0 \le s \le t  }$, representing the questions we can answer about the model given the observation of the random variables $\set{y_s\ :\ 0 \le s \le t  }$.   Then for every event $A_i \in \mathcal{Y}$, (assuming $\mathbbm{P}(A_i) \ne 0$) we are able to compute the conditional expectation value of \emph{any} random variable $x$, via Eq. (\ref{chMath:eq:condtionalExpectValue})
\begin{equation}
  \mathbbm{E}(x| A_i) = \frac{\mathbbm{E}(x \indicate{A_i})}{\mathbbm{E}(\indicate{A_i})}.
\end{equation}
In mapping quantum mechanics on to classical probability theory everything still applies, as long as the operators $\set{Y_s\ :\ 0 \le s \le t  }$ all mutually commute.  If we have a projection $\proj_{\lambda}$ generated from these operators, we can define a \emph{quantum conditional expectation value},
\begin{equation}
  \mathbbm{E}(X| \lambda ) \define \frac{\expect{X \proj_{\lambda}}}{\expect{\proj_{\lambda}}}  = \frac{\Tr (\rho\, X \proj_{\lambda} )}{\Tr (\rho\,\proj_{\lambda} )}.
\end{equation}
This equation shows how crucial that $X$ and $\proj_\lambda$ commute in order for this equation to make sense as a classical analogy.  Not only do we need $[X, \proj_\lambda] = 0$ in order for $X$ to be block diagonalized via the labels $\lambda$, but also in order for this expression to be interpretable as a conditional expectation value, we needed to have $\expect{X \proj_{\lambda}} = \expect{\proj_{\lambda} X}$ for any $\lambda$.   Classically, we took a conditional expectation value to a conditional expectation by multiplying by the ``projector'' onto the event $A_i$, see Eq. (\ref{chMath:eq:condtionalExpectation}).  The same is true in the quantum case, as can be seen in a finite dimensional example.

From a set mutually commuting observables $\set{Y_1, Y_2, \dots Y_n}$, we can form a commutative $*$-algebra $\mathscr{Y}$ that is spanned by the orthogonal projectors $\set{\proj_\lambda}$.  These projectors form a resolution of the identity so that $\sum_\lambda \proj_\lambda = \ident$.  For any operator $X$ in the commutant of $\mathcal{Y}$, the \emph{quantum conditional expectation} is defined as \cite{bouten_introduction_2007}
\begin{equation}\label{chMath:eq:quantumConditional}
  \mathbbm{E}(X|\mathscr{Y}) \define \sum_{\lambda} \frac{\expect{X \proj_{\lambda}}}{\expect{\proj_{\lambda}}}\ \proj_{\lambda}.
\end{equation}
The quantum conditional expectation has a number of properties in common with the classical conditional expectation.  Specifically,
\begin{align}\label{chMath:eq:quantumConditionalProperties}
    \mathbbm{E}\left( \mathbbm{E}\left(X  \middle\vert \mathscr{Y} \right) \right) &= \mathbbm{E}\left( X \right) & \\
    \mathbbm{E}\left( Y_1\, X\, Y_2  \middle\vert  \mathscr{Y} \right)\ &= Y_1\, \mathbbm{E}\left( X \middle\vert \mathscr{Y} \right)\, Y_2 \qquad \forall \ Y_1, Y_2 \in \mathscr{Y}\\
    \mathbbm{E}\left( X  \middle\vert  \mathscr{Y} \right)\ & = \mathbbm{E}(X)\, \ident \hspace{2.2 cm} \forall \ X \text{ independent of }  \mathscr{Y}.
\end{align}
The quantum conditional expectation has some operator specific properties,
\begin{align}\label{chMath:eq:quantumConditionalProperties2}
    \mathbbm{E}\left(\ident  \middle\vert \mathscr{Y} \right) &= \ident  \\
    \mathbbm{E}\left( X^\dag  \middle\vert  \mathscr{Y} \right)\ &= \mathbbm{E}\left( X \middle\vert \mathscr{Y} \right)^\dag\\
    \mathbbm{E}\left( X^\dag X  \middle\vert  \mathscr{Y} \right)\ & \ge 0.
\end{align}
Finally, an extremely important property that also carries over from the classical conditional expectation is that for all $Y \in \mathscr{Y}$ and $X \in \mathscr{Y}'$
\begin{align}\label{chMath:eq:quantumConditionalDefinition}
    \mathbbm{E}\left( \mathbbm{E}\left(X  \middle\vert \mathscr{Y} \right) Y \right) &= \mathbbm{E}\left(X Y \right).
\end{align}
This property is as important here as it is in the classical case, because it is often used to identify\footnote{\emph{i.e.} guess and check.} what the conditional expectation should be when it is intractable to find an explicit representation for $\set{\proj_\lambda}$.  This is particularly true in the infinite dimensional case and then Eq. (\ref{chMath:eq:quantumConditionalDefinition}) is taken as the defining characteristic for the conditional expectation.  In other words, if you can find an operator $\hat{X} \in \mathscr{Y}$ that satisfies this equation, then you have the conditional expectation for $X$ given $\mathscr{Y}$ \cite{bouten_introduction_2007}.

\subsection{The conditional expectation and generalized measurements}
Before considering a specific example of a continuous-time quantum conditional expectation, we briefly pause to discuss the connection between the quantum conditional expectation and generalized measurements, as traditionally formulated in quantum mechanics.   We argue here that these two ideas are essentially equivalent and will specifically show that any measurement given in terms of a countable set of distinct Kraus operators $\set{M_m}$ is equally well represented in terms of a quantum system mapped to a classical probability model.  In particular, the posterior state $\rho|_{m}$ is a Schr\"odinger picture version of a quantum conditional expectation value.

A general quantum measurement on a Hilbert space $\Hilbert$ is specified by a set of Kraus measurement operators $\set{M_m}$ where the indices $m$ label the outcomes of the measurement \cite{nielsen_quantum_2000}.  The measurement operators are required to satisfy a completeness relation
\begin{equation}
  \sum_m M^\dag_m M_m  = \ident.
\end{equation}
The completeness relation means that $\set{ M^\dag_m M_m }$ is a valid POVM, and in particular that under the state $\rho$ the expectation values, $\set{ \Tr(\rho\, M^\dag_m M_m ) }$ define a probability measure for a sample space $\Omega = \set{1,\dots, m,\dots}$.   Upon receiving the outcome $m$, a mixed state $\rho$ updates to the posterior state $\rho|_m$ via the map
\begin{equation}
  \rho|_m \define \frac{M_m \rho M_m^\dag}{ \Tr(\rho\, M^\dag_m M_m ) }.
\end{equation}

Our claim is that there exists a Heisenberg picture formulation where the use of this posterior state is replaced by a conditional expectation.  Proving this is not difficult, by using the fact that any generalized measurement is equivalent to a projective measurement performed on an ancillary system after an entangling unitary operation \cite{nielsen_quantum_2000}.  The equivalent Heisenberg/quantum probability picture is then to evolve all of the operators with the entangling unitary and then calculate a conditional expectation value for a post interaction projector.

The specific relation is that every measurement outcome can be modeled as a state in an ancillary system with Hilbert space $\Hilbert_A$ where there are basis vectors $\ket{m}$ that correspond to the outcomes of the measurement.  Clearly for this to make sense, the dimension of $\Hilbert_A$ must be as least as big as the number of measurement outcomes.  The entangling unitary operator $U$ then maps a fiducial pure state $\ket{0}$ to the basis vectors $\ket{m}$ and when doing so applies the operator $M_m$ to system.  Thus, there always exists a unitary $U$ such that for every system state vector $\psi$
\begin{equation}
  U\, \ket{\psi}\ket{0} = \sum_{m'} M_{m'} \ket{\psi}\ket{m'}.
\end{equation}
Operating the projector $\ident \otimes \ketbra{m}{m}$ on the post unitary state, results in
\begin{equation}
 ( \ident \otimes \ketbra{m}{m} )\, U \ket{\psi}\ket{0} = \sum_{m'} M_{m'} \ket{\psi}\ket{m}\, \braket{m}{m'} = M_m\ket{\psi}\ket{m}.
\end{equation}
In other words by applying the projector $\ketbra{m}{m}$ to the post interaction state, we have applied the measurement operator $M_m$ to the system and projected the ancilla into the measurement eigenstate $m$.  For a general system state $\rho$, the probability for obtaining the outcome $m$ is then given by
\begin{multline}
\mathbbm{P}(m) = \Tr\Big( \left(U\, \rho\otimes \ketbra{0}{0}\, U^\dag\right)\ \ident \otimes \ketbra{m}{m}\, \Big)\\
 = \Tr\left( \rho\otimes \ketbra{0}{0}\ \left( U^\dag\, \ident \otimes\ketbra{m}{m}\, U \right)\,\right).
\end{multline}
Applying a unitary transformation to an operator does not change its spectrum and so a unitary evolved projector is still a projector and in this case, one that is no longer acting solely on the ancilla.

The quantum probability description between the generalized measurement with operators $\set{M_m}$ is to use a Heisenberg picture version of the ``purification'' of that measurement.  Specifically the commuting set of operators that we are conditioning on is simply the unitarily evolved projectors $\set{ \proj_m \define U^\dag\,\ident \otimes \ketbra{m}{m} \,U }$.

It is not difficult to show that the partial trace of the posterior state $\rho|_m$ with the system operator $X$ is equivalent to the conditional expectation value of the operator $U^\dag (X \otimes \ident) U$ conditioned on the projector $\proj_m$, under the joint state $\rho \otimes \ketbra{0}{0}$.  In other words we have the equality
\begin{equation}
  \operatorname{tr}_{sys}(\rho|_m X) = \frac{\Tr\left(\rho\otimes\ketbra{0}{0}\ U^\dag( X \otimes \ident) U\  \proj_m \right)}{\Tr\left(\rho\otimes\ketbra{0}{0}\, \proj_m \right) } = \mathbbm{E}\left( U^\dag( X \otimes \ident) U\middle|\, \proj_m \right)      .
\end{equation}

\section{Quantum Filtering Theory \label{chMath:sec:quantumFiltering} }
Quantum filtering theory has a particularly grandiose title but in actuality it is not much more than what we have already developed here.  \citeauthor{bouten_introduction_2007}, wrote an award winning introduction to the problem quantum filtering and quantum stochastic calculus \cite{bouten_introduction_2007}.  This section does little more than quote their final results.  The quantum filter is in essence nothing more than the conditional expectation for a system observable $X$, based upon a light observable, \emph{e.g.} $Q^i_t$, after both have interacted though a unitary $U_t$.  The two light measurements that are typically considered are that of measuring an output quadrature, \emph{e.g.} $U^\dag_t Q^i_t U_t$, or a direct photon number measurement, \emph{e.g.} $U^\dag_t\, \Lambda^{ii}_{t}\, U_t$.  Here we have focused on classical and quantum diffusion, and so we will assume that we are measuring the quadrature $Q^i_t$.  In addition, to simplify the notation, we assume that we are considering a single field mode and will drop the label $i$.  More general expressions are not difficult to derive once the formalism is in place; for examples see \cite{bouten_quantum_2005}.

The quantum filter for time independent system observable $X$ is written as a time indexed map $\pi_t(X)$ and is the conditional expectation of the unitarily evolved operator $U_t^\dag X U_t$, conditioned on the (continuous) set of measurements of an output process $\set{Y_t}_{t \ge 0}$.  In the diffusive case
\begin{equation}
  Y_t \define U_t^\dag Q_t U_t = U_t^\dag (A_t + A_t^\dag ) U_t.
\end{equation}
When $U_t$ is given by a general single mode, it is the solution to the QSDE given in Appendix \ref{app:QSDEs}
 Eq. (\ref{appQSDEs:eq:generaldUt}).  In the 1D case with no scattering interactions we are able to calculate that,
\begin{equation}\label{chMath:eq:Yt}
    Y_t = Q_t + \int_0^t U_s^\dag ( L + L^\dag) U_s\, ds.
\end{equation}
The general expressions for the unitary evolution of any system operator $X$ and the fundamental field processes $A_t^j$, $A_t^{i\dag}$ and $\Lambda_t^{ij}$ are given in Appendix \ref{app:QSDEs}, Sec. \ref{appQSDEs:sec:Jt}.   

Sec. \ref{chQuLight:sec:QuWienerCharacteristic} showed that in vacuum expectation, $Q_t$ has statistics of a Wiener process.  Because of this one may be tempted to interpret Eq. (\ref{chMath:eq:Yt}) as the time integral of a system operator plus quantum white noise.  We urge the reader to avoid this temptation because, as Sec. \ref{appQSDEs:sec:Jt} shows, $U_t^\dag(L + L^\dag)U_t$ is generally a very complicated expression involving integrals with respect to $d \Lambda_t$, $d A_t$, and $d A_t^\dag$.  $Y_t$ is a fully coherent operator acting on the joint Hilbert space $\Hilbert \otimes \Fock(\hilbert_{[0, t]})$ and does not generally commute with $Q_t$.

It is, however, not difficult to show that $Y_t$ commutes with itself at different times, \emph{i.e.} $[Y_t, Y_s]  =0$ for any times $t$ and $s$.  Therefore a continuous observation of $Y$ between the times $0 \le s \le t$ makes a set of commuting observables $\set{Y_s\ :\ 0 \le s \le t}$.  This set of observations can then be used to form a commutative von Neumann algebra $\mathscr{Y}_t$.    The quantum filter $\pi_t(X)$ is then given by the conditional expectation
\begin{equation}
  \pi_t(X) \define \mathbbm{E}(U_t^\dag X U_t | \mathscr{Y}_t).
\end{equation}

Finding an expression for $\pi_t(X)$ requires implementing the conditional expectation in the form given in Eq. (\ref{chMath:eq:quantumConditionalDefinition}).  Note that in general, the conditional expectation depends upon the properties of the joint system field state and so you will arrive at different filtering equations if the system is in vacuum, \cite{bouten_introduction_2007}, a coherent state \cite{gough_quantum_2010}, or a state with nonclassical photon statistics \cite{gough_quantum_2012}.  The quadrature measurement of a single mode in vacuum expectation is arguably the simplest of all cases, and is what we will use exclusively here.  The bottom line result is that the quantum filter for any system operator $X$ is given by the recursive QSDE
\begin{equation}\label{chMath:eq:QuFilter}
    \begin{split}
        d\pi_t(X) & = \pi_t(\, \mathcal{L}_{00}(X) \,)\, dt \\
              &\ + \big( \pi_t(\, L^\dag X + X L \,) - \pi_t(L^\dag + L)\, \pi_t(X)\, \big ) \big( dY_t - \pi_t( L + L^\dag )\, dt \big).
    \end{split}
\end{equation}
with the initial condition $\pi_0(X) = \mathbbm{E}(X)$.  This is very analogous to the classical Kushner–-Stratonovich equation of nonlinear filtering \cite{bouten_introduction_2007}.  The operator map
\begin{equation}\label{chMath:eq:Liouvillian}
\mathcal{L}_{0 0}(X) = +i [H, X] + L^{ \dag} X L - \half L^{ \dag} L\, X - \half X\, L^{\dag} L
\end{equation}
is the $00$ Evens-Hudson map, (see Sec. \ref{appQSDEs:sec:Jt}) and is essentially the Heisenberg picture version of the Lindblad master equation.
A serious draw back to the quantum filter is that because it is recursive, it will very rarely close.   In order to propagate Eq. (\ref{chMath:eq:QuFilter}) for the operator $X$, we need to also calculate in parallel the filter for the operators $A  = L^\dag + L$, $B = L^\dag X + X L$, and $C = \mathcal{L}_{0 0}(X)$.  It's also highly likely that the space of possible system operators is not generated by simply these four operators.  By calculating $\pi_t( A)$, we will also need to know the filter for $\pi_t( \mathcal{L}_{0 0}(A) )$, which itself will likely generate more complicated operators.  Fortunately a saving grace is that we can invert this equation to find an effective ``noisy'' system operator $\rho_t$.  The equation of motion for $\rho_t$ is the conditional master equation which will discuss in Sec. \ref{chMath:sec:CME}.

Before doing so, we would like to highlight one important issue that makes a strong distinction between quantum and classical filtering.  In the classical case the filter is one of a couple of operations that one is interested in computing conditioned on an observation process $\set{y_t}$.  Another process that one is interested in is a \emph{smoother}, which is defined classically as
\begin{equation}
  \pi_{s,t}(x) \define\mathbbm{E}(x_s |\set{y_{t'}: 0 \le t' \le t} )\quad \text{for } s \le t.
\end{equation}
Classically this is a perfectly well defined thing to do, as long as $x_s$ is measurable with respect to the $\sigma$-algebra defining the global probability space $(\Omega, \mathcal{F}, \mathbbm{P})$.  One would generally then be tempted to define a quantum mechanical smoother,
\begin{equation}\label{chMath:eq:quantumSmoother}
  \pi_{s,t}(X) \define \mathbbm{E}(U_s^\dag X U_s |\mathscr{Y}_t)\quad \text{for } s < t.
\end{equation}
Unfortunately this object is not well defined for any system operator, $X$, because $U_s^\dag X U_s$ is not in the commutant of $\mathscr{Y}_t$.  To see why, consider that $Y_t = Q_t + \int_0^t U_r^\dag ( L + L^\dag) U_r\, dr$, which certainly has support upon the system Hilbert space via the time integral of $U_r^\dag ( L + L^\dag) U_r$.  There is no guarentee that $[U_s^\dag X U_s, U_t^\dag (L + L^\dag) U_t]$ for any $X$ and times $t,s$.  The reason that $U_t^\dag X U_t$ is in the commutant of $\mathscr{Y}_t$ is because we can show that $U_s^\dag Q_s U_s = U^\dag_t Q_s U_t$,  for  $s \ge t$.  This property then shows us that 
\begin{equation}
[U_t^\dag X U_t,\, Y_s] = [U_t^\dag X U_t,\, U^\dag_t Q_s U_t] = U_t^\dag [X ,\, Q_s] U_t = 0
\end{equation}
for $s \le t$.  This means that the post interaction system operator at time $t$ is able to be conditioned on \emph{past} measurements.  However the same ``advancement'' trick is not possible for the system observable, and therefore there is no guarantee that Eq. (\ref{chMath:eq:quantumSmoother}) is well defined.  If you simply threw caution to the wind and went through a smoothing calculation, even though $U_s^\dag X U_s$ is not in the commutant of $\mathscr{Y}_t$, then it quite possible that by conditioning you could take positive operators to negative ones or even Hermitian observables into non-Hermitian operators \cite{bouten_introduction_2007}.

\citeauthor{tsang_time-symmetric_2009} proposed a time-symmetric quantum smoother where one calculates a smoothing operation for a \emph{classical} signal imprinted on a quantum system \cite{tsang_time-symmetric_2009}.  In this case, the smoother is calculating a conditional estimate for the classical signal and therefore commutes with both the system and field operators.   In Chap. \ref{chap:QubitState} we wish to form an estimate for the system state at the initial time $t = 0$ given measurements up to time $t$.  One might be tempted to try and formulate a quantum smoothing equation $\mathbbm{E}(X|\mathscr{Y}_t)$, but as we just showed such an object is not in general well defined.  Therefore we have to resort to different methods.

\section{The Conditional Master Equation \label{chMath:sec:CME}}
Sec. \ref{chMath:sec:quantumFiltering} just showed how one could form a conditional estimate for system obervables based upon a measurement of an output light quadrature via the Heisenberg picture formalism of quantum probability.  A serious drawback is that the filtering equations are recursive and hardly ever close.  The saving grace of this is to convert to a randomized Schr\"{o}dinger picture and work with a \emph{Conditional Master Equation} (CME).

We know from Sec. \ref{chMath:sec:quantumFiltering}, that every commutative space of operators is mappable to a classical probability space.  We also know that from the definition of the conditional expectation, the filter $\pi_t(X) = \mathbbm{E}(U_t^\dag X U_t| \mathscr{Y}_t) $ is an operator in $\mathscr{Y}_t$.   And so if we generate a classical probability space $(\Omega, \mathcal{F}, \mathbbm{P})$ for $\mathscr{Y}_t$ then the filter $\pi_t(X)$ should be representable in that space.  Furthermore in a given experiment, the eigenvalues we receive from measuring $Y$ form a realization of a \emph{classical} stochastic process $y_t$ defined on that probability space.

What this means in practice is that we will now focus our attention to solely system variables a treat the measurement record $y_t$ as a classical stochastic process.  It is in this sense that we call the conditional master equation a \emph{semiclassical} equation.  Specifically, it treats the output measurements $\set{Y_t}_{t \ge 0}$ as a classical random variable while the system undergoes a noisy quantum evolution.   In our opinion, it cannot be over emphasized that this process has its origin as a quantum object and so not every operator will commute with $Y_t$ -- particularly past system observables.

With that warning to tread lightly, finding a semiclassical equation for a noisy system state $\rho_t$ is remarkably easy. Such an equation begins by enforcing that for every system operator $X$, we must have\footnote{Mathematically, this equivalence may seem strange as the left hand side is a scalar valued random variable while the right hand side is an operator in $\mathscr{Y}_t$.  The equivalence is made though the classical outcome $\omega$, that labels the set of eigenvalues we receive from the measurement.}
\begin{equation}\label{chMath:eq:conditionalDensity}
    \Tr( \rho_t\, X ) \cong \pi_t(X).
\end{equation}
To find an SDE for $\rho_t$, we simply notice two things.  In every term of Eq. (\ref{chMath:eq:QuFilter}), there is a coefficient $\pi_t( Y )$ of some operator $Y$ which is in turn relatable to $\Tr(\rho_t Y)$.  The second is that the only quantum stochastic differential in Eq. (\ref{chMath:eq:QuFilter}) is $dY_t$, which from Eq. (\ref{chMath:eq:Yt}), satisfies the quantum It\=o rule,
\begin{equation}
  d Y_t\, dY_t = dt.
\end{equation}
Therefore in the semiclassical mapping $dy_t$ also has the It\=o rule
\begin{equation}
  dy_t dy_t = dt.
\end{equation}
With these two observations we have,
\begin{multline}
        \Tr(d\rho_t\, X)  = \Tr(\rho_t\,\mathcal{L}_{00}(X)\,)\, dt \\
              + \big( \Tr \left(\rho_t\,( L^\dag X + X L ) \right) - \Tr \left( \rho_t\,( L^\dag + L ) \right)\, \Tr \left(\rho_t X \right)\, \big ) \big( dy_t - \Tr(\rho_t (L^\dag + L) )\, dt \big).
\end{multline}
We can use the cyclic property of the trace to decompose $\mathcal{L}_{00}(X)$ into an adjoint map acting on $\rho_t$,
\begin{equation}
    \Tr(\rho_t\, \mathcal{L}_{00}(X) ) = \Tr\left( \left( -i [H_t, \rho_t] +  L \rho_t L^\dag\,  -\half L^\dag L \rho_t - \half \rho_t L^\dag L \right)\, X \right).
\end{equation}
By making the same kind of transformation of the remaining terms and noting that it is true for any system operator $X$, we arrive at the conditional master equation (CME)
\begin{equation}\label{chMath:eq:CME}
    d \rho_t = -i [H_t, \rho_t]\, dt + \mathcal{D}[L](\rho_t)\, dt + \mathcal{H}[L](\rho_t)\, dv_t,
\end{equation}
with the initial condition is $\rho_0 = \rho(0)$ and we made the following definitions.
$\mathcal{D}[L](\rho_t)$ is the Lindblad operator map commonly found in open quantum systems and is defined as
\begin{equation}\label{chMath:eq:D}
    \mathcal{D}[L](\rho_t) \define L\, \rho_t\, L^\dag - \half L^\dag L\, \rho_t - \half \rho_t\,L^\dag L.
\end{equation}
$\mathcal{H}[L](\rho_t)$ is the state update map defined as
\begin{equation}\label{chMath:eq:H}
    \mathcal{H}[L](\rho_t) \define L\, \rho_t + \rho_t\, L^\dag -  \Tr( (L + L^\dag)\, \rho_t)\, \rho_t.
\end{equation}
This map shows how the state updates, weighted by the strength of the stochastic process,
\begin{equation}\label{chMath:eq:dv}
    dv_t = dy_t - \Tr( (L + L^\dag)\, \rho_t)\, dt.
\end{equation}

The random process $v_t$, called \emph{the innovation process}, plays an important role as it is the only random contribution to the CME.  In the next section we will review the proof that when everything about the measurement $y_t$ is properly specified, then $dv_t$ is a realization of a Wiener process.

\subsection{The innovation process \label{chMath:sec:Innovation}}

Here we will show that in the innovation process $v_t$ transforms $y_t$ into a Wiener process by subtracting off the conditional expected mean.  In classical probability, L\'{e}vy's theorem is an important result because it gives necessary and sufficient conditions for showing that a given process is in fact a Wiener process.  Roughly stated, if a stochastic process $m_t$ is a ``local martingale'' and obeys the It\={o} rule that $(dm_t)^2 = dt$ then it must be a Wiener process \cite{bouten_introduction_2007}.   Martingales are an important kind of stochastic process that play a crucial role in classical probability theory (see Sec. \ref{chMath:sec:processes}).  In essence it is a random process where the conditional mean of any future increment is zero \cite{oksendal_stochastic_2002, van_handel_stochastic_2007}.

The proof that $v_t$ is a Wiener process is given in theorem 7.1 of reference \cite{bouten_introduction_2007} and relies on some fundamental properties of the conditional expectation.  We quote this result in Lemma \ref{chMath:lemma:innovations}, for two reasons.  The first is simply because it is easily shown and is a rather elegant result.  The second is that Chap. \ref{chap:QubitState} uses the fact that $v_t$ is Wiener process only when $\rho_t$ is ``consistent'' with the actual statistics of $\set{y_t}_{t \ge 0}$.  Here consistency means that the correspondence $\Tr( \rho_t\, X ) \cong \pi_t(X)$ holds in the sense that $\pi_t(X)$ is a conditional expectation of $X$ with respect to $\mathscr{Y}_t$, under the true quantum state.  If $\rho_t$ does not exactly match $\pi_t(\cdot)$ because its initial condition is wrong or any number of other approximations, then $v_t$ will not generally have the statistics of a Wiener process.  See Secs. \ref{chQubit:sec:likelihood} and \ref{chQubit:sec:Simulations} for further discussion.

\begin{lemma}\label{chMath:lemma:innovations}
    In vacuum expectation, the quantum stochastic process $M_t \define Y_t - \int_0^t \pi_s(L + L^\dag)\, ds$ is an instance of a quantum Wiener process in that its finite dimensional statistics are independent mean zero Gaussian random variables with variances equal to the time differences.
\end{lemma}
\begin{proof}
    In Sec. \ref{chMath:sec:processes}, we review that the classical definition of a martingale is that it satisfies the property $\mathbbm{E}(m_t| \mathscr{F}_{s}) = m_s$ for $s \le t$.  In the quantum case this is equivalent to showing that $\mathbbm{E}(\,(M_t - M_s) \,|\mathscr{Y}_s) = 0$.  The reason for this is because the conditional expectation obeys the property that for every $K \in \mathscr{Y}_s$, $\mathbbm{E}(\,K \,|\mathscr{Y}_s) = K$.

    By the definition of the conditional expectation, we have that for every $K \in \mathscr{Y}_s$
    \begin{equation}
        \mathbbm{E}\left( \mathbbm{E}( M_t - M_s | \mathscr{Y}_s)\, K \, \right) = \mathbbm{E } \left(  (M_t - M_s) \, K \, \right).
    \end{equation}
    Substituting the definition of $M_t$,
    \begin{equation}\label{chMath:eq:InnovationMartingale}
        \mathbbm{E } \left(  (M_t - M_s) \, K \, \right) = \mathbbm{E } \left( (Y_t - Y_s) \, K \, \right) - \mathbbm{E }( \int_s^tds'\, \pi_{s'}(L + L^\dag) \, K \, ).
    \end{equation}
    Notice, however, that $\pi_{s'}(X) = \mathbbm{E}(U_{s'}^\dag X U_{s'}|\mathscr{Y}_{s'})$ we can again use the definition of the conditional expectation to convert the second term into an expectation of an integral of $U_{s'}^\dag(L+ L^\dag) U_{s'} $.  In Eq. (\ref{chMath:eq:Yt}) we solved for $Y_t$, and found that $Y_t = Q_t + \int_0^t ds'\, U_{s'}^\dag(L+ L^\dag) U_{s'}$.  After substituting that solution, Eq. (\ref{chMath:eq:InnovationMartingale}) simplifies to
    \begin{equation}
        \begin{split}
           \mathbbm{E } \left(  (M_t - M_s) \, K \, \right) &= \mathbbm{E } \left( (Y_t - Y_s) \, K \, \right) - \mathbbm{E } \left( \int_0^t ds'\, U_{s'}^\dag(L+ L^\dag) U_{s'} \, K \, \right) \\
                &= \mathbbm{E } \left( (Q_t - Q_s) \, K \, \right).
        \end{split}
    \end{equation}
    Any operator $K \in \mathscr{Y}_s$ is an operator which acts on the system Hilbert space and the Fock space associated with light operators defined for times $s' \in [0, s)$.  The operator $Q_t - Q_s$ acts on light field states defined on the time interval $[s, t]$.  This means that this final expectation value factorizes to show,
    \begin{equation}
        \mathbbm{E } \left(  (M_t - M_s) \, K \, \right)  = \mathbbm{E}(Q_t - Q_s)\, \mathbbm{E}(K) = 0.
    \end{equation}
    This is zero because the quadrature operator $Q_t$ is mean zero in vacuum, and so $M_t$ is indeed a martingale, when we condition on $\mathscr{Y}_s$.  The proof is finished by simply observing that $dM_t dM_t = dt$ and so $M_t$ is a quantum Wiener process by L\'{e}vy's thoerem.  \qedhere
\end{proof}

\subsection{The It\={o} correction in the conditional master equation\label{chMath:sec:QuantumItoCorrection} }

The quantum filter $\pi_t(\cdot)$ is given by an It\=o form quantum stochastic differential equation and therefore the conditional master equation is a semiclassical It\=o equation.  In addition to an It\=o integral, there is also a Stratonovich integral where the rules of standard calculus still apply, but the statistical properties are more subtle (see Appendix \ref{app:SDEs} for their respective definitions).  While the two forms of integration are distinct, they are related by a conversion formula, resulting in the ``It\=o correction'', derived in Appendix \ref{appSDEs:sec:ItoConversion}.   In Chap. \ref{chap:projection}, we are required to work with a conditional master equation written as a Stratonovich integral and so we derive this conversion here.

For a general measurement operator $L$ and Hamiltonian $H$, the conditional master equation is
\begin{equation}
    d\rho_t = -i[H, \rho_t]\, dt + \mathcal{D}[L](\rho_t)\, dt + \mathcal{H}[L](\rho_t)\, dv_t.
\end{equation}
The first two terms are simple deterministic integrals and are unaffected by the choice of stochastic integral and can be ignored.   The integrand in It\=o integral is the conditioning map and for reference is,
\begin{equation}
    \mathcal{H}[L](\rho_t) =  L\, \rho_t + \rho_t\, L^\dag - \Tr( L \rho_t + \rho_t L^\dag  )\, \rho_t.
\end{equation}
For the remainder of this section we will suppress the parameterizing argument and simply write $\mathcal{H}(\rho_t)$.

A one-dimensional It\=o integral that is typically considered in an It\=o-Stratonovich conversion has a differential
\begin{equation}
  dx_t = b(x_t) dw_t,
\end{equation}
for a smooth integrand $b(x)$.  When written as a Stratonovich equation, this differential is notated as
\begin{equation}
  dx_t = b(x_t) \circ dw_t.
\end{equation}
The It\=o correction is what results when you enforce that both integrals must give the same process $x_t$, and the final result is that \begin{equation}
     b( x_t ) \circ d w_t =  b( x_t )\, d w_t + \half\, \frac{d b}{d x}(x_{t})\,b( x_{t} )\, d t.
\end{equation}
The additional term is known as the It\=o correction.

To immediately apply this result to the conditional master equation would involve defining what it means to take the derivative the $\mathcal{H}(\rho )$ operator with respect to $\rho$.  Rather than defining a calculus of super--operators, we will return to the roots of the relation (see Appendix \ref{appSDEs:sec:ItoConversion}) and write the correction as
\begin{equation}\label{chMath:eq:CMEItoCorrection}
    d I = \mathcal{H}(\rho_t) \circ dv_t - \mathcal{H}(\rho_t) dv_t = \left( \mathcal{H}(\rho_t + \half\, d\rho_t) - \mathcal{H}(\rho_t)\right) dv_t.
\end{equation}
The map $\mathcal{H}$ is unfortunately not a linear operator in $\rho_t$ and so the integrand on right hand side is not simply $\mathcal{H}( \half d\rho_t)$.  After a little algebra we find that
\begin{multline}
d I = \half \Big( L\, d \rho_t +d \rho_t \, L^\dag  -  \Tr( L \rho_t + \rho_t L^\dag ) d \rho_t \\
- \Tr( L\, d\rho_t + d\rho_t\, L^\dag  )\, ( \rho_t  + \half\, d  \rho_t )\, \Big) dv_t.
\end{multline}

To simplify this expression into a final form, we will substitute the It\={o} equation expression for $d\rho_t$ and apply the It\={o} rule that $dv_t dv_t = dt$ with all other differential products being zero.  This means that when substituting $d\rho_t$ we need to only use the stochastic term as any deterministic term will result in a product $dt dv_t = 0$.  Furthermore any term with two powers of $d\rho_t$ will also be zero as that will result in three powers of $dv_t$. The simplified expression is then,
\begin{equation}
\begin{split}
   dI =&\, \half\left( L\, \mathcal{H}(\rho_t) +\mathcal{H}(\rho_t) \, L^\dag -  \Tr\left( (L + L^\dag) \, \mathcal{H}(\rho_t)  \right)  \rho_t  -  \Tr( (L  +  L^\dag) \rho_t )\, \mathcal{H}(\rho_t)\,  \right) \,  dt\\
    \enifed&\, \mathcal{I}_c[L](\rho_t)\, dt.
\end{split}
\end{equation}

Substituting the definition of $\mathcal{H}[L](\rho_t)$, the It\={o} correction map, $\mathcal{I}_c[L](\rho_t)$, simplifies to
\begin{equation}\label{chMath:eq:itoCorrectionMap}
\begin{split}
   \mathcal{I}_c[L](\rho_t) =&\, \left( L \rho_t L^\dag + \half L^2\rho_t + \half \rho_t L^{\dag\, 2 } \right)\\
    &-\, \left( \langle L^\dag L \rangle + \half \langle L ^2 \rangle + \half \langle L ^{\dag 2} \rangle \right) \rho_t\\
    &-\, \langle L + L^\dag \rangle\, (L\, \rho_t + \rho_t\, L^\dag - \langle L + L^\dag \rangle\, \rho_t )
\end{split}
\end{equation}
where $\langle X \rangle = \Tr (X \rho_t)$.

Ultimately the Stratonovich form of the conditional master equation is then given by
\begin{equation}\label{chMath:eq:CMEStratonovich}
    d \rho_t = -i [H_t, \rho_t]\, dt + \mathcal{D}[L](\rho_t)\, dt - \mathcal{I}_c[L](\rho_t)\, dt + \mathcal{H}[L](\rho_t)\circ dv_t.
\end{equation}

\subsection{The conditional Schr\"{o}dinger equation  \label{chMath:sec:CSE}  }

In this chapter we have focused solely on the interpretation of quantum mechanics in terms of probability spaces.  That description lead to a quantum conditional expectation and a quantum filter, which is described in a Heisenberg picture.  From that Heisenberg picture description we found a conditional master equation (CME).  When the state of the system is pure and the dynamics are such that it will remain pure, then propagating a full density matrix is unnecessary and a conditional Schr\"{o}dinger equation (CSE) is sufficient.   Chap. \ref{chap:QubitState} uses this fact for computational efficiency and therefore we include the general expression for a CSE based upon the CME in Eq. (\ref{chMath:eq:CME}).  The details of the conversion can be found in \cite{chase_parameter_2009}.

The CME gives the evolution for $\rho_t$ in terms of an It\={o} differential $d \rho_t$, Eq. (\ref{chMath:eq:CME}).  Any density matrix whose purity is $1$ can be represented by an outer product of a normalized state vector in Hilbert space \cite{nielsen_quantum_2000},
\begin{equation}
  \rho_t = \ketbra{\psi_t}{\psi_t}  \quad \text{if and only if } \Tr(\rho^2)= 1.
\end{equation}
Furthermore $\ket{\psi_t}$ is unique up to an arbitrary constant phase.  While we have worked quite hard to derive the CME and give it physical meaning, practically speaking it is ``nothing'' more than a matrix valued stochastic differential equation defined on a classical probability space.  Therefore a method for moving from a CME to a CSE is to hypothesize the existence of a random state vector $\ket{\psi_t}$ satisfying some vector valued SDE $d \ket{\psi_t}$ and then solve for the differential that give the differential $d (\ketbra{\psi_t}{\psi_t})$ equal to the CME.

We note that this is not a standard derivation.  Typically in quantum optics, one first derives a stochastic Schrodinger equation via an unraveling of a master equation that considers photon counting and then takes a diffusive limit \cite{gardiner_quantum_2004}.  Having already developed the CME from the quantum filter, it is much simpler to perform the above calculation, rather than including an independent derivation.   The resulting equations are identical.

The derivation of $d\ket{\psi_t}$ is not difficult as we can see that the only random process that enters the Eq. (\ref{chMath:eq:CME}) is through the innovations $v_t$ and it does so linearly.  We also know that $v_t$ satisfies the It\={o} rule, $dv_t dv_t = dt$.  Therefore a reasonable form for $d\ket{\psi_t}$ is
\begin{equation}
  d\ket{\psi_t} = A_t \ket{\psi_t} \,dt + B_t\, \ket{\psi_t}\, dv_t.
\end{equation}
for some time-adapted but possibly state dependent operators $A_t$ and $B_t$.
The adjoint of this equation is then
\begin{equation}
  d\bra{\psi_t} = \bra{\psi_t} A^\dag_t \,dt +  \bra{\psi_t}\, B^\dag_t\, dv_t.
\end{equation}
And so we need to solve for $A_t$ and $B_t$ subject to the constraint,
\begin{equation}
  d \rho_t = d\ket{\psi_t}\ \bra{\psi_t} + \ket{\psi_t}\ d\bra{\psi_t} + d \ket{\psi_t}\ d \bra{\psi_t}.
\end{equation}
Doing so is not too difficult and the operators turn out to be
\begin{equation}
  A_t = -i H_t - \frac{1}{2} \left( L^\dag L - 2 \expect{L^\dag} L + \expect{L^\dag} \expect{L} \right)
\end{equation}
and
\begin{equation}
  B_t =  L - \expect{L}.
\end{equation}

Traditionally the operator $L$ is Hermitian, and so if we choose our favorite example of $L = \sqrt{\kappa}\, J_z$ then
\begin{equation}\label{chMath:eq:CSE}
  d \ket{\psi_t} = \Big(\,-i H_t - \half \kappa \left(J_z  - \expect{J_z}\right)^2\, \Big) \ket{\psi_t}\, dt + \sqrt{\kappa}\,\big( J_z - \expect{J_z} \big)\, \ket{\psi_t}\,  d v_t
\end{equation}
with
\begin{equation}
  d v_t = d y_t - 2 \sqrt{\kappa} \expect{J_z}\, dt.
\end{equation}
This is the equation used in Chaps. \ref{chap:projection} and \ref{chap:QubitState}.

\chapter{Projection Filtering for Qubit Ensembles\label{chap:projection}}

This chapter derives an approximate form for the conditional dynamics of an ensemble of $n$ qubits under the assumption that the state will remain nearly an identical separable state.  We assume that the system is undergoing a diffusive measurement of the collective angular momentum operator $J_z$ while simultaneously experiencing strong global rotations.  The approximation is made by formulating a projection filter from the exact conditional master equation.  The projection is made through the technique of orthogonal projections in differential geometry.  Here we identify the space of identical separable states as a Riemannian manifold and then project the conditional master equation into its tangent space.  We also review the elements of differential geometry that make such a mapping possible.  Finally we test the accuracy of the projection filter numerical by comparing it to simulations of a stochastic schrodinger equation.  We find that it matches the conditional mean spin projections to within a $5 \%$ RMS error.

\section{Introduction}
Numerical integration of a conditional master equation is generally a resource intensive exercise.  Specifying a general mixed state for a $d$-dimensional quantum system requires $d^2 - 1$ real parameters.  Furthermore the total dimension of a many body system grows exponentially.  A system of $n$ qubits generates a $2^n$-dimensional Hilbert space, requiring $2^{2 n} - 1$ parameters.  This ``curse of dimensionality'' is true even in a unconditioned system and so physicists often search for symmetries that allow for a more efficient description.  The nonlinearity in the conditional master equation means that a number of symmetries that are often preserved in an uncondition map are no longer exploitable.

A projection filter is a tool that was developed in the context of classical filtering theory and provides a general method for constraining nonlinear estimators to remain in a lower dimensional space \citep{brigo_differential_1998, brigo_approximate_1999}.  Within the past decade these tools have also been applied to quantum systems, specifically for cavity QED systems \citep{van_handel_quantum_2005, mabuchi_derivation_2008, hopkins_reduced_2009, nielsen_quantum_2009}, collective spin systems \citep{chase_parameter_2009}, and low rank approximations for general master equations \cite{bris_low_2012}.  The flexibility of the projection method is provided by its formulation in the language of differential geometry.  In the quantum framework we have a high, possibly infinite, dimensional manifold representing the space of possible states.  It is often the case that the system is initialized in a state with a large amount of symmetry thereby initially allowing an efficient, lower dimensional representation.  The project filter modifies the exact evolution in such a way as to constrain the system to remain in the lower dimensional submanifold.   It does so by projecting the differential into the lower dimensional tangent space.

Here we focus on an ensemble of $n$ qubits initially prepared in an identical tensor product state.  In other words, the total state of the system $\rho_\tot$ is initialized in as a $n$-fold tensor product of a single qubit state $\rho$,
\begin{equation}
  \rho_\tot = \rho^{\otimes n}.
\end{equation}
Clearly this is a highly symmetric and easily represented state, as a single qubit state requires only $3$ parameters to be specified uniquely.  If  the master equation acts on each qubit individually then the total system will remain in an identical separable state for all future times.  However for a joint qubit system undergoing a weak, diffusive measurement of the \emph{collective} angular momentum variable $J_z$, the conditional master equation is generally entangling.  In the long time limit, this kind of measurement most often results in the system projecting into a nonseparable Dicke state.

In this chapter we demonstrate, through numerical simulation, that if the system also undergoes strong, randomized rotations in addition to the collective measurement then the system will remain nearly sparable.  Under this assumption that this is the case, we apply the technique of projection filtering to the conditional master equation so that it maps identical separable states to identical separable states.

\subsection{An introduction to differential projections}

The general technique of differential projections can be understood though the following example.   Consider an ordinary scalar function defined on three dimensions, $f(x, y, z)$. The chain rule shows that the differential for $f$ is
\begin{equation}
  df = \frac{\partial f}{\partial x}\, dx + \frac{\partial f}{\partial y}\, dy + \frac{\partial f}{\partial z}\, dz.
\end{equation}
Suppose that we have a particle with position vector $\Vx(t)$, and at each time $t$ we evaluate $f(x(t), y(t), z(t) )$.  In order to have a complete description for $f$ we clearly need to keep track of all three components because a change in $x$, $y$ or $z$ induces a change in $f$.  Now suppose that keeping track of $z$ is too much of a hassle and we are only interested in tracking $x(t)$ and $y(t)$.  The question posed by the projection filter is, ``how should we modify $f$ so that we only need to track $x$ and $y$?''  The answer comes from the fact that if $\frac{\partial f}{\partial z} = 0$ everywhere then $f$ doesn't change with $z$ and ultimately $z$ can be ignored.  Therefore the modification we should make it set the gradient of $f$ to point only in the $xy$-plane, \emph{i.e.} set $\frac{\partial f}{\partial z} = 0$.  This modification is the differential projection of $f$.  Therefore we have a modified function $\restrict{f}{\set{x,y}}$, whose differential is simply,
\begin{equation}
  \restrict{df}{\set{x,y}} = \frac{\partial f}{\partial x}\, dx + \frac{\partial f}{\partial y}\, dy + 0\, dz.
\end{equation}
The difficulty in forming a projection filter is that $f$ is not usually written in terms of, $\set{x, y, z}$, but instead some other set of parameters, $\set{x', y', z'}$, or even just $t$.  Furthermore the desired subspace might be some complicated 2D surface with parameters $v$ and $w$.   It is very likely that $v$ and $w$ may not even be orthogonal, at least not in the same sense $x$, $y$ and $z$ are orthogonal.  The first challenge in developing a projection filter is to give the desired objective a geometric interpretation.

\subsection{The conditional master equation}

Before embarking on a description of the geometry of quantum states, we will first collect all of the necessary equations from previous chapters here for a single point of reference.  Sec. \ref{chMath:sec:CME} found that the state of an atomic system conditioned on a continuous diffusive measurement is easily represented by the conditional master equation (CME) given by the It\={o} differential,
\begin{equation}\label{chProj:eq:GeneralCME}
    d \rho_t = -i [ H,\, \rho_t] dt + \mathcal{D}[L](\rho_t) dt + \mathcal{H}[L](\rho_t) dw_t.
\end{equation}
(See Appendix \ref{app:SDEs} for a review of classical stochastic differential equations.)   The dissipation and conditioning maps, $\mathcal{D}[L](\cdot)$ and $\mathcal{H}[L](\cdot)$, are parameterized by the measurement operator $L$ and are defined as
\begin{equation}\label{chProj:eq:dissipator}
    \mathcal{D}[L](\rho) = L \rho L^\dag  - \half L^\dag L \rho - \half \rho L^\dag L
\end{equation}
and
\begin{equation}\label{chProj:eq:conditioning}
    \mathcal{H}[L](\rho_t) = L \rho_t  + \rho_t L^\dag  - \Tr( (L + L^\dag) \rho_t)\, \rho_t.
\end{equation}
Often we will omit the parameterizing argument and simply write $\mathcal{D}(\rho_t)$ and $\mathcal{H}(\rho_t)$.  Also note that $\hbar$ has been set equal to one, so that the Hamiltonian operator $H$ has units of frequency and the measurement operator $L$ has units of root frequency.

Note that Sec. \ref{chMath:sec:CME} used a slightly different notation, referring to the innovation as $dv_t$ rather than $dw_t$.  Sec. \ref{chMath:sec:Innovation} showed that innovation computed from the measurement record $y_t$ has the statistics of a Wiener process, when the initial condition $\rho_0$ coincides with the ``true'' initial state.  In Chap. \ref{chap:QubitState} this will not be the case, however here we are assuming that the initial condition is known, and in particular, that it can be written as $\rho^{\otimes n}$.  Therefore, throughout this chapter we will consider the innovation to be Wiener process and write it as $dw_t$.  Sec. \ref{chMath:sec:Wiener} reviews the statistical and defining properties of the Wiener process.

The physical system that we have in mind is the idealized linear Faraday interaction in Sec. \ref{chQuLight:sec:faraday}, meaning the measurement operator is
\begin{equation}\label{chProj:eq:JzMeasurement}
    L = \sqrt{\kappa}\, J_z
\end{equation}
where $\kappa$ is a constant rate.  In addition to this measurement, we consider applying a uniform but time varying magnetic field, leading to the Hamiltonian
\begin{equation}\label{chProj:eq:controlHt}
    H = f^x(t) J_x +f^y(t) J_y + f^z(t) J_z.
\end{equation}
The control fields $f^i(t)$ are assumed to be real valued, deterministic functions of time\footnote{ In Chap. \ref{chap:QubitState} the control fields are written as $\V{b}(t)$, however in this chapter the coordinates $b^i$ indicate the projected coefficients for the stochastic terms, so here we use $f^i(t)$ instead.}.

For reasons made apparent in Sec. \ref{chProj:sec:StochCalcManifolds}, we will also need to work with the \emph{Stratonovich} form of the CME,
\begin{equation}\label{chProj:eq:StratonovichCME}
    d \rho_t = -i [ H,\, \rho_t] dt + \mathcal{D}[L](\rho_t) dt - \mathcal{I}_c[L](\rho_t) dt + \mathcal{H}[L] \circ dw_t.
\end{equation}
The conversion from the It\={o} form generated the It\={o} correction map, derived in Sec. \ref{chMath:sec:QuantumItoCorrection}, which is
\begin{equation}\label{chProj:eq:ItoCorrectionMapIntro}
\begin{split}
   \mathcal{I}_c[L](\rho_t) =&\,  L \rho_t L^\dag + \half L^2\rho_t + \half \rho_t L^{\dag\, 2 } \\
    &-\, \left( \expect{L^\dag L} + \half \expect{ L ^2 } + \half \expect{ L ^{\dag 2} } \right) \rho_t\\
    &-\, \expect{ L + L^\dag }\, (L\, \rho_t + \rho_t\, L^\dag - \expect{ L + L^\dag }\, \rho_t ).
\end{split}
\end{equation}
Here the expectation value of the operator $X$ has been written as $\expect{X} \define \Tr(\rho_t\, X)$.

\section{Differential Manifolds \label{chProj:sec:DifferentialManifolds} }
A manifold $\man{M}$ is most generally a continuous set of point that can be locally mapped to a $d$-dimensional Euclidean space.  In a neighborhood of any point in $\man{M}$ we can define a smooth mapping points in that neighborhood to a flat space of dimension $d$.  How smooth this mapping needs to be, often depends upon the author and the context, generally it must be smooth enough so that the tools of differential calculus can be applied.  The concept of \emph{smooth} is quite at odds with the random nature of Brownian motion, as the Wiener process is provably nondifferentiable with probability one.  Here we will be ultimately considering random trajectories on a differential manifold.  The resolution between these two conflicting notions is that while a diffusive trajectory is nondifferentiable, it is a trajectory in a smooth space.  \emph{i.e.} a two-dimensional Brownian motion is not a differentiable curve, but it is defined on a 2-D plane which is smooth.

The specific manifold we need is the space of all valid density operators for $n$ qubits.  For a single qubit, the Bloch vector defines a perfectly respectable one-to-one mapping between a quantum state and the 3-dimensional Euclidian ball.  The conditional master equation then has a representation as a diffusive trajectory within the Bloch ball.  Defining an equivalent representation for a $d-$dimensional quantum state is nontrivial and is still the subject of current research.  While there does exist an equivalent mapping to a ball living in a $2^d -1$-dimensional space, the boundary and smoothness of this mapping is quite complex and not well understood \cite{kimura_bloch-vector_2005}.  Here we will only be interested in a geometric representation of states that can be written as $n$ copies of a single qubit state.  Ultimately a Bloch vector representation is sufficient for our purposes.


\subsection{Tangent spaces \label{chProj:sec:tangentSpaces}}

The differential projection we ultimately want to preform requires a deeper understanding of how to define a gradient in a more abstract setting.  A key conceptual point is that we make an association between basis vectors in a $d$ dimension space and the partial derivatives we can take of a function defined on the manifold.  Specifically, a point $p$ in the manifold $\man{M}$ is representable by the coordinates $\set{x^1, x^2, \dots, x^d}$.  Any smooth function $f$ on the manifold, evaluated at this point $p$ can therefore also be represented as a function of these coordinates, $f(x^1, x^2,\dots x^d)$.  At the point $p$ the partial derivative of the function $f$ with respect to the coordinate $x^i$ defines the rate of change of $f$ as $x^i$ is varied, \emph{i.e.} it defines a line tangent to $f$ pointing in the direction of $x^i$.

The relation between partial derivatives and vectors can be formed by associating the basis element $\Ve_i$ with the partial derivative operator $\dydx{}{x^i}$.  Differential geometry is concerned with defining structures that are independent of any given coordinate system.  Calculating a partial derivative with respect to a different coordinate system, $\set{y^i}$, is easily accomplished by applying the usual chain rule,
\begin{equation}
 \dydx{}{y^i} = \dydx{x^j}{y^i} \dydx{}{x^j}.
\end{equation}
The coordinate independent quantity here is the space of all possible partial derivatives we could take at this point.  At first glance this may seem like a rather large object, however the chain rule just showed that a partial derivative in one basis is simply a linear combination of partial derivatives in another basis.  Therefore the space of all possible partial derivatives is simply the linear span of partials taken with respect to some basis.  This space is called the tangent space of $\man{M}$ at point $p$, denoted by,
\begin{equation} \label{chProj:eq:tangentSpace}
    T_p \man{M} = \lspan{ \left. \tfrac{\partial}{\partial x^i} \right|_p\ :\ i = 1\dots d}.
\end{equation}
Note that the tangent space is a $d$-dimensional vector space, as we are taking linear combinations of $d$ basis vectors.  Often we will discuss a \emph{directional derivative}, meaning that we will be taking a derivative in the direction of another point in the manifold.  But as this could have any relation to a given coordinate system, the direction derivative defines a vector in the tangent space.  Another useful bit of jargon is that if you have a tangent vector defined for ever point in the manifold then this defines a vector field.

\subsection{Riemannian Metrics and orthogonal projections\label{chProj:sec:Metrics}}

The tangent space $T_p \man{M}$ defines the set of all possible partial derivatives one could make at the point $p$.  However it does not describe how those derivatives are related.  While in a Cartesian basis we have a sense that $\Ve_x$ is orthogonal to $\Ve_y$, in general its hard to tell how the arbitrary vector $\Ve_u$ is related to $\Ve_w$.  The missing element is a \emph{metric}, $\inprod{\cdot}{\cdot}_p$, describing a positive definite inner product between any two tangent vector fields.  At each point $p$ we can take the dot product of $\Ve_u$ and $\Ve_v$ to see how they are related.  If the space is Euclidean, then the metric well report the fact that $\Ve_x$ is orthogonal to $\Ve_y$, which is not true in general.

A \emph{Riemannian manifold} is a manifold $\man{M}$ that is equipped with a metric that varies continuously between different points.  While in a Euclidean space the inner product between two vectors doesn't change between different points, this is not true in a general space leading to much richer geometries.  For a basis of vectors $\set{\Ve_j}$ spanning the tangent space $T_p\man{M}$ the metric at that point can be written as a $d \times d$ matrix with components,
\begin{equation}
  g_{ij}(p) \define \inprod{\Ve_i}{\Ve_j}_{p}.
\end{equation}
In addition to being positive definite, a metric is also symmetric in that $\inprod{\Ve_i}{\Ve_j}_{p} = \inprod{\Ve_j}{\Ve_i}_{p}$.

A metric gives a notion of two vectors being orthogonal and from that we are able to make an orthogonal projection.  This is crucially important as we wish to project the conditional master equation into the tangent space of states that are $n$ copies of a single qubit state.  For a Euclidean space, the orthogonal projection of the vector $\V{v} = v^1 \Ve_1 + v^2 \Ve_2 + v^3 \Ve_3$ onto the $XY$ plane is trivial to compute, as it simply discards the $\Ve_3$ component.   Given a metric and a general manifold we can make a similar formulation.

Suppose for a Riemannian manifold $\man{M}$ we have a submanifold $\man{N} \subseteq \man{M}$ of dimension $n \le d$.  Without explicitly constructing an orthogonal basis for every tangent space $T_p \man{M}$, we would like to find a map that discards the vector components orthogonal to $T_p\man{N}$.  This can easily be done, given a basis of vectors $\set{ \V{v}_i \, :\, i = 1,\dots, n}$ that span $T_p\man{N}$.  The metric $\inprod{\cdot}{\cdot}_{p}$ taken from $\man{M}$, can equally well be applied to $\man{N}$ as their tangent spaces overlap.   Applying this metric to $\set{\V{v}_i}$ we have the $n \times n$ matrix with elements
\begin{equation}
  g_{ij}(p) = \inprod{\V{v}_i}{\V{v}_j}_{p}.
\end{equation}
As the metric is positive definite,  this matrix is invertible whose entries are often written as, $g^{ij}(p) \define \big(\, \V{g}(p)\, \big)^{-1}_{ij}$.  We can now show that that the map $\Pi_{\man{N}}\ :\ T_p\man{M} \rightarrow T_p\man{N}$,
\begin{equation}\label{chProj:eq:projMap}
  \Pi_{\man{N}} ( \cdot )  = g^{ij}(p) \inprod{\V{v}_j}{\cdot\, }_{p}\, \V{v}_i
\end{equation}
is equivalent to discarding the component of $\V{w}$ orthogonal to $T_p\man{N}$.

The projection map should operate as the identity for any vector $\V{u} \in T_p\man{N}$.  To check that this is true, $T_p\man{N} = \lspan{ \V{v}_i\, :\, i = 1\dots n}$, and so $\V{u}$ be written as $\V{u} = u^k\, \V{v}_k$ for some coefficients $u^k$.  Then we have
\begin{equation}
 \begin{split}
  \Pi_{\man{N}}( \V{u} )  &= g^{ij}(p)\, \inprod{\V{v}_j}{ u^k \V{v}_k }_{p}\, \V{v}_i\\
   &= g^{ij}(p)\, \left(\,u^k\, g_{jk}(p)\, \right) \, \V{v}_i\\
   &= u^k\, g^{ij}(p)\,g_{jk}(p)\, \V{v}_i\\
   &= u^k  \delta^{i}_{k} \, \V{v}_i = \V{u}.
 \end{split}
\end{equation}
$\Pi_{\man{N}}$ should also return zero for every vector orthogonal to $T_p\man{N}$.  This is also easy to check, as for every $\V{v}^\perp$ in the orthogonal complement of $T_p\man{N}$, we have $\inprod{\V{v}}{\V{v}^\perp}_p = 0$, if $\V{v} \in T_p \man{N}$.  Therefore,
\begin{equation}
  \Pi_{\man{N}}( \V{v}^\perp )  = g^{ij}(p) \inprod{\V{v}_j}{ \V{v}^\perp }_{p}\, \V{v}_i = 0.
\end{equation}
But as $T_p\man{M} = T_p\man{N} \cup (T_p\man{N})^c$, $\Pi_{\man{N}}$ is the correct mapping.

Note that Eq. (\ref{chProj:eq:projMap}) required only specifying the metric $g_{ij}(p)$ on the submanifold $\man{N}$ and does not require an explicit representation for tangent vectors outside of this subspace.  This is the reason why it is not necessary to find an explicit mapping between the space of $n$-qubit density matrices to a $2^d - 1$-dimensional Euclidean space in order to use the projection filtering methods.  All we need is the valid metric for density matrices and a spanning set of tangent vectors in the submanifold we wish to project onto.

\subsection{Differentials on abstract manifolds \label{chProj:sec:Differentials}}

As the conditional master equation is written in terms of stochastic differentials, we must see how a differential operates in a geometric context.  In multivariable calculus the fundamental object is the differential of the coordinates, \emph{e.g.} $dx$, $dy$, \emph{etc.}.  In a more abstract space, its difficult to intuit what the differential means.  For instance how would one define a differential of a matrix, say the Pauli matrix $\sigma_x$.  Would it be the differential of its entries, the differential of its eigenvalues or maybe even a differential of both the eigenvalues and eigenvectors?  The solution to this problem is to consider the differential not the individual points themselves, but the differential after the application of a smooth map to the Euclidean space.  The differential in the abstract space is then inferred from the Euclidean differential.  This process of inference  is called the \emph{pullback}, in that you are pulling back from the original mapping.   Our ultimate goal is to interpret the conditional master equation $d \rho_t$ in the language of differential geometry and so we need to understand how it relates to a Euclidean mapping.

%
%

Basic multivariable calculus shows that the total differential of the scalar function $f$ is given by, 
\begin{equation}\label{chProj:eq:df}
    df = \dydx{f}{x^i} d x^i.
\end{equation}
A differential can also be view as a linear map acting on tangent vectors.  The action of $df$ on the tangent vector $\dydx{}{x^i}$ is defined to be
\begin{equation}\label{chProj:eq:dfOnDi}
    df\left( \frac{\partial}{\partial x^i} \right) \define \dydx{f}{x^i}.
\end{equation}
While this may seem a bit obtuse at first, it is actually a very useful concept.  To see why, consider the most basis function we can consider, namely the coordinate function $x^i$.  The differential $d x^i$ has an action on the basis vector $\frac{\partial}{\partial x^j}$,
\begin{equation}
    dx^i\left(\dydx{}{x^j}\right) = \dydx{x^i}{x^j} = \delta^i_j.
\end{equation}
This shows that the coordinate differential is biorthogonal to $\frac{\partial}{\partial x^j}$, and therefore can be thought of as a dual basis vector.  When defining the tangent space in Sec. \ref{chProj:sec:tangentSpaces} we found that the partial derivatives spanned that space and the coordinate transformation coefficients were simply linear expansion coefficients.  The same is true for the differential $df$ in that $\dydx{f}{x^i}$ are the expansion coefficients in a dual space spanned by the basis vectors $dx^i$.  The dual space is often called the \emph{cotangent space}.

A differential of a function between two spaces can also be defined.  While we just considered the differential of a scalar valued function, $f$, we can also consider the differential of a vector, matrix, or operator valued function.  When $df$ acted on a basis vector $\dydx{}{x^i}$ it returned a scalar value $\dydx{f}{x^i}$, but with a more general mapping function the returned value should be something other than a scalar.   It turns out that when you have a function $\varphi: \R^3 \rightarrow \man{M}$, the differential of this is a function $d\varphi: T_\Vx \R^3 \rightarrow T_{\varphi(\Vx)} \man{M}$.  The point being that when a function maps one space into another, the differential maps tangent vectors to tangent vectors.  This is best illustrated though a concrete example, which we will give in Sec. \ref{chProj:sec:projUnconditional}, after formulating the Bloch vector representation as a Riemannian manifold.

\subsection{Stochastic calculus on differential manifolds \label{chProj:sec:StochCalcManifolds} }
There seems to be a fundamental inconsistency between a \emph{smooth}, infinitely differentiable manifold and the nowhere differentiable path of a Wiener process.  From the Wong-Zakai theorem, (see Appendix \ref{app:QuWongZakai}) we know that if there exists a smooth, ordinary differential equation that limits to a stochastic differential equation, then the limit should be interpreted as a \emph{Stratonovich} SDE.  When trying to incorporate stochastic differential equations into the language of differential forms, one approach would be to enforce a smooth approximation, apply the differential technique and then take a stochastic limit at the end.  However a skeptical mathematician might wonder if such a result could be believed as the end result might depend heavily on how the smooth approximation was made.

At a practical level, the second order nature of the It\={o} rule is difficult to reconcile with the notion of constrained motion on submanifold.  A simple example of this is made in \cite{brigo_approximate_1999}, which we will reproduce here.  Consider the ordinary differentials,
\begin{equation}
\begin{split}
  dx &=\, dt \\
  dy &=\, 2 t\, dt.
\end{split}
\end{equation}
We can easily see that this describes the parabola $y = x^2$, which can also be considered an immersion of a one-dimensional manifold into $\R^2$.  (The parameterizing function is $\varphi(t) = (t, t^2)$.)  Furthermore we can see that the coefficients of these equations, describe a vector, $(1, 2 x)$, tangent to the parabola.  Were these equations used to describe the evolution of a system whose initial condition is on the parabola, we expect the system to remain on this submanifold.

In contrast, consider an equivalent system of \emph{It\={o}} stochastic differential equations
\begin{equation}
\begin{split}
  dx_t &=\, dw_t \\
  dy_t &=\, 2 x_t\, dw_t.
\end{split}
\end{equation}
The coefficients still describe a vector, $(1, 2 x_t)$, tangent to the parabola $\varphi(t) = (x_t,\, x_t^2)$.  However, a simple application of the It\={o} rule shows that these SDEs have the solution,
\begin{equation}
\begin{split}
  x_t &=\,x_0 + w_t \\
  y_t &=\,y_0 + w_t^2 - t.
\end{split}
\end{equation}
So even if $(x_0, x_0) = (0,0)$, these equations clearly does not remain on the parabola, even though they are described by a vector field in its tangent space.
Conversely, the \emph{Stratonovich} SDEs
\begin{equation}
\begin{split}
  dx_t &=\,dw_t \\
  dy_t &=\,2 x_t\circ dw_t
\end{split}
\end{equation}
have the solution
\begin{equation}
\begin{split}
  x_t &=\,x_0 + w_t \\
  y_t &=\,y_0 + w_t^2,
\end{split}
\end{equation}
which properly describes diffusion on the parabolic manifold.

Our ultimate goal is to take a system of stochastic differential equations and modify their coefficients so that they remain constrained to a particular submanifold.  This example demonstrates that in order for the tangent space projection to be effective, we must first express the It\={o} equation in a Stratonovich form.

\section{The Bloch Sphere as a Riemannian Manifold\label{chProj:sec:BlochSphereManifold}}

In order to describe the space of density matrices in geometric terms, we need to \emph{choose} a metric.  There are an infinite number we can choose from and it is likely that any results we arrive at will depend upon this choice.  In the classical projection filtering problem, the metric \citeauthor{brigo_differential_1998} choose the the Fisher information, as it endows information theory with a nontrivial geometry \cite{brigo_differential_1998}.  \citeauthor{van_handel_quantum_2005} follow this example and use a quantum version of the Fisher information \cite{van_handel_quantum_2005}.  Later authors choose a different metric, namely the trace inner product \citep{mabuchi_derivation_2008, hopkins_reduced_2009, nielsen_quantum_2009, chase_parameter_2009}.  While the trace inner product does not have an immediate connection to quantum information theory, it is significantly simpler to work with and, as we will shortly show, under this metric the Bloch sphere for a single qubit is Euclidean.   In showing this, we will also formally construct the state space for a single qubit as a Riemannian manifold.

For a Hilbert space of dimension $d$, we will follow \cite{kimura_bloch-vector_2005} and refer to the set of all valid density operators as $\man{S}(d)$.  In the case of a qubit with $d =2$, we already know that the Bloch sphere is an incredibly useful parametrization of this set.  Formally, we define this as the map $\rho: \man{B} \subset \R^3 \rightarrow \man{S}(2)$ so that
\begin{equation}\label{chProj:eq:BlochSphMap}
\rho(\Vx) = \tfrac{1}{2}\left(\ident + \Vx \cdot \V{\sigma} \right).
\end{equation}
As every valid quantum state is required to be trace $1$ and positive semi-definite, we have the constraint that $\abs{\Vx} \le 1$, implying that $\man{B}$ is the unit ball.

Through the Bloch sphere mapping, we can construct a tangents space for $\man{S}(2)$.  This is first done by defining a directional derivative for $\man{S}(2)$. Consider the Bloch vectors $\V{x}$ and $\V{y}$ ($ \abs{\V{x}}, \abs{\V{y}} \in (0, 1) $ ).  The derivative of $\rho(\Vx)$ in the direction of $\V{y}$ is defined to be
\begin{equation}
    D_{\V{y}} \define \lim_{\lambda \rightarrow 0} \frac{\rho(\Vx + \lambda \V{y}) - \rho(\Vx) }{\lambda} = \frac{\V{y} \cdot \V{\sigma} }{2}.
\end{equation}
Then assuming the standard Cartesian coordinate system $\set{x^1, x^2, x^3}$, we have the basis of tangent vectors
\begin{equation}\label{chProj:eq:TangentBasis}
    D_i \define \half \sigma_i.
\end{equation}
The tangent space at the point $\rho(\Vx) \in \man{S}(2)$ is then
\begin{equation}\label{chProj:eq:TangentSpace}
    T_{\rho(\Vx)}\man{S}(2) = \lspan{ D_i \ : \ i = 1, 2, 3}.
\end{equation}

Armed with these tangent vectors we will choose, with some foresight, the trace inner product as a metric.  For two tangent vectors $D_i$ and $D_j$ we have the metric
\begin{equation}
  g_{ij} = \inprod{D_i}{ D_j}_{\rho} \define \Tr(D_i^\dag\, D_j).
\end{equation}
While this could result in a complex metric, we can see that for the qubit the basis vectors are Hermitian and therefor the metric is real.  Also note that due to the cyclic property of the trace, it is also symmetric.  Then for the qubit, simply calculating shows
\begin{equation}\label{chProj:eq:metric}
    \inprod{D_i}{ D_j}_{\rho} = \frac{1}{4} \Tr(\sigma_i \sigma_j) = \frac{1}{2}\delta_{i j}.
\end{equation}
Up to a factor of a half, the Bloch sphere is Euclidian under this metric.

\subsection{Projecting the unconditional master equation\label{chProj:sec:projUnconditional}}

In this section we work though an example of explicitly expressing an unconditional master equation for a single qubit in terms of a differential form $d \rho$. We will also do so generally, without assuming a Euclidean metric.  Most generally $\rho(t)$ is a map $\rho : \R^+ \rightarrow \man{S}(2)$. For any time, $t$, $\rho(t)$ returns a valid density matrix.  Then as a differential, the master equation is the map $d \rho : T_t R^+ \rightarrow T_{\rho(t)}\man{S}(2)$, which is specifically
\begin{equation}\label{chProj:eq:GeneralMasterEquation}
    d \rho = -i [ H,\, \rho] dt + \mathcal{D}[L](\rho) dt,
\end{equation}
for a general Hamiltonian $H$ and jump operator $L$.  Instead of the direct mapping between time and density matrices, we would like to consider this in terms of the Bloch sphere mapping of Eq. (\ref{chProj:eq:BlochSphMap}). This can be done if we consider the time component as a kind of functional composition, so that $\rho(t) = \rho(\Vx(t))$ for a map $\Vx(t)$ between time and Bloch vectors.
From Eq. (\ref{chProj:eq:BlochSphMap}), the general expression for $d \rho: T_{\Vx} \R^3 \rightarrow T_{\rho(\Vx)} \man{S}(2)$ is
\begin{equation}\label{chProj:eq:drhoGeneral}
    d \rho = \half a^i(\Vx)\, \sigma_i\, d x^i.
\end{equation}
To see that this is indeed a map the two tangent spaces we can simply calculate its action on the basis vector $\dydx{}{x^j}$
\begin{equation}
    d \rho \left( \dydx{}{x^j} \right) = \sum_i \half a^i(\Vx)\, \sigma_i d x^i \left( \dydx{}{x^j} \right) = \sum_i \half a^i(\Vx)\, \sigma_i \, \delta^i_j = \half a^j (\Vx) \sigma_j
\end{equation}
where there is no sum in the final expression.  This is clearly in the tangent space $T_{\rho(\Vx)} \man{S}(2)$ as it is proportional to $D_j = \half \sigma_j$.   Our ultimate goal is then to solve for the coefficients $a^i(\Vx)$.

Any traceless matrix $2 \times 2$ matrix can be written as a linear combination of Pauli matrices.  As both the commutator $[H,\, \rho]$ and the map $\mathcal{D}[L](\rho)$ are traceless, both of these operations have some expansion coefficient in terms of the Pauli matrices.  Sec. \ref{chProj:sec:BlochSphereManifold} found that the tangent space $T_{\rho}\man{S}(2)$ is also spanned by the Pauli matrices, meaning that $-i [H,\, \rho]$ and $\mathcal{D}[L](\rho)$ are vectors in this space.  Thus finding the coefficients $a^i(\Vx)$ simply comes to projecting these maps onto the basis vectors $D_i$.


Sec. \ref{chProj:sec:Metrics} established that the general projection map $\Pi_{\man{N}}$ can be written as Eq. (\ref{chProj:eq:projMap}), in terms of the metric and its inverse $g^{ij}(\Vx)$.  We are able to write $d\rho(t)$ as
\begin{equation}\label{chProj:eq:qubitBlochVectorMasterDifferential}
    d \rho(\Vx(t)) = - i\, g^{i j}(\Vx) \, \inprod{D_i}{ [ H, \rho]}_{\rho} D_j dt +  g^{i j}(\Vx)\, \inprod{D_i}{\mathcal{D}[L](\rho) }_{\rho} D_j dt.
\end{equation}
But as this is a differential with respect to $dt$ and not $dx^j$ we can define the differentials for the time-dependent coordinates $\set{x^j(t)}$
\begin{equation}
    d x^j = -i\, g^{i j}(\Vx) \, \inprod{D_i}{ [ H, \rho] }_{\rho} dt +  g^{ij}(\Vx)\, \inprod{D_i}{ \mathcal{D}[L](\rho)}_{\rho}  dt,
\end{equation}
meaning that
\begin{equation}
    d \rho = D_j\, dx^j.
\end{equation}

\section{Projections in the tensor product submanifold}

Our ultimate goal is to form a projection from a general state over $n$ qubits, to the closest $n$-fold tensor product of a single qubit state.  We will define $\man{P}$ to be the submanifold of $\man{S}(2^n)$ which describes the space of all states of the form $\rho(\Vx)^{\otimes n}$.  This space has the simple parameterization $ \varrho: \man{B} \subset \R^3 \rightarrow \man{P} \subset \man{S}(2^n)$ such that
\begin{equation}\label{chProj:eq:rhoOtimesN}
    \varrho(\Vx) \define \rho(\Vx)^{\otimes n} = \frac{1}{2^n}\left(\ident + \Vx \cdot \V{\sigma} \right)^{\otimes n}.
\end{equation}

We also need to identify the tangent spaces for each point in the submanifold.  Because of the linear nature of the one qubit map $\rho$ the directional derivative of $\varrho(\Vx)$ with respect to $\V{y}$ is simply
\begin{equation}
    D_{\V{y}} = \left.\dydx{}{\lambda} \varrho(\Vx + \lambda \V{y})\right|_{\lambda =0}.
\end{equation}
A derivative acting on a tensor product must obey the Leibnitz rule.   The directional derivative of $\rho(\Vx)^{\otimes\, n}$ in the direction $\V{y}$ must then be equal to
\begin{equation}
D_{\V{y}}(\varrho(\Vx)) = \left.\dydx{}{\lambda} \rho(\Vx + \lambda \V{y})^{\otimes n} \right|_{\lambda =0} = \sum_{i = 1}^n\rho(\Vx)^{\otimes\, i -1} \otimes \frac{1}{2}\V{y}\cdot\V{\sigma} \otimes \rho(\Vx)^{\otimes\, n-i}.
\end{equation}

For the single qubit, the directional derivative was uniform over the manifold, which implied the Euclidean geometry for our simple metric.  For multiple qubits, this is no longer the case, which implies that  $\man{P}$ has a richer geometry.
With a slight abuse of notation, the basis vector associate with the coordinate $x^i$, evaluated at the state $\rho(\Vx)^{\otimes\, n}$ will be notated $D_i(\Vx)$ and is given by
\begin{equation}\label{chProj:eq:tangentVectorCollective}
    D_{i}(\Vx) = \sum_{j = 1}^n \rho(\Vx)^{\otimes\, j -1} \otimes \half \sigma_i \otimes \rho(\Vx)^{\otimes\, n-j}.
\end{equation}
The tangent space at $\varrho(\Vx)$ is then
\begin{equation}
    T_{\varrho(\Vx)}\man{P} = \lspan{ D_{i}(\Vx)\, :\ i = x, y, z}.
\end{equation}

The metric on $\man{P}$ induced from the trace inner product is now easily calculated.  The product of the two basis vectors $D_i$ and $D_j$ is equal to
\begin{equation}
\begin{split}\label{chProj:eq:DiDj}
D_i D_j =&\, \sum_{p, q = 1}^n (\rho^{\otimes\, p -1} \otimes \half \sigma_i \otimes \rho^{\otimes\, n-p} ) (\rho^{\otimes\, q -1} \otimes \half \sigma_j \otimes \rho^{\otimes\, n-q} ) \\
  =&\, \sum_{p = q = 1}^n  {\rho^2}^{\otimes\, p -1} \otimes \tfrac{1}{4} \sigma_i\, \sigma_j \otimes {\rho^2}^{\otimes\, n-p} \\
  &+\,\sum_{q > p = 1}^n  {\rho^2}^{\otimes\, p -1}\otimes \half \sigma_i \rho \otimes {\rho^2}^{\otimes\,q- p -1} \otimes \half \rho \sigma_j \otimes {\rho^2}^{\otimes\, n-q} \\
  &+\,\sum_{q < p = 1}^n {\rho^2}^{ \otimes\, p -1}\otimes \half \rho \sigma_j \otimes {\rho^2}^{\otimes\,p- q-1}\,\otimes \half \sigma_i \rho  \otimes {\rho^2}^{ \otimes\, n-q}.
  \end{split}
\end{equation}
The metric coefficient is then
\begin{equation}\label{chProj:eq:MetricTensorProduct}
  \begin{split}
    \inprod{D_i}{D_j}_{\varrho(\Vx)} =&\, \Tr(D_{i}\, D_{j}\,)\\
        =&\,\frac{n}{4}\, \Tr(\rho^2)^{n-1}\, \Tr(\,\sigma_i\, \sigma_j\,) +
            \frac{n(n-1)}{4} \Tr(\rho^2)^{n-2}\, \Tr(\rho\, \sigma_i)\, \Tr(\rho\, \sigma_j)\\
        =&\,\frac{n}{2^n} \left(1 + \abs{\Vx}^2 \right)^{n-1}\, \delta_{ij}  +
            \frac{n(n-1)}{2^n} \left(1 + \abs{\Vx}^2 \right)^{n-2} \, x^k x^\ell\, \delta_{k i} \delta_{\ell j}.
  \end{split}
\end{equation}

We will often need to calculate the product between several collective operators and then take the trace.  While Eq. (\ref{chProj:eq:DiDj}) has a distinct ordering to the tensor products, resulting in the two sums $p < q$ and $q < p$, upon taking the trace this order becomes irrelevant.  Thus, there are only two relevant terms: $p = q$ and $p \ne q$.

\subsection{The metric in spherical coordinates}
The metric as given by Eq. (\ref{chProj:eq:MetricTensorProduct}) has a simple form when written in spherical coordinates.  In terms of the Cartesian basis vectors $\set{\Ve_x, \Ve_y, \Ve_z}$ the standard spherical basis vectors are defined as,
\begin{equation}\label{chProj:eq:sphericalBasis}
  \begin{split}
    \Ve_r =&\, \sin \theta\, \cos \phi\, \Ve_x + \sin \theta\, \sin \phi\, \Ve_y +  \cos \theta\, \Ve_z\\
    \Ve_\theta =&\, \cos \theta\, \cos \phi\, \Ve_x + \cos \theta\, \sin \phi\, \Ve_y -  \sin \theta\, \Ve_z\\
    \Ve_\phi =&\, - \sin \phi\, \Ve_x + \cos \phi\, \Ve_y.
  \end{split}
\end{equation}
In analogy, we will define the associated tangent vectors,
\begin{equation}\label{chProj:eq:sphericalTangents}
  \begin{split}
    D_r(\Vx) =&\, \sin \theta\, \cos \phi\, D_x(\Vx) + \sin \theta\, \sin \phi\, D_y(\Vx) +  \cos \theta\, D_z(\Vx)\\
    D_\theta(\Vx) =&\, \cos \theta\, \cos \phi\, D_x(\Vx) + \cos \theta\, \sin \phi\, D_y(\Vx) -  \sin \theta\, D_z(\Vx)\\
    D_\phi(\Vx) =&\, - \sin \phi\, D_x(\Vx) + \cos \phi\, D_y(\Vx).
  \end{split}
\end{equation}
When $\Vx$ is in the subset $\set{\Vx \in \man{B}\, : \, (0 < r < 1,\, 0 < \theta < \pi ,\, 0 < \phi < 2 \pi)}$, these vector fields form a perfectly valid basis for each tangent space $T_{\varrho(\Vx)} \man{P}$.

It will also be convenient to define ``spherical'' Pauli matrices,
\begin{equation}\label{chProj:eq:sphericalPauli}
  \begin{split}
    \sigma_r \define&\, \sin \theta\, \cos \phi\, \sigma_x + \sin \theta\, \sin \phi\, \sigma_y +  \cos \theta\, \sigma_z\\
    \sigma_\theta \define&\, \cos \theta\, \cos \phi\, \sigma_x + \cos \theta\, \sin \phi\, \sigma_y - \sin \theta\, \sigma_z\\
    \sigma_\phi \define&\, -  \sin \phi\, \sigma_x + \cos \phi\, \sigma_y.
  \end{split}
\end{equation}
These operators obey the usual properties associated with Pauli matrices, in that for $i,j, k \in \set{ r, \theta, \phi}$
\begin{subequations}
\begin{align}
    \Tr(\sigma_i) &= 0\\
    \Tr(\sigma_i\,\sigma_j) &= 2 \delta_{ij} \\
    [\sigma_i,\, \sigma_j] &=  i\, \varepsilon_{ijk}\, 2 \sigma_k\\
    (\sigma_i \sigma_j + \sigma_j \sigma_i) &= \delta_{ij}\, 2 \ident.
\end{align}
\end{subequations}
Furthermore, we have that
\begin{equation}
    D_i(\Vx) = \sum_j \rho(\Vx)^{\otimes j-1}\otimes \half \sigma_i \otimes \rho(\Vx)^{\otimes n-j}
\end{equation}
for both Cartesian and spherical bases.
And the state $\rho(\Vx)$ can now be written as
\begin{equation}\label{chProj:eq:rhoSpherical}
    \rho( \Vx ) = \half\left(\ident + r\, \sigma_r \right).
\end{equation}

We can now use the fact that the Pauli matrices are orthogonal, and the fact that the state $\rho$ is now orthogonal to $\sigma_\theta$ and $\sigma_\phi$ to evaluate the inner product between the spherical tangent vectors, and thus write the metric as a matrix in spherical coordinates. From the general expression
\begin{equation}
    \inprod{D_i}{D_j} = \tfrac{1}{4} n \Tr(\rho^2)^{n-1}\, \Tr(\sigma_i \sigma_j)  +
           \tfrac{1}{4}  n(n-1)\Tr(\rho^2)^{n-2}\, \Tr(\sigma_i \rho)\, \Tr(\sigma_j \rho),
\end{equation}
we have
\begin{equation}\label{chProj:eq:MetricTensorSpherical}
\begin{split}
    \inprod{D_r}{ D_r} =&\, \frac{n}{2^n} \left(1 + r^2 \right)^{n-1}\frac{1 + n r^2}{1 + r^2},\\
    \inprod{D_\theta}{ D_\theta} =&\, \frac{n}{2^n} \left(1 + r^2 \right)^{n-1},\\
    \inprod{D_\phi}{ D_\phi} =&\, \frac{n}{2^n} \left(1 + r^2 \right)^{n-1},
\end{split}
\end{equation}
and
\begin{equation}
   \inprod{D_r}{D_\theta} =\, \inprod{D_r}{D_\phi} = \inprod{D_\theta}{D_\phi} = 0.
\end{equation}

As a matrix, the metric in spherical coordinates is given by
\begin{equation}\label{chProj:eq:metricMatrixSpherical}
    G(\Vx) = \frac{n}{2^n} \left(1 + r^2 \right)^{n-1} \left(
       \begin{array}{ccc}
         \frac{1 + n r^2}{1 + r^2} & 0 & 0 \\
         0 &1 & 0 \\
         0 & 0 & 1
       \end{array}
     \right)
\end{equation}
and its inverse is
\begin{equation}\label{chProj:eq:metricInverseMatrixSpherical}
    G^{-1}(\Vx) = \frac{2^n}{n} \left(1 + r^2 \right)^{-(n-1)} \left(
       \begin{array}{ccc}
         \frac{1 + r^2}{1 + n r^2} & 0 & 0 \\
         0 & 1 & 0 \\
         0 & 0 & 1
       \end{array}
     \right).
\end{equation}
Notice that when $n = 1$ we recover the simple Euclidean metric of $g_{ij} = \half\, \delta_{ij}$.

\subsection{Calculating collective operator inner products\label{chProj:sec:CollectiveProjections}}
This section contains the detailed calculations necessary for projecting the various components of the conditional and unconditional master equations onto the space of identical separable states.  We will first derive the projection for a general conditional master equation of the form
\begin{equation}
        d \rho_\tot = -i [ H_\tot,\, \rho_\tot] dt + \mathcal{D}[L_\tot](\rho_\tot) dt + \mathcal{I}_c[L_\tot] (\rho_\tot) dt +  \mathcal{H}[L_\tot](\rho_\tot)\circ dw_t.
\end{equation}
The subscript $_\tot$ is used to specify that these operators are operators on the total Hilbert space consisting of $N$ particles.  Any single particle operator $A$, acting on the $n^{th}$ particle of the ensemble, is denoted by $A^{(n)}$ and is given by the tensor product,
\begin{equation}
    A^{(n)} \define \ident^{\otimes n -1}\otimes A \otimes \ident^{\otimes N -n}.
\end{equation}

The fundamental assumption for this derivation is that the operators $H_\tot$ and $L_\tot$, act independently and identically on each each qubit and may be written as
\begin{subequations}
\begin{align}\label{chProj:eq:collectiveOperators}
    H_\tot =&\, \sum_{n = 1}^{N} H^{(n)} = \sum_{n = 1}^{N}\ident^{\otimes n -1}\otimes H \otimes \ident^{\otimes N -n}\\
    L_\tot =&\, \sum_{n = 1}^{N} L^{(n)} = \sum_{n = 1}^{N}\ident^{\otimes n -1}\otimes L \otimes \ident^{\otimes N -n}.
\end{align}
\end{subequations}
Furthermore, the tangent vectors $D_i$, for the single particle state $\rho$ are
\begin{equation}
    D_{i} = \sum_{n = 1}^N \rho^{\otimes\, n -1} \otimes \half \sigma_i \otimes \rho^{\otimes\, N-n}.
\end{equation}

In projecting the collective master equation onto the identical product states, we will need to calculate the product of up to three collective operators and then take the trace.  Each collective operator is composed of a sum over single particle operator, each acting on $n^{th}$ member.  When taking the product of sums there will be $N$ terms where both single particle operators act on the same subsystem, as well as $N(N-1)$ terms where the constituent operators act on different systems.

The simplest case is when there are no collective operators \emph{i.e.}, simply calculating the overlap between $D_i$ and the state $\varrho$.  This is equal to
\begin{equation}\label{chProj:eq:Dirho}
\begin{split}
    \inprod{D_i}{\varrho} =&\, \sum_{n = 1}^N \Tr \left( (\,\rho^{\otimes\, n -1} \otimes \half \sigma_i \otimes \rho^{\otimes\, N-n} )\, \varrho \right)\\
    =&\, N \Tr( \rho^2)^{N-1}\, \Tr(\half \sigma_i\, \rho).
\end{split}
\end{equation}
The next step up in complexity is to include a single collective operator, $A_\tot$.   This requires two sums, the sum from $D_i$ and the sum from $A_\tot$.  We then have
\begin{equation}\label{chProj:eq:DiArho}
    \begin{split}
        \inprod{D_i}{A_\tot\,\varrho} =&\,\sum_{n, m = 1}^N \Tr \left( \left( \rho^{\otimes\, m -1} \otimes\half \sigma_i \otimes \rho^{\otimes\, N-m}\right)  A^{(n)} \rho^{\otimes N} \right)\\
        =&\,\sum_{n = m = 1}^N \Tr \left(  {\rho^2}^{ \otimes\, m -1} \otimes \half \sigma_i\, A \rho  \otimes {\rho^2}^{ \otimes\, N-m} \right) \\
        &+\,\sum_{n <  m = 1}^N \Tr \left(  {\rho^2}^{ \otimes\, n -1} \otimes \rho A \rho \otimes {\rho^2}^{ \otimes\, m-n-1} \otimes \half \sigma_i\, \rho \otimes {\rho^2}^{ \otimes\, N -m} \right) \\
        &+\,\sum_{n >  m = 1}^N \Tr \left(  {\rho^2}^{ \otimes\, m -1} \otimes \half \sigma_i\, \rho \otimes {\rho^2}^{ \otimes\, n-m-1} \otimes \rho A \rho \otimes {\rho^2}^{ \otimes\, N -n} \right)\\
        =&\,N \Tr(\rho^2)^{N-1} \Tr( \half \sigma_i\,  A \rho) + N ( N - 1 ) \Tr(\rho^2 )^{N -2}\,  \Tr( \half \sigma_i\,\rho )\, \Tr( A\, \rho^2).
    \end{split}
\end{equation}
The cyclic property of the trace shows us that upon switching the order of $A_\tot$ and $\varrho$, (\emph{i.e.} to instead calculate $\inprod{D_i}{\varrho\,A_\tot}$) the second term will be left unchanged, so
\begin{equation}\label{chProj:eq:DirhoA}
    \inprod{D_i}{\varrho\, A_\tot} =\,N \Tr(\rho^2)^{N-1} \Tr( \half \sigma_i\, \rho A) + N ( N - 1 ) \Tr(\rho^2 )^{N -2}\,  \Tr( \half \sigma_i\,\rho )\, \Tr( A\, \rho^2).
\end{equation}

When calculating the projection of the dissipator terms, we need to calculate the product of two collective operators, which will have a triple sum.  For two collective operators $A_\tot$ and $B_\tot$ we have,
\begin{align}
    \inprod{D_i}{A_\tot\,B_\tot\, \varrho } =&\,\sum_{n, m, l
     = 1}^N \Tr \left( \left( \rho^{\otimes\, n -1} \otimes\half \sigma_i \otimes \rho^{\otimes\, N-n}\right) A^{(m)}\, B^{(l)}\,\rho^{\otimes N} \right).
\end{align}
The previous result shows us that there will be five distinct terms,  for cases where $n = m = l$, $n \ne m = l$, $n = m \ne l$, $n = l \ne m$, and  $n \ne m \ne l$.  This expression then simplifies to
\begin{equation}\label{chProj:eq:DiABrho}
    \begin{split}
    \inprod{D_i}{A_\tot\,B_\tot\, \varrho } =&\, N \Tr(\rho^2)^{N-1} \Tr(\half \sigma_i\, A\, B\, \rho)\\
    &+\, N ( N - 1 ) \Tr(\rho^2 )^{N -2}\,\Tr(\half \sigma_i\, \rho )\, \Tr(  \rho\, A\, B\, \rho )\\
    &+\, N ( N - 1 ) \Tr(\rho^2 )^{N -2}\, \Tr(\half \sigma_i A \rho)\, \Tr(\rho\, B\, \rho )\\
    &+\, N ( N - 1 ) \Tr(\rho^2 )^{N -2}\,\Tr(\half \sigma_i B \rho )\, \Tr( \rho\, A\, \rho )\\
    &+\, N ( N - 1 )(N-2) \Tr(\rho^2 )^{N -3}\,\Tr( \half \sigma_i \,\rho )\,\Tr(\rho\,  A\, \rho )\, \Tr(\rho\, B\,  \rho ).
    \end{split}
\end{equation}
The terms $\Tr(\rho X \rho)$ can be simplified to $\Tr(X \rho^2)$ but were left to make it more explicit.
The order of the collective operators can be exchanged, but this won't effect the five term structure.  Thus we calculate that
\begin{equation}\label{chProj:eq:DiArhoB}
    \begin{split}
    \inprod{D_i}{A_\tot\, \varrho\, B_\tot } =&\, N \Tr(\rho^2)^{N-1} \Tr(\half \sigma_i\, A\, \rho\, B)\\
    &+\, N ( N - 1 ) \Tr(\rho^2 )^{N -2}\,\Tr(\half \sigma_i\, \rho )\, \Tr(  \rho\, A\,\rho\, B)\\
    &+\, N ( N - 1 ) \Tr(\rho^2 )^{N -2}\, \Tr(\half \sigma_i A\, \rho)\, \Tr( B\, \rho^2 )\\
    &+\, N ( N - 1 ) \Tr(\rho^2 )^{N -2}\,\Tr(\half \sigma_i \rho\, B)\, \Tr(A\, \rho^2 )\\
    &+\, N ( N - 1 )(N-2) \Tr(\rho^2 )^{N -3}\,\Tr( \half \sigma_i \,\rho )\,\Tr(A\, \rho^2 )\, \Tr(B\, \rho^2 )
    \end{split}
\end{equation}
and
\begin{equation}\label{chProj:eq:DirhoAB}
    \begin{split}
    \inprod{D_i}{\varrho\,A_\tot\,B_\tot } =&\, N \Tr(\rho^2)^{N-1} \Tr(\half \sigma_i\,\rho\, A\, B)\\
    &+\, N ( N - 1 ) \Tr(\rho^2 )^{N -2}\,\Tr(\half \sigma_i\, \rho )\, \Tr(  A\, B\, \rho^2 )\\
    &+\, N ( N - 1 ) \Tr(\rho^2 )^{N -2}\, \Tr(\half \sigma_i\rho A )\, \Tr(B\, \rho^2 )\\
    &+\, N ( N - 1 ) \Tr(\rho^2 )^{N -2}\,\Tr(\half \sigma_i\rho B )\, \Tr(A\, \rho^2 )\\
    &+\, N ( N - 1 )(N-2) \Tr(\rho^2 )^{N -3}\,\Tr( \half \sigma_i \,\rho )\,\Tr( A\, \rho^2 )\, \Tr(B\,  \rho^2 ).
    \end{split}
\end{equation}

The final two calculations involving collective operators are the expectation values, $\expect{A_\tot}$ and $\expect{A_\tot\, B_\tot}$.  They are
\begin{equation}\label{chProj:eq:Arho}
    \expect{A_\tot} =\Tr(A_\tot \varrho)= N  \Tr(A\,\rho)
\end{equation}
and
\begin{equation}\label{chProj:eq:ABrho}
 \expect{A_\tot B_\tot} =\Tr(A_\tot B_\tot\,\varrho) = N  \Tr(A\, B\,\rho) + N(N-1) \Tr(A\, \rho) \Tr(B\,\rho).
\end{equation}

The general expressions in Eqs. (\ref{chProj:eq:Dirho}-\ref{chProj:eq:ABrho}) are all we need in order to calculate all of the terms in the conditional master equation and a tangent vector $D_i$.  Starting with the Hamiltonian commutator,  Eq. (\ref{chProj:eq:DiArho}) and its permutated version gives,
\begin{equation}\label{chProj:eq:DiHamTerm}
        \inprod{D_i}{[H_\tot,\,\varrho]} = N \Tr(\rho^2)^{N-1} \Tr( \half \sigma_i\,  [H,\, \rho] ).
\end{equation}

For the dissipator term, substituting in to Eq. (\ref{chProj:eq:DiABrho}-\ref{chProj:eq:DirhoAB}) with the appropriate collective operator $L_\tot$ or $L_\tot^\dag$,
\begin{equation}\label{chProj:eq:dissipatorTerm}
    \begin{split}
    \inprod{D_i}{\mathcal{D}[L_\tot]( \varrho)} =&\, \Tr \left( D_i\, ( L_\tot \varrho L_\tot^\dag  - \half L_\tot^\dag L_\tot \varrho - \half \varrho L_\tot^\dag L_\tot )\,\right)\\
    =&\, N \Tr(\rho^2)^{N-1} \Tr(\half \sigma_i\, \mathcal{D}[L](\rho)\, ) \\
    &+\, N ( N - 1 ) \Tr(\rho^2 )^{N -2}\, \Tr(\half \sigma_i\, \rho)\, \Tr(\rho\, \mathcal{D}[L](\rho) )\\
    &+\, N ( N - 1 ) \Tr(\rho^2 )^{N -2}\, \Tr(\tfrac{1}{4} \sigma_i\, [L,\,\rho] )\, \Tr( L^\dag \rho^2 )\, \\
    &+\, N ( N - 1 ) \Tr(\rho^2 )^{N -2}\, \Tr(\tfrac{1}{4} \sigma_i\, [\rho,\, L^\dag])\,\Tr( L \rho^2 ).
    \end{split}
\end{equation}
Note that final lines in Eqs. (\ref{chProj:eq:DiABrho}-\ref{chProj:eq:DirhoAB}) are all equal and hence cancel in the dissipator term.

For the conditioning map $\mathcal{H}[L_\tot]$  we again need Eq. (\ref{chProj:eq:DiArho}), with the addition of Eq. (\ref{chProj:eq:Dirho}) and the single operator expectation value Eq. (\ref{chProj:eq:Arho}).  Its overlap then reduces to
\begin{equation}\label{chProj:eq:DiConditioning}
    \begin{split}
    \inprod{D_i}{\mathcal{H}[L_\tot]( \varrho)} =&\, \Tr\left( D_i\, \left( L_\tot \varrho + \varrho L_\tot^\dag - \expect{L_\tot + L_\tot^\dag} {}\, \varrho\, \right) \right)\\
    =&\,N \Tr(\rho^2)^{N-1}\, \Tr( \half \sigma_i (L \rho + \rho L^\dag))\\
     &+ N ( N - 1 ) \Tr(\rho^2 )^{N -2}\, \Tr( (L + L^\dag)\, \rho^2)\, \Tr( \half \sigma_i\,\rho )\\
     &- N^2 \Tr(\rho^2)^{N-1} \Tr( (L +L^\dag)\rho )\, \Tr( \half \sigma_i \rho ) .
    \end{split}
\end{equation}

Finally for the general It\={o} correction map Eq. (\ref{chMath:eq:itoCorrectionMap}),
\begin{equation}\label{chProj:eq:DiItoCorrect}
\begin{split}
   \inprod{D_i}{\mathcal{I}_c[L_\tot](\varrho)} =&\, N \Tr(\rho^2)^{N-1} \Tr \left( \half \sigma_i\, ( L \rho L^\dag  + \half L^2 \rho + \half \rho L^{\dag\, 2}  )\,\right)\\
   &+N(N-1)\Tr(\rho^2)^{N-2} \Tr(\half \sigma_i\, \rho)\,\Tr \left( \rho ( L \rho L^\dag  + \half L^2 \rho + \half \rho L^{\dag\, 2} )  \right) \\
   &+N(N-1)\Tr(\rho^2)^{N-2} \Tr(\half \sigma_i (L \rho + \rho L^\dag)\, ) \Tr( (L + L^\dag) \rho^2) \\
   &-N^2\Tr(\rho^2)^{N-1} \Big( \expect{L^\dag L + \half L^2 + \half L^{2\, \dag}} \\
   &\qquad \qquad \qquad\qquad \left.+ \half (N - 1) \expect{L + L^\dag}^2 \right) \Tr(\half \sigma_i\, \rho)\\
   &-N \expect{L + L^\dag}\, \inprod{D_i}{\mathcal{H}[L_\tot]( \varrho)}.
\end{split}
\end{equation}

These expressions simplify, in the case of a Hermitian operator $L$ acting on a single qubit state.   Specifically when $L = J_z$ and $H$ is the control Hamiltonian in Eq. (\ref{chProj:eq:controlHt}), Eqs. (\ref{chProj:eq:DiHamTerm}-\ref{chProj:eq:DiItoCorrect}) simplify to
\begin{equation}\label{chProj:eq:DiControlHt}
        \inprod{D_i}{[H_\tot,\,\varrho]} = \tfrac{1}{4} N \Tr(\rho^2)^{N-1} \Tr \left( \sigma_i\, [f^j(t) \sigma_j,\, \rho] \right),
\end{equation}
\begin{equation}\label{chProj:eq:DiLiouJz}
    \begin{split}
    \inprod{D_i}{\mathcal{D}[J_z]( \varrho)} =&\, \tfrac{1}{8} N  \Tr(\rho^2)^{N-1}\left(\, \Tr \left( \sigma_i\, \sigma_z \rho\, \sigma_z\,\right) - N \Tr ( \sigma_i\, \rho )\, \right)\\
    &+\,\tfrac{1}{8} N ( N - 1 ) \Tr(\rho^2 )^{N -2}\, \Tr(\rho \sigma_z \rho \sigma_z  )\,\Tr( \sigma_i\, \rho),
    \end{split}
\end{equation}
\begin{equation}\label{chProj:eq:DiHJz}
    \begin{split}
    \inprod{D_i}{\mathcal{H}[J_z]( \varrho)} =&\, \tfrac{1}{4} N  \Tr(\rho^2)^{N-1}\, \Tr \left( \sigma_i\, (\sigma_z \rho + \rho\,  \sigma_z)\,\right)\\
     &- \half N^2 \Tr(\rho^2)^{N-1} \Tr ( \sigma_z\, \rho )\, \Tr ( \sigma_i\, \rho ) \\
    &+ \tfrac{1}{2} N ( N - 1 ) \Tr(\rho^2 )^{N -2}\, \Tr(\sigma_z \rho^2 )\, \Tr( \sigma_i\, \rho),
    \end{split}
\end{equation}
\begin{equation}\label{chProj:eq:DiICJz}
  \begin{split}
    \inprod{D_i}{ \mathcal{I}_c[ J_z](\varrho ) } =&\, \inprod{D_i}{ \mathcal{D}[J_z]( \varrho)}\\
    &+ \tfrac{1}{4} N ( N - 1 ) \Tr(\rho^2 )^{N -2}\, \Tr(  \sigma_z \rho^2 )\, \Tr \left( \sigma_i\, (\sigma_z \rho + \rho  \sigma_z)\,\right)\\
    &- \tfrac{1}{4} N^2(N-1) \Tr(\rho^2)^{N-1} \Tr(\sigma_z\, \rho)^2\Tr( \sigma_i  \,\rho)\\
    &+ \tfrac{1}{4} N ( N - 1 )(N-2) \Tr(\rho^2 )^{N -3}\,\Tr(  \sigma_z \rho^2 )^2\, \Tr(  \sigma_i\,\rho )\\
    &- N \Tr(\sigma_z\, \rho) \inprod{D_i}{\mathcal{H}[J_z]( \varrho) }.
  \end{split}
\end{equation}
It is worth noting that in Eqs. (\ref{chProj:eq:DiLiouJz}-\ref{chProj:eq:DiICJz}), a majority of the terms are proportional to $\Tr(  \sigma_i\,\rho )$.  If the state has zero expectation along the $\sigma_i$ axis, then the overlap with that tangent vector will be greatly simplified.    However, any qubit state (which is not completely mixed) has a Bloch vector pointing along \emph{some} axis, leaving the orthogonal axes with zero expectation.  When the state happens to align with the $D_i$, \emph{i.e.} the Bloch vector is $\Vx = r \Ve_i$, $\Tr(  \sigma_i\,\rho ) = r$.  This suggests that these terms may simplify for any state, if we choose to work in \emph{spherical} coordinates.

\subsection{The spherical projection of the CME}

With the metric inverse in Eq. (\ref{chProj:eq:metricInverseMatrixSpherical}), and the inner product expressions in Eqs. (\ref{chProj:eq:DiHJz}-\ref{chProj:eq:DiICJz}) we can finally calculate the projection coefficients $h^i,\, l^i,\, c^i,\, b^i$.    To fully simplify the inner products we need the following relations, which are easily calculated:
\begin{subequations}\label{chProj:eq:TrRhoRelations}
\begin{equation}
\Tr(\sigma_i\, \rho) = \left\{
                        \begin{array}{cc}
                          r & \text{for } i = r \\
                          0 & \text{for } i = \theta \\
                          0 & \text{for } i = \phi
                        \end{array}
                      \right.,
\end{equation}
\begin{equation}\label{chProj:eq:siAntiSzrho}
\Tr(\sigma_i\, ( \sigma_z \rho + \rho \sigma_z) ) = \left\{
                        \begin{array}{cc}
                          2 \cos(\theta) & \text{for } i = r\\
                          -2 \sin(\theta) & \text{for } i = \theta \\
                          0 & \text{for } i = \phi
                        \end{array}
                      \right.,
\end{equation}
\begin{equation}\label{chProj:eq:siSzrhoSz}
\Tr( \sigma_i\,  \sigma_z\,\rho\,  \sigma_z) = \left\{
                        \begin{array}{cc}
                          r\cos(2 \theta)& \text{for } i = r \\
                          - r \sin(2 \theta) & \text{for } i = \theta \\
                          0& \text{for } i = \phi
                        \end{array}
                      \right.,
\end{equation}
\begin{equation}\label{chProj:eq:SzRhoSq}
\Tr( \sigma_z\, \rho) = \Tr( \sigma_z\, \rho^2) = r \cos(\theta),
\end{equation}
and
\begin{equation}\label{chProj:eq:rhoSzrhoSz}
\Tr( \rho\,  \sigma_z\, \rho\,  \sigma_z) = \half (1 + r^2 \cos(2\theta) ).
\end{equation}
\end{subequations}
The azimuthal symmetry of the problem is directly apparent in these expressions, as the $J_z$ projection carries no information about $\phi$.

Substituting the spherical Pauli matrices into Eq. (\ref{chProj:eq:DiHJz}), the Hamiltonian inner products simplify to    \footnotesize
\begin{equation}\label{chProj:eq:DiHamJzSp}
  \begin{split}
    \inprod{D_r}{ -i [H,\,\varrho ] } &= 0\\
    \inprod{D_\theta}{ -i [H,\,\varrho ] } &= \half n \Tr(\rho^2)^{n-1}\, r  \left(f^2(t)\, \cos\, \phi - f^1(t) \sin\, \phi\, \right)\\
    \inprod{D_\phi}{ -i [H,\,\varrho ] } &= \half n \Tr(\rho^2)^{n-1} \, r \left(f^3(t)\, \sin\, \theta - f^1(t)\, \cos\,\theta\, \cos\,\phi - f^2(t)\, \cos\,\theta\, \sin\, \phi \right).
  \end{split}
\end{equation}
\normalsize
Physically, applying magnetic fields to a spin ensemble cannot change the total magnetization, a fact that is confirmed by having the $D_r$ projection be zero.

The remaining projections we need to calculate are all based on the $J_z$ measurement operator and so contain no information about the $\phi$ coordinate.  This can be verified by substituting the $D_\phi$ tangent vector into Eqs. (\ref{chProj:eq:DiLiouJz}-\ref{chProj:eq:DiICJz}), which all evaluate to zero.

The $D_\theta$ projections are greatly simplified by the fact that $\Tr(\sigma_\theta\, \rho) = 0$.   Using Eq. (\ref{chProj:eq:siSzrhoSz}) we find that
\begin{equation}\label{chProj:eq:DthetaLiouJz}
    \inprod{D_\theta}{\mathcal{D}[J_z]( \varrho)} =- \tfrac{1}{8} n \Tr(\rho^2)^{n-1}\, r \sin(2 \theta).
\end{equation}
From Eq. (\ref{chProj:eq:siAntiSzrho}) the conditioning product reduces to
\begin{equation}\label{chProj:eq:DthetaHJz}
    \inprod{D_\theta}{\mathcal{H}[J_z]( \varrho)} = - \half n \Tr(\rho^2)^{n-1}\, \sin(\theta).
\end{equation}
And combing these two results into Eq. (\ref{chProj:eq:DiICJz}) we have
\begin{equation}\label{chProj:eq:DthetaICJz}
  \begin{split}
    \inprod{D_\theta}{ \mathcal{I}_c[ J_z](\varrho ) } =&\,- \tfrac{1}{8} n \Tr(\rho^2)^{n-1}\, r \sin(2 \theta)\\
    &- \tfrac{1}{4} n ( n - 1 ) \Tr(\rho^2 )^{n -2}\, r \sin(2 \theta)\\
    &+ \tfrac{1}{4} n^2 \Tr(\rho^2)^{n-1}\, r \sin(2 \theta).
  \end{split}
\end{equation}

Simplifying the $r$ projections is obviously a more complicated task.  However, the dissipator product with $D_r$ is not particularly more difficult than the $D_\theta$ product. By including the fact that $\Tr(\sigma_r \rho) = r$ and substituting Eqs. (\ref{chProj:eq:siSzrhoSz} and \ref{chProj:eq:rhoSzrhoSz}) into Eq. (\ref{chProj:eq:DiLiouJz}) we find
\begin{equation}\label{chProj:eq:DrLiouJz}
    \inprod{D_r}{\mathcal{D}[J_z]( \varrho)} = -\tfrac{1}{8} n \Tr(\rho^2)^{n-2} (1 + n r^2)\, r \sin^2(\theta).
\end{equation}
Evaluating the conditioning map requires Eqs. (\ref{chProj:eq:siSzrhoSz}) and (\ref{chProj:eq:SzRhoSq}) which reduce Eq. (\ref{chProj:eq:DiHJz}) to
\begin{equation}\label{chProj:eq:DrHJz}
\inprod{D_r}{\mathcal{H}[J_z]( \varrho)} = -\tfrac{1}{4} n \Tr(\rho^2)^{n-2} (1 + n r^2) (r^2 - 1 ) \cos(\theta).
\end{equation}
Finally when we combining these past results into the It\={o} correction product, Eq. (\ref{chProj:eq:DiICJz}) simplifies to
\begin{equation}
  \begin{split}\label{chProj:eq:DrICJz}
    \inprod{D_r}{ \mathcal{I}_c[ J_z](\varrho ) } =&\, -\tfrac{1}{8} n \Tr(\rho^2)^{n-2} (1 + n r^2)\, r \sin^2(\theta)\\
    &+ \tfrac{1}{2} n ( n - 1 ) \Tr(\rho^2 )^{n -2}\, r \cos^2(\theta) \\
    &- \tfrac{1}{4} n^2(n-1) \Tr(\rho^2)^{n-1} r^3 \cos^2(\theta)\\
    &+ \tfrac{1}{4} n ( n - 1 )(n-2) \Tr(\rho^2 )^{n -3}\,r^3 \cos^2(\theta)\\
    &+\tfrac{1}{4} n^2 \Tr(\rho^2)^{n-2} (1 + n r^2) (r^2 - 1 )r  \cos^2(\theta).
  \end{split}
\end{equation}

Having now simplified the the projections, we are able to include the inverse metric components in Eq. (\ref{chProj:eq:metricInverseMatrixSpherical}) to arrive at the proper projections.
The Hamiltonian projections are
\footnotesize
\begin{equation}\label{chProj:eq:HamProjections}
  \begin{split}
    h^r(\Vx, t) &= 0\\
    h^\theta(\Vx, t) &= g^{\theta \theta}\inprod{D_\theta}{ -i [H,\,\varrho ] } = f^2(t)\, r \cos\, \phi - f^1(t)\,r \sin\, \phi \\
    h^\phi(\Vx, t) &= g^{\phi \phi} \inprod{D_\phi}{ -i [H,\,\varrho ] } =  f^3(t)\, r \sin\, \theta - f^1(t)\,r \cos\,\theta\, \cos\,\phi - f^2(t)\, r \cos\,\theta\, \sin\, \phi.
  \end{split}
\end{equation}
\normalsize
The dissipator projections are
\begin{equation}\label{chProj:eq:LiouProjections}
  \begin{split}
    l^r(\Vx) &=  g^{r r}\inprod{D_r}{ \kappa\,\mathcal{D}[ J_z](\varrho ) }= -\half \kappa \, r \sin^2\theta\\
    l^\theta(\Vx) &= g^{\theta \theta}\inprod{D_\theta}{\kappa\, \mathcal{D}[ J_z](\varrho )  } = -\tfrac{1}{4} \kappa \, r \sin 2\theta \\
    l^\phi(\Vx) &=0.
  \end{split}
\end{equation}
The conditioning projections are
\begin{equation}\label{chProj:eq:StochProjections}
  \begin{split}
    b^r(\Vx) &= g^{r r}\inprod{D_r}{ \sqrt{\kappa}\,\mathcal{H}[ J_z](\varrho ) }= - \sqrt{\kappa}(1 - r^2) \cos \theta\\
    b^\theta(\Vx) &= g^{\theta \theta}\inprod{D_\theta}{\sqrt{\kappa}\, \mathcal{H}[ J_z](\varrho )  } = -\sqrt{\kappa} \sin \theta \\
    b^\phi(\Vx) &=0.
  \end{split}
\end{equation}
The It\={o} correction projections are
\begin{equation}\label{chProj:eq:ItoProjections}
  \begin{split}
    c^r(\Vx) &= g^{r r}\inprod{D_r}{\kappa\,  \mathcal{I}_c[ J_z](\varrho ) } = l^r(\Vx) +\kappa\,\alpha(r)\,  r \cos^2\theta\\
    c^\theta(\Vx) &= g^{\theta \theta}\inprod{D_\theta}{\kappa\,  \mathcal{I}_c[ J_z](\varrho )  } = l^\theta(\Vx) + \half  \kappa\, \beta(r)\, r \sin 2 \theta \\
    c^\phi(\Vx) &=0
  \end{split}
\end{equation}
where we defined the coefficients
\begin{equation}\label{chProj:eq:alphaCoeff}
\begin{split}
      \alpha(r) \define &\,  n (r^2-1 ) + \frac{2(n-1)}{(1 + n r^2)}
      - \frac{n (n - 1) (1 + r^2)}{2 (1 + n r^2)}\, r^2
      + \frac{2 (n - 2) (n - 1)}{(1 + r^2) (1 + n r^2)}\, r^2
\end{split}
\end{equation}
and
\begin{equation}\label{chProj:eq:betaCoeff}
    \beta(r) \define n - 2\frac{  n - 1}{1 + r^2}.
\end{equation}

To bring all of this together, we are reminded that the projected conditional master equation is
\begin{equation}
   \begin{split}
     d \rho_t = \,\left(\,  h^i + l^i - c^i\, \right) D_i \, dt + b^i\, D_i\circ dw_t.
   \end{split}
\end{equation}
Furthermore a general Stratonovich SDE is traditionally written as $d x_t = \tilde{A}(x_t)\, dt + B(x_t)\circ dw_t$.  To conform to this convention we will define the coefficients
\begin{equation}\label{chProj:eq:DetStratonovich}
   \begin{split}
    \tilde{a}^r(\Vx, t) \define&\, - \kappa\,\alpha(r)\,  r \cos^2\theta\\
    \tilde{a}^\theta(\Vx, t) \define&\, h^\theta(\Vx,t) - \half \kappa\,\beta(r)\, r \sin 2 \theta\\
    \tilde{a}^\phi(\Vx, t) \define&\,h^\phi(\Vx,t).
   \end{split}
\end{equation}
with the Hamiltonian projections $h^i(\Vx, t)$ defined in Eq. (\ref{chProj:eq:HamProjections}).
The projected conditional master equation is finally given by the Stratonovich SDE
\begin{equation}
   \begin{split}\label{chProj:eq:projectedCME}
     d\rho_t =\,\tilde{a}^i(\Vx, t) D_i(\Vx) \, dt + b^i(\Vx)\, D_i(\Vx) \circ dw_t
   \end{split}
\end{equation}
with the spherical tangent vectors defined in Eq. (\ref{chProj:eq:sphericalTangents}).

\section{The Projection Filter\label{chProj:sec:projectionFilter}}
From the projected conditional master equation in Eq. (\ref{chProj:eq:projectedCME}), we would like to find a closed set of easily simulated stochastic differential equations.  We now have a nonlinear matrix-valued SDE that propagates the closest identical product state to the exact conditional state.  A single qubit state is completely characterized by its Bloch vector, and so to simulate $\rho_t$ we need only find SDEs for the three Bloch components.

In the notation of quantum filtering theory \cite{bouten_introduction_2007}, the filter is a QSDE that propagates the conditional expectation of a given observable,
\begin{equation}\label{chProj:eq:filterDefinition}
\pi_t(X) \cong \Tr(\rho_t\, X).
\end{equation}
The filter is generally expressed as a differential, so that for the observable $X$ (on $n$ qubits), we have
\begin{equation}\label{chProj:eq:filterDifferential}
d\pi_t(X) = \Tr(d \rho_t\, X) = \tilde{a}^i(\Vx, t)\, \Tr(D_i(\Vx) X)\, dt + b^i(\Vx)\, \Tr(D_i(\Vx) X) \circ dw_t.
\end{equation}

The Bloch vector components can of course be identified by the expectation value of the Pauli operators.  For $n$ qubits, we also have the relation that $\Tr(J_i\, \varrho) = \frac{n}{2}\Tr(\sigma_i\, \rho) = \frac{n}{2}\, x^i,$  so to extract SDEs for the Cartesian Bloch components $x^i$ from $d\rho_t$, we simply need to calculate the filtering equations for the operators $2 J_i/n$.  Thus,
\begin{equation}\label{chProj:eq:dxDefinition}
d x^i_t = \frac{2}{n} \Tr(d \rho_t\, J_i) = \frac{2}{n}\left( \tilde{a}^\alpha(\Vx_t, t)\, \Tr(D_\alpha(\Vx_t) J_i )\, dt + b^\alpha(\Vx_t)\,\Tr(D_\alpha(\Vx_t) J_i ) \circ dw_t \right)
\end{equation}
for $\alpha \in \set{ r, \theta, \phi}$ and $i \in \set{ 1,2,3} $.  Explicit calculation shows that
\begin{equation} \label{chProj:eq:DiJiSp}
\begin{split}
    \Tr(D_r(\Vx) J_i ) &= \frac{n}{2}\left(\, \sin\theta\, \cos \phi\, \delta_{i\, 1} + \sin\theta\, \sin \phi\, \delta_{i\, 2} +  \cos \theta\, \delta_{i\, 3}\,\right)\\
    \Tr(D_\theta(\Vx) J_i ) &=\frac{n}{2}\left(\, \cos \theta\, \cos \phi\, \delta_{i\, 1} + \cos \theta\, \sin \phi\, \delta_{i\, 2} -  \sin \theta\, \delta_{i\, 3}\,\right)\\
    \Tr(D_\phi(\Vx) J_i ) &= \frac{n}{2}\left(\,-\sin \phi\, \delta_{i\, 1} + \cos  \phi\, \delta_{i\, 2}\,\right).
\end{split}
\end{equation}

These equations show how to find a mixed coordinate expression for the projection filter, \emph{i.e.} it expresses $dx$ in terms of the variables $r, \theta, \phi$. We would like to combine these results with the expressions for $\tilde{a}^\alpha$ and $b^\alpha$, Eqs. (\ref{chProj:eq:DetStratonovich} and \ref{chProj:eq:StochProjections}) to obtain deterministic and stochastic coefficients expressed in Cartesian coordinates.  We seek the functions, $\tilde{a}^i(\Vx, t)$ and $b^i(\Vx)$ such that
\begin{equation}\label{chProj:eq:dxSimple}
d x^i_t = \tilde{a}^i(\Vx_t, t)\, dt + b^i(\Vx_t)\,\circ dw_t.
\end{equation}
With the standard conversion between spherical and Cartesian coordinates, Eqs. (\ref{chProj:eq:HamProjections}, \ref{chProj:eq:DetStratonovich}, \ref{chProj:eq:StochProjections}, and \ref{chProj:eq:DiJiSp}) we can easily find these coefficients.

In Cartesian coordinates the deterministic Stratonovich coefficients are
\begin{equation}\label{chProj:eq:atildeCartesian}
  \begin{split}
    \tilde{a}^1(\Vx,t) =&\, f^2(t)\,x^3 - f^3(t)\, x^2 - \kappa\,\left(\, \alpha(r)  + \beta(r)\, \right) x^1\,\frac{(x^3)^2}{r^2},\\
    \tilde{a}^2(\Vx,t) =&\, f^3(t)\,x^1 - f^1(t)\, x^3 - \kappa\,\left(\, \alpha(r)  + \beta(r)\, \right) x^2\,\frac{(x^3)^2}{r^2},\\
    \tilde{a}^3(\Vx,t) =&\, f^1(t)\,x^2  -f^2(t)\, x^1 + \kappa\, \beta(r)\, x^3 - \kappa\,\left(\, \alpha(r)\,  + \beta(r)\, \right) x^3\, \frac{(x^3)^2}{r^2},
  \end{split}
\end{equation}
and the stochastic coefficients are
\begin{equation}\label{chProj:eq:StochCartesian}
  \begin{split}
    b^1(\Vx) &=  - \sqrt{\kappa}\, x^1\, x^3, \\
    b^2(\Vx) &=  - \sqrt{\kappa}\, x^2\, x^3, \\
    b^3(\Vx) &= \sqrt{\kappa}(1 - (x^3)^2).
  \end{split}
\end{equation}
Note that the functions $\alpha(r)$ and $\beta(r)$ actually only depend upon $r^2 = \norm{\Vx}^2$.

To complete the derivation we will convert these Stratonovich equations back to the It\={o} form.  The It\={o} correction for the multivariable $b^i$ coefficients is given by,
\begin{equation}\label{chProj:eq:multiVariableItoCorrection}
    \Delta\,a^i(\Vx) =  a^i(\Vx,t) -  \tilde{a}^i(\Vx,t) = \frac{1}{2} b^j(\Vx) \dydx{b^i(\Vx)}{x^j}.
\end{equation}
Substituting $b^{i}(\Vx)$ into this formula, the It\={o} corrections simplify to
\begin{equation}\label{chProj:eq:ICprojCartesian}
    \begin{split}
      \Delta\,a^1(\Vx) =&\, \kappa\, x^1 \left((x^3)^2 -\half\right) \\
      \Delta\,a^2(\Vx) =&\, \kappa\, x^2 \left((x^3)^2 -\half\right) \\
      \Delta\,a^3(\Vx) =&\, \kappa\, x^3 \left((x^3)^2 - 1 \right).
    \end{split}
\end{equation}
By adding these terms to the deterministic coefficients in Eq. (\ref{chProj:eq:DetStratonovich}) we find
\begin{equation}\label{chProj:eq:aCartesian}
  \begin{split}
    a^1(\Vx,t) =&\, f^2(t)\,x^3 - f^3(t)\, x^2 - \half \kappa\, x^1 + \kappa\, \gamma(r)\, x^1\, (x^3)^2,\\
    a^2(\Vx,t) =&\, f^3(t)\,x^1 - f^1(t)\, x^3 - \half \kappa\, x^2 + \kappa\, \gamma(r) \, x^2 \,(x^3)^2,\\
    a^3(\Vx,t) =&\, f^1(t)\,x^2  -f^2(t)\, x^1 + \kappa\, (\beta(r) - 1) \, x^3 - \kappa\, \gamma(r) (x^3)^3
  \end{split}
\end{equation}
were we defined the new coefficient function $\gamma(r)$ as
\begin{equation}\label{chProj:eq:gammaCoeff}
  \begin{split}
    \gamma(r) &\define 1 -  \frac{\alpha(r) + \beta(r)}{r^2} = (1 - r^2) \left( \frac{n\,(n+1)}{2\,(1 + n\, r^2)} -  \frac{1}{1 + r^2} \right)
  \end{split}
\end{equation}
and the $\beta(r)$ coefficient is given in Eq. (\ref{chProj:eq:betaCoeff}).  Note that for any $n$ and $0 \le r \le 1$, the coefficient $\gamma(r)$ is strictly nonnegative.  The zeros of this function however are two special cases, which we will discuss next.

If we substitute the stochastic coefficients $b^i(\Vx)$ into Eq. (\ref{chProj:eq:multiVariableItoCorrection}) the projection filtering equations are
\begin{equation}\label{chProj:eq:projectionFilterIto}
\begin{split}
    dx^1_t =&\, a^1(\Vx_t, t)\, dt - \sqrt{\kappa}\, x^1_t\, x^3_t\, dw_t\\
    dx^2_t =&\, a^2(\Vx_t, t)\, dt - \sqrt{\kappa}\, x^2_t\, x^3_t\, dw_t\\
    dx^3_t =&\, a^3(\Vx_t, t)\, dt + \sqrt{\kappa}(1 - (x^3_t)^2)\, dw_t.
\end{split}
\end{equation}

\subsection{Special cases for the projection filter\label{chProj:sec:SpecialProjectionFilters}}

There exist two very interesting special cases for the projection filter. The first is when we only have a single qubit, and the second is when the state is pure.
We have already shown how, when $n = 1$, the metric is simply a Euclidean metric.  (Up to a factor of a half.) Therefore, we expect the projection filtering equations to dramatically simplify for a single qubit.  The first thing to notice is that when $n = 1$ the $\beta(r)$ coefficient Eq. (\ref{chProj:eq:betaCoeff}) simplifies to
\begin{equation}
  \beta(r) = 1 \qquad \text{ for  } n = 1.
\end{equation}
Furthermore, for $n = 1$ the $\gamma(r)$ coefficient Eq. (\ref{chProj:eq:gammaCoeff}) simplifies to
\begin{equation}
  \gamma(r) = 0 \qquad \text{ for  } n = 1.
\end{equation}
When we evaluate the projection filter for pure states, we arrive at remarkably similar results.  In other words for any $n = 1, 2, \dots$ we also have
\begin{equation*}
    \beta(r = 1) = 1  \quad \text{ and } \quad  \gamma(r = 1) = 0.
\end{equation*}

These fantastic simplifications means that for a (possibly mixed) single qubit or \emph{any} pure multi-qubit state, the projection filtering equations are simply
\begin{equation} \label{chProj:eq:projectionFilterItoSingleQubit}
  \begin{split}
    dx^1_t =&\, \left( f^2(t)\,x^3_t - f^3(t)\, x^2_t - \half \kappa\, x^1_t\,\right)\, dt  - \sqrt{\kappa}\, x^1_t\, x^3_t\, dw_t\\
    dx^2_t =&\, \left( f^3(t)\,x^1_t - f^1(t)\, x^3_t - \half \kappa\, x^2_t\,\right)\, dt - \sqrt{\kappa}\, x^2_t\, x^3_t\, dw_t\\
    dx^3_t =&\, \left( f^1(t)\,x^2_t  -f^2(t)\, x^1_t\, \right)\, dt + \sqrt{\kappa}\, (1 - (x^3_t)^2)\, dw_t.
  \end{split}
\end{equation}
It is an interesting exercise to see that if one simply computed the Heisenberg picture filtering equations for the Pauli operators, $\pi_t(\sigma_i)$, one arrives at these vary same 3 coupled SDEs.

What the pure state evaluation says is that when $r = 1$, the dynamics of the separable system no-longer depends on the total number of qubits and evolve simply as $n$ identical copies of a single qubit state.  The only remaining dependence upon $n$ is in the innovations process, where we have the differential $d w_t \rightarrow dv_t = d y_t - \sqrt{\kappa}\, n\, x^3_t\, dt$, where $ dy_t$ is the integrated measurement current in time $[t, t + dt]$.

One finial question about the pure state projection filter is, ``Will the filter remain pure, once it becomes pure?''  In other words, is $r = 1$ a trap for the system so that if $r_s = 1$ for some time $s$, will $r_{t} =1$ for all $t \ge s$.  The physics of the situation dictates that it must, as with no sources of decoherence the state will purify and stay pure.  However, it is difficult to see that this is indeed the case simply by inspecting Eq. (\ref{chProj:eq:projectionFilterItoSingleQubit}).  Where this not the case, then it would likely indicate an error with the model.   Thankfully the answer is decidedly yes and can be seen in the spherical basis representation of $d \rho_t$, with the Stratonovich coefficients $\V{\tilde{a}}(\Vx, t)$ and $\V{b}(\Vx,t)$ given in Eqs. (\ref{chProj:eq:DetStratonovich}) and (\ref{chProj:eq:StochProjections}).  From these two equations we see that
\begin{equation}
    \tilde{a}^r(\Vx, t) = - \kappa\,\alpha(r)\,  r \cos^2\theta
\end{equation}
and that
\begin{equation}
  b^r(\Vx, t) = - \sqrt{\kappa}(1 - r^2) \cos \theta.
\end{equation}
Substituting in for $r = 1$ into Eq. (\ref{chProj:eq:alphaCoeff}), we find that $\alpha(r = 1) = 0$.  But this means that $\tilde{a}^r(\Vx, t)|_{r = 1} = 0$ and $ b^r(\Vx, t)|_{r = 1} = 0$.
Following this logic to its conclusion we find that
\begin{equation}
\begin{split}
   d \rho|_{r = 1} &= \tilde{a}^\theta(\Vx, t)|_{r = 1}\, D_{\theta} \, dt + b^\theta(\Vx, t)|_{r = 1}\, D_{\theta} \circ dw_t\\
   &\quad +\tilde{a}^\phi(\Vx, t)|_{r = 1}\, D_{\phi} \, dt + b^\phi(\Vx, t)|_{r = 1}\, D_{\phi} \circ dw_t .
\end{split}
\end{equation}
In other words, by evaluating $d \rho $ at $r = 1$ we see that it is independent of the tangent vector $D_r$ and thus will remain on the submanifold of identical separable states defined by $r = 1$.

\section{Simulations and Performance \label{chProj:sec:Simulations} }

The primary purpose for deriving the projection filter was to find a reduced dimensional description of a system of qubits undergoing a continuous measurement of the collective angular momentum $J_z$.  However we know that the end result of a continuous measurement of $J_z$, absent of any other influence, is an eigenstate of $J_z$, a so called Dicke states.   With the exception of the stretched states, Dicke states are \emph{not} separable.  Even after a relatively short time, the reduction of the $J_z$ spin component results in a nonclassical spin squeezed state, a phenomena observed in several experiments \cite{kuzmich_generation_2000,hald_spin_1999,schleier-smith_states_2010}.

In this section we test through numerical simulation how well the projection filter reproduces certain properties of the exact conditional state.  Here we show that by adding strong randomized external control fields during the collective measurement the joint state remains highly separable and the approximate description of the projection filter reproduces collective expectation values much more faithfully.

Chap. \ref{chap:QubitState} uses the projection filter in an algorithm to reconstruct the initial condition of a SCS from a continuous measurement of $J_z$, characterized by the rate $\kappa$.  In order to obtain information about observables other than $J_z$, an external control Hamiltonian must be applied.  For reasons discussed in Sec. \ref{chQubit:sec:ControlLaw}, this takes form of a sequence of global $\pi/2$ rotations, where each rotation is about an axis $\V{n}$ that was independently sampled from a uniform distribution.  Therefore, we will test the performance of the projection filter with this control law in mind.

Fully characterizing the control amplitude $\V{f}(t)$ requires specifying the amplitude and duration of each pulse, as a larger Larmor frequency is needed to enact the same rotation in a shorter time.  For simplicity, we will fix $\V{f}(t)$ to have a constant magnitude and will only vary its direction.  This constrains the $\pi/2$ rotations to be square-wave pulses, each of duration $\tau$,
\begin{equation}\label{chProj:eq:randomControl}
\V{f}(t) = \frac{\pi}{2\, \tau} \sum_{m =1} \indicate{[{m-1}, m)}\hspace{-4pt}\left(\,t/\tau\right)\, \V{n}_m
\end{equation}
where $\indicate{[a, b)}\hspace{-4pt}(t)$ is the indicator function for the interval $[a, b)$ and  $\set{ \V{n}_m} $ are \emph{i.i.d.} unit vectors drawn from a isotropic distribution.

To efficiently simulate the exact dynamics we utilize two conserved quantities.  The first is that because the system Hamiltonian $H_t$ and measurement operator $L$ commute with $J^2$, the total angular momentum of the atomic system will be conserved.  Furthermore, the states we ultimate use are all initialized in states with a maximum projection of angular momentum along some direction, thereby always possessing $n/2$ units of angular momentum. This allows us to restrict the simulations to a $d = n +1$-dimensional space.  In other words we simulate a single $J = n/2$ spin system.

The second conservation property we will use is the fact that without any additional sources of decoherence, the conditional master equation maps pure quantum states to pure quantum states.  Sec. \ref{chMath:sec:CSE} discusses the \emph{conditional Schr\"{o}dinger equation} (CSE) and how it can be derived from a conditional master equation (CME).  Using a CSE generates significant computational savings, as each time step propgates a single a complex vector, rather than a complex matrix.  The general form of the CSE is given in Eq. (\ref{chMath:eq:CSE}).  In our case, $L = \sqrt{\kappa} J_z$ for a real, positive $\kappa$ this simplifies to
\begin{equation}\label{chProj:eq:CSE}
  d \ket{\psi_t} = \left(\,-i H_t - \half \kappa ( J_z - \expect{J_z} \,)^2\, \right) \ket{\psi_t}\, dt + \sqrt{\kappa }\big( J_z - \expect{J_z} \big)\, \ket{\psi_t}\,  d w_t.
\end{equation}

\subsection{Simulation parameters}

In absence of the control Hamiltonian $H_t$, the one universal timescale in the CSE is set by the measurement strength, \emph{i.e.} the characteristic time $\kappa^{-1}$.  Therefore these simulations are all reported in time units of this characteristic time.  The range of qubits total qubit numbers we will test are between $25 - 100$, meaning that the simulations will be of  collective spin values of $ 12.5 \le J \le 50$.  In addition to these collective spins, we will compare the projection filtering equations to the exact simulations for a single qubit, proving that they generate the same dynamics.

The remaining parameters, namely the gate duration $\tau$ and the fixed terminal simulation time $t_f$ will be chosen to correspond to the parameters that will be ultimately used in Chap. \ref{chap:QubitState}.    Specifically, $\tau = 5 \times 10^{-3}\,\kappa^{-1}$ and $t_f = 0.2\,\kappa^{-1}$.

The actual simulations are implemented in the MATLAB computing environment using a hand coded, weak second order predictor-corrector stochastic differential equation integrator.  The algorithm is described by
 \citet[page 200]{kloeden_numerical_1994} and was implemented in MATLAB by Brad Chase for his PhD dissertation \cite{chase_parameter_2009}.

\subsection{Spin squeezing comparisons\label{chProj:sec:spinSqueezing}}

Spin squeezing is a much sought after and well studied effect in atomic spin ensembles.  The general phenomena describes the reduction in uncertainty in a expected value of a spin component transverse to the mean spin direction.  Due to a Heisenberg uncertainty relationship, this reduction in uncertainty is accompanied by an increase in uncertainty in the orthogonal quadrature.

The standard example is to consider an collective spin system composed of $n$ qubits, initialized in a SCS pointing along the $+\Ve_x$ direction.  This state is clearly an eigenstate of $J_x$ with eigenvalue $m_x = J = n/2$.  It is also easy to show that this state is a minimum uncertainty state so that it minimizes the Heisenberg uncertain relation
\begin{equation}\label{chProj:eq:spinUncertainty}
  \expect{\Delta J_z^2} \expect{\Delta J_y^2} \ge \frac{1}{4} \expect{J_x}^2
\end{equation}
with equal uncertainties $\expect{\Delta J_z^2} = \expect{\Delta J_y^2} = \half J$.  An example of a spin squeezed state is a state that is still mostly polarized along the $\Ve_x$ axis but also contains quantum correlations so that $\expect{\Delta J_z^2} < \expect{\Delta J_y^2}$ but still maintains the equality of Eq. (\ref{chProj:eq:spinUncertainty}) \cite{kitagawa_squeezed_1993}.

A spin squeezed state is one of the immediate consequences of a continuous measurement of a collective angular momentum variable, such as $J_z$ \cite{kuzmich_atomic_1998}, for a SCS prepared transverse to the measurement axis.  A number of papers have investigated the relation between spin squeezing and various measures of entanglement (see \emph{e.g.} \citep{yin_spin_2011} and references therein).  One particular measure of spin squeezing, $\xi^2_T$, has been shown by \citeauthor{yin_spin_2011} to be directly related to the concurrence, a measure of pairwise entanglement \citep{yin_spin_2011}.  They show that when the concurrence $C$ is greater then zero, indicating entanglement, then $\xi^2_T < 1$ and that when $\xi^2_T \ge 1$, $C = 0$ and the state is unentangled.  $\xi_T^2$ takes on the following definition.

For each component of angular momentum we can compose the symmetrize correlation and covariance matrices ($i, j = x, y, z$)
\begin{equation}
  \operatorname{Corr}_{i,j} = \half \expect{J_i J_j + J_j J_i}
\end{equation}
and
\begin{equation}
  \operatorname{Covar}_{i,j} = \operatorname{Corr}_{i,j} - \expect{J_i}\expect{J_j}.
\end{equation}
From these matrices we can also form the Hermitian matrix
\begin{equation}
  \Gamma = (n-1) \operatorname{Covar} + \operatorname{Corr}.
\end{equation}
The squeezing parameter $\xi^2_T$ is then defined as
\begin{equation}\label{chProj:eq:TothSqueezing}
  \xi_T^2 \define \frac{\lambda_{\min}}{\expect{J^2} - \frac{n}{2}}
\end{equation}
where $\lambda_{\min}$ is the minimum eigenvalue of the matrix $\Gamma$.

\subsection{Squeezing simulations\label{chProj:sec:SqueezingSimulations}}
This section presents simulations that benchmark the typical effect the measurement has upon the states of interest, as well as how much the control law mitigates these effects.  We test here five classes of states, each composed of $n = \{1, 25, 50, 75, 100\}$ qubits.  The one qubit case is included as a control, testing that the numerics produce reasonable results.

As discussed previously, a continuous measurement of $J_z$ is a standard protocol for producing a spin squeezed state.  However, a spin coherent state (SCS) prepared along to the $\Ve_z$ axis will not squeeze at all, while the states prepared in the equatorial plane squeeze the most.  Therefore to demonstrate the maximum amount of quantum correlations a typical measurement realization can produce, a natural choice is a SCS prepared along the $\Ve_x$ axis.
\begin{figure}[bht]
	\begin{center}
		\includegraphics[width=1\hsize]{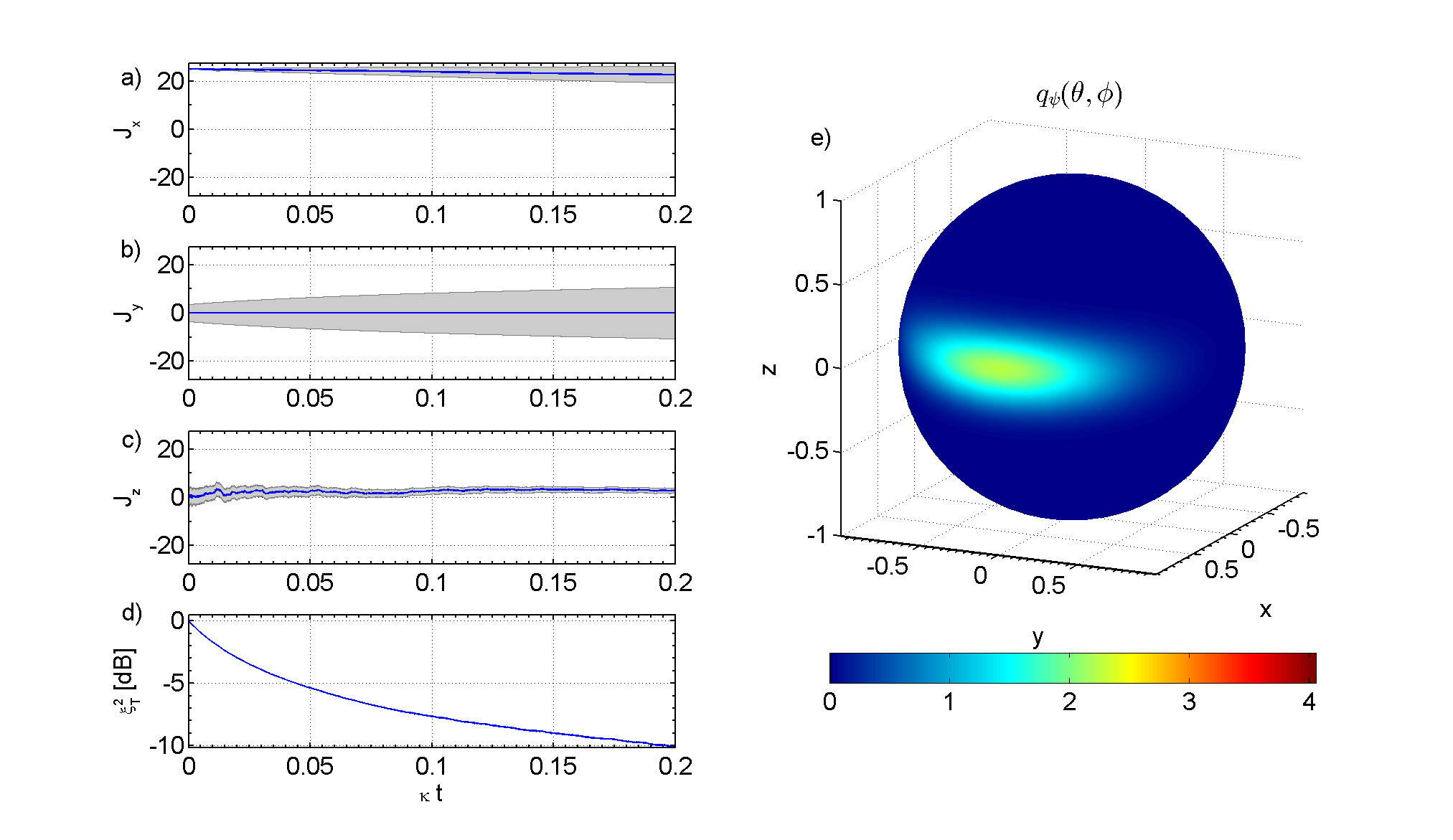}
   		 \caption{Typical Uncontrolled Evolution.  CSE evolution for a $+\Ve_x$ SCS initial state, containing $n = 50$ qubits, ($J = 25$). a-c: The conditional expectation values for $J_x$, $J_y$, and $J_z$ vs time, including a $1 \sigma$ region of confidence. d: The squeezing parameter $\xi_T^2$ in dB vs time.  e: The Husimi Q-function $q_\psi(\theta, \phi)$ for the conditional state, $\ket{\psi_{t_f}}$, at the final time $t_f = 0.2\,\kappa^{-1}$. \label{chProj:fig:FrameNoControl} }
		\end{center}
\end{figure}

Fig. \ref{chProj:fig:FrameNoControl} characterizes the typical results of a CSE simulation for a $+\Ve_x$ SCS initial state, containing $n = 50$ qubits, (spin $J = 25$).  Figs. \ref{chProj:fig:FrameNoControl}.a - \ref{chProj:fig:FrameNoControl}.c show the conditional expectation values of $J_x$, $J_y$ and $J_z$ as a function of time.  Also plotted are $1\sigma$ regions of confidence indicating the expected deviations from these mean values. In other words, the grey regions are bounded by the values $\expect{J_i} \pm \sqrt{\expect{\Delta J_i^2}}$.  Fig. \ref{chProj:fig:FrameNoControl}.a shows how $\expect{J_x}$ tends to decrease as the state squeezes around the sphere and how its uncertainty grows.  Fig. \ref{chProj:fig:FrameNoControl}.b shows how $\expect{J_y}$ remains zero throughout the measurement, while its variance increase due to the characteristic anti-squeezing.  Conversely, Fig. \ref{chProj:fig:FrameNoControl}.c shows the decrease in $\expect{\Delta J_z^2}$ due to the squeezing as well as the deviation of $\expect{J_z}$ from zero as the system evolves towards an eigenstate of $J_z$.
Fig. \ref{chProj:fig:FrameNoControl}.d plots the evolution of the squeezing parameter $\xi_T^2$, in dB, as a function of time.  ( Where $\xi_T^2$ in decibels is $10 \log_{10}( \xi_T^2 )$. )  Fig. \ref{chProj:fig:FrameNoControl}.e  shows a single 3D plot of the Husimi Q-function quasi-probability distribution for the state at the final time $t_f$.   The Q-function for a pure spin state $\psi$ with total angular momentum $J$  is defined as
\begin{equation}
   q_\psi(\theta, \phi) \define \frac{2 J + 1}{4 \pi}\, \abs{\braket{\theta,\phi}{\psi} }^2
\end{equation}
where $\ket{\theta,\phi}$ is a SCS parameterized by the polar angles $\theta$ and $\phi$.  The constant factor ensures normalization.  The color scale of Fig. \ref{chProj:fig:FrameNoControl}.e has been normalized to the maximum value of $\frac{2 J + 1}{4 \pi}$.  The fact that the Q-function has a maximum value of $\sim 2.2$ shows that this squeezed state has poor overlap with SCSs.

\begin{figure}[htb]
	\begin{center}
		\includegraphics[width=1\hsize]{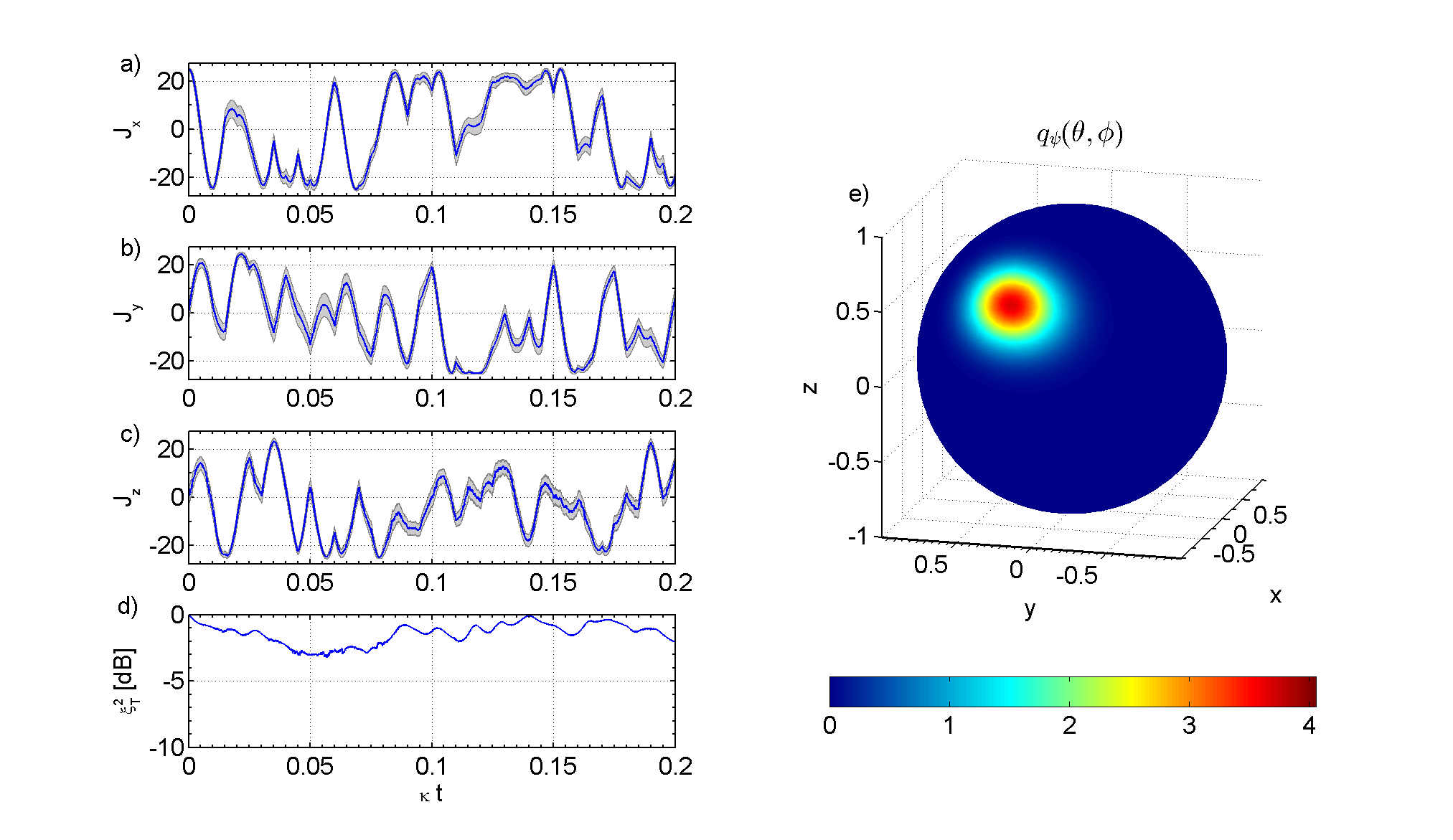}
   		 \caption{Typical Controlled Evolution.  CSE evolution for a $+\Ve_x$ SCS initial state, containing $n = 50$ qubits, with 40 randomized $\pi/2$ rotations. a-c: The conditional expectation values for $J_x$, $J_y$, and $J_z$ vs time, including a $1 \sigma$ region of confidence. d: The squeezing parameter $\xi_T^2$ in dB vs time.  e: The Husimi Q-function $q_\psi(\theta, \phi)$ for the conditional state, $\ket{\psi_{t_f}}$, at the final time $t_f = 0.2\,\kappa^{-1}$. \label{chProj:fig:FrameControl} }
		\end{center}
\end{figure}

Fig. \ref{chProj:fig:FrameControl} presents a typical realization for the same initial condition in the presence of the randomized $\pi/2$ rotations, with a duration of $\tau = 5 \times 10^{-3}\, \kappa^{-1}$.  By the time $t_f = 0.2\, \kappa^{-1}$, this period leads to a total of 40 rotations or one for every horizontal tick mark in Figs. \ref{chProj:fig:FrameControl}.a - \ref{chProj:fig:FrameControl}.d.  Figs. \ref{chProj:fig:FrameControl}.a - \ref{chProj:fig:FrameControl}.c show the conditional expectation values $\expect{J_x}$, $\expect{J_y}$ and $\expect{J_z}$ as a function of time.  In contrast to the example, lacking the randomized controls, these Figs. show that there is little qualitative difference between the three expectation values.  Fig. \ref{chProj:fig:FrameControl}.d indicates that there is a significant reduction in the amount of squeezing produce during measurement compared to the uncontrolled system. $\xi^2_{T \min{}} = -3.21 \text{ dB}$ in the presence of controls while $\xi^2_{T \min{}} = -10.1 \text{ dB}$ without them.  The $1\sigma$ confidence regions indicate that this squeezing is not with respect to a fixed coordinate axis but is rotated between all three and so there is a substantial averaging effect that leads to far less squeezing than the uncontrolled case.  In fact there is nearly a factor of 5 decrease in the maximum amount of squeezing and therefore the controlled state is kept much more separable.  This separability is indicated in the Q-function of the final state, shown in Fig. \ref{chProj:fig:FrameControl}.e, with its more spherical appearance and near perfect overlap with a spin coherent state, indicated by the maximum value $\sim 3.7$.

The amount of squeezing is significantly reduced because the randomized controls tends to mix both the squeezed and anti-squeezed components leading to a near zero average.  Not only does the mean spin rotate, but the orientation of the squeezing ellipse also rotates.  As the rotation axes are chosen from a uniform distribution, the squeezed component is just as likely as the anti-squeezed component to be oriented along the measurement axis.  At any given time, the uncertainty in the $J_z$ component is equally likely to be above or below the uncertainty of an equivalent spin coherent state.  Therefore it is difficult for any significant squeezing to develop.

\subsection{Projection filter simulations}

Ultimately we need to compare how well the projection filter performs when it calculates an innovation from a measurement record $y_t$ that \emph{is not} generated by a separable state.  In this case $dw_t$ is actually given by the innovation process, $d v_t = dy_t - \sqrt{\kappa}\, n\, x^3_t\ dt$, which is not a Wiener process in all cases.  We make this comparison through two measures.  The first is to see how well the projection filter is able to reproduce the expectation values $\expect{J_x}$, $\expect{J_y}$ and $\expect{J_z}$ compared to the exact conditional state.  The second is through the fidelity, squared overlap, between the exact and approximate states.
\begin{figure}[hbt]
	\begin{center}
		\includegraphics[width=1\hsize]{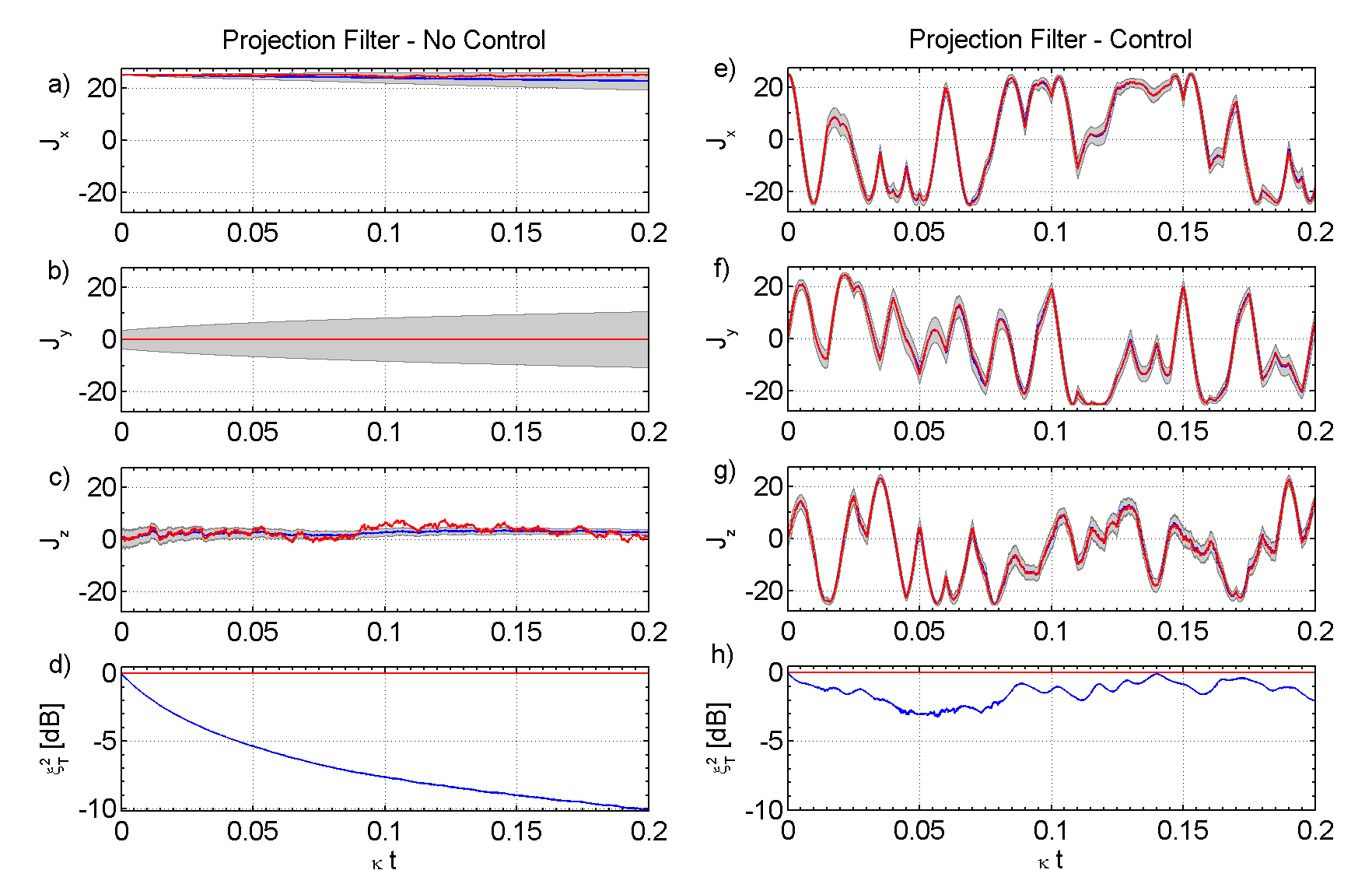}
   		 \caption{Projection Filter Tracking example. The comparison between the projection filter predictions and the exact conditional expectation values for the simulations shown in Figs. \ref{chProj:fig:FrameNoControl} and \ref{chProj:fig:FrameControl}. a-c: The conditional expectation values $\expect{J_i}$  are shown for the exact uncontrolled state (blue) with the $1\sigma$ regions of confidence (grey).  Also shown are the projection filtered values $\frac{n}{2}x^i_t$ (red).  d:  The squeezing in the exact uncontrolled state (blue) and the squeezing reported by the projection filter (red).  e-h: Same as a-d but with the controls now applied. \label{chProj:fig:ProjectionControlComparison}}
		\end{center}
\end{figure}
Fig. \ref{chProj:fig:ProjectionControlComparison} makes the comparison between the projection filter and the expectation value shown in Figs. \ref{chProj:fig:FrameNoControl} and \ref{chProj:fig:FrameNoControl}.  Figs. \ref{chProj:fig:ProjectionControlComparison}.a - \ref{chProj:fig:ProjectionControlComparison}.c and \ref{chProj:fig:ProjectionControlComparison}.e - \ref{chProj:fig:ProjectionControlComparison}.g re-plot the true conditional expectation values, $\expect{J_i}$, as well as the projection filter values, given simply as $\frac{n}{2} \, x^i_t$.  Fig. \ref{chProj:fig:ProjectionControlComparison}.a shows that as the uncontrolled state becomes significantly squeezed, the $\expect{J_x}$ value reduces accordingly. The projection filter is unable to account for this and therefore has a noticeable error.  Fig. \ref{chProj:fig:ProjectionControlComparison}.c shows that after a time $t \sim 0.05 \, \kappa^{-1}$, there is an increase in the difference between the conditional expectation value $\expect{J_z}$ and the value calculated from the projection filter.  These differences are in stark contrast to the tracking results in Figs. \ref{chProj:fig:ProjectionControlComparison}.e - \ref{chProj:fig:ProjectionControlComparison}.g where the differences between all three expectation values are almost all within the line thicknesses.  Figs. \ref{chProj:fig:ProjectionControlComparison}.d and \ref{chProj:fig:ProjectionControlComparison}.h emphasise the fact that the projection filter reports separable states and so the projected squeezing parameter $\xi_T^2$ remains fixed at 0 dB.

Beyond these two sample trajectories, we also test the quality of the projection filter for a variety of initial states and qubit numbers.  This comparison is made in Fig. \ref{chProj:fig:ProjAverageError}, showing a trial averaged RMS error between the projection filter and the exact conditional expectation values.
\begin{figure}[bht]
	\begin{center}
		\includegraphics[width=1\hsize]{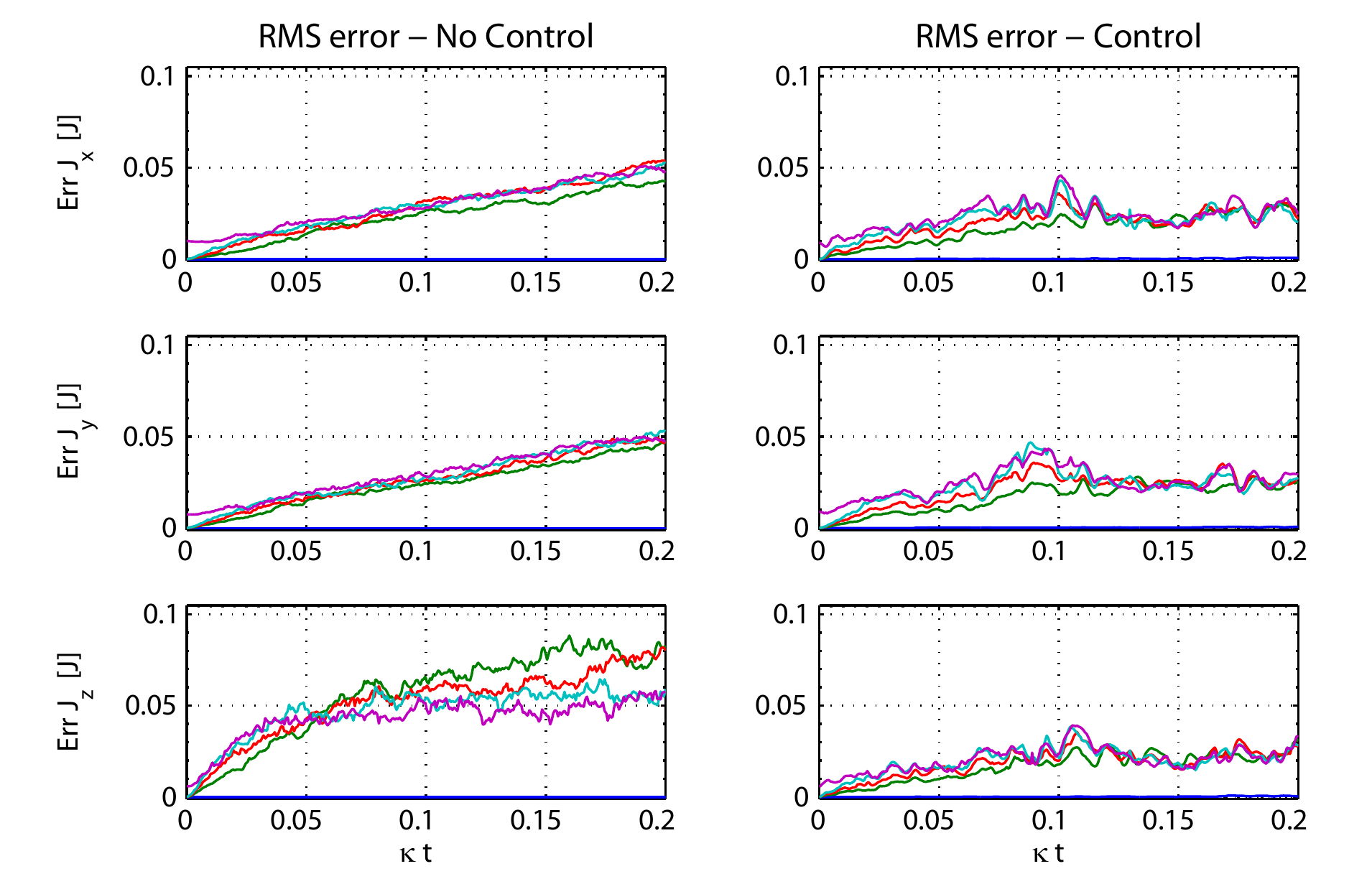}
   		 \caption{Average RMS Tracking error vs time. The left column shows $\text{Err } J_i$, with $i = x, y, z$ (in descending order) in the case of no control fields.  The right column shows the same but in the presence of control fields.  The average is over $\nu = 100$ uniformly random Bloch angles with a single noise realization per state.  For each Bloch vector, the RMS error is computed for $n = 1$ (blue), $n = 25$ (green), $n = 50$ (red), $n = 75$ (cyan) and $n = 100$ (purple) qubits.\label{chProj:fig:ProjAverageError} }
		\end{center}
\end{figure}
Here we test five different qubit values, $n \in \{1, 25, 50, 75, 100\}$.  For each $n$, the average is made over $\nu = 100$ input SCSs chosen at random with a uniform distribution over the Bloch sphere.  These same Bloch angles are used for each $n$.  For each input state we run a single simulation to compute the three exact conditional expectation values, $\expect{J_i}$, as well as the projection filter Bloch components $x^i_t$.  Then for each run we compute the RMS errors
\begin{equation}
\text{Err } J_i \define \sqrt{\expect{ \left(\,\tfrac{1}{J}\braOket{\psi_t}{J_i}{\psi_t}  - x^i_t\,\right)^2 }_\nu},
\end{equation}
as a function of time.  The expectation value $\expect{\cdot}_\nu$ represents the athermic mean over the $\nu$ trials.  The normalization of the exact expectation value means that $0 \le \text{Err } J_i \le 1$ or in other words, is in units of the total spin length $J$.  With the exception of the single qubit case (showing only numerical integration error), the scaled RMS errors are relatively independent of number of qubits.  This is likely due to the fact that when the system is in a pure state, the projection filtering equations are independent of $n$ (see Sec. \ref{chProj:sec:SpecialProjectionFilters}).

However, the presence of the strong randomized controls has a significant effect.  Sans the controls, Fig. \ref{chProj:fig:ProjAverageError} shows a near linear increase in $\text{Err }J_x$ and $\text{Err }J_y$.   In Fig. \ref{chProj:fig:ProjectionControlComparison}.a, the projection filter was unable to track the decrease in the $\expect{J_x}$ component as the uncontrolled state became squeezed and developed significant curvature on the sphere.  Were the system initialized in a $+ \Ve_y$ spin coherent state, the roles of $J_x$ and $J_y$ would be reversed but still have the same behavior.  We attribute the increase in $\text{Err }J_x$ and $\text{Err }J_y$ to this effect.  In contrast, when the randomized controls are applied, the RMS error is equally distributed across all expectation values and remain $\lesssim  5\%$ of the total spin length.  Additionally, the $\text{Err }J_z$ values is significantly worse in the uncontrolled case.  The $\sim 1\%$ error at time $t = 0$ in the $n = 100$ simulations is attributed to using Stirling's approximation to calculate the $J_z$ basis coefficients in the initial SCS.
\begin{figure}[hbt]
	\begin{center}
		\includegraphics[width=1\hsize]{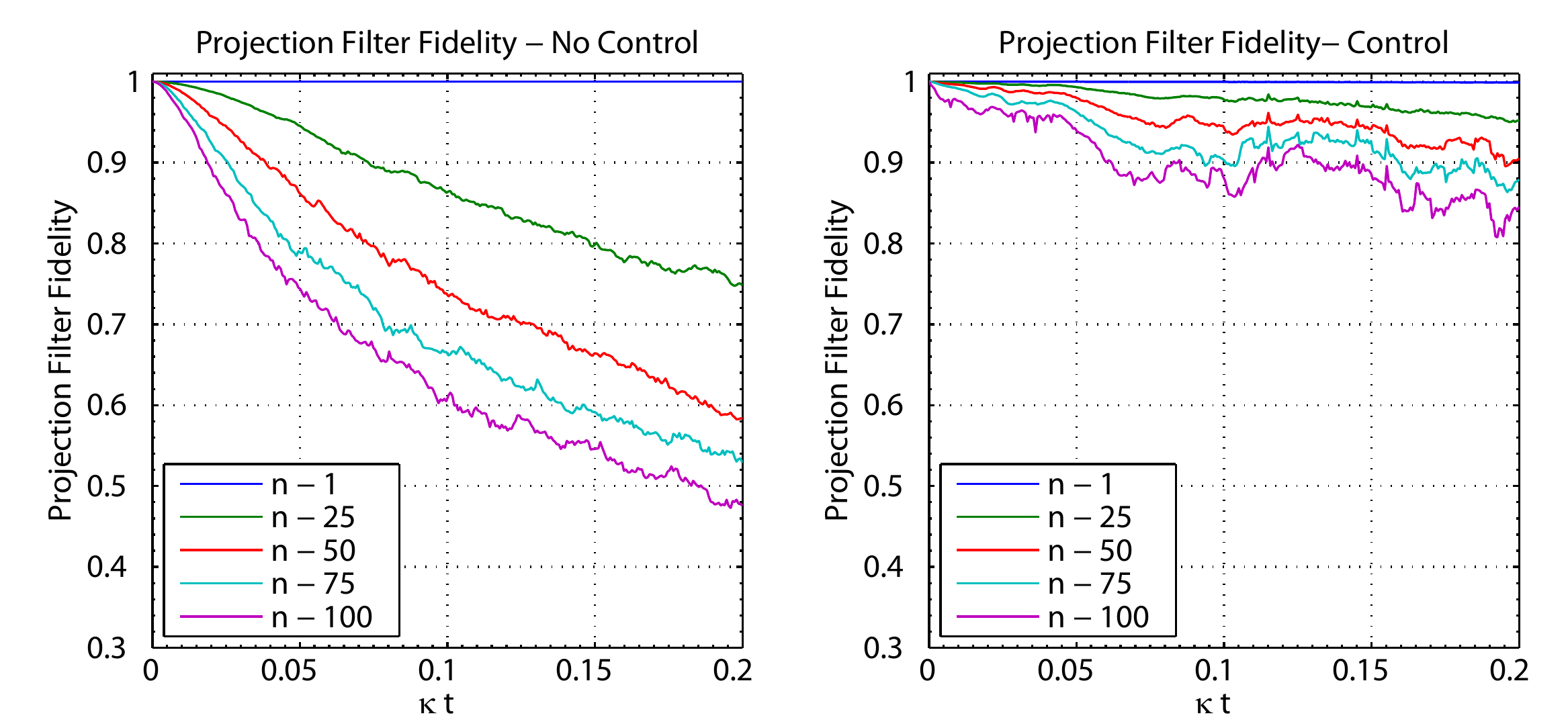}
   		 \caption{Average Projection Filter Fidelities.  These plots show the average fidelity between the exact CSE simulation and the SCS given by the projection filter.  The left plot shows the fidelity in the uncontrolled case with the right adding the 40 $\pi/2$ gates per simulation.  The average is over $\nu = 100$ uniformly random Bloch angles and a single noise realization per state.  For each Bloch vector, the average fidelity is computed for $n = 1$ (blue), $n = 25$ (green), $n = 50$ (red), $n = 75$ (cyan) and $n = 100$ (purple) qubits.\label{chProj:fig:ProjAvgerageFidelity} }
		\end{center}
\end{figure}

While these results indicate that the projection filter performs well in the presence of rapid, randomized rotations, an arbitrary spin state with $J > \half$ contains more information than simply three expectation values.  To characterize the general performance, we turn to the second comparison and calculate the average fidelity between the exact state and a SCS given by the projection filter.  Fig. \ref{chProj:fig:ProjAvgerageFidelity} makes this comparison, averaged over $\nu = 100$ uniformly sampled states. This is made both with and without controls and again for $n = \{ 1, 25, 50, 75, 100\}$ qubits.  The overall state fidelity for SCSs with $n >1$ shows a poor performance as the number of qubits increases, indicating an increase in the squeezing produced during the fixed measurement duration.  In the worse case with $n = 100$ and no controls applied, the average fidelity reaches a minimum value $\sim 0.47$.  However with the controls, the fidelity is $> 0.80$ for any $n$.  The non-monotonic decrease in the controlled fidelity suggests that the specifics of the control law impacts this fidelity and so it might be possible to optimize the control law so that the average fidelity is maximized.

\chapter{Qubit State Reconstruction \label{chap:QubitState} }

This chapter describes how to use the quantum filtering formalism to construct a tomographic estimate for an unknown initial quantum state from an ensemble of identical copies experiencing a joint continuous measurement.  We make a maximum likelihood estimate of this state, based upon the statistics of a continuous measurement of an output field quadrature.  The purpose of this work is to extend previous results \cite{silberfarb_quantum_2005,riofrio_quantum_2011,smith_efficient_2006} into a regime where the quantum backaction significantly effects the measurement statistics.

We consider here the case of an ensemble of $n$ qubits coupled to a single traveling wave quantum light field.  The qubit ensemble is assumed to be in a pure spin coherent state characterized by the unknown polar angles $(\theta,\, \phi)$.  The quantum state estimation problem is mapped to a parameter estimation problem, which is then approximated by a Monte Carlo sampling algorithm.  Numerical experiments show that the ultimate performance of the estimate approaches an optimum fidelity bound, found by \citet{massar_optimal_1995}.  The deficit in the reconstruction fidelity is attributed to a separability approximation in the Monte Carlo algorithm.  This algorithm is compared to, and significantly out performs, an equivalent ``Schr\"{o}dinger'' estimate that ignores the backaction of the measurement.  At long times the Schr\"{o}dinger estimate is shown to be biased away from the true state, indicating the significance of the conditional dynamics and the utility of the quantum filtering framework.

\section{Previous reconstruction results}

A fundamental task in quantum information processing is the ability to both reliably prepare an arbitrary quantum state and to experimentally verify its production.  Traditional quantum state estimation relies on an exhaustive tomographic procedure where the target state is repeatedly prepared and then destructively measured in an informationally complete number of measurement settings.  Such a procedure is often extremely time intensive, requiring both a tremendous amount of data as well as significant post processing time \citep{haffner_scalable_2005,leibfried_creation_2005}.  

In an alternative protocol proposed by \citeauthor{silberfarb_quantum_2005}, these inefficiencies can be largely side-stepped though a weak continuous measurement of an identically prepared ensemble in conjunction with a well chosen dynamical control \citep{silberfarb_quantum_2005}.  In particular, an atomic ensemble is prepared in an identical tensor product state $\rho_{\text{tot}} = \rho_0^{\otimes n}$ and experiences a known Hamiltonian while simultaneously coupled to a traveling wave probe, via a collective degree of freedom.  A continuous measurement of this probe then generates a measurement record that is strongly correlated with the evolution of the system.  If the dynamic drives the system in such a way as to make the measurement informationally complete, then a statistical estimate of an unknown initial system state should have a high fidelity with the true initial condition.

Such a system naturally arises in the field of laser cooled atoms, were an ensemble of $n$ atoms are easily assembled and then weakly coupled to an off-resonant probe laser.  One can then measure a collective spin state of the ensemble via the amount of polarization rotation induced by the Faraday effect. 
This protocol has been implemented in several experiments, ultimately reconstructing the full 16-dimensional hyperfine ground state manifold \citep{smith_efficient_2006, riofrio_quantum_2011}.  However, they were performed in a parameter regime where the intrinsically quantum nature of the continuous measurement could be ignored.  The amount of state disturbance caused by the nonlinear measurement process, the so called \emph{backaction}, was negligibly small when compared to the decoherence induced by diffuse light scattering as well as inhomogeneous effects in the control fields.

This work investigates, through theoretical analysis and numerical simulation, the fundamental limits of this protocol.  We do so in an idealized model, where the effects of decoherence are absent and thus the backaction becomes a significant effect.  To avoid unnecessary complications, we will also reduce the dimensionality of our fundamental system and consider only pure qubits, initialized in an identical tensor product state $\ket{\psi_{\text{tot}}} = \ket{\psi_0}^{\otimes n}$.  With a fully quantum model of the atom-light interaction, we formulate a maximum likelihood (ML) estimate of the single particle initial state, which we will denote as $\ket{\hat{\psi}_0}$.

\section{The Estimation Procedure}
When ignoring backaction, the linearity of an unconditioned master equation means that the measurement signal can be considered as a linear function of the initial state of the atomic system, $\rho_0$.  In an additive white noise model, the instantaneous polarimetry signal $y(t)$ can be modeled as,
\begin{equation}\label{chQubit:eq:y_t:semiclassical}
        y(t) = g\, \Tr \left( \mathcal{V}(t, O_0 )\, \rho_0 \right) + \text{``white noise''}
\end{equation}
where $g$ is a measurement gain relating to the signal-to-noise ratio, $\mathcal{V}(t, \cdot )$ is the Heisenberg picture equivalent to a dissipative master equation and $O_0$ is the initial system coupling observable \citep{riofrio_quantum_2011}.  The problem of state reconstruction in this model then becomes a constrained linear estimation problem.

In a generalized measurement model, the set of possible outcomes is described by a positive operator-valued measure (POVM), with elements $\set{E_\alpha}$ indexed by a discrete outcome $\alpha$.  In a given model of this measurement there exist a (possibly not unique) decomposition of a POVM into a set of Kraus operators $\set{A_\alpha}$, which satisfy the relation $E_\alpha = A^\dag_\alpha\, A_\alpha$ for every outcome $\alpha$.  Then upon obtaining the outcome $\alpha$, a pure state $\ket{\psi}$ updates via the transformation
\begin{equation}\label{chQubit:eq:krausUpdate}
    \ket{\psi} \rightarrow \frac{1}{\sqrt{\left\langle \psi \middle|E_\alpha \middle| \psi \right\rangle} }\, A_\alpha\, \ket{\psi}.
\end{equation}
Due to the renormalization factor, this update map is inherently nonlinear in the state vector.  Any generalized measurement scheme can be decomposed into a continuous measurement process \citep{varbanov_decomposing_2007}.  Conversely, a continuous measurement process can be modeled as a limiting sequence of weak generalized measurements.  Then in general, the nonlinearity of a repeated application of a time-dependent update map means that a measurement sequence is no longer a linear functional of the initial state $\rho_0$.

Much is known about the fundamental quantum limits of reconstructing pure qubit states from a finite number of measurements. \citeauthor{massar_optimal_1995} showed that given $n$ copies of a pure qubit state, it is possible to find a generalized measurement that optimizes the average fidelity $\langle \mathcal{F}\rangle$ between of the state estimate and the true state, averaged over all possible input states \citep{massar_optimal_1995}.  The fidelity for the pure states $\psi_1$ and $\psi_2$ is
\begin{equation}
  \mathcal{F} \define \abs{\braket{\psi_1}{\psi_2}}^2.
\end{equation}
(For mixed states this corresponds to the Uhlmann fidelity, but here we will only be concerned with pure states.)  With this definition the optimum average fidelity bound is simply given
\begin{equation}\label{chQubit:eq:optimumBound}
    \langle\mathcal{F}\rangle_\text{opt}  = \frac{n+1}{n+2}.
\end{equation}
They also showed that such a generalized measurement is necessarily a joint measurement involving all $n$ qubits, and no single measurement applied in series to each qubit can achieve this bound.  Later, \citeauthor{bagan_comprehensive_2005} found that a generalized measurement scheme that achieves this bound is a measurement that is uniform over all possible spin coherent states (SCS) composed from $n$ qubits \cite{bagan_comprehensive_2005}.  While \citeauthor{varbanov_decomposing_2007} gave a constructive proof for a continuous time stochastic process that reproduces a given generalize measurement, it is often quite difficult to obtain a closed form expression for what POVM the entirety of a given continuous measurement implements.  Instead of pursuing this track however, we instead turn to a Monte Carlo sampling framework.

At its most basic level, the initial state estimation problem is a parameter estimation problem, in that we observe a time varying signal whose statistics parametrically depend upon the initial state of the atomic system.  The simplest of all initial state estimation problems is binary state discrimination.  In this problem, the initial condition is know to be one of two possibilities, $\psi^a$ or $\psi^b$.  Then based upon a sequence of measurements, $\set{y_t}$, we wish to identify which state was most likely to generate these data.

In our more general problem, we have a data set $\set{y_t}$ and a detailed model of the dynamical system that generated the data, but only with the knowledge that the initial state is a SCS.  To deal with the continuous nature of this parameter estimation problem we resort to Monte Carlo sampling.  We randomly generate a collection of $m$ sample SCS, $\set{\psi^{j}\ : \ j = 1,\dots m}$, picked from some prior distribution. In Sec. \ref{chQubit:sec:likelihood} we describe how we choose the prior distribution though a two step resampling procedure, seeded from a uniform distribution over spherical angles.   Because the space of qubit SCS is isomorphic to the surface of the sphere, with just a few hundred samples we can easily cover that space so that any discretization error is well below the infidelity implied by the optimum bound $\langle \mathcal{F}\rangle_\text{opt}$.

Irrespective of how the candidate states are chosen, we have reduced the continuous parameter estimation problem to a much simpler state discrimination problem.  We will choose the state $\ket{\psi^{m'}}\in \set{\ket{\psi^m}}$ that maximizes the likelihood function $P(\set{y_t}|\psi^{m})$.  In other words, the ML state $\ket{\hat{\psi}}$ defined as
\begin{equation}\label{chQubit:eq:MLEstimate}
    \ket{\hat{\psi}} = \ket{\psi} \in \set{\ket{\psi^m}}\, : \quad  p(\set{y_t}|\psi)  = \argmax_m \set{\, p(\set{y_t}|\psi^m)\, }.
\end{equation}   In order to evaluate the likelihood function, we are still left with the problem of solving the recursive POVM expression or finding an equivalent method for calculating it.

Here we choose to formulate an equivalent expression.  Because we are working with a finite set of hypothesis states, we find that it is more efficient to propagate $m$ (approximate) conditional states from their initial values and calculate the likelihood for seeing the \emph{next} increment, given the current estimates.  This method is discussed in detail in Sec. \ref{chQubit:sec:likelihood}.

\section{The Model}

Sec. \ref{chQuLight:sec:faraday} reviews how the Faraday interaction can be modeled as a collective angular momentum $\V{J}$ coupled to a single $P$ quadrature in vacuum.  We align our coordinates so that we couple to the $J_z$ projection of angular momentum, with the collective angular momentum operators
\begin{equation}\label{chQubit:eq:CollectiveAngularMomentum}
    J_i \define \frac{1}{2} \sum_{j=1}^n \sigma^{(j)}_i,
\end{equation}
where $\sigma^{(j)}_i$ is the $i^{th}$ Pauli operator for the $j^{th}$ qubit and we have set $\hbar = 1$.  The coupling rate $\kappa$ between $J_z$ and $P$ is proportional to the local power in the drive laser field, which in general could be a time varying quantity.  For simplicity, we will assume that the laser is operated in a switched mode, where at time $t=0$ it achieves a constant value and that the measurement record ends before it is turned off.

In order to make the measurement record informationally complete, (or in the language of filter stability, make the system observable), we need to add an external control Hamiltonian $H_t$, acting solely on the collective spin system.  The exact form for $H_t$ to make it observable will be discussed in Sec. \ref{chQubit:sec:ControlLaw}.  Under these parameters the system field interaction is given by the unitary propagator $U_t$, which is the solution to the QSDE
\begin{equation}\label{chQubit:eq:dUt}
    d U_t = \left( \sqrt{\kappa}\, J_z\, dA^{\dag}_t - \sqrt{\kappa}\, J_z\, dA_t  - \half \kappa\, J_z^2\, dt -i H_t\, dt \right) U_t, \quad U_0 = \ident.
\end{equation}
From this stochastic propagator we are able to apply the results of Sec. \ref{chMath:sec:quantumFiltering} and work with a conditional master equation (CME).  For reference, upon the receipt of the measurement realization $\set{y_t}_{t \ge 0}$, the CME for this model is given by the SDE
\begin{equation}\label{chQubit:eq:CME}
    d \rho_t = -i [H_t, \rho_t]\, dt + \kappa\, \mathcal{D}[J_z](\rho_t)\, dt + \sqrt{\kappa}\, \mathcal{H}[J_z](\rho_t)\, dv_t
\end{equation}
with the initial condition $\rho_0 = \rho(0)$, where we have the following definitions.
$\mathcal{D}[J_z](\rho_t)$ is the Lindblad map commonly found in open quantum systems and is defined as
\begin{equation}\label{chQubit:eq:D}
    \mathcal{D}[J_z](\rho_t) \define J_z\, \rho_t\, J_z - \half J_z^2\, \rho_t - \half \rho_t\,J_z^2.
\end{equation}
$\mathcal{H}[J_z](\rho_t)$ is the state update map defined as
\begin{equation}\label{chQubit:eq:H}
    \mathcal{H}[J_z](\rho_t) \define J_z\, \rho_t + \rho_t\, J_z - 2 \Tr( J_z\, \rho_t)\, \rho_t.
\end{equation}
This map shows how the state updates, weighted by the strength of the innovation process,
\begin{equation}\label{chQubit:eq:dv}
    dv_t = dy_t - 2 \sqrt{\kappa} \Tr( J_z\, \rho_t)\, dt.
\end{equation}

\subsection{Observability and randomized controls.\label{chQubit:sec:ControlLaw} }
In reconstructing the full Cs ground state manifold, \citeauthor{riofrio_quantum_2011} used a randomized control policy to generate an informationally complete measurement record \citep{riofrio_quantum_2011}.  \citeauthor{merkel_quantum_2008} showed that by combining traverse RF magnetic fields and microwave radiation, with fixed magnitudes and time varying phases, the 16-dimensional ground state manifold is controllable \citep{merkel_quantum_2008}.  In other words, through these fundamental operations it is possible to generate any ground state operation and thereby map any state to any other state.

The connection between controllability and observability is a natural one.  Imagine that at time $t=0$ the probe couples to the operator $J_z$.  In order for the measurement statistics of this probe to depend upon the $J_y$ Bloch component, an external control must at some point rotate the system so that field now couples to the part of Hilbert space spanned by the projectors of $J_y$.  If the controls are unable to effect some hidden subspace, then the only other way to know about that part of Hilbert space is to apply an additional probe.  Not every observable system needs to be controllable, however.  One can certainly observe a system completely without being able to affect it in an arbitrary way.

The strictest definition for a system to be observable is that if there are two quantum states $\rho^A$ and $\rho^B$ where $\rho^A \ne \rho^B$ then there cannot exist a projector $\proj$ in the von Neumann algebra generated by the observation process $\set{Y_t}_{t \ge 0}$ such that $\Tr(\rho^A\, \proj) = \Tr(\rho^B\, \proj)$ \cite{van_handel_stability_2009}.  (See Sec. \ref{chMath:sec:QuProb} for a discussion of von Neumann algebras and quantum stochastic processes.)  This definition guarantees that after many trials, one will always be able to distinguish $\rho^A$ from $\rho^B$ by looking at the statistics of $Y$.

However, even if a given system is observable, this does not guarantee that it is well observed in a given measurement realization.  In order for the statistics of a single realization to give a high fidelity estimate, the space of possible initial states, \emph{e.g.} the space of all spin coherent states, should be well represented throughout the measurement record.  If the goal was to measure $J_z$ to a high degree of accuracy, the optimum control policy would be to apply no control at all.  However our objective is to measure every spin coherent state with equal weight, there by hopefully achieving the optimum POVM fidelity bound.

\citeauthor{riofrio_quantum_2011} found that high fidelity reconstructions were possible by choosing random, piecewise constant phase angles, thereby randomly cycling though a controllable set of operations.  Here we choose to implement a control policy that is randomized between a set of generators that rapidly spans the space of spin coherent states.  This policy then guarantees that these states will be well represented in the measurement statistics.  To achieve this, the control Hamiltonian $H_t$ is chosen to have the form
\begin{equation}\label{chQubit:eq:control}
    H_t = \V{b}(t) \cdot \V{J} = b^x(t) J_x + b^y(t) J_y + b^z(t) J_z,
\end{equation}
where the control field components $b^i(t)$ are drawn from a random distribution but are predetermined before the start of the measurement, \emph{i.e.} are without measurement feedback.

For simplicity, we further emulate the control policy of the Cs experiments and fix the magnitude of the control field while varying its direction in a randomized but piecewise constant way.  Furthermore we will constrain the magnitude so that for each direction, the Bloch vector will rotate by $\pi/2$.  Switching the field direction with a period of $\tau$ then requires $\norm{\V{b}(t)} = \pi/(2\tau)$.    With this constraint, the control law is fully defined.

To generate a control waveform with $m$ randomized $\pi/2$ gates with a period $\tau$, we first generate a set of $m$ of unit vectors $\set{\V{e}_i}$ so that each vector $\V{e}_i$ is drawn from a uniform distribution across the unit sphere.  The control field is then
\begin{equation} \label{chQubit:eq:controlField}
    \V{b}(t) =  \frac{\pi}{2\, \tau} \sum_{i =1}^m \indicate{[{i-1}, i)}\hspace{-3pt}\left(\,t/\tau\right)\, \V{e}_i .
\end{equation}

\section{The Likelihood Function \label{chQubit:sec:likelihood} }

In a discrete setting where the space of all possible outcomes, (the entire measurement record $\set{y_t}_{t \ge 0}$) can only have a finite number of outcomes, the likelihood function is simply the probability of receiving the observed values, given a parameter value.  The maximum likelihood estimate is then the parameter value that maximizes the probability for obtaining the observed data.  When the measurement takes on a continuous number of outcomes the probability for receiving a specific outcome is in fact zero.  However, we can still formulate a likelihood function by instead considering the probability density for the observed value.  

Things become a bit more complicated when considering stochastic processes in continuous time.  In Chap. \ref{chap:Math}, we found that the probability measure for a Wiener process was defined by Wiener's discrete path integral.  This means that for a sequence of $n$ times $\set{0 = t_0 <  \dots < t_i \dots < t_n = t_f}$ we can ask for the probability that the Wiener process evaluated at time $t_i$ will be within the interval $(a_i, b_i)$.  The resulting probability is given by the integral
\begin{equation}\label{chQubit:eq:nTimeWiener}
  P( \set{w_{t_i} \in I_i} ) = \int_{a_1}^{b_1} dw_1 \int_{a_2}^{b_2} dw_2 \dots \prod_{i} \left( \frac{1}{\sqrt{2 \pi \Delta t_i}} \exp \left( -\frac{(w_i - w_{i -1})^2}{2 \Delta t_i} \right) \right).
\end{equation}
If one attempts to take a continuous limit of this expression you find something rather peculiar \cite{hunter_lecture_2009}.  By focusing on just the product of exponentials, one finds
\begin{equation}
\begin{split}
    \lim_{n \rightarrow \infty} \prod_{i = 1}^n \exp\left(-\frac{(w_i - w_{i -1})^2}{2 \Delta t_i} \right)
    &= \lim_{n \rightarrow \infty} \exp\left( -\half \sum_{i=1}^n \Delta t_i \left( \frac{w_i - w_{i -1}}{\Delta t_i} \right)^2  \right) \\
      &= \exp\left(- \half \int_0^{t_f} ds \left(\tfrac{d w_s}{ds}\right)^2 \right).
\end{split}
\end{equation}
Were the Wiener process in anyway differentiable, this expression might be exceedingly useful.  However with our limited knowledge of stochastic analysis, it merely indicates the subtleties in working with densities of continuous time, nondifferentiable processes.  Attempting to make sense of these kinds of objects lead to the formulation of a stochastic calculus of variations, which has proved exceedingly useful for extending an It\={o} integral for anticipative integrands \cite{nualart_malliavin_1995} as well as a theory of white noise stochastic partial differential equations \cite{holden_stochastic_2010}.   We will not follow this path here.

There is an additional consideration as we know that $\set{y_t}_{t \ge 0}$ is decidedly not a Winer process.  Even making a discrete approximation, we still need to find an expression for the discrete density and how it depends upon the initial system state.  In this problem, we can make some progress.  Sec. \ref{chMath:sec:Innovation} showed that the innovation process, $v_t = y_t - 2 \sqrt{\kappa} \int_0^t ds \Tr ( J_z \rho_s)$ is an instance of a Wiener process.  More specifically, $v_t$ is a Wiener process when the filtered state $\rho_s$ accurately represents the conditional state of the system.  In our Monte Carlo setting we do not have just a single conditional state $\rho_t$, we in fact have a set of $m$ conditional states $\set{\rho_t^m}$, as the proper initial condition is unknown.  It is possible that a candidate state $\rho_t^m$ will differ in some aspects from the conditional state we would calculate, had we know then true initial condition.  For each hypothetical state we will have a set of possible innovations $\set{v_t^m}$, each a function of the measurement record $y_t$ and the filtered state $\rho_t^m$.  It should be clear that not every $v_t^m$ will be an instance of a Wiener process.  In fact the maximum likelihood estimate that we will construct \emph{hinges} upon the fact that not every $v^m_t$ will be a Wiener process.  This is because rather than computing the entire unknown and highly complicated statistics of $y_t$, we will compute the statistics of the known and simple statistics of the Wiener process $v_t$.  We then seek the candidate initial condition that makes the statistics of $v^m_t$ most resemble a Wiener process.

Stepping back from the mathematics for a moment, converting Eq. (\ref{chQubit:eq:nTimeWiener}) into an expression for $p(\set{y_t}|\rho^{m})$ via the innovation is deceptively simple.  We can write $v_t$ as
\begin{equation}
  v_t = v_s + y_t - y_s - 2 \sqrt{\kappa} \int_s^t ds\, \Tr(J_z\, \rho_s),
\end{equation}
or in other words,
\begin{equation}
\Delta v_{i} \define y_{t_{i}} - y_{t_{i-1}} - 2 \sqrt{\kappa} \int_{t_{i-1}}^{t_{i}} ds\, \Tr(J_z\, \rho_s).
\end{equation}
From the nonanticipative construction of the It\={o} integral, we have, for the smallest of possible time differences,
\begin{equation}
  \Delta v_{i-1} \approx \Delta y_{i} - 2 \sqrt{\kappa} \Delta t_i \Tr(J_z\, \rho_{t_{i-1}}).
\end{equation}
The density for $\set{y_t}_{t \ge 0}$ is then made by simply substituting $\Delta v_i$ into Eq. (\ref{chQubit:eq:nTimeWiener}).  In other words, the likelihood for the increment variables $\set{\y_i\define \Delta y_i}$ is then given by
\begin{multline} \label{chQubit:eq:approxLikelihood}
    p_n(\y_1, \y_2,\dots, \y_n| \rho^{m} ) \approx \\
    \left(\prod_{i = 1}^n (2 \pi \Delta t_i)^{-\half} \right)\, \exp\left(- \sum_{i=1}^n\, \frac{( \y_i - 2 \sqrt{\kappa} \Delta t_i \Tr(J_z\, \rho^{m}_{t_{i-1}}) )^2}{2\, \Delta t_i} \right).
\end{multline}
The only possible way to maximizing Eq. (\ref{chQubit:eq:approxLikelihood}) with respect to the initial condition $\rho^{m}$ is by minimizing the argument of the exponential.  If we set all of the time increments to be equal, $\Delta t_i = \Delta t$, we can even factor out the denominator and so the maximum likelihood estimate then becomes a problem of minimizing the sum,
\begin{equation}\label{chQubit:eq:QV}
     \mathrm{QV}( v^{m}_t ) \define \sum_{i= 1}^n (\Delta v^{m}_i)^2 = \sum_{i=1}^n\, \left( \Delta y_i - 2 \sqrt{\kappa} \Delta t \Tr(J_z\, \rho^{m}_{t_{i-1}}) \right)^2.
\end{equation}
This kind of object is called the \emph{quadratic variation}\footnote{Technically the quadratic variation is given in the infinitesimal limit \cite{oksendal_stochastic_2002}.} and Appendix \ref{app:SDEs} showed that it is ultimately what gives rise to the rules of It\={o} calculus.  So while it is a relatively delicate mathematical object, it is well defined in the infinitesimal limit. Furthermore, in proving the It\={o} rules, one shows that $\mathrm{QV}(w_t) = t$ with probability one, so that we expect, and observe numerically, that $\mathrm{QV}( v^{m}_t ) \sim t$ for most candidate initial conditions.  It is often the case that in Guassian problems such as ours, a maximum likelihood estimate over a Gaussian probability density simple becomes the least squared estimate.  So in our Monte Carlo search we have
\begin{equation}
    \argmax_m p( y_t|\rho^{m}) = \argmin_{m} \mathrm{QV}( v^{m}_t ).
\end{equation}

\subsection{The reconstruction procedure}
The Monte Carlo sampling estimator we have outlined follows this rough procedure:
\begin{enumerate}
  \item Sample $m$ pure Bloch vectors uniformly from the unit sphere.
  \item For each sample state compute the forward time evolution conditional on the measurement record $y_t$.
  \item Compute the quadratic variations of the innovation processes for each conditional state.
  \item Select as the estimate, the sample state that minimizes the quadratic variation at the final time.
\end{enumerate}
In practice we need to modify this procedure in two respects.  The first is that due to involving the stability of Markov Filters \cite{van_handel_stability_2009}, the above procedure suffers from poor numerical stability when the hypothesis initial condition has very little overlap with the true initial condition.  To rectify this problem, we implement a two step procedure, by first sampling $m$ \emph{mixed} initial conditions and then resampling, within some solid angle, pure states about the direction of the most probable mixed state.  This issue will be discussed in detail in Sec. \ref{chQubit:sec:Stability}.

The second modification stems from the fact that propagating the full conditional Schr\"odinger equation for a sufficient number of samples requires a large amount of computer time.  To fully propagate a spin $J$  \emph{pure} state requires $2 J + 1$ complex numbers.  The stochastic integrator we choose to use implements a weak second-order predictor-corrector method ( \citet[page 200]{kloeden_numerical_1994} ) and empirically requires a time step $\Delta t \sim 10^{-6}\, \kappa^{-1}$ to produce reliable expectation values.  When considering ensembles of mixed qubits it is not sufficient to consider the maximum projection of the collective angular momentum, but instead requires considering all possible total angular momentum values one could construct with $n$ spin-$\half$ particles.  This requires a total density matrix of order $n^2 \times n^2$ in size \cite{chase_collective_2008}.

In Chap. \ref{chap:projection} we developed a projection filter that projected the conditional master equation for the collection of $n$ qubits onto the manifold of identical separable states, which greatly reduces the computation demand.  We also showed that in the presence of strong randomized control, the projection filter tends to track the exact expectation values with a RMS error of less than $5\%$ of the total spin length. For mixed initial conditions, rather than propagating matrices of dimension $\sim n^2 \times n^2$ \emph{for each sample state}, the projection filter allows us to reduce this to tracking a single mixed Bloch vector, \emph{i.e.} three real numbers.  With these modifications, the Monte Carlo separable least squares estimate is computed though the pseudocode algorithm \ref{chQubit:alg:MonteCarlo}.
\begin{algorithm}
\caption{A Monte Carlo Separable Least Squares Estimate  \label{chQubit:alg:MonteCarlo}}
\begin{algorithmic}
    \State $\{\, \Vr_m\,\} \longleftarrow m$ uniformly random Bloch vectors with $r = r_\text{mixed} < 1$
    \ForAll{ $\Vr_m \in \{\, \Vr_m\,\}$ }
        \State $\Vr^{m}_t \longleftarrow$ Integrate Eq. (\ref{chProj:eq:projectionFilterIto}) with record $y_t$ and initial value $\Vr^{m}_0 = \Vr_m $.
        \State $\operatorname{QV}( v^{m}_t) \longleftarrow \sum_{i} \left( \Delta y_{i} - \sqrt{\kappa}\,n\,x^{3\,m}_t\,\Delta t_i\, \right)^2 $
    \EndFor
    \State $ \Vr_{\min{}} \longleftarrow  \Vr_{m'} \in \{\, \Vr_m\,\} :\ \operatorname{QV}( v^{(m')}_t) = \min\, \{\, \operatorname{QV} ( v^{m}_t)\, \} $
    \State $\{ \, \Vr'_m \} \longleftarrow m $ random Bloch vectors with $r' = 1$ and $ \Vr'_m \cdot \Vr_{\min{}}/ r_{\text{mixed}} \le \cos( \Theta_{\max{}})$
    \ForAll{ $\Vr'_m \in \{\, \Vr'_m\,\}$ }
        \State $\hat{\Vr}^{m}_t \longleftarrow$ Integrate Eq. (\ref{chProj:eq:projectionFilterItoSingleQubit}) with record $y_t$ and initial value $\hat{\Vr}^{m}_0 = \Vr_m' $.
        \State $\operatorname{QV}( v^{m}_t) \longleftarrow \sum_{i} \left( \Delta y_{i} - \sqrt{\kappa}\,n\,x^{3\,m}_t\,\Delta t_i\, \right)^2 $
    \EndFor
    \State $ \Vr'_{\min{}} \longleftarrow  \Vr'_{m'} \in \{\, \Vr'_m\,\} :\ \operatorname{QV}( v^{(m')}_t) = \min\, \{\, \operatorname{QV} ( v^{m}_t)\, \} $\\
    \Return $\rho(\Vr'_{\min{} })$
\end{algorithmic}

\end{algorithm}

\subsection{Coupled CMEs\\ and filter stability \label{chQubit:sec:Stability}}
For a given hypothesis state $\rho^{m}$, there is a hidden ``true'' state $\rho^{\star}$ that generated the measurement record $y_t$.  When using this measurement record to propagate the hypothesis state, that ends up coupling $\rho^{m}$ to $\rho^{\star}$.  This leads to a coupled set of stochastic differential equations.  By definition, the conditional state $\rho^{\star}_t$ results in an innovations process that is a \emph{true} Wiener process.  Explicitly, $\rho^{\star}_t$ is evolves according to the stochastic differential equations,
\begin{subequations}
   \begin{align}
      d \rho^{\star}_t =&-i [H_t, \rho^{\star}_t]\, dt + \kappa \mathcal{D}[J_z](\rho^{\star}_t)\, dt +  \sqrt{\kappa}\mathcal{H}[J_z](\rho^{\star}_t)\, d w_t,  \\
      d y_t =&\, d w_t + 2 \sqrt{\kappa} \Tr( J_z \rho^{\star}_t) dt. \label{chQubit:eq:dy_t}
   \end{align}
\end{subequations}
where $w_t$ is the correct innovation process and is an unobserved Wiener process.  The candidate initial condition $\rho^{m}_0$, and the measurement record $y_t$ are the elements necessary for propagating $\rho^{m}_t$ in time via,
\begin{subequations}
    \begin{align}
      d \rho^{m}_t =&-i [H_t, \rho^{m}_t]\, dt + \kappa \mathcal{D}[J_z](\rho^{m}_t)\, dt +  \sqrt{\kappa} \mathcal{H}[J_z](\rho^{m}_t)\, d v^{m}_t, \label{chQubit:eq:drho'_t}  \\
      dv^{m}_t =&\, dy_t - \Tr( ( L + L^\dag) \rho^{m}_t) dt. \label{chQubit:eq:dv'_t}
    \end{align}
\end{subequations}
In terms of the unobserved Wiener process $w_t$, $\rho^{m}_t$ evolves as,
\begin{subequations}
  \begin{align}
      d \rho^{m}_t =&-i [H_t, \rho^{m}_t]\, dt + \kappa \mathcal{D}[J_z](\rho^{m}_t)\, dt +  \sqrt{\kappa} \mathcal{H}[J_z](\rho^{m}_t)\, d w_t \\
      & - 2 \kappa\, \mathcal{H}[J_z](\rho^{m}_t)\, \Tr \left( J_z (\rho^{m}_t-  \rho^\star_t) \right) dt. \nonumber
  \end{align}
\end{subequations}
Note that 
if at some time $t$ we happen to have the equality $\rho^\star_t = \rho^{m}_t$, then we also have $d \rho^{\star}_t = d \rho^{m}_t$.  But the general state $\rho_\tau$ at some time $\tau > t$ can be written as $\rho_{\tau} = \rho_t + \int_t^\tau d \rho_{t'}$.  This implies that because the differentials are equal whenever the states are equal we will have $\rho^\star_\tau = \rho^{m}_\tau$ for every $\tau > t$ if $\rho^\star_t = \rho^{m}_t$.

One possible result of this coupling is that it acts as an attractor, always decreasing the ``distance'' between the $J_z$ projections of $\rho^\star_t$ and $\rho^{m}_t$.  This correction effect is known as \emph{filter stability}.  If the filter is able to correct for certain modeling errors, it is stable.  The differences in the two initial states $\rho^\star_0$ and $\rho^{m}_0$ can be viewed as a modeling error and the convergence of $\rho^{m}_t \rightarrow \rho^*_t$ is a correction of this error.  This is a well studied effect in both the quantum and classical settings, see \cite{van_handel_stability_2009} and references there in.

In \citep{van_handel_stability_2009}, \citeauthor{van_handel_stability_2009} gave explicit criteria for when a quantum filter is stable for an incorrect initial conditions.  For our purposes these criteria boiled down to the following two issues.  The first is that the system must be \emph{observable}, in that the measurement record must be informationally complete.  If we did not have a transverse magnetic field, then the measurement statistics would only include information about the eigenstates of $J_z$ and so the system is not observable.  The second issue is that the probability density for $y_t$ calculated with the true state $\rho^\star$ must be \emph{absolutely continuous} with respect to the density calculated under the guessed state $\rho^{m}$.   This is a term borrowed from classical probability theory and embodies the concept that a probability measure $\mathbbm{P}^{B}$ is compatible with observations that are actually governed by $\mathbbm{P}^{A}$.  The quantum version is given by the following definition.  In order for $\rho^A$ to be absolutely continuous with respect to $\rho^B$, then for any projector $\proj$ in the von Neumann algebra generated by $\set{Y_t}_{t \ge 0}$, we must have $\Tr(\proj \rho^B) = 0$ implying that $\Tr(\proj \rho^A) = 0$.  This need not be a two sided relationship so that $\rho^B$ need not be absolutely continuous with respect to $\rho^A$. These requirements are not just important to the question of filter stability but also apply to the Monte Carlo sampling procedure.   As it has been discussed previously, the observability condition is vital in order to obtain a high fidelity estimate.  However absolute continuity is also quite important.

In a Kraus operator formulation of a continuous measurement, the state after the measurement outcome $i$ is updated as
\begin{equation}
  \rho \mapsto \rho|_i =\frac{ A_i \rho A_i^\dag}{\Tr(A_i^\dag A_i \rho)}.
\end{equation}
If the denominator $\Tr(A_i^\dag A_i \rho) = 0$, then the update cannot be made as it requires dividing by 0.  However $\Tr(A_i^\dag A_i \rho)$ is also the probability for obtaining the outcome $i$, as calculated according to the state $\rho$.   Therefore, if the event $\set{i}$  occurs with this probability, dividing by zero is not an issue as it will never happen.  Suppose we obtain the outcome $i$ that occurred with probability $\Tr(A_i^\dag A_i \rho^\star) = p_i^\star$.  Furthermore, suppose we tried to update a state $\rho^{m}$ that had the audacity to assert $p^{m}_i = \Tr(A_i^\dag A_i \rho^{m}) = 0$.  This results in a crisis of conscious, as there is no way to incorporate this incompatible information into our world view.  The condition that $\rho^\star$ must be absolutely continuous with respect to $\rho^{m}$ means that $p^{m}_i$ will never be zero without $p^\star_i$ also equal to zero.

In principle any valid initial spin state could generate a given diffusive measurement record.  This can be easily seen by noting that the ``true'' innovations process is given by
\begin{equation}\label{chQubit:eq:vtStar}
  v^\star_t = y_t - 2 \sqrt{\kappa} \int_0^t ds \Tr(J_z \rho^\star_s)
\end{equation}
and is a Brownian motion.  Because $J_z$ takes on eigenvalues in the range $- n/2 \le m_z \le n/2 $, any candidate innovation $v^{m}_t$ will be within the range,
\begin{equation}
   y_t - n \sqrt{\kappa}\, t \le v_t^{m} \le y_t + n \sqrt{\kappa}\, t.
\end{equation}
For finite, $n$, $\kappa$, and $t$, it is perfectly possible for a Brownian motion to obtain any of these values, it is just increasingly unlikely.  Therefore if the measurement record is observable, the conditional master equation is in principle stable.

In practice, the numerical stability of states conditioned on highly improbable measurements becomes a issue.   By not taking this into account preliminary results that did not consider the possibilities of unstable trajectories, showed nearly a $1\%$  drop in the average reconstruction fidelity from what we ultimately achieve.  Investigating the cause of this sub-optimal performance showed that the average fidelity was significantly biased by outlier trajectories that gave estimated states that were nearly orthogonal to the true state.   The cause of these outliers was the numerical stability of Monte Carlo sample points with very poor overlap with the true state.

By switching to the two step sampling procedure in algorithm \ref{chQubit:alg:MonteCarlo}, every initial mixed single qubit state can be viewed as a convex combination of pure states pointing along opposite directions.  That is, if we have the possibly mixed single qubit Bloch vector $\Vr$ with length $0 \le r \le 1$ we have
\begin{equation}
  \rho( \Vr ) = \tfrac{1 + r}{2}\rho(\Ve_{\Vr}) + \tfrac{1 - r}{2}\, \rho(\Ve_{-\Vr} ),
\end{equation}
where $\rho(\Ve_{\Vr})$ is a projector on the the pure SCS pointing in the $\Ve_{\Vr}$ direction.  This implies that by using initial mixed vectors each initial state has some support over the orthogonal spin coherent state $\psi(\Ve_{-\Vr})^{\otimes n}$.  In the numerical simulations presented in Sec. \ref{chQubit:sec:Simulations} an initial mixed state vector of radius $r_{\text{mixed}} = \tfrac{3}{4}$ provides enough of a signal to choose an appropriate direction for the pure state resample as well as enough orthogonal support for the trajectories remain stable.

\subsection{Backaction in continuous quantum measurement}
In order to identify what impact the backaction has on the reconstruction fidelity, we need to construct a similar but backaction-free estimator.  A figure of merit commonly used to consider the importance of backaction is the ratio of the ``projection noise'' to the ``shot-noise''.  The projection noise is a description of the fluctuations (\emph{i.e.} noise) in a given observable if a projective measurement is made.   As we are considering a continuous measurement of $J_z$ with respect to a SCS, the relevant projection noise is
\begin{equation}
  \expect{\Delta J_z^2}_{\psi^{\otimes n}} = \frac{n}{4} \expect{\Delta \sigma_z^2}_{\psi} = n\, p_{+1}(1-p_{+1})
\end{equation}
where $p_{+1} = \abs{\braket{+1}{\psi}}^2$ is the probability to observe the individual spin state to be in the $+1$ eigenstate of $\sigma_z$ \cite{itano_quantum_1993}.

The shot-noise describes the noise added by making a continuous measurement over a finite time.  To identify the order of magnitude of this additive noise, note that from Eq. (\ref{chQubit:eq:vtStar}) we have
\begin{equation}\label{chQubit:eq:ytVtStar}
  y_t = v_t^\star + 2 \sqrt{\kappa}\int_0^t ds\ \Tr(\rho_s^\star\, J_z)
\end{equation}
and that $v_t^\star$ is a realization of Brownian motion.  We would like to invert this formula to arrive at a random variable whose statistics allow for an estimate of $\Tr( \rho_0^\star\, J_z)$.  Suppose that we wished to model the system completely ignoring the theory of continuous quantum measurement and that for times $0 \le s \le t$, $\rho^\star_s$ evolves according to the Schr\"{o}dinger equation.  If we further assume that $H_t  = 0$, we then have that $\Tr( \rho_s^\star\, J_z) = \Tr( \rho_0^\star\, J_z)$ and so the classical random variable
\begin{equation}
    \mathrm{j}_z \define \frac{y_t}{2 \sqrt{\kappa}\, t} = \Tr(J_z\, \rho_0^\star) + \frac{1}{2 \sqrt{\kappa}\, t} \,v_t^\star
\end{equation}
is Gaussian distributed with mean $\Tr(J_z\, \rho_0^\star)$ and $\text{Var}(\,\mathrm{j}_z) = ( 4\,\kappa \, t)^{-1}$.  It is this variance that is referred to as the shot-noise added by the probe.   Looking at the ratio of these two fluctuations we have
\begin{equation}
  \zeta \define \frac{\expect{\Delta J_z^2}_{\psi^{\otimes n}}}{\text{Var}(\,\mathrm{j}_z)} = 4\, n\,\kappa\, t\, p_{+1}(1-p_{+1}).
\end{equation}
If the system is prepared in a SCS with $p_{+1} = \half$ then Sec. \ref{chProj:sec:SqueezingSimulations} showed that we then expect a maximum amount of spin squeezing or equivalently a large amount of bipartite entanglement.  In this case $\expect{\Delta J_z^2}_{\psi^{\otimes n}}$ takes on its maximum value of $n/4$ and so $\zeta = n\,\kappa\, t$.  When $\zeta \gg 1$ then one expects a significant contribution of quantum backaction in the system and therefore the measurement effects must be accounted for \cite{deutsch_quantum_2010}.   In the uncontrolled spin squeezing simulations of Sec. \ref{chProj:sec:SqueezingSimulations}, we found that for $n = 100$ and $\kappa\, t = 0.2$ we found $\xi_T^2 \sim 10 \text{ dB}$ and so that in the absence of strong Hamiltonian controls, $\zeta = 20$ indeed leads to a strongly nonclassical state.

However, the above discussion assumed no controls.  It is possible that with the randomized controls considering only the Hamiltonian evolution is sufficient to obtain a high fidelity estimate.  To make this comparison we formulate a backaction-free estimator, one that only includes the Hamiltonian in the model for the forward time dynamics.  Rather than considering a measurement record where $y_t$ is given by Eq. (\ref{chQubit:eq:ytVtStar}), we instead propose a model were
\begin{equation}\label{chQubit:eq:ytSchrodinger}
    y_t \approx w_t + 2 \sqrt{\kappa}\int_0^t ds\ \Tr(J_z\, \tilde{\rho}^\star(s) )
\end{equation}
and $\tilde{\rho}^\star(t)$ is the solution to the Schr\"{o}dinger equation
\begin{equation}\label{chQubit:eq:Schrodinger}
    \frac{d}{dt} \tilde{\rho}^\star(t) = -i[H_t, \tilde{\rho}^\star(t) ]
\end{equation}
and $w_t$ is a Wiener process.

To make a fair comparison, this backaction-free estimator will also be implemented though a Monte Carlo sampling procedure.  We use a algorithm similar to algorithm \ref{chQubit:alg:MonteCarlo}, but with two modifications.  The first is that the two step sampling procedure is unnecessary because there are no conditional dynamics to cause numerical stability. The second is that because the dynamics are linear, the Schr\"{o}dinger evolution in Eq. (\ref{chQubit:eq:Schrodinger}) is most efficiently computed in the Heisenberg picture.  In the Heisenberg picture, we simply need to integrate the time evolution of the $J_z$ observable \emph{once} and then compute its expectation value with each candidate state.  Furthermore in this decoherence free model the system state will always remain in a separable state and so we need only consider the Heisenberg evolution for the single qubit Pauli operator, $\sigma_z$.  In other words,
\begin{equation}
    2 \sqrt{\kappa} \Tr( J_z \tilde \rho^{m}(t) ) = \sqrt{\kappa}\, n \, \braOket{\psi^{m}}{\sigma_z(t)}{\psi^{m}}
\end{equation}
where $\sigma_z(t)$ is the solution to the Heisenberg equation of motion
\begin{equation}\label{chQubit:eq:HeisenbergSz}
    \frac{d}{dt} \sigma_z(t) = +i [ H_t,\, \sigma_z(t)], \quad\text{ with }  \sigma_z(0) = \sigma_z.
\end{equation}
The pseudocode for the backaction-free estimator is given in algorithm \ref{chQubit:alg:Schordinger}

\begin{algorithm}
\caption{A Monte Carlo Backaction-free Estimate  \label{chQubit:alg:Schordinger}}
\begin{algorithmic}
    \State $\set{\sigma_z(t_i)} \longleftarrow$ Integrate Eq. (\ref{chQubit:eq:HeisenbergSz}) and evaluate at times $\set{t_i}$.
    \State $\{\, \Vr_m\,\} \longleftarrow m$ uniformly random Bloch vectors with $r = 1$
    \ForAll{ $\Vr_m \in \{\, \Vr_m\,\}$ }
        \State $\operatorname{QV}( v^{m}_t) \longleftarrow \sum_{i} \left( \Delta y_{i} - \sqrt{\kappa}\, n\, \Tr(\sigma_{z}(t_{i-1})\, \rho(\Vr_m)\,) \, \Delta t_i\, \right)^2 $
    \EndFor
    \State $ \Vr_{\min{}} \longleftarrow  \Vr_{m'} \in \{\, \Vr_m\,\} :\ \operatorname{QV}( v^{(m')}_t) = \min\, \{\, \operatorname{QV} ( v^{m}_t)\, \} $\\
    \Return $\rho(\Vr_{\min{} })$
\end{algorithmic}

\end{algorithm}

\section{Numeric Simulations\label{chQubit:sec:Simulations}}
This section presents the results of numerical simulations, comparing algorithms \ref{chQubit:alg:MonteCarlo} and $\ref{chQubit:alg:Schordinger}$ to the optimum POVM bound in Eq. (\ref{chQubit:eq:optimumBound}).   The bound $\langle\mathcal{F}\rangle_\text{opt}  = (n+1)/(n+2)$ gives the average fidelity of a single POVM where the average is taken over measurement outcomes \emph{as well as} an average over possible input SCS.  Therefore, the results of these simulations are reported as an average of ensemble of $\nu$ trials.  All results in this section use $\nu = 1000$ trials.

\subsection{Simulation parameters}
For each trial, we choose a single qubit Bloch vector from a distribution that is uniform over the surface of the unit sphere.  We then use this vector to generate SCSs composed of $n$ qubits.  This simulations use the qubit numbers $n = 25, 40, 55, 70, 85, \text{ and } 100$.  Then for each initial state and each number of qubits, we  generate a single measurement realization $y_t$ and use this record to then estimate the initial Bloch vector.  In total 6000 measurement records were generated.

Every simulation uses the same control Hamiltonian, where the randomized piecewise constant control vector $\V{b}(t)$ was generated at the start of the simulation.  The directions of rotation are again distributed uniformly across the unit sphere and no attempt was made to select an optimum realization.  The parameters that fully constrains the simulation are the measurement strength $\kappa$, the final measurement time $t_f$ and the control gate period $\tau$. With no other scales in the problem we choose to essentially set $\kappa$ to one and discuss the remaining two parameters in units of $\kappa^{-1}$.

In Chap. \ref{chap:projection} we found that the separable approximation is valid in regime where the randomizing magnetic field strength $\kappa \ll b_0$.  By fixing the strength to generate a $\pi/2$ rotation in one gate period $\tau$ this means that $b_0 = \pi/(2 \tau)$, implying that $\kappa\, \tau \ll 1$.  We also found that  a gate period $\tau = 5 \times 10^{-3} \kappa^{-1}$ gave less than a $5\%$ RMS tracking error for the separable projection filter, (see Sec. \ref{chProj:sec:Simulations}), and places $b_0$ two orders of magnitude greater than $\kappa$.

For $n = 25-100$ qubits we find that the reconstruction fidelities have saturated by a time $t \sim t_f = 0.2\,\kappa^{-1}$, which we fix as final time for every simulation run.  With this final time and gate period, each simulation has 40 randomized $\pi/2$ rotations.

To efficiently implement these simulations we exploit two conservation properties of the system.  The first is that because the total angular momentum operator $J^2$ commutes with the stochastic unitary of Eq. (\ref{chQubit:eq:dUt}), the total angular momentum of the system is conserved.  This means that by initializing the system in a state of maximum projection of angular momentum (\emph{i.e.} in a pure SCS) we are initializing the system in the eigenspace with total angular momentum $J = n/2$.  Rather than considering the entire $d = 2^n$ dimensional Hilbert space we only need to simulate a spin $J = n/2$ particle and work in its $d = 2 J + 1 = n + 1$ dimensional Hilbert space.  The second property is that the conditional master equation we consider here maps pure states to other pure states, because it has no additional loss channel.  This means that we can in fact integrate a conditional Schr\"{o}dinger equation rather than a conditional master equation.  This makes a substantial savings in computational overhead as we be propagating a $d = n+1$ complex vector in time, rather than a $d\times d$ complex matrix.  These two properties that makes it computationally feasible to generate 1000 measurement records for a system containing 100 qubits.

The actual simulations are implemented in the MATLAB computing environment using a hand coded weak second-order predictor-corrector stochastic differential equation integrator.  The algorithm is described in \citet[page 200]{kloeden_numerical_1994} and was implemented in MATLAB by Brad Chase for his PhD dissertation \citep{chase_parameter_2009}.

\subsubsection{Monte Carlo Parameters}

The Monte Carlo separable estimator used 250 sample states for each part of the two-step estimation.  In the initial step, the 250 mixed states produce a sparse but uniform covering of all possible SCS directions.  A typical sampling has an average angular separation between adjacent points of $\sim 6^\circ$ and a maximum separation of $\sim 20^\circ$.  As mentioned above, the mixed state radius of the Bloch vector used in these simulations is $r_{\text{mixed}} = 0.75$.
In the second step, we sample 250 pure states that are constrained to be no more than $45^\circ$ from the most likely mixed state direction.  Example first and second step sampling distributions are shown in Fig. \ref{chQubit:fig:monteCarloResampling}.

\begin{figure}[ht]
    \begin{center}
        \includegraphics[width=1\hsize]{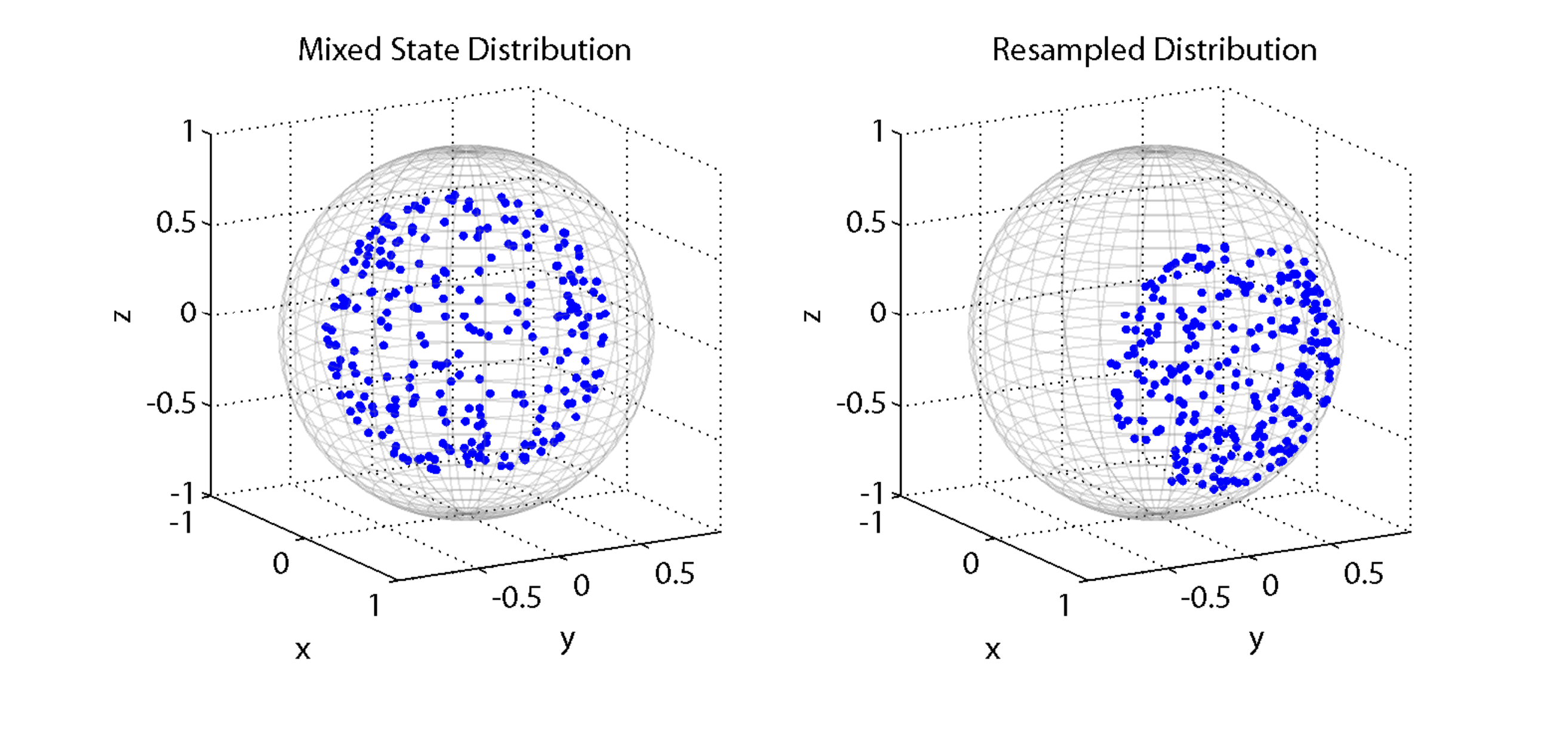}
        \caption{An Example Monte Carlo State Sampling Distribution.  (left) An example of the initial mixed state sampling for the Monte Carlo algorithm, with $m = 250$ and $r_{\text{mixed}} = 0.75$.  (right) An example of resampling $m = 250$ states about the $+\Ve_x$ axis with $\Theta_{\text{max}} = 45^\circ$. \label{chQubit:fig:monteCarloResampling} }
    \end{center}
\end{figure}

For the backaction-free comparison, we use an number of samples matching the density of points in the second resample step.  The resampled solid angle covers approximately $15\%$ of the Bloch sphere, meaning that $m = 1700$ cover the whole sphere with roughly the same density of states.  This number of samples lead to an average fidelity between nearest neighbors of $\langle \mathcal{F} \rangle_{\text{sample}} = .9994$, meaning that if the true ML estimate falls between two sample points, on average, the infidelity caused by the Monte Carlo sampling will be on the order of $10^{-4}$.  This is well below the optimum POVM bound for the simulated qubit number and so any loss in fidelity should not be attributable to sampling errors.

\subsection{Results and discussions \label{chQubit:sec:results} }

\begin{figure}[ht]
    \begin{center}
        \includegraphics[width=1\hsize]{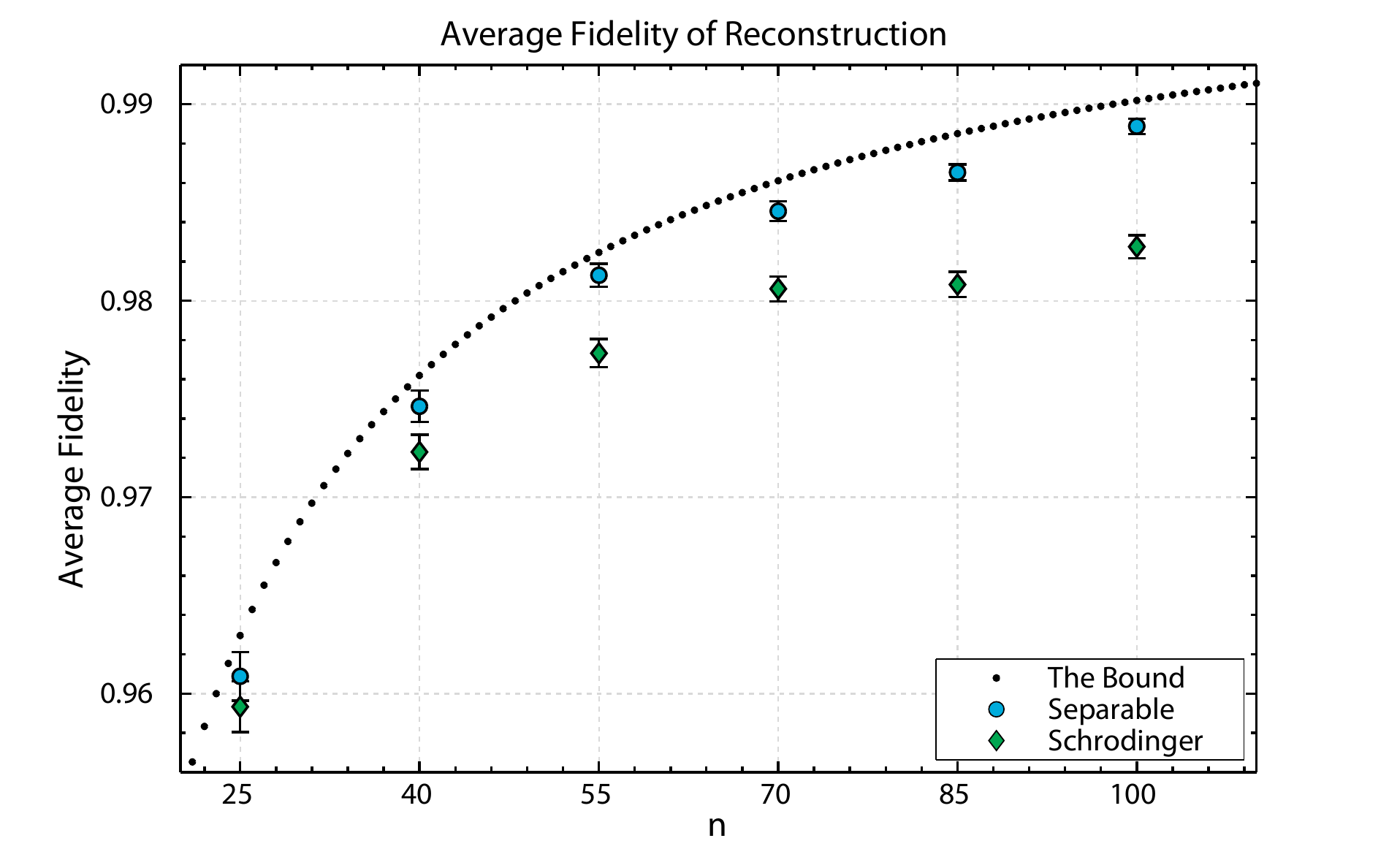}
        \caption{ \label{chQubit:fig:FidelPlot} A comparison of numerical reconstructions to the optimum bound (Color Online.)  Data points show the average fidelity of single shot reconstruction as a function of the number of qubits $n$, averaged over $\nu = 1000$ randomly chosen pure initial states states.  Blue circles show the separable estimator.  Green diamonds show the backaction-free Schr\"{o}dinger equation estimator.  The optimum POVM bound is show in as a dotted line.  Error bars show a standard error of $\pm \sqrt{\text{Var}[\mathcal{F}]/\nu}$. }
	\end{center}
\end{figure}
Fig. \ref{chQubit:fig:FidelPlot} shows the results of the numerical simulations.  The trial-averaged reconstruction fidelity is plotted as a function of the number of qubits in the system for both the separable estimate (\emph{i.e.} with backaction) and the Schr\"{o}dinger evolved, backaction-free estimator.  The fidelity is computed by taking the squared overlap between the single qubit state for that measurement record with the single qubit state estimate.  In other words, if the true qubit state is given by the Bloch vector $\Vr_0$ and the estimate reports the Bloch vector $\Vr_{m}$ then the fidelity of that reconstruction is given by $\mathcal{F} = \half ( 1 + \Vr_0\cdot \Vr_m\, )$.

In these numerical experiments, the separable Monte Carlo estimator shows a significant improvement over a simple backaction-free estimator that considers only the unitary evolution of the state due to the control fields.  The discrepancy increases as the number of qubits increase, keeping the duration of the measurement fixed.

Furthermore, the separable estimator almost achieves the optimum bound.  The deficit between the bound and the numerical averages never exceeds $0.21\%$ with an average of $0.16\%$, which is still above the expected error caused by the discrete Monte Carlo sampling.   A possible source for this deficit could be the separability assumption in the projection filtering method, which is known to have a non-negligible tracking error in the $J_z$ expectation value (see Sec. \ref{chProj:sec:Simulations}).

The performance of the backaction-free Schr\"{o}dinger estimator is best understood by considering not just the estimate for the initial state given the entire measurement record, but to instead consider the family of estimates created by only taking part of the measurement record.

\subsubsection{Estimator Bias}
The Monte Carlo estimators take as input a measurement record $y$ containing data for times $t \in [0, t_f]$ and returns an estimate for the initial state $\hat{\rho}_0$.  It is just as easy to consider a whole family estimates computed with only part of the total measurement record, \emph{i.e.}, instead of using the entirety of $y$ we use $y_s$ for $0 < s \le t_f$ in computing the estimate.  Ideally, having more data should only improve the estimate.  However, in order to use the data at times $t > s$ we are required to compute an estimate for the state of the system at time $s$.  If this estimate is in fact inaccurate, then any modeling errors might bias the conclusions drawn from future measurements.

Moreover \emph{both} of the estimators considered here have modeling errors.  The separable estimator uses the projection filtering equations, which explicitly remove any entangling dynamics.  The estimator based simply upon the unitary Schr\"{o}dinger dynamics makes a much greater sin.  This estimate completely ignores any effect the measurement has on the system of qubits.  Figures \ref{chQubit:fig:FilterBias} and \ref{chQubit:fig:SchBias} indicate what affect these modeling errors have on the average reconstruction fidelity.

\begin{figure}[ht]
    \begin{center}
    	\includegraphics[width=0.85\hsize]{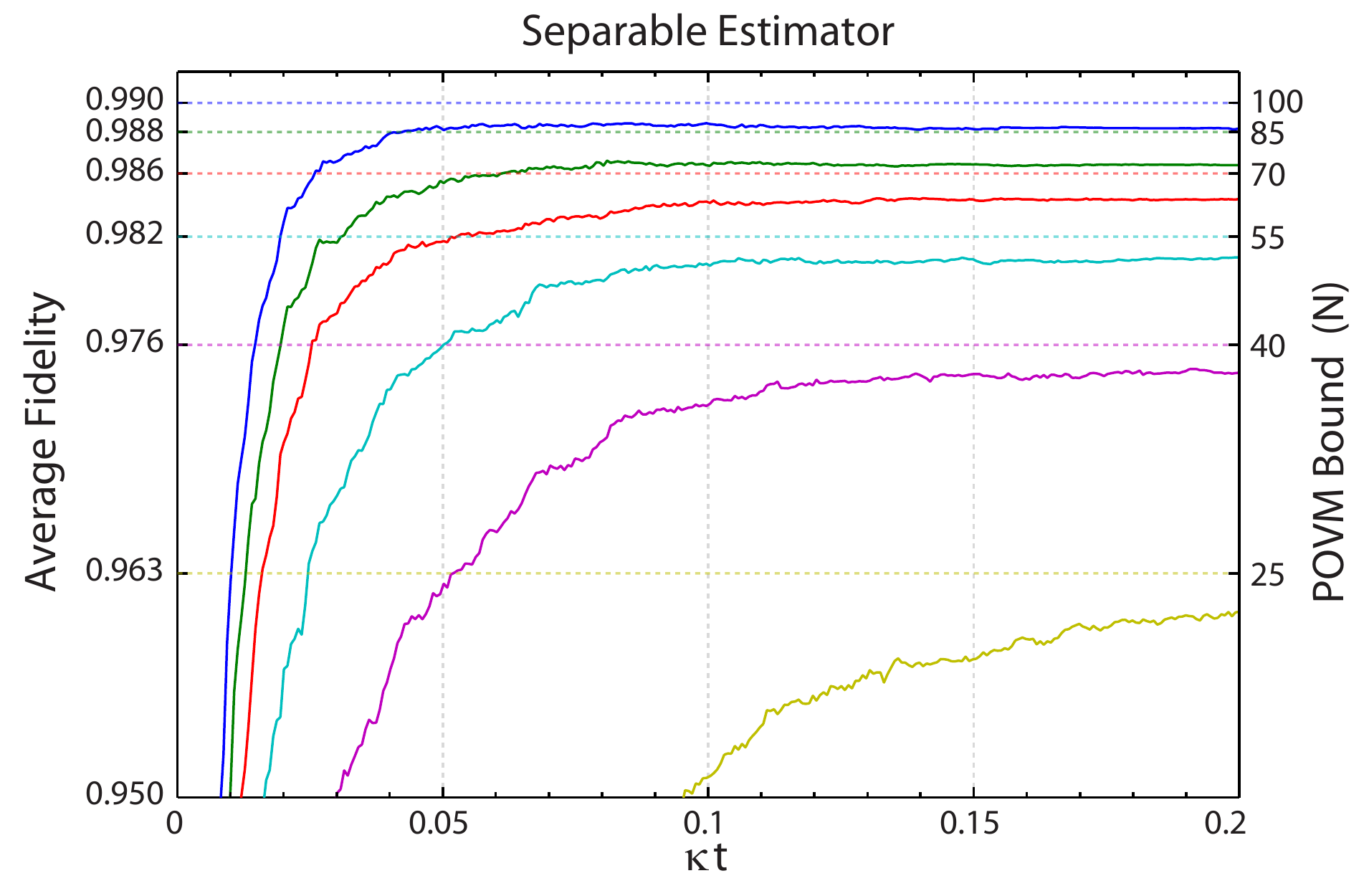}
         \caption{\label{chQubit:fig:FilterBias} Reconstruction fidelity vs measurement duration.  This plot shows the average reconstruction fidelities for the separable filter estimate, as a function of the length of the measurement record.  Shown are traces for the $6$ qubit numbers considered, which are (in order of  decreasing reconstruction fidelity) 100, 85, 70, 55, 40, and 25 qubits respectively.  The vertical axis is a linear scale, with grid lines indicating the optimum fidelity bound for these same number of qubits.  The averaging was over $\nu = 1000$ randomly chosen pure initial states.}
	\end{center}
\end{figure}

Fig. \ref{chQubit:fig:FilterBias} shows for the separable filter, the trial averaged reconstruction fidelities for all 6 qubit numbers plotted against the duration of the measurement record.  It is clear from this figure that having a larger signal composed of more qubits improves the final fidelity.  It also shows how, as the number of qubits increases, the fidelity improves at a faster rate.  Furthermore, the modeling error introduced by the separable approximation does not seem to significantly bias the estimate away from an optimum sample state.

\begin{figure}[ht]
    \begin{center}
    	\includegraphics[width=0.85\hsize]{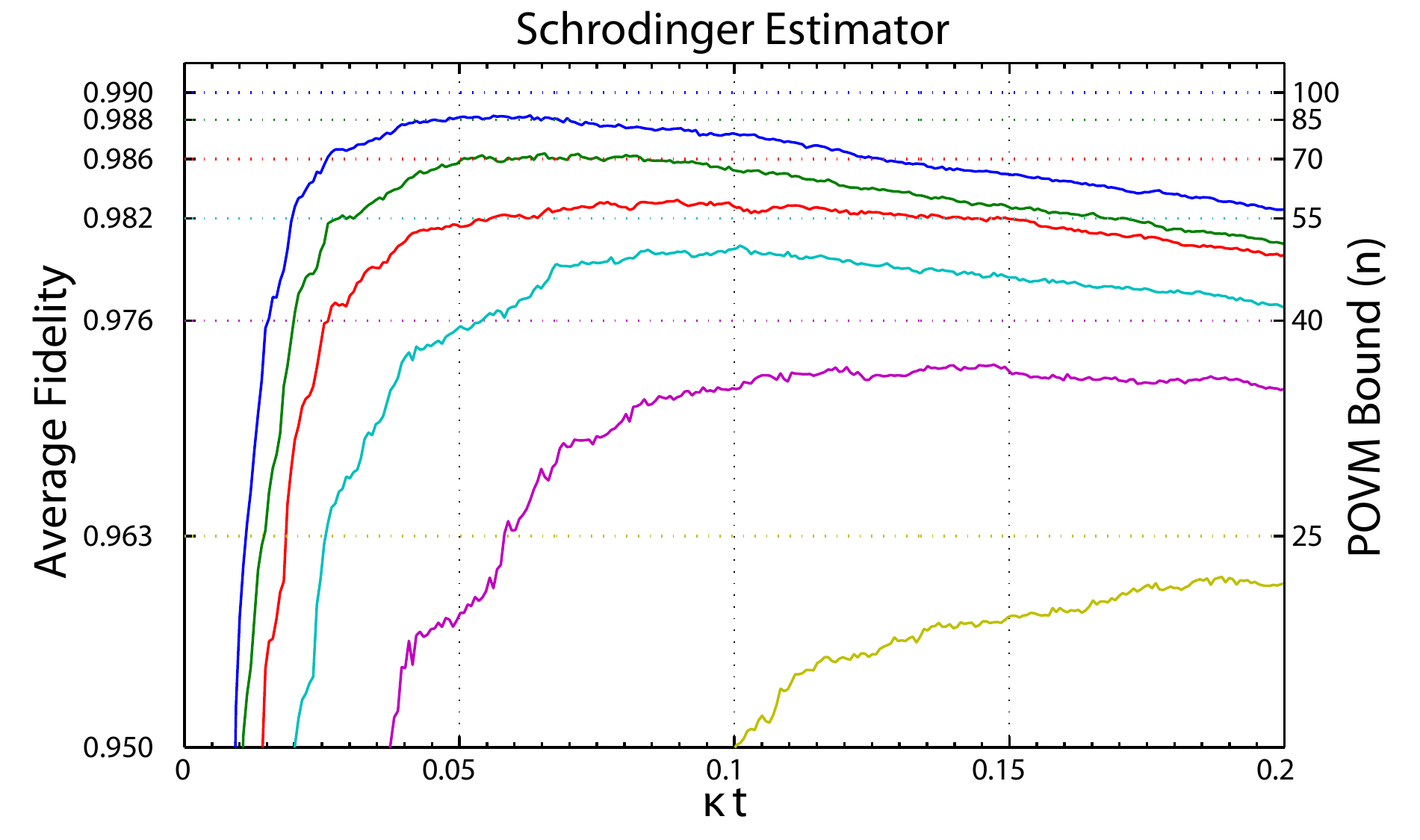}
         \caption{\label{chQubit:fig:SchBias} Reconstruction fidelity vs measurement duration.  This plot shows the average reconstruction fidelities for the backaction-free Schr\"{o}dinger estimate, as a function of the length of the measurement record.  Shown are traces for the $6$ qubit numbers considered, which are (in order of  decreasing reconstruction fidelity) 100, 85, 70, 55, 40, and 25 qubits respectively.  The vertical axis is a linear scale, with grid lines indicating the optimum fidelity bound for these same number of qubits.  The averaging was over $\nu = 1000$ randomly chosen pure initial states.}
	\end{center}
\end{figure}

Fig. \ref{chQubit:fig:SchBias} shows an identical plot for the backaction-free Schr\"{o}dinger estimate, showing that a larger number of qubits improves the reconstruction fidelity with a higher fidelity estimate at shorter measurement times.  However, it also shows that not including the backaction into the model significantly \emph{decreases} the reconstruction fidelity at longer measurement times.  This bias tends to be more pronounced as number of qubits increases. 

While the \emph{peak} reconstruction fidelities between the two methods seem to be comparable (certainly within error bars), the fact that the Schr\"{o}dinger evolution is biased away from an optimum state shows the importance of including the backaction in the dynamical model.

\chapter{Summary and Outlook\label{chap:Outtro}}
We conclude with a summary of each research chapter and a discussion of the possible avenues this research might take in the future.

\section{Quantum optics and quantum stochastic differential equations}

Chap. \ref{chap:QuantumLight} derived a relation between quasi-monochromatic traveling wave packets and the bosonic Fock space necessary for defining a formal quantum It\={o} stochastic calculus.  We identified a limit where the continuous-time tensor product decomposition is consistent with a quasi-monochromatic approximation.  The limit was ultimately enforced by convolving any bounded, square-integrable complex function with a smoothing kernel, constraining the resulting object to be slowly-varying in time.  A suitable white noise limit was identified when the kernel approached a delta function in conjunction with a limit where its derivative remained remained infinitesimal when compared to an optical period.  This produced a separation of three timescales.  The resulting quantum stochastic integral is on the slowest scale, the delta correlated smoothing kernel is in the middle, and the fastest is an optical period.

This version of a quantum white noise limit is new and distinct from two existing explanations.  The first is a static picture of a bosonic heat bath lacking any dynamical flow of information, which does not capture the fundamental propagation of a traveling wave field \cite{gardiner_input_1985,gardiner_quantum_2004}.  In an alternative description, \citeauthor{accardi_quantum_2002} derive a quantum white noise limit through a rescaling argument \cite{accardi_weak_1990,accardi_quantum_2002}.  If the system field coupling Hamiltonian had a fundamental interaction strength of $\lambda$, then by rescaling time as $t \rightarrow t/\lambda^2$ and taking the limit $\lambda \rightarrow 0$, a quantum white noise operator appears.  While this is in no doubt mathematically correct, it is in our opinion \emph{ad hoc}, as requiring time to be not just big, but specifically $1/\lambda^2$-big is an artificial constraint.  Here we do require a white noise limit in a ``middle timescale'', meaning that the smoothing envelop approaches a delta function ($\sigma \rightarrow 0$) while the timescale of one carrier oscillation tends to zero \emph{faster} ($(\sigma \omega_0)^{-1} \rightarrow 0$).  We do not require any fixed or delicate scaling law.  Additionally, in our model the QSDE treatment stands independently from any system-field interaction, as long as that coupling respects the above approximations.

After formulating this a quantum white noise approximation, we rederived a QSDE description of the propagator.  This derivation applies the recently derived quantum Wong-Zakai theorem \cite{gough_quantum_2006}.   This result describes how a general interaction involving scattering operations converges to a valid QSDE.  This result allows for a stochastic description of a dispersive Faraday interaction in a regime of a weak drive but high optical density.  A fundamental characteristic of any scattering based propagator is that admits for the possibility of having multiple scattering events in the intermediate timescale.  It effectively renormalizes over this effect.

The scattering propagator we derive is similar to a propagator derived by \citeauthor{bouten_adiabatic_2008} \cite{bouten_adiabatic_2008}.  They derive a polarizability interaction for a 4-level atom starting from a QSDE expression for a quantum field coupling two ground states, via two excited states.  They then adiabatically eliminated the excited atomic states under the usual approximations of weak excitation.  The resulting propagator contained similar but not identical scattering processes.  The difference is that the resulting integrands for the $d \Lambda^{rr}_t$ and $d \Lambda^{ll}_t$ terms only contain single scattering events.  By starting from a QSDE for a dipole interaction and then eliminated an atomic manifold, the atoms are only capable of making a single scattering transition in one intermediate time increment.  It is unclear at this point which model is more applicable for describing the underlying physics or if both are equally ``wrong'', just in different ways. It should be noted that both models agree in the limit of a weak forward scattering rate and negligible spontaneous emission.  More analysis is clearly needed to settle any debate.

\section{Classical and quantum probability theory}

Chap. \ref{chap:Math} served as a review of both classical and quantum probability theory.  It ultimately focused on the mapping, via the spectral theorem, between sets of commuting observables and a classical probability space. When a family of operators $\set{Y_i}_{i \in I}$ pairwise commute, the underlying projectors, \emph{i.e} the spectral measure $\proj(d\lambda)$, define the classical probability model $(\Omega, \mathcal{F}, \mathbbm{P})$.  The sample space $\Omega$ is set of labels $\set{\lambda}$, $\mathcal{F}$ is the smallest $\sigma$-algebra over $\Omega$, and the probability measure is defined by the quantum expectation value for a projector associated with any element in $\mathcal{F}$,   $\mathbbm{P}(d \lambda) = \Tr(\rho \proj(d\lambda)) $.

The utility of this mapping is that it allows us to define a classical stochastic process that is in some sense equivalent to a QND observable.  When restricting our attention to operators that commute with the projectors defining the classical space, we are able to treat any operator $A = \int_\Omega a(\lambda)\, \proj(d\lambda)$ in terms of the classical random variable $a$.  When the eigenvectors of $A$ do not form a complete basis in the underlying Hilbert space, there exist operators $X$ that commutes with $A$ but still must be treated quantum mechanically.   Rather than being a liability, this is in fact a feature, as it allows us to simplify a system-probe interaction by treating the system quantum mechanically and the probe as a classical random variable.  We call the resulting description semiclassical, but it still captures the physics of a system-probe-measurement model, when $\set{Y_i}_{i \in I}$ are the measured probe observables.  However a careful physicist should always check that any operator $X$ considered in this model commutes with every $Y_i$.

We then apply this formalism to rederive the conditional master equation, when measurements are generated from the observation process $\{ Y_t = U_t^\dag (A_t + A_t^\dag) U_t \}_{t \ge 0}$ in vacuum expectation.  This derivation is initially performed in the Heisenberg picture and is then referred to a quantum filter.  In the quantum filtering language the filtered observables $U_t^\dag X U_t$ are still operators, written as $\pi_t(X)$, and are linear combinations of the projectors $\set{\proj(d \lambda)}$.  When converted to a conditional master equation, we compute a system state $\rho_t$ that is defined on the system Hilbert space only.  When taking an expectation value of a system operator $X$, we have the equality
\begin{equation}
  \pi_t(X)|_{\set{Y_t = y_t}} = \Tr(\rho_t X)
\end{equation}
for every system operator $X$ and time $t$.  We also described how the quantum filter is equivalent to a purification of any generalized measurement scheme.

When the conditional master equation has complete information, and the initial condition is a pure state, then it is sufficient to use an equivalent Schr\"{o}dinger equation and calculate the random system state vector $\ket{\psi_t}$.  Given a correct and complete description of the system, a conditional Schr\"{o}dinger equation is identical to the stochastic Schr\"{o}dinger equation derived in the quantum optics literature and the choice of the measurement observables defines the specific unraveling of the unconditioned master equation.

In a recent paper by \citeauthor{tsang_evading_2012}, they define a \emph{quantum mechanics free subsystem} as a set of time-dependent operators subject to the constraint that the operators commute for any time \cite{tsang_evading_2012}.    Given the set of Heisenberg picture operators $\set{O_i(t)\, :\, i = 1, \dots, n}$, if $[O_i(t),\, O_j(t')] = 0$ for all $i,j, t, \text{and } t'$, then they are free of the laws of quantum mechanics.  This is exactly the same idea as the quantum to classical mapping via the spectral theorem and relies on exactly the same principle.  We mention this result here as it is an example showing active research utilizing the mapping between quantum and classical structures.  The tools of classical probability theory should have an important role to play in this line of research.  In particular, the concept of a commutant should be invaluable as it describes the set of operators that are compatible with this subsystem.

The results of Chap. \ref{chap:QubitState} would not have been possible without noting that the fundamental noise process driving the conditional master equation is not simply a Wiener process, but is instead the measurement realization itself.  With different initializations, the conditional master equation produces different innovations and only a few of them will have the statistics of a Wiener process.  From a statistical perspective, the conditional master equation is an estimator that allows us to predict the outcomes of measurements performed on the system, given an ancilla coupled measurement record.  Furthermore in classical estimation theory, the concepts of robustness and stability play an important roles, as it is a desirable for an estimator to be robust to modeling imperfections and incorrect initializations.  The stability of the quantum filter shows that in most cases the conditional master equation is able to correct itself given bad initial information.  This supports the case that a conditional quantum state should be viewed as a quantum analog to a classical estimator.  Using this perspective, it is possible that one might be able to formulate a variant of the quantum filter, that is more robust to modeling errors or corrupting noise in the measurement signal.

For real optical beams, the continuous-time tensor product decomposition was only an approximation, and that for short enough timescales the operators did not strictly commute.  Outside of this approximation, it is not immediately clear if one could formulate a truly commutative space of operators for the purpose of defining a conditional expectation.  As the technology of ultrafast lasers progresses, it might be possible to experimentally test a regime where subsequent optical measurements \emph{almost} commute and the conditional dynamics might then reveal surprising quantum effects.  Doing so would likely require formulating a conditional expectation on a noncommutative von Neumann algebra, which in some cases is possible \cite{van_handel_stability_2009}, but the classical probabilistic interpretation is lost \cite{bouten_introduction_2007}.

Von Neumann argued that a commuting approximation to almost commuting observables was always possible \cite{von_neumann_mathematical_1996}, however this has been shown to be not the case \cite{davidson_local_2001}.  In the context of spin chains, \citeauthor{ogata_approximating_2011} recently established the existence of commuting approximations for ``macroscopic'' observations \cite{ogata_approximating_2011}.  It is likely that these mathematically rigorous results will be invaluable in identifying the consistent information embedded in a sequence of noncommuting observations.

\section{Projection filtering for qudit ensembles}
The methods of differential geometry are a set of powerful and flexible tools, readily applied to a wide variety of problems.  Orthogonal projections are fundamental to quantum theory.  Chap. \ref{chap:projection} combines both of these tools to derive an approximation to the conditional master equation for an ensemble of $n$ qubits, given a diffusive measurement of the collective angular momentum projection $J_z$, and in the presence of strong global rotations.  The approximation was based on the ansatz that if the system was initialized in an identical tensor product state, $\varrho = \rho^{\otimes n}$, then it should remain close to a state $\rho^{\prime\, \otimes n}$ of some single qubit state $\rho'$.   The approximation was made to find a modified evolution that preserved this symmetry.  It was formulated by projecting the conditional master equation, acting on the state $\varrho$, into the tangent space of the manifold of states
\begin{equation}
  \mathcal{P} \define \set{ \rho^{\otimes n}\ : \ \rho\text{ is a valid qubit state} }.
\end{equation}
We worked in a parametrization where the single qubit state is mapped to a vector defined within the unit Bloch ball.  We were able to derive an analytic expression for a projection filter that describes the diffusion of the Bloch vector in a non-Euclidean but isotropic space.

We subsequently tested the quality of the resulting approximation numerically for pure spin coherent states.  These simulations were performed for a systems composed of $25 \le n \le 100$ qubits under a variety of conditions.  We make this comparison first without an external Hamiltonian, allowing the system to evolve only under the action of the measurement.  In an exact description this model would produce a significant amount of spin squeezing and is  confirmed in the simulations.  The projection filter tracked the mean expectation values with $\sim 90\%$ accuracy but failed to describe the correlations induced by the squeezing, as it was designed to do.  We then performed the same analysis in the presence of a Hamiltonian driving strong randomized rotations.  In this case, spin squeezing failed to significantly accumulate, leading to a $\gtrsim 95\%$ agreement between the exact and projected mean expectation values and an average fidelity $> 80\%$ between the projected and exact states for all qubit numbers tested.

A natural extension of the projection filter is to move beyond qubits and consider higher spin systems.  Unfortunately, the simplicity of the Bloch sphere is lost for $d > 2$.  There certainly exist $d^2 - 1$ traceless, orthogonal, Hermitian matrices for decomposing a $d$-dimensional quantum state.  The problem is that in attempting to formulate a mapping between valid quantum states and a $d^2-1$-dimensional ball, you find that not every point inside the ball, or its surface, corresponds to a valid quantum state \cite{kimura_bloch-vector_2005}.  The problem is that while the orthonormal matrices have a number of useful features, they do not share the same spectra, and so the boundary between valid and invalid states is not isotropic.  For qubits, we happily ignored any issues involving the boundary between valid and invalid states.  It is likely that for qudits willful ignorance may lead to disaster.  The best course of action may be to seek a more abstract representation of the state.

In addition to moving to a higher spin system, we can also consider correlated states.  One family of correlated states of general interest are spin squeezed states.  Finding a smooth parametrization for pure spin squeezed states is not difficult as the canonical example of spin squeezing is generated by a specific Hamiltonian \cite{kitagawa_squeezed_1993}.  By composing the one parameter group of squeezers with the group of $SU(d)$ rotations, it is likely that one can describe the space of pure spin-$d$ squeezed state as a $d^2$-dimensional manifold, baring any issues with linear independence.  Whether or not there is a wieldy metric induced on this space is a whole other question entirely.

An additional complication is the unavoidable fact that for a model to be at all experimentally useful, it must be able to handle mixed states and decoherence.  Adding single qubit decoherence to the separable projection filter is a trivial under taking.  Any map that acts identically and independently on each qubit is, by definition, in the tangent space of identical separable states.  The reason why our simulations only considered pure state dynamics is because generating exact simulations for $n \sim 50$ qubits is quite challenging when the total angular momentum is not a conserved quantity.   \citeauthor{chase_collective_2008} derive a simulation technique that required only order $n^2$ parameters for exactly propagating $n$ qubits under symmetry preserving local decoherence \cite{chase_collective_2008}.  By applying this or a similar method, we expect to be able to extend our numerical tests to include some decoherence to the model.

While the algorithmic and ``optimal'' nature of the projection filter is appealing, control and system engineers confronted by nonlinear problems have derived a number of suboptimal but highly successful estimation techniques.  Some of which are a linearized extended Kalman filters \cite{jazwinski_stochastic_2007}, ``unscented'' Kalman filters \cite{julier_unscented_2004}, Monte Carlo ``particle'' filters \cite{arulampalam_tutorial_2002}, symmetry preserving filters \cite{bonnabel_symmetry-preserving_2008}, and so on.  Some or all of these techniques may prove useful for partially observed quantum systems.  Although, without a general mapping between quantum observables and classical statistics, none of these tools are applicable.

\section{Qubit State Reconstruction}
Chap. \ref{chap:QubitState} applied the quantum filtering formalism to construct a tomographic estimate for an unknown initial quantum state from an ensemble of identical copies experiencing a joint continuous measurement.  We found a maximum likelihood estimate of the initial state, based upon the statistics of a \emph{single} continuous measurement realization.  The purpose of this work was to extend previous results using a continuous measurement for quantum state tomography, into a regime where the quantum backaction significantly affect the measurement statistics.  In a numerical study with ideal conditions, we found that our estimate nearly saturate an optimum bound.  Derived by \citeauthor{massar_optimal_1995}, this bound states that the average reconstruction fidelity given $n$ copies of a pure qubit state and no other prior information, the best average fidelity is given by $\langle\mathcal{F}\rangle_\text{opt}  = (n+1)/(n+2).$

The problem of identifying a tomographic estimate was mapped to a parameter estimation problem, where the statistics of the measurement record parametrically depended upon the initial qubit state.  We then found that the likelihood function for the measurement record was ultimately Gaussian, leading to an equivalence between a maximum likelihood estimate and a least-squares estimate.   Maximizing the likelihood function then ultimately reduced to minimizing the quadratic variation of an innovation process, computed from the measurement record and a conditional state estimate.  When the conditional state corresponds to a ``correct'' description, then the innovation is a Wiener process, setting the minimum value to be $\sim t$.

In order to make a numerical implementation computationally feasible, we approximated an exact innovation by one computed with the projection filter.  As this reconstruction procedure is tied to the quality of the projection filter, any improvements in its accuracy will almost surely improve the reconstruction fidelity.  We expect that by extending the projection filter to include squeezed states, there will be near perfect agreement between an innovation computed from the projection filter and an innovation computed from the exact conditional master equation.

The nature extensions of the projection filter carries over to the case of state reconstruction.  The principle of finding a least-squared estimate is system independent as long as evolution still described by a diffusive conditional master equation.  However when moving to qudits, the number of parameters we need to estimate grows unfavorably with $d$.  It is likely that in the general case, it will no longer be feasible to simply sample from the compact parameter space and select the most likely candidate.  Parameter estimation in nonlinear statistical models is a well studied problem for classical systems and we believe that a classical solution will be adaptable to the quantum case.  One possible avenue to investigate is a statistical importance and resampling technique, a ``particle filter'' \cite{arulampalam_tutorial_2002}, which has already been adapted to a quantum parameter estimation problem \cite{chase_single-shot_2009}.

From our perspective, it is an open question as to whether or not minimizing the innovation's quadratic variation is the optimum statistical test to use.  In hypothesis testing, comparing the ratio of two likelihood functions has been shown to have the most predictive power out of all statistical tests.  From that fact, one possible method for improving the reconstruction procedure is to compute the likelihood ratio between each candidate state and a master equation initialized in the completely mixed state.  In hypothesis testing, the ratio is made by comparing the likelihood of the data being generated from your model compared to a null hypothesis.  For a quantum system the most logical null hypothesis is the completely mixed state.  It is possible that by computing this likelihood ratio we will be able to better discriminate the signal arising from the initial state from the signal caused by the quantum backaction.

It is still an open question as to why this continuous measurement scheme approaches the optimum bound computed by \citeauthor{massar_optimal_1995}.  Because the numerical results perform so well, it is likely that the randomized controls are mapping the continuous measurement to a uniform measure over all spin coherent states, as this is known to achieve the optimum bound \cite{bagan_comprehensive_2005}.  Understanding what effective POVM a given controlled-continuous measurement implements will likely be a powerful result in itself.  Armed with that knowledge, a clever experimentalist could engineer any number of complicated measurement protocols.  Not the least of which being unambiguous state discrimination \cite{nielsen_quantum_2000}.



\appendix

\chapter{Paraxial Optics \label{app:paraxialOptics}}
This appendix review the paraxial wave equation and ultimately calculates Fourier transform of a paraxial mode function.  The paraxial wave equation begins by assuming a quasi-monochromatic solution to the wave equation that takes the form of a rapidly oscillating plane wave, $\exp( + i( k_0 z - \omega_0 t) ) $ modulate by an envelope function that changes slowly in both space and time.  If there is a vector valued function $\V{\mathcal{U}}(\Vx,t)$ satisfying the wave equation
\begin{equation}
    \grad^2 \V{\mathcal{U}}(\Vx,t) - \frac{1}{c}\frac{\partial^2}{\partial t^2} \V{\mathcal{U}}(\Vx,t) = 0,
\end{equation}
then we hypothesize a real-valued solution of the form $\V{\mathcal{U}}(\Vx,t) = \V{\mathcal{U}}^{\pf}(\Vx,t) + c.c.$ where
\begin{equation}
    \V{\mathcal{U}}^\pf(\Vx,t) = \V{u}^\pf(\Vx, t) e^{+ i ( k_0 z - \omega_0 t)}
\end{equation}
and $\V{u}^\pf(\Vx, t)$ is a slowly varying function.  Slowly varying is characterized by the inequalities
\begin{subequations}\label{chPara:eq:paraxialApproximation}
\begin{equation}
    \abs{\frac{\partial^2 \V{u}^\pf}{\partial z^2} }  \ll k_0 \abs{\frac{\partial \V{u}^\pf}{\partial z} } \ll k_0^2 \abs{\V{u}^\pf}
\end{equation}
and
\begin{equation}
    \abs{\frac{\partial^2 \V{u}^\pf}{\partial t^2} } \ll \omega_0 \abs{\frac{\partial \V{u}^\pf}{\partial t} }\ll \omega_0^2 \abs{\V{u}^\pf}.
\end{equation}
\end{subequations}
Based upon this assumption one can then neglect terms in the wave equation that are second derivatives with respect to $z$ and $t$.  The result is the paraxial wave equation,
\begin{equation}
    \begin{split}\label{chPara:eq:paraxialWaveEq}
        \frac{1}{2 k_0} \grad_T^2 \V{u}^\pf + i \left(\frac{\partial \V{u}^\pf}{\partial z} + \frac{1}{c}\frac{\partial \V{u}^\pf}{\partial t} \right) = 0.
    \end{split}
\end{equation}
$\grad_T^2$ denotes the Laplacian with respect to the remaining transverse coordinates $\set{x, y}$, and we will denote the transverse direction as $\Vx_T$.  It is often said that the paraxial approximation is valid under the assumption that the envelope function $\V{u}^\pf$ varies slowly when compared to an optical period, $2 \pi/\omega_0$.

Under the change of variables $(\Vx; t) \rightarrow (\Vx_T, z; t - z/c)$ the paraxial wave equation becomes independent of retarded time $t_r = t - z/c$.  Without a loss of generality it can be assumed that 
\begin{equation} \label{chPara:eq:paraxialFactoredSolution}
    \V{\mathcal{U}}^{\pf}(\Vx,t) = f(t_r )\, \V{u}^\pf_T(\Vx_T, z)\, e^{- i \omega_0\, t_r }
\end{equation}
for any slowly varying $f$. In this case we have,
\begin{equation}\label{chPara:eq:paraxialSchrodingerEquation}
    \frac{1}{2 k_0} \grad_T^2 \V{u}^\pf_T(\Vx_T, z)  = - i \frac{\partial}{\partial z} \V{u}^\pf_T(\Vx_T, z).
\end{equation}
If we make the replacement $z \rightarrow t$ and $k_0 \rightarrow m/\hbar$ then Eq. (\ref{chPara:eq:paraxialWaveEq}) is identical to a two-dimensional Schr\"{o}dinger equation in free space.  Furthermore in the Fourier domain, each plane wave component is an eigenstate of the ``Hamiltonian'' $\frac{1}{2 k_0} \grad_T^2$ ultimately implying,
\begin{equation}
    \VF{u}^{\pf}_T(\Vk_T, z) = \VF{u}^{\pf}_T(\Vk_T, 0) e^{- i \frac{\abs{\Vk_T}^2}{2 k_0} z}.
\end{equation}
where $\Vk_T = k_x \Ve_x + k_y \Ve_y$.  Note that this is Fourier transform is with respect to the transverse components only and is still a function of the longitudinal component $z$.  Taking the Fourier transform of the full solution $\V{\mathcal{U}}(\Vx, t) \rightarrow \VF{\mathcal{U}}(\Vk, t)$,
\begin{equation}
    \VF{\mathcal{U}}^{\pf}(\Vk, t) = \VF{u}^{\pf}_T(\Vk_T, z = 0) \int \frac{dz}{\sqrt{2 \pi}} e^{-i k_z z} e^{- i \frac{\abs{\Vk_T}^2}{2 k_0} z} f(t - z/c )e^{i k_0 (z - c t)}.
\end{equation}
If we change variable from taking the spatial transform with respect to $z$ to transforming with respect to the retarded time $t_r$ we find that
\begin{equation}
    \VF{\mathcal{U}}^{\pf}(\Vk, t) = c\, \F{f}(\,c |\Vk_T|^2/(2 k_0) + c k_z - \omega_0)\, \VF{u}^{\pf}_T(\Vk_T, 0)\,e^{-i c \left(\frac{\abs{\Vk_T}^2}{2 k_0} + k_z\right) t}
\end{equation}
where $\F{f}(\omega)$ is the temporal Fourier transform of $f$.

We know that regardless of any approximations the positive frequency component of a traveling wave solution evolves according to $\VF{\mathcal{U}}^{\pf}(\Vk, t) = \VF{\mathcal{U}}^{\pf}(\Vk, 0) e^{-i c \abs{\Vk} t}$.  Evidently, the paraxial approximation is an approximation that
\begin{equation}\label{chPara:eq:Dispersion}
    \omega(\Vk) = c \abs{\Vk} \approx c\,(\frac{\abs{\Vk_T}^2}{2 k_0} + k_z).
\end{equation}
This approximation seems slightly at odds with the fact that $\abs{\Vk}$ is strictly nonnegative, while the right-hand side of Eq. (\ref{chPara:eq:Dispersion}) extends to negative frequencies.  This issue is resolved by the assumption that $\V{u}^{\pf}(\Vx, t)$ is a slowly varying function or equivalently $\VF{u}^{\pf}(\Vk, t)$ is a sharply peaked function about $\Vk = 0$.  The resolution is that because of the carrier plane wave, $\VF{\mathcal{U}}^{\pf}(\Vk, t)$ is a sharply peaked function about $\Vk_0 = k_0 \, \Ve_z$.  While in principle  $\omega(\Vk)$ can be negative, these negative components never contribute as long as the paraxial approximation holds.  Finally we find that
\begin{equation}
    \VF{\mathcal{U}}^{\pf}(\Vk, t) = c\, \F{f}\left(\omega(\Vk) - \omega_0\right)\, \VF{u}^{\pf}_T(\Vk_T, 0)\,e^{-i \omega(\Vk) t}.
\end{equation}

\chapter{Classical Stochastic Calculus \label{app:SDEs}}
This appendix reviews the derivation of the Stratonovich and It\=o integrals as well as the rules of It\=o calculus.  We attempt to describe the salient points found in standard texts while leaving out the proofs.

Stochastic integration, in either the quantum or classical sense begins from a form of functional integration.  The traditional Riemann integral takes a function $f(t)$ and integrates with respect to a small difference in its argument $\Delta t$. In a functional integral, as defined by Stieltjes, the function $f(t)$ is integrated with respect to a small difference in another function $g(t)$, \emph{i.e.}
\begin{equation}\label{appSDEs:eq:stieltjes}
  \int_0^t f(s)\, dg(s)\define \lim_{n \rightarrow \infty} \sum_{i = 1}^n f(t^*_{i})\ \big(\, g(t_i) - g(t_{i-1})\, \big)
\end{equation}
where $t_{i-1} \le t^*_i \le t_i$.  This limit can be shown to make sense if $f$ and $g$ are reasonably well behaved.  One limitation is that $g$ can't vary ``too much'' over a time interval $\Delta t$, (the total variation of $g$ must be finite) \cite{van_handel_stochastic_2007}.  Like the usual definition of a Riemann integral, the convergence of this integral does not depend upon where $t^*_i$ lies in the interval $[t_{i-1}, t_i]$.

A stochastic integral, often called a stochastic differential equation (SDE), replaces both functions $f$ and $g$ by stochastic processes.  However in this replacement the problems of working with nondifferentiable functions leads to a more delicate situation.  The fact that Brownian motion has a nowhere smooth trajectory means that its total variation is infinite, leading to a divergence in a Riemann-Stieltjes limit \cite{oksendal_stochastic_2002}.  Not only does this mean that we are forced to consider a different kind of limit, but the choice of $t^*_i$ makes a dramatic difference on its mathematical and statistical properties. Specifically, if $t^*_i = t_{i-1}$, one arrives with an It\={o} integral, which has several desirable statistical properties but does not obey the chain rule as seen in ordinary calculus.  If, however, $t_i^*$ is taken at the midpoint of the interval, called a Stratonovich integral, then the rules of calculus are preserved but the statistical properties are more involved.  Fortunately there exists a simple conversion between the two integral definitions.  We will discuss all of these issues in greater detail in the following section.

Concretely, the It\={o} and Stratonovich integrals begin with the following definitions.  Consider the partitioning of the time interval $[0, \tau]$ into in increasingly dense mesh $n$-ordered times $\set{t_n : n \in \mathbbm{N} }$, $0 < t_1 < \dots < t_n = \tau$.  The It\={o} integral takes the (time-adapted) process $\set{x_t}_{t \ge 0}$ and defines the integral of the well-behaved function $b(x_t)$, with respect to the Wiener process $w_t$, to be
\begin{equation}\label{appSDEs:eq:classicalIto}
    \int_0^\tau b(x_t) \, d w_t \define \lim_{n \rightarrow \infty} \sum_{i = 1}^n b(x_{t_{i-1}})\, ( w_{t_i} - w_{t_i-1}).
\end{equation}
This limit is then shown to converge to an almost unique object, with probability one\footnote{There is a slight caveat where one could add another random process which happens to have zero probability of ever occurring.}  The Stratonovich integral takes a different definition,
\begin{equation}\label{appSDEs:eq:classicalStratonovich}
    \int_0^\tau b( x_t ) \circ d w_t \define \lim_{n \rightarrow \infty} \sum_{i = 1}^n b \left(  \tfrac{x_{t_{i}} + x_{t_{i-1}}}{2} \right)\, ( w_{t_i} - w_{t_i-1}).
\end{equation}
These two definitions arrive at fundamentally different, but not unrelated, integrals.  One integral can be converted to another by using a simple trick, derived in Sec. \ref{appSDEs:sec:ItoConversion}.

Before moving on to discussing the operational and statistical properties of these integrals it is worth noting, that in attempting to model a classical physical system, with SDEs the choice of calculus is crucial.  Fortunately, the question as to which calculus to use is answered by the Wong-Zakai theorem \citep{wong_convergence_1965} (see the introduction to Appendix \ref{app:QuWongZakai}).  In this paper they showed that an ordinary differential equation containing a piecewise smooth approximation to Brownian motion, limits to a Stratonovich equation and not an It\={o} equation.  If one derives an equation of motion for a system including an approximation to Brownian motion then the solution to that equation must be interpreted in the Stratonovich sense.

\section{It\={o} Calculus}
The rules of It\={o} calculus can be derived with varying levels of detail and sophistication.  Their practical purpose is to give a method for manipulating and combining multiple It\={o} integrals into new and different expressions.  The bottom line result is that the standard differential chain rule, $d (f g) = f' dg + g' df$  must be extended to include a second order correction, see Eqs. (\ref{appSDEs:eq:ItoExpansion} - \ref{appSDEs:eq:ItoExampleProducts}).  An often cited reference for the derivation of the It\={o} integral and It\={o} calculus is the book by \citeauthor{oksendal_stochastic_2002} \cite{oksendal_stochastic_2002}.  There he shows how the limits in Eqs. (\ref{appSDEs:eq:classicalIto}) and (\ref{appSDEs:eq:classicalStratonovich}) may or may not converge.  Specifically, the standard techniques for defining a integral with respect to a Riemann sum fails, because in doings you ultimately consider the quantity, called the total variation,
\begin{equation}
    \lim_{\Delta t \rightarrow 0} \sum_{i = 1}^n \abs{ w_{t_i} - w_{t_i-1} }
\end{equation}
for the partition of times $a \le t_0 < \dots < t_n = b$.  What you can show is that with probability one, this is infinite for the Wiener process.
However you can also show that instead of summing the absolute value of each increment, $|w_{t_i} - w_{t_i-1}|$, you sum the square of each increment, $(w_{t_i} - w_{t_i-1})^2$, then you obtain a finite quantity.  This is called the quadratic variation, and with probability one,
\begin{equation}
    \lim_{\Delta t \rightarrow 0} \sum_{i = 1}^n( w_{t_i} - w_{t_i-1} )^2 = b - a.
\end{equation}
What this is saying is that while the Wiener process travels an infinite absolute distance in any finite time, it RMS displacement only grows like the square root of time.  Using the fact that the Wiener process has a well behaved quadratic variation, the limits of the It\=o and Stratonovich integrals make sense if you consider their squared expectation value, the so called $\mathcal{L}^2(\mathbbm{P})$ limit.

For a flavor for how a squared expectation value might make sense for a Wiener process, observe that as it is constructed to have Gaussian statistics for any finite interval,
\begin{equation}
\lim_{n \rightarrow \infty} \mathbbm{E}\left( \left( \sum_{i = 1}^n\,( w_{t_{i}} - w_{t_{i-1}}\, ) \right)^2 \right) = \mathbbm{E}\left(  (\, w_{t} - w_s\, )^2 \right) =  t - s .
\end{equation}
The way this is turned into an integral is that if one has the \emph{time-adapted} stochastic process $b_t$, \emph{i.e.} it is assumed to be independent of $w_{t'}$ for times $t' > t$, then it can be shown that
\begin{equation}
\lim_{n \rightarrow \infty} \mathbbm{E}\left( \Big(\, \sum_{i = 1}^n\,b_{t_{i-1}}( w_{t_{i}} - w_{t_{i-1}})\, \Big)^2 \right) = \lim_{n \rightarrow \infty} \mathbbm{E}\left(  \sum_{i = 1}^n\,b_{t_{i-1}}^2( t_{i} - t_{i-1}\, ) \right)
\end{equation}
as long as $\mathbbm{E}\left(  \int_s^t\, b_{t}^2\, dt \right) < \infty$  \cite{oksendal_stochastic_2002}.  A fair amount of analysis goes into showing for what kinds of processes a piecewise constant approximation.  More work is needed to extend the proof to hold with probability one.  We mention this here only because it indicates the line of reasoning that relates the product of two It\={o} integrals to a integral over time of the product of the integrands.

Moving from a consistent definition of an integral to a mature calculus involves placing a constraint on what kind of integrands we are able to use in a stochastic integral.  Consider the example of the recursively defined It\={o} process
\begin{equation}\label{appSDEs:eq:classicalItoExample}
  x_t = x_0 + \int_0^t a(s, x_s)\, ds + \int_0^t b(s, x_s)\, dw_s,
\end{equation}
where $a$ and $b$ are continuous functions that are once differentiable in time and twice in $x$.
An extremely common and useful notational device is to write an It\={o} integral in a differential form,
\begin{equation}\label{appSDEs:eq:classicalItoDifferential}
  dx_t = a(t, x_t)\, dt +  b(s, x_t)\, dw_t
\end{equation}
to represent that integral.  The utility of this notation is apparent in that if $x_t$ were an ordinary deterministic equation, (setting $b(t, x_t) = 0$) then $x_t$ is the solution to the equation
\begin{equation}
  \frac{d x}{dt} = a(t, x).
\end{equation}
This is why that these kinds of equations are called \emph{stochastic differential equations}.

A typical application is to consider two such equations, so in addition to $x_t$ we have the differential
\begin{equation}
  dy_t = c(t, y_t)\, dt +  d(t, y_t)\, dw_t.
\end{equation}
We wish to find a differential for the product $x_t\, y_t$ or even some other function $f(x_t, y_t)$.  The answer to the first example is that
\begin{equation}\label{appSDEs:eq:ItoProductRule}
  \begin{split}
  d(x_t\,y_t) &= a(t, x_t)\,y_t\, dt + b(t, x_t)\, y_t\, dw_t\\
              &\quad + x_t\, c(t, y_t)\, dt +  x_t\, d(t, y_t)\, dw_t\\
              &\quad + b(t,x_t)\, d(t, y_t)\, dt.
  \end{split}
\end{equation}
This expression can be derived by computing the $\mathcal{L}^2(\mathbbm{P})$, $\Delta t_i \rightarrow 0$ limit of the definition of the It\={o} integral.
It is important to emphasize that the right-hand side of this equation is itself another It\={o} integral.  By multiplying, adding, and subtracting It\={o} integrals one still finds ``just'' another It\={o} integral, with the whole closing upon itself to form an algebra.  Compare this example to a second order Taylor expansion of the product $x_t\, y_t$,
\begin{equation}\label{appSDEs:eq:secondOrderProductRule}
  d(x_t\,y_t) = dx_t\, y_t + x_t\, dy_t + dx_t\, dy_t.
\end{equation}
If we apply the It\={o} rules to the second order product $dx_t\, dy_t$, the only surviving term is $(b\, dw_t)\, (d\, dw_t) = b\, d\, dt$.  This means that there is a consistency between Eqs. (\ref{appSDEs:eq:ItoProductRule}) and (\ref{appSDEs:eq:secondOrderProductRule}).
The general case for some function $f(x_t, y_t)$ is given by the second order Taylor expansion,
\begin{multline} \label{appSDEs:eq:ItoExpansion}
    df(x_t, y_t) = \frac{\partial f(x_t, y_t) }{\partial x_t}\, d x_t + \frac{\partial f(x_t, y_t) }{\partial y_t}\, d y_t\\
 + \frac{1}{2} \frac{\partial^2 f(x_t, y_t) }{\partial x_t^2}\, d x_t\, d x_t + \frac{1}{2} \frac{\partial^2 f(x_t, y_t) }{\partial y_t^2}\, d y_t\, d y_t + \frac{\partial^2 f(x_t, y_t) }{\partial x_t\, \partial y_t}\, d x_t\, d y_t
\end{multline}
where
\begin{align} \label{appSDEs:eq:ItoExampleProducts}
  dx_t\, dx_t &= b^2(t, x_t)\, dt, \\
  dy_t\, dy_t &= d^2(t, y_t)\, dt,\quad \text{ and}\\
  dx_t\, dy_t &= b(t, x_t)\,d(t, x_t)\, dt.
\end{align}
In order for this general It\={o} expansion to be well defined, $f(x, y)$ must have a finite first and second derivatives.

\subsection{The It\={o} conversion \label{appSDEs:sec:ItoConversion}}
Supposed we have the recursive It\={o} form SDE such that
\begin{equation}
    x_\tau = \int_0^\tau a(x_t)\, dt + \int_0^\tau b(x_t)\, dw_t.
\end{equation}
We then assert that this same process has a corresponding Stratonovich form,
\begin{equation}
    x_\tau = \int_0^\tau \tilde{a} (x_t)\, dt + \int_0^\tau b(x_t)\circ dw_t.
\end{equation}
($\tilde{a}$ has no relation to the Fourier transform.)  Our goal is to then find a relation between the functions $a(x)$ and $\tilde{a}(x)$ that makes this assertion true.  Subtracting these two expressions lead to the equality,
\begin{equation}\label{appSDEs:eq:itoCorrectionDefinition}
    I_c \define \int_0^\tau b( x_t ) \circ d w_t - \int_0^\tau b( x_t )\, d w_t  = \int_{0}^\tau a(x_t)\, dt - \int_{0}^\tau \tilde{a}(x_t)\, dt.
\end{equation}
This difference is known as the \emph{It\={o} correction term}.

To lighten the notation we will define the intervals, $\Delta x_{i} \define x_{t_i} - x_{t_{i-1}}$, $\Delta w_{i} \define w_{t_i} - w_{t_{i-1}}$ and $\Delta t_{i} \define t_i -  t_{i-1}$.  The It\={o} correction can then be written as
\begin{equation}\label{appSDEs:eq:itoCorrectionLimit}
  \begin{split}
    I_c =&\, \lim_{n \rightarrow \infty} \sum_{i = 1}^n \left( b \left(  \tfrac{x_{t_{i}} + x_{t_{i-1}}}{2} \right) - b( x_{i-1}) \right)\, \Delta w_i\\
    &\, \lim_{n \rightarrow \infty} \sum_{i = 1}^n \left( b \left( x_{t_{i-1}} +  \half\, \Delta x_i \right) - b( x_{i-1}) \right)\, \Delta w_i
  \end{split}
\end{equation}
We can write the integrand very suggestively in terms of a prelimit form of the derivative of $b(x)$.
By defining
\begin{equation}
    \frac{\Delta b}{\Delta x}(x) \define \frac{b(x + \Delta x) - b(x) }{\Delta x}
\end{equation}
we then have
\begin{equation}
    b \left( x_{t_{i-1}} +  \half\, \Delta x_i \right) - b( x_{i-1}) =  \frac{1}{2}\, \frac{\Delta b}{\Delta x}( x_{t_{i-1}})\, \Delta x_i.
\end{equation}
With a recursive definition for $x_t$,  we can substitute $\Delta x_i$ to find
\begin{equation}
    I_c = \lim_{n \rightarrow \infty} \sum_{i = 1}^n  \frac{1}{2} \frac{\Delta b}{\Delta x}(x_{t_{i-1}})\,\big( a(x_{t_{i-1}})\,\Delta t_i + b(x_{t_{i-1}} )\, \Delta w_i \big)\, \Delta w_i.
\end{equation}
However, by the rules of It\={o} calculus (\emph{i.e.} $\Delta t_i\, \Delta w_i \rightarrow 0$ and $\Delta w_i\, \Delta w_i \rightarrow \Delta t_i$ with probability 1) the whole expression converges to
\begin{equation}\label{appSDEs:eq:itoCorrection}
    I_c =  \frac{1}{2} \int_0^\tau\,\frac{d b}{d x}(x_{t})\,b(x_{t} )\, d t.
\end{equation}
And so we ultimately find that
\begin{equation}
    \tilde{a}(x) = a(x) -  \frac{1}{2}\frac{d b}{d x}(x)\,b(x).
\end{equation}
This conversion between It\={o} and Stratonovich equations is vitally important in Chap. \ref{chap:projection}, where the first order rules of ordinary calculus allows for the application of differential geometry to a stochastic system. Sec. \ref{chProj:sec:StochCalcManifolds} gives an example for why using Stratonovich calculus is necessary.

\chapter{Quantum Stochastic Calculus\label{app:QSDEs} }

This appendix reviews the basis notation and properties of an It\={o} form quantum stochastic differential equation (QSDE).  Sec. \ref{chQuLight:sec:QuWienerProcesses} discusses at length how a bosonic Fock space $\Fock(\hilbert)$, defined over the single particle Hilbert space $\hilbert = \mathcal{L}^2(\R^+)\otimes\Cn{d}$, has the operators $Q^i_t$ and $P^i_t$.  Each of these operators have the statistics of Wiener processes when taken in vacuum expectation.  They are constructed though linear combinations of the annihilation and creation operators, $A_t^i = \ahat[\indicate{[0,t]}\Ve_i]$ and $A_t^{i\, \dag} = \ahat^\dag[\indicate{[0,t]}\Ve_i]$, which are also used to form non-Hermitian quantum It\={o} integrals.    In addition to these processes, Sec. \ref{chQuLight:sec:Lambda} encountered a different kind of operator, the scatter or conservation processes $\Lambda^{ij}_t$.  Here we will gloss over how $\Lambda^{ij}_t$ can be defined in terms of an operation acting on the single particle Hilbert space (see \cite{barchielli_continual_2006}).  Note that $\left(\Lambda^{ij}_t\right)^\dag = \Lambda^{ji}_t$.   To define an integral with respect to $A_t^i$, $A_t^{i\, \dag}$, and $\Lambda^{ij}_t$, it is sufficient to know the matrix elements\footnote{For technical reasons the amplitudes of these states are assumed to be square integrable and have a large but finite upper bound \cite{barchielli_continual_2006, bouten_introduction_2007}. },
\begin{subequations}\label{appQSDEs:eq:quantumNoisesMatrixElements}
\begin{align}
  \braOket{\e[\Vf]}{A^{i}_t}{\e[\Vh]} &= \int_0^t ds\ h_i(s)\ \braket{\e[\Vf]}{\e[\Vg]},\\
  \braOket{\e[\Vf]}{A^{j\,\dag}_t}{\e[\Vh]} &= \int_0^t ds\ f^*_j(s)\ \braket{\e[\Vf]}{\e[\Vg]},\quad\text{and}\\
  \braOket{\e[\Vf]}{\Lambda^{ij}_t}{\e[\Vh]} &= \int_0^t ds\ f^*_i(s)\, h_j(s)\ \braket{\e[\Vf]}{\e[\Vg]}.
\end{align}
\end{subequations}

The derivation of the quantum It\={o} integral and the equivalent stochastic calculus begins in much of the say way as the classical It\={o} integral.  Eq. (\ref{appSDEs:eq:classicalIto}) gives the limiting sequence used in defining the It\={o} integral, with its characteristic nonanticipative integrand.  The quantum It\={o} integral is given by a similar form, where the integral with respect to $dA^j_t$ is given by the limiting sum, (with the mesh of times $t_0 = 0 < t_1\, <\dots t_n < t$)
\begin{equation}\label{appQSDEs:eq:qudAIntegral}
    Y_t = \int_0^t X_s\, d A^j_s \define \lim_{n \rightarrow \infty} \sum_{i = 1}^n X_{t_{i-1}}\, (A^j_{t_{i}} - A^j_{t_{i-1}} ).
\end{equation}
Similar definitions are made for integrals with respect to $A^{j\dag}_t$, $\Lambda^{ij}_t$, and time.  The integrand $X_s$ is required to be a time-adapted process meaning that it is required to act as the identity on the Fock space $\Fock(\hilbert_{(s,\infty)})$.  Note that in general, $X_{s}$ is not required to act solely on the Fock space $\Fock(\hilbert_{[0, s]})$.  It could be an operator defined over a joint space $\Hilbert_{sys}\otimes\Fock(\hilbert_{[0, s]})$.  Because the integrand is required to be time-adapted it commutes with the differential and so we can write,
\begin{equation}
    Y_t = \int_0^t X_s\ d A^j_s  = \int_0^t dA^j_s\ X_s.
\end{equation}
We will use both forms, whichever is more convenient.  It should not be that surprising that proving that the above integrals exists and are finite is more difficult than in a classical setting.  We will not reproduce the full result here, see \cite{hudson_quantum_1984, parthasarathy_introduction_1992} for the proof.

However to get a sense of where the quantum It\={o} rule comes from, we will discuss a few key points.  The proof of convergence of the infinitesimal limit is shown by taking a piecewise constant approximation, and then proving convergence of the matrix elements.  The set of vectors chosen for those matrix elements are the tensor product of any system pure state and an exponential vector.  By showing that it hold for these matrix elements, you can then extend the result to hold in expectation with any state composed of convex linear combinations of these vectors.  For example, consider the integral $Y_t$ in Eq. (\ref{appQSDEs:eq:qudAIntegral}),
\begin{equation}
    \braOket{\psi\otimes \e[\Vf]}{Y_t}{\psi\otimes \e[\Vh]} = \lim_{n \rightarrow \infty} \sum_{i = 1}^n \int_{t_{i-1}}^{t_i} ds\ h_i(s)\ \braOket{\psi\otimes \e[\Vf]}{X_{t_{i-1}}}{\psi\otimes \e[\Vh]} .
\end{equation}
For a reasonably large class of integrands the piecewise constant approximation is appropriate and converges to
\begin{equation}
    \braOket{\psi\otimes \e[\Vf]}{Y_t}{\psi\otimes \e[\Vh]} = \int_{0}^{t} ds\ h_i(s)\ \braOket{\psi\otimes \e[\Vf]}{X_{s}}{\psi\otimes \e[\Vh]}.
\end{equation}
Equivalent expressions occur for integrals with respect to $dA_t^{j\, \dag}$ and $d\Lambda_t^{ij}$, where $h_i(s)$ is replaced by the correct amplitudes as given in Eq. (\ref{appQSDEs:eq:quantumNoisesMatrixElements}).

In order to define a proper calculus, one must also consider how products of the integrals behave.  The elementary step is to consider the matrix elements of two operator combinations of $A^{i}_t$, $A^{j\,\dag}_t$ and $\Lambda^{ij}_t$.  The matrix elements $\braOket{\e[\Vf]}{A^{i}_t\, A^{j}_t}{\e[\Vh]}$, $\braOket{\e[\Vf]}{A^{i\dag}_t\, A^{j\,\dag}_t}{\e[\Vh]}$, and $\braOket{\e[\Vf]}{\Lambda^{ij}_t\, A^{k}_t}{\e[\Vh]}$ are easily to calculated, as they involve an eigenvalue relationship of $A_t$ and the matrix elements given in Eq. (\ref{appQSDEs:eq:quantumNoisesMatrixElements}).  The nonobvious two operator matrix elements are \citep[Proposition 20.13]{parthasarathy_introduction_1992}
\begin{subequations}
\begin{align}
  \braOket{\e[\Vf]}{A^{i}_t\, A^{j\,\dag}_t}{\e[\Vh]} &= \left(\int_0^t ds\ f^*_j(s)\ \int_0^t ds\ h_i(s) \right. \nonumber \\
    &\qquad + \delta_{ij} \left.\int_0^t ds\ \right)\braket{\e[\Vf]}{\e[\Vg]},\\
  \braOket{\e[\Vf]}{A^i_t\, \Lambda^{jk}_t}{\e[\Vh]} &= \left( \int_0^t ds\ f^*_j(s)\, h_k(s) \int_0^t ds\ h_i(s)\right. \nonumber \\
      &\qquad + \left.\delta_{ij} \int_0^t ds\ h_k(s) \right) \braket{\e[\Vf]}{\e[\Vg]},\\
  \braOket{\e[\Vf]}{\Lambda^{ij}_t A^{k\dag}_t\, }{\e[\Vh]} &= \left( \int_0^t ds\ f^*_i(s)\, h_j(s) \int_0^t ds\ f_k^*(s)\right. \nonumber \\
      &\qquad + \left.\delta_{jk} \int_0^t ds\ f^*_i(s) \right) \braket{\e[\Vf]}{\e[\Vg]},\quad \text{and} \\
    \braOket{\e[\Vf]}{\Lambda^{ij}_t\, \Lambda^{k \ell}_t}{\e[\Vh]} &= \left( \int_0^t ds\ f^*_i(s)\, h_j(s) \int_0^t ds\ f_k^*(s) h_\ell(s)\right. \nonumber \\
     & \qquad +\left. \delta_{jk} \int_0^t ds\ f_i^*(s) h_\ell(s) \right) \braket{\e[\Vf]}{\e[\Vg]}.
\end{align}
\end{subequations}
Calculating these expressions requires knowledge of the commutation relations between all of the noises and/or their relation to the exponential vectors.  Without discussing the origin of $\Lambda_t$, writing down the commutation relations are no more intuitive and less useful than above matrix elements. Inspecting these four equations shows that each has two parts.  The first is essentially the product of the single operator matrix elements, up to a factor of $\braket{\e[\Vf]}{\e[\Vg]}$.  The second part is an additional term due to the noncommuting structure of the processes.   If we express these four matrix elements in terms of differentials,
\begin{subequations}
\begin{align}
  \braOket{\e[\Vf]}{dA^{i}_t\, dA^{j\,\dag}_t}{\e[\Vh]} &= \Big( O(dt^2)  + \delta_{ij}\,  dt \Big) \braket{\e[\Vf]}{\e[\Vg]},\\
  \braOket{\e[\Vf]}{dA^i_t\, d\Lambda^{jk}_t}{\e[\Vh]} &= \Big( O(dt^2) + \delta_{ij}\, dt\, h_k(t)\, \Big) \braket{\e[\Vf]}{\e[\Vg]},\\
  \braOket{\e[\Vf]}{d\Lambda^{ij}_t d A^{k\dag}_t\, }{\e[\Vh]} &= \Big( O(dt^2) + \delta_{jk}\, dt\, f^*_i(t)\,\Big) \braket{\e[\Vf]}{\e[\Vg]},\quad \text{and} \\
  \braOket{\e[\Vf]}{d\Lambda^{ij}_t\, d\Lambda^{k \ell}_t}{\e[\Vh]} &= \Big( O(dt^2)  +\delta_{jk}\, dt f_i^*(t) h_\ell(t) \Big) \braket{\e[\Vf]}{\e[\Vg]}.
\end{align}
\end{subequations}
Notice that each term on the order of $dt$ is expressible in terms of the matrix element of a differential, $dt$, $dA_t^k$, \emph{etc.}  Taking the terms  $O(dt^2) \rightarrow 0$, and asserting that knowing these matrix elements is sufficient, we have the quantum It\={o} rules
\begin{subequations} \label{appQSDEs:eq:quantumItoRules}
\begin{align}
  dA^{i}_t\, dA^{j\,\dag}_t &=   \delta_{ij}\,  dt, \\
  dA^i_t\, d\Lambda^{jk}_t &= \delta_{ij}\, dA_t, \\
  d\Lambda^{ij}_t\, d A^{k\dag}_t &= \delta_{jk}\, dA_t^{i\,\dag}, \\
  d\Lambda^{ij}_t\, d\Lambda^{k \ell}_t &= \delta_{jk}\, d\Lambda^{i\ell}_t.
\end{align}
\end{subequations}
All other differential products are zero.

The remainder of the construction of the quantum It\={o} calculus is technical, and involves showing how the product of two piecewise approximations converge as $\Delta t_i \rightarrow 0$, as well as dealing with the equality 
\begin{equation}
\left\langle X_t\, \psi\otimes \e[\Vf] \middle\vert Y_t\, \psi\otimes \e[\Vh]\right \rangle = \left\langle  \psi\otimes \e[\Vf]\middle\vert X^\dag_t Y_t\, \psi\otimes \e[\Vh]\right \rangle
\end{equation}
for two integrals $X_t$ and $Y_t$.   These issues are well beyond our scope, except that we will note that when $X_t$ and $Y_t$ are bounded in a suitable sense, all things work out nicely \cite{bouten_introduction_2007, barchielli_continual_2006}.

Before moving on to the specifics of a QSDE for a unitary propagator, we will give the following general example for using the quantum It\=o rule. (See \citep[proposition 2.4]{barchielli_continual_2006} for all of the necessary qualifiers and assumptions.)  Consider the quantum stochastic process $X_t$,  given by the It\={o} integral (implied sums on repeated indices),
\begin{equation}\label{appQSDEs:eq:generalItoIntegral}
  X_t = X_0 + \int_0^t d\Lambda^{ij}_s\, F^{ij}_s + \int_0^t d A^{i\, \dag}_s\, F^{i0}_s + \int_0^t d A^{j}_s\, F^{0j}_s + \int_0^t ds\, F^{00}_s.
\end{equation}
The processes $F^{\alpha \beta}_s$ are assumed to act nontrivially only on the joint Hilbert space $\Hilbert_{sys} \otimes \Fock(\hilbert_{[0, s]})$ and are integrable (without really defining what that means).  The initial value $X_0$ is assumed to be a bounded operator that acts as the identity on $\Fock(\hilbert)$.  This integral can also be notated differentially as the QSDE
\begin{equation}
  d X_t = F^{ij}_t\, d\Lambda^{ij}_t  + F^{i0}_t\, d A^{i\, \dag}_t  + F^{0j}_t\, d A^{j}_t  + F^{00}_t\, dt.
\end{equation}
Given another process $Y_t$ whose differential is
\begin{equation}
  d Y_t = K^{ij}_t\, d\Lambda^{ij}_t  + K^{i0}_t\, d A^{i\, \dag}_t  + K^{0j}_t\, d A^{j}_t  + K^{00}_t\, dt,
\end{equation}
the product $X_t Y_t$ can also be expressed in terms of a process $Z_t = X_t Y_t$ whose differential is given by
\begin{equation}
  d Z_t = M^{ij}_t\, d\Lambda^{ij}_t  + M^{i0}_t\, d A^{i\, \dag}_t  + M^{0j}_t\, d A^{j}_t  + M^{00}_t\, dt,
\end{equation}
where the resulting integrands $M^{\alpha \beta}_t$ are
\begin{subequations} \label{appQSDEs:eq:quantumItoProductCoefficients}
\begin{align}
  M^{ij}_t &= X_t\, K^{ij}_t + F^{ij}_t\, Y_t + F^{i \ell}_t\,K^{\ell j}_t,\\
  M^{i0}_t &= X_t\, K^{i0}_t + F^{i0}_t\, Y_t + F^{i j}_t\,K^{j 0}_t,\\
  M^{0j}_t &= X_t\, K^{0j}_t + F^{0j}_t\, Y_t + F^{0 i}_t\,K^{i j}_t,\quad \text{and}\\
  M^{00}_t &= X_t\, K^{00}_t + F^{00}_t\, Y_t + F^{0 i}_t\,K^{i 0}_t.
\end{align}
\end{subequations}
This result can be viewed as an application of the quantum It\={o} product rule
\begin{equation}
  d (X_t Y_t) = X_t\, dY_t + dX_t\ Y_t + dX_t\, dY_t,
\end{equation}
where $dX_t\, dY_t$ is multiplied out and the second order differentials are evaluated according to the quantum It\={o} rules in Eq. (\ref{appQSDEs:eq:quantumItoRules}).

\section{The Quantum Stochastic Unitary\label{appQSDEs:sec:dUt}}
With the general quantum It\={o} rule firmly in hand, we would like to apply it to find a universal expression for a unitary process $U_t$.  This is actually quite straight forward by first noting that because its unitary, $U_t^\dag U_t = \ident$.  We also know that if $U_t$ is independent from the fundamental processes $A_t$, $A^\dag_t$ and $\Lambda_t$, than $U_t$ is the solution to the ordinary differential equation $dU_t = -i H_t\, U_t\, dt$.  When including the fundamental processes, the objective is to write $U_t$ as a general QSDE and find how unitary constrains the various integrands.  Taking the ``noise free'' solution as a starting point we hypothesize the coefficients $G^{\alpha\beta}_t$ so that $U_t$ is given by the QSDE
\begin{equation}
  d U_t = G^{ij}_t\,U_t\, d\Lambda^{ij}_t  + G^{i0}_t\,U_t\, d A^{i\, \dag}_t  + G^{0j}_t\,U_t\, d A^{j}_t  + G^{00}_t\,U_t\, dt
\end{equation}
and its adjoint is
\begin{equation}
  d U_t^\dag = U^\dag_t\,G^{ij\dag}_t\, d\Lambda^{ji}_t  +  U^\dag_t\,G^{i0\dag}_t\, d A^{i}_t  +  U^\dag_t\,G^{0j\dag}_t\, d A^{j\,\dag}_t  +  U^\dag_t\,G^{00 \dag}_t\, dt.
\end{equation}
The unitary constraint's impact on the differential is that $d(U_t^\dag\, U_t) = d(U_t\, U^\dag_t) = 0$.  The general It\={o} product coefficients in Eq. (\ref{appQSDEs:eq:quantumItoProductCoefficients}) then says that in order for this to unitary,
\begin{equation}
\begin{split}
  U^\dag_t\, G^{\alpha \beta}_t\,U_t + U^\dag_t\,G^{\beta \alpha \, \dag}_t\, U_t + U^\dag_t\,G^{\ell \alpha\,\dag}_t\,G^{\ell \beta}_t\,U_t = 0\\
  U_t\,U^\dag_t\, G^{\beta \alpha \, \dag}_t + G^{\alpha \beta}_t\,U_t\, U^\dag_t + G^{\alpha \ell}_t\,U_t\, U_t^\dag\, G^{ \beta \ell\,\dag}_t = 0
\end{split}
\end{equation}
for $\alpha,\beta$ starting at zero and the implied sum over $\ell$ starting from $1$.  Eliminating $U_t$ and $U_t^\dag$ from the constraints, they simplify to
\begin{align}\label{appQSDEs:eq:dUtConstraints}
  G^{\alpha \beta}_t + G^{\beta \alpha\, \dag}_t + G^{\ell \alpha \,\dag}_t\,G^{\ell \beta}_t = G^{\alpha \beta}_t + G^{\beta \alpha\, \dag}_t + G^{\alpha \ell}_t G^{\beta \ell \, \dag}_t = 0.
\end{align}

The coefficients $G^{\alpha \beta}$ are typically written in terms of a different set of operators, $S^{ij}_t$, $L^i_t$ and $H_t$.  The reason for this transformation is that $(S^{ij}_t, L^i_t, H_t)$ have more desirable and physically relevant properties than $G^{\alpha \beta}$.  Immediately we can see that some part of $G^{00}_t$ should be $-i H_t$, as the general QSDE solution contains the case where $U(t) = \exp(-i H t)$ for a time independent Hamiltonian $H$. Also if $G^{ 0 i} = G^{i 0} = 0$ then Eq.(\ref{appQSDEs:eq:dUtConstraints}) reads as $G^{00} = - G^{00\, \dag}$, implying that $G^{00}_t = -i H_t$ for Hermitian $H_t$.

To identify how $S^{ij}_t$ fits into the picture, consider for the moment the case where each $G^{ij}_t = g_{ij} \ident$ for some complex coefficients $g_{ij}$.  Then the constraints for $G^{ij}_t$ are
\begin{align}
  g_{ij} + g_{ji}^* + g_{\ell i}^{*}\,g_{\ell j} = g_{ij} + g_{ji}^* + g_{\ell i} g_{\ell j}^{*} = 0
\end{align}
Writing the constants in term of a matrix $\mathbbm{G}$ we have
\begin{align}
  \mathbbm{G} + \mathbbm{G}^\dag + \mathbbm{G}^\dag\,\mathbbm{G} = 0.
\end{align}
If we define a matrix $\mathbbm{S} \define \mathbbm{G} + \ident$ then this constraint reads,
\begin{equation}
\begin{split}
    0&= \mathbbm{S} - \ident + \mathbbm{S}^\dag - \ident + (\mathbbm{S}^\dag - \ident)(\mathbbm{S} - \ident) \\
    0&=\mathbbm{S} - \ident + \mathbbm{S}^\dag - \ident + (\mathbbm{S}^\dag\mathbbm{S} - \mathbbm{S}^\dag - \mathbbm{S} + \ident)\\
    \ident&=\mathbbm{S}^\dag \mathbbm{S}.
\end{split}
\end{equation}
In other words $\mathbbm{S}$ is a unitary matrix.    Returning to the general case, we can still define the operators
\begin{equation}
  S^{ij}_t \define G^{ij}_t + \delta_{ij}.
\end{equation}
Then Eq. (\ref{appQSDEs:eq:dUtConstraints}) transforms the constraint for $G^{ij}_t$ into the constraint,
\begin{equation}
   S^{ij \dag}_t S^{jk}_t = S^{i j }_t S^{k j\, \dag}_t = \delta_{ik}.
\end{equation}
In other words, $S^{ij}_t$ is a unitary matrix of operators.

By introducing $S^{ij}_t$, the constraint for  $G^{0i}_t$ is also significantly simpler.  Specifically,
\begin{equation}
\begin{split}
  0 &= G^{0 i}_t + G^{i 0\, \dag}_t + G^{\ell 0 \,\dag}_t\,G^{\ell i}_t \\
  0&=G^{0 i}_t + G^{i 0\, \dag}_t + G^{\ell 0 \,\dag}_t\,(S^{\ell i}_t - \delta_{\ell i}) \\
  G^{0 i}_t &= - G^{\ell 0 \,\dag}_t\,S^{\ell i}_t.
\end{split}
\end{equation}
The remaining coefficients $G^{i 0}_t$ are essentially arbitrary, which is relabeled as the operators $L^i_t$.  Writing the constraint for $G_t^{00}$ in terms of $L^i$, means that $G^{0 0}_t + G^{0 0\,\dag}_t = -  L^{i\dag}_t L^i_t$.

Bringing all of these results together we can reexpress $G^{\alpha \beta}_t$ in terms of $(S^{ij}_t, L^{i}_t, H_t)$,
\begin{equation} \label{appQSDEs:eq:GtoSLHconversion}
\begin{split}
  G^{i j}_t &= S^{i j}_t - \delta_{ij},\\
  G^{i 0}_t &= L_t^i,\\
  G^{0 j}_t &= - L_t^{i \dag} S^{i j}_t,\\
  G^{0 0}_t &= -i H_t - \half L^{i\dag}_t L^i_t.
\end{split}
\end{equation}
In other words,
\begin{equation} \label{appQSDEs:eq:generaldUt}
  d U_t = \Big(  \left(S^{i j}_t - \delta_{ij} \right)\, d\Lambda^{ij}_t  + L^{i}_t\, d A^{i\, \dag}_t  - L_t^{i \dag} S^{i j}_t\, d A^{j}_t   - \half L^{i\dag}_t L^i_t\, dt -i H_t\, dt\ \Big) U_t.
\end{equation}
This is the standard form for the QSDE for a general propagator $U_t$.  While in principle, the initial value for $U_0$ could be any unitary operator acting on a system $\Hilbert_{sys}$, typically $U_t$ describes an interaction picture representation of the system-field dynamic, and then $U_0 = \ident$.

One final remark is that if we take each coefficient to be its own stochastic process, $\{\,G^{\alpha\, \beta}_t U_t\, \}_{t \ge 0}$, then they are required to still be time-adapted.  This results in the constraint that the initial values $S^{ij}_0$, $L^i_0$ and $H_0$ must all be system operators only, as they must act as the identity on the field at that time.  Furthermore if $S^{ij}_t$, $L^i_t$ and $H_t$ are known to be time independent, then they must be system operators only.

\subsection{Unitary evolution\label{appQSDEs:sec:Jt}}
One final calculation we will include is the unitary evolution of a time independent system operator $X$.  In the quantum stochastic literature this unitary evolution is written in terms of a map $j_t(\cdot)$ generating the ``flow'' or current of the operator,
\begin{equation}
  j_t(X) \define U_t^\dag X U_t.
\end{equation}
This map can a written as a solution to a QSDE,
\begin{equation}
  j_t(X) = U_0^\dag X U_0 + \int_0^t dj_s(X),
\end{equation}
with a differential
\begin{equation}
  dj_t(X) = dU_t^\dag\ X U_t + U_t^\dag X\ dU_t + dU_t^\dag X dU_t.
\end{equation}
After an exercise in quantum stochastic calculus, one finds the recursive QSDE
\begin{equation} \label{appQSDEs:eq:dJtX}
   dj_t(X) = j_t(\mathcal{L}^{ij}_t(X) )\, d \Lambda^{ij}_t + j_t(\mathcal{L}^{i0}_t(X))\, dA_t^{i\, \dag} + j_t(\mathcal{L}^{0j}_t(X))\, dA_t^{j} + j_t(\mathcal{L}^{00}_t(X))\, dt
\end{equation}
where $\mathcal{L}^{\alpha \beta}_t(\cdot)$ are known as the Evens-Hudson maps and in terms of $G^{\alpha \beta}_t$ are
\begin{equation}
    \mathcal{L}^{\alpha \beta}_t(X) = G^{\beta \alpha\, \dag}_t X +  X G^{\alpha \beta}_t + G^{k \alpha\, \dag}_t X G^{k \beta}_t.
\end{equation}
When written in terms of $(S^{ij}_t, L^i_t, H_t)$ these maps are,
\begin{subequations}
\begin{align}
  \mathcal{L}^{i j}_t(X) &=  S^{k i\, \dag}_t X S^{k j}_t - \delta_{ij} X, \\
  \mathcal{L}^{i 0}_t(X) &= S^{k i\,\dag}_t \left[X,\, L_t^k\right],\\
  \mathcal{L}^{0 j}_t(X) &=  - \big[X,\,  L_t^{k \dag}\big] S^{k j}_t,\\
  \mathcal{L}^{0 0}_t(X) &= +i [H_t, X] + L^{i\, \dag}_t X L^i_t - \half L^{i\, \dag}_t L^i_t X - \half X L^{i\, \dag}_t L^i_t .
\end{align}
\end{subequations}
For an arbitrary model the unitary evolution becomes exceedingly complicated very rapidly.  Each coefficient in Eq. (\ref{appQSDEs:eq:dJtX}) is itself given by the unitary flow of the operator $Z_t \define \mathcal{L}^{\alpha \beta}_t(X)$.  Systems lacking any kind of fundamental symmetry will rarely close on a useful subspace of operators meaning that after repeated applications of $\mathcal{L}^{\alpha \beta}_t(\cdot)$ more complicated operators will be generated, spanning a larger and larger space of operators. 

In addition to calculating the unitary output of system operators, it is also useful to calculate the output for the fundamental field operators $A_t^{j \dag}$, $A_t^i$, $\Lambda_t^{ij}$.  In a rather tedious exercise in manipulating the quantum It\={o} rules it can be shown that
\begin{subequations}\label{appQSDEs:eq:dJtFields}
\begin{align}
   dj_t(A^j_t) =& j_t(S^{jk}_t )\, dA_t^{k} + j_t(L^j_t)\, dt,\\
   dj_t(A^{i\,\dag}_t) =& j_t(S^{ik\, \dag}_t )\, dA_t^{k\, \dag} + j_t(L^{i\,\dag}_t)\, dt,\\
   dj_t(\Lambda^{ij}_t) =& j_t(S^{ik\, \dag}_t S^{j \ell}_t)\, d\Lambda_t^{k \ell} + j_t(S^{ik\, \dag}_t L^{j}_t)\, dA_t^{k\, \dag}  + j_t(L^{i\, \dag}_t S^{j \ell}_t)\, dA_t^{\ell} +  j_t(L^{i\,\dag}_t L^j_t)\, dt.
\end{align}
\end{subequations}

\chapter{The Quantum Wong-Zakai Theorem \label{app:QuWongZakai}}

In a classical system, the convergence of an ordinary differential equation to a stochastic one was treated in the work of Wong and Zakai \citep{wong_convergence_1965}.  There they show that an ODE containing a piecewise-smooth approximation to white noise, $\xi^{(\lambda)}_t$, converges to a Stratonovich integral as $\xi^{(\lambda)}_t \rightarrow \xi_t$ with $\lambda \rightarrow 0$.
For instance, Consider the ODE
\begin{equation}
    \frac{\partial x^{(\lambda)}(t)}{\partial t} = f(t, x^{(\lambda)}(t)) + g(t, x^{(\lambda)}(t))\, \xi^{(\lambda)}_t.
\end{equation}
The Wong-Zakai theorem states that the integrated solution
\begin{equation}
    x^{(\lambda)}(t) = x^{(\lambda)}(0) + \int_0^tds\,  f(s, x^{(\lambda)}(t))  + \int_0^t ds\,  g(s, x^{(\lambda)}(s))\,\xi^{(\lambda)}_s
\end{equation}
converges to the Stratonovich integral
\begin{equation}
    x_t = x_0 + \int_0^t f(x(t), s)\, ds  + \int_0^t g(x(s), s)\circ d w_s.
\end{equation}

Appendix \ref{app:SDEs} reviews the distinctions between the two most common forms of classical stochastic integration, the It\={o} integral and the Stratonovich integral and Appendix \ref{app:QSDEs}  discusses their quantum analogs.  This appendix reviews the quantum analog of this result where the specific ODE is for a propagator whose Hamiltonian contains field operators that are limiting to quantum white noise.

In \citeyear{gough_quantum_2006}, \citeauthor{gough_quantum_2006} derived a quantum limit, equivalent to the  Wong-Zakai theorem \citep{gough_quantum_2006}.  Specifically, he investigated the convergence of the Schr\"{o}dinger equation, written in terms of the time evolution operator $U(t)$, as field operators in the Hamiltonian converge to singular, delta-commuting operators.   The specific Hamiltonian consider is given in Eq. (\ref{appQuWZT:eq:QuWongZakaiHamiltonian}), but before discussing it, we will first describe the quantum version of $\xi^{(\lambda)}_t$ and how it can be interpreted as quantum white noise.

\section{Quantum white noise \label{appQuWZT:sec:QuWongZakaiWhiteNoise}}

In order to make a connection with a quantum It\={o} integral as formulated by Hudson and Parthasarthy, the limiting field operators clearly must be reference to a Fock space $\Fock(\hilbert')$, $\hilbert' = \mathcal{L}^2(\R^+)\otimes \Cn{d}$.  The delta commuting limit is introduced by considering the \emph{differentiable} functions $\{ \xi^{(\lambda)}_i(t) \in \hilbert' :\quad i = 1,\dots, d \}$ parameterized by $\lambda > 0$ so that we have the field operators $\ahat[\xi^{(\lambda)}_i(t)]$ and $\ahat^\dag[\xi^{(\lambda)}_j(t')]$, with
\begin{equation}
  \left[ \ahat[\xi^{(\lambda)}_i(t)],\, \ahat^\dag[\xi^{(\lambda)}_j(t')] \right] = \inprod{\xi^{(\lambda)}_i(t)}{\xi^{(\lambda)}_j(t')} \define c_{ij}(\lambda, t- t').
\end{equation}
This inner product is assumed to satisfy the properties:
\begin{subequations}\label{appQuWZT:eq:quWongZakaiLimitCriteria}
 \begin{align}
    \int_{-\infty}^{\infty} dt\ c_{ij}(\lambda, t)  & < \infty, \\
    c_{ij}(\lambda, t) &= c_{ji}^*(\lambda, -t), \quad \text{and}\\
    \lim_{\lambda \rightarrow 0} c_{ij}(\lambda, t) & = \delta_{i j}\, \delta(t). \label{appQuWZT:eq:goughwhitenoise}
 \end{align}
\end{subequations}
To simplify the notation, it is convenient to write $\ahat_i(\lambda, t) \define \ahat[\xi^{(\lambda)}_i(t)]$.

These operators end up serving two purposes in the quantum Wong-Zakai theorem.  The first is of course to act in the limiting Hamiltonian and the second is to generate the smoothed exponential vectors,
\begin{equation}\label{appQuWZT:eq:wongZakaiExpVec}
    \e[\Vg(\lambda)] \define \exp\left(\int_0^\infty dt\ g_i(t)\, \ahat^\dag_i(\lambda, t) \right)\, \ket{\vac}.
\end{equation}
Note that in relation to the one-dimensional representation of Sec. \ref{chQuLight:sec:oneDimensionalLimit}, $\ahat_i(\lambda, t)$ is \emph{almost} equivalent to $\ahat[\V{\varphi}^{(\sigma)}(t)]$. The differences lies in how the smoothed wave packets are defined.  \emph{One} possible mapping between paraxial optics and the abstract operators $\ahat_i(\lambda, t)$ is to identify $d$ paraxial spatial mode functions $\V{u}_{i}^{(+)}(\Vx_T, z)$, which satisfy the orthogonality relation, $\int d^2 x_T\, \V{u}^*_{i}(\Vx_T, z)\cdot \V{u}_{j}(\Vx_T, z) = \delta_{ij}\, \sigma_T$.  For each paraxial mode there are $d$ independent complex wave packet envelopes, inducing the smoothing operators $\ahat[\V{\varphi}^{(\sigma)}_i(t)]$.  In the case of a single paraxial mode, Eq. (\ref{chQuLight:eq:SmoothWavePacketOperator}) gave the expression for a smoothed paraxial wavepacket $\Vg^{(\sigma)}(\Vk, t)$.   In order for this expression to be equivalent to the argument of Eq. (\ref{appQuWZT:eq:wongZakaiExpVec}), we require that
\begin{equation}\label{appQuWZTeqn:prelimitahat}
    \ahat_i(\lambda,t) \cong e^{+i \omega_0 t} \ahat[\V{\varphi}^{(\sigma)}_i(-t)].
\end{equation}
The fact that we require the time reversed version of $\V{\varphi}^{(\sigma)}_i$ is an artifact of defining the smoothing with respect to a convolution.  The inclusion of the carrier phase is both mathematically and physically interesting.  Physically it is a reminder that the elements of $\hilbert'$ represent the part of the light existing on the measurement timescale, which is \emph{much} slower than the carrier frequency.  In fact the appearance of this phase is intimately related to the rotating wave approximation.  In the rotating frame atomic transition operators develop explicit time dependence at the carrier frequency.  The mathematical relevance of this carrier phase is that it cancels any rapidly oscillating phases in the inner product, $\inprod{\V{\varphi}^{(\sigma)}_i(-t)}{\V{\varphi}^{(\sigma)}_i(-t')}$.  This cancelation is explicitly apparent by computing that
\begin{equation}\label{appQuWZT:eq:goughCommutator}
    \begin{split}
      \left[\ahat_i(\lambda,t),\,\ahat_j^\dag(\lambda,t') \right]  &= c_{ij}(\lambda, t - t')\\
      &=e^{+i \omega_0 (t-t')} \inprod{\V{\varphi}^{(\sigma)}_i(-t)}{\V{\varphi}^{(\sigma)}_j(-t')}\\
      &= \delta_{ij}\,\left( \varphi^{(\sigma)} \star \varphi^{(\sigma)}\ (t - t') - i\frac{1}{\omega_0} \frac{d\varphi^{(\sigma)} }{dt}\star \varphi^{(\sigma)}\ (t - t') \right).
    \end{split}
\end{equation}
where in the final line we inserted the unequal time inner product in Eq. (\ref{chQuLight:eq:paraxialInnerProduct}), as well as remember that for the smoother $\V{\varphi}^{(\sigma)}_i$, $\norm{\Vg}/\sqrt{\tau} \rightarrow 1$.   $\lambda$ is simply a parameter representing the formal limit that as $\lambda \rightarrow 0$, $\sigma \rightarrow 0$ and $(\sigma\, \omega_0)^{-1} \rightarrow 0$.

\section{The quantum Wong-Zakai theorem}\label{appQuWZT:sec:quWongZakaiHamiltonian}
The Hamiltonian that Gough ultimately considers is,
\begin{equation} \label{appQuWZT:eq:QuWongZakaiHamiltonian}
  H_{int}(\lambda, t) =\hbar \left( \sum_{i,j = 1}^d E_{ij}\, \ahat^\dag_i(\lambda, t)\, \ahat_j(\lambda, t) + \sum_{i = 1}^d E_{i0}\, \ahat_i^\dag (\lambda, t) +  \sum_{j = 1}^d E_{0j}\, \ahat_j(\lambda, t) + E_{00} \right)
\end{equation}
The quantum Wong-Zakai theorem then takes the solution to the equation
\begin{equation}\label{appQuWZT:eq:SchrodingerEquation}
  \frac{d}{d t} U(\lambda, t) = -\frac{i}{\hbar} H_{int}(\lambda, t) \, U(\lambda, t),  \qquad U(\lambda, 0) = \ident
\end{equation}
and proves that the limit $\lim_{\lambda\rightarrow 0} U(\lambda, t) \define U_t$ is a quantum stochastic unitary process, solving a quantum Stratonovich differential equation.  Fortunately, a quantum Stratonovich integral is also expressible in terms of a quantum It\={o} integral.  Gough provides a well constructed conversion between both forms, which we review shortly.   The specifics of this limit is that it is shown to hold ``weakly'', in that
\begin{equation}
    \lim_{\lambda \rightarrow 0} \braOket{\psi\otimes\e[\Vf(\lambda)]}{U(\lambda, t)}{\phi\otimes \e[\Vg(\lambda)]} = \braOket{\psi\otimes\e[\Vf]}{U_t}{\phi\otimes \e[\Vg]},
\end{equation}
for any system state vectors $\psi$ and $\phi$ and the exponential vectors $\e[\Vf(\lambda)]$ and $\e[\Vg(\lambda)]$ as defined in Eq. (\ref{appQuWZT:eq:wongZakaiExpVec}).  In addition to the matrix elements of the propagator, it is also shown to hold weakly for the Heisenberg picture evolution of an operator $X$, so that
\begin{equation}
  \lim_{\lambda \rightarrow 0} U^\dag(\lambda, t)\, X U(\lambda, t) = U_t^\dag X U_t.
\end{equation}

Before describing the resulting unitary process, it is worth stressing several advantages of the quantum Wong-Zakai theorem.  So long as $H_{int}(\lambda, t)$ and the $\lambda \rightarrow 0$ limit are physically justifiable, then the specifics of the total state $\rho_{tot}$ are almost irreverent.  The only constraint is that the total state must be expressible in terms of convergent sequence of the matrix elements $\lim_{\lambda \rightarrow 0} \braOket{\psi_i\otimes\e[\Vf_i(\lambda)]}{\rho_{tot}}{\phi_j\otimes\e[\Vg_j(\lambda)]}$.  This means that we allowed the possibility of nonclassical superpositions.  A second advantage is that the presence of the scattering interaction $E_{ij}$ allows for a much broader class of interactions than the linear interactions typically considered in quantum optics.  While this dissertation will ultimately be considering a linear Hamiltonian ($E_{ij}$ will be negligibly small) the fact that the theory could consider a system coupled to an instantaneous number operator $\ahat^{i \dag}(\lambda, t)\, \ahat^i(\lambda, t)$  in a ``white noise'' limit is no small feat.  One possible example is to engineering a quantum-optical router, where for some set of modes the operators $E_{ij}$ coherently scatter quanta in a system dependent way.  Finally, much can be said for the fact that the limiting interaction is still described by a unitary operation.  While we have constructed the limiting field operators in terms of a measurement timescale, we have not specified what kind of measurement we will be performing.  Formulating a conditional estimate for the system given a measurement of the field is one of the main purposes of Chap. \ref{chap:Math} but at the level of the system field interaction, everything remains fully coherent.

\subsection{Quantum stochastic calculus and operator ordering}
In proving the quantum Wong-Zakai theorem, Gough also found an intuitive correspondence between operator orderings and quantum stochastic differential equations.    In order to describe the limiting propagator $U_t$ as a solution to a standard It\={o} form QSDE we need to review this correspondence.  Ultimately, the correspondence is between quantum Stratonovich equations and \emph{time-ordered} solutions to a given recursive differential equation.  Conversely, a quantum It\={o} equation is identified with a \emph{normally-ordered} solution \cite{gough_quantum_2006}.

The solution to the Schr\"{o}dinger equation with a time-dependent Hamiltonian is given by a time-ordered exponential, \begin{equation}
  U(\lambda, t) = \vec{\mathcal{T}} \exp \left(-\frac{i}{\hbar}\int_0^t ds\, H_{int}(\lambda, s) \right).
\end{equation}
The time-ordered exponential is a compact short hand for the iterated integrals,
\begin{multline}
 \vec{\mathcal{T}} \exp \left(-\frac{i}{\hbar}\int_0^t ds\, H_{int}(\lambda, s) \right) =\\ \sum_{n = 0}^{\infty} \left(\frac{-i}{\hbar}\right)^n \int_0^t dt_n \dots \int_0^{t_{2}} dt_1\ H_{int}(\lambda, t_n) \dots H_{int}(\lambda, t_1).
\end{multline}
Note that $H_{int}(\lambda,t)$ need not commute with $H_{int}(\lambda, s)$ and thus the operator ordering is critical in this expression.  There is nothing in the $\lambda \rightarrow 0$ limit that changes the operator ordering and so identifying a Stratonovich equation with a time-ordered equation is simply a statement of this fact.  The heart of the quantum Wong-Zakai theorem is relating this time-ordered exponential to a quantum It\={o} integral.  As the Wong-Zakai theorem considers the limit of matrix elements between two exponential vectors, this relation is made by comparing the matrix elements between this expression and the matrix element of an iterated It\={o} integral.

Sec. \ref{app:SDEs} discusses the classical It\={o} integral and shows how an integral $x_t = \int_0^t b_s\, dw_s$ for a time-adapted process $b_t$ leads to the It\={o} rules of calculus and how it crucially depend upon the statistical independence of $b_t$ from $dw_t$.  Not surprisingly the quantum It\={o} integral also relies on a similar independence of the integrand from the differential.  As Sec. \ref{chQuLight:sec:continuousTimeTensor} discussed, the continuous-time tensor product decomposition allows for defining field operators $\ahat[\chi_{(t, t+dt]}\, \Ve_j] = A^{j}_{t + dt} - A^j_t$, which commute with any operator adapted to the the time interval $[0, t]$.  The quantum It\={o} integral with respect to $dA_t$ is defined as
\begin{equation}
  Y_t = \int_0^t X_s\, d A^j_s \define \lim_{n \rightarrow \infty} \sum_{i = 1}^n X_{t_{i-1}}\, (A^j_{t_{i}} - A^j_{t_{i-1}} )
\end{equation}
and likewise an integral with respect to $dA^{k\dag}_t$ is
\begin{equation}
  Z_t = \int_0^t Y_s\, d A^{k\dag}_s \define \lim_{n \rightarrow \infty} \sum_{i = 1}^n Y_{t_{i-1}}\, (A^{j\dag}_{t_{i}} - A^{j \dag}_{t_{i-1}} )
\end{equation}
Taking the matrix element of $Y_t$ between two exponential vectors results in
\begin{equation}
  \braOket{\e[\Vf]}{\int_0^t X_s\, d A^j_s}{\e[\Vh]} = \int_0^t ds\,  \braOket{\e[\Vf]}{X_s}{\e[\Vh]} \, h_j(s)
\end{equation}
and equivalently
\begin{equation}\label{appQuWZT:eq:dAdagItoInt}
  \braOket{\e[\Vf]}{Z_t}{\e[\Vh]} = \int_0^t ds\, f^\ast_j(s)\, \braOket{\e[\Vf]}{Y_s}{\e[\Vh]}.
\end{equation}
Eq. (\ref{appQuWZT:eq:dAdagItoInt}) is valid because $Y_s$ is time-adapted and therefore commutes with $dA^{j\,\dag}_s$.  This commutating property carries over to the iterated integral, so by substituting in for $Y_s$,
\begin{equation}
  \braOket{\e[\Vf]}{Z_s}{\e[\Vh]} = \int_0^t ds\int_0^s ds'\ f^\ast_k(s')\, \braOket{\e[\Vf]}{X_{s'}}{\e[\Vh]}\, h_j(s).
\end{equation}
If we hypothesize the existence of an operator $\ahat_j(t)$ by the eigenvalue relationship
\begin{equation}
  \ahat_j(t)\ket{\e[\Vf]} = f_j(t)\, \ket{\e[\Vf]}
\end{equation}
we then have
\begin{equation}
  \braOket{\e[\Vf]}{Z_t}{\e[\Vh]} \cong \int_0^t ds\int_0^s ds'\  \braOket{\e[\Vf]}{\ahat_k^\dag(s')\,X_{s'}\, \ahat_j(s)}{\e[\Vh]}.
\end{equation}
This relation holds for any iterated It\={o} integral as long as the integrand is expressed in \emph{normal order}, with all of the creation operators on the left and all of the annihilation operators on the right.  But as we have
\begin{equation}
  \lim_{\lambda \rightarrow 0}\ahat_j(\lambda, t)\ket{\e[\Vf]} = \lim_{\lambda \rightarrow 0}\int_{0}^\infty ds\  c_{j k}(\lambda,t-s)\,f_k(s)\, \ket{\e[\Vf]} = f_j(t) \ket{\e[\Vf]},
\end{equation}
we will ultimately find an equivalence between the operator $\ahat_j(t)$ and the limiting form of $\ahat_j(\lambda, t)$.

The proof of the quantum Wong-Zakai theorem follows the procedure of converting the time-ordered exponential into normal order, showing that the matrix elements converge to a finite quantity and then proving a correspondence with an equivalent It\={o} form QSDE.   

\subsection{Gauge freedom in the It\={o} correction\label{appQuWZT:sec:gaugeFreedom}}
The difference between a Stratonovich equation and an It\={o} equation is often called the It\={o} correction term.  As we have identified an It\={o} equation with the normally order version of the iterated integral, the It\={o} correction term is intimately related to this conversion.  Converting any product of field operators into normal order, is given by Wick's theorem \citep{wick_evaluation_1950}.  It states that any product of creation and annihilation operators can be written as the sum over the normal ordering of all possible contractions between all pairs of operators.  A contraction between the operators $\ahat$ and $\hat{b}$ is defined as
\begin{equation}
    \ahat^{\bullet} \hat{b}^{\bullet} \define \ahat \hat{b}\, -\, :\ahat \hat{b}:
\end{equation}
where $:\ahat \hat{b}:$ is the normal ordering of the two operators.  By Wick's theorem we can write the product
\begin{equation}
    \ahat \hat{b} \hat{c} =\ :\ahat \hat{b} \hat{c}: \,+\, :\ahat^\bullet \hat{b}^\bullet \hat{c}: \,+\, :\ahat^\bullet \hat{b} \hat{c}^\bullet: \,+ :\ahat \hat{b}^\bullet \hat{c}^\bullet:.
\end{equation}
For the boson operator considered here, the only nonzero contraction is
\begin{equation}\label{appQuWZT:eq:ahatRescaledContraction}
    \ahat_i^\bullet(\lambda,t) \ahat_j^{\dag\bullet}(\lambda,s) =[\ahat_i(\lambda,t),\, \ahat_j^{\dag}(\lambda,s)] = c_{ij}(\lambda,t - s).
\end{equation}

The heart of finding the equivalent It\={o} QSDE from the time-ordered exponential is to first apply Wick's theorem to each term in the time-ordered exponential, then take the $\lambda \rightarrow 0$ limit, and finally re-sum the series.  We will not be reproducing this result here, where the details of such a limit can be found in the following references.  In the absence of the scattering terms the proof is detailed in the book by \citeauthor{accardi_quantum_2002} \citep{accardi_quantum_2002}. The scattering terms were subsequently added by Gough \citep{gough_quantum_2005}.  However, one important aspect of the limit must be discussed as it affects the final limiting QSDE, as well as takes its root in the physical origin of $c_{ij}(\lambda, t-s)$.

In each term of the time-ordered exponential, the operators on the right are always constrained to be at an earlier time than the operators on the left.  Therefore when applying Wick's theorem, the contraction $\ahat_i^\bullet(\lambda,t) \ahat_j^{\dag\bullet}(\lambda,s)$  will always be constrained to have $t \ge s$.  This constraint means that when $\lambda \rightarrow 0$, only \emph{half} of the $c_{ij}(\lambda,t - s) \rightarrow \delta(t-s)$ limit will apply.  It is often the case that when a causal constraint is applied to a delta function limit, an additional complex term appears involving a Cauchy principle value. For each $c_{ij}(\lambda,\tau)$ the extra complex term is called a gauge freedom and generates, among other things, an effective level shift in the $E_{00}$ term.

As a concrete example, consider the second order term in the time-ordered expansion
\begin{equation}
\left(\frac{-i}{\hbar}\right)^2 \int_0^\tau dt \int_0^t ds\, H_{int}(\lambda, t) H_{int}(\lambda, s).
\end{equation}
The operator product $H_{int}(\lambda, t) H_{int}(\lambda, s)$ contains $16$ terms with at most $4$ field operators (from the scattering terms) and in the case of $E_{00}(t)\,E_{00}(s)$, no operators.  One part of this expression is the integral
\begin{equation}
 - \sum_{ij} \int_0^\tau dt \int_0^t ds\, E_{0i}\, \ahat_i(\lambda, t)\, E_{j0}\, \ahat_j^{\dag}(\lambda, s).
\end{equation}
Applying Wick's theorem means that
\begin{multline}
  -\sum_{ij} \int_0^\tau dt \int_0^t ds\, E_{0i}\, \ahat_i(\lambda, t)\, E_{j0}\, \ahat_j^{\dag}(\lambda, s) = \\
  -\sum_{ij} \int_0^\tau dt \int_0^t ds\, E_{0i}\,E_{j0}\,\, \left(: \ahat_i(\lambda, t)  \ahat_j^{\dag}(\lambda, s) :  +\, c_{ij}(\lambda, t-s) \right).
\end{multline}
This commutator term on the right-hand side will ultimately contribute to the It\={o} correction term and generate the gauge shift, as long as it survives the $\lambda \rightarrow 0$ limit.  Therefore the most basic contribution is made by the limit
\begin{equation}
    \lim_{\lambda \rightarrow 0} \int_0^\tau dt \int_0^t ds\, c_{ij}(\lambda, t-s).
\end{equation}
By substituting Eq. (\ref{appQuWZT:eq:goughCommutator}) for $c_{ij}(\lambda, t-s)$ and dropping the term proportional to $1/\omega_0$,
\begin{multline}\label{appQuWZT:eq:secondOrderCausalCommutator}
    \lim_{\lambda \rightarrow 0} \int_0^\tau dt \int_0^t ds\, c_{ij}(\lambda, t-s) = \delta_{ij}\, \lim_{\sigma \rightarrow 0} \int_0^\tau dt \int_0^t ds\, \varphi^{(\sigma)} \star \varphi^{(\sigma)}\ (t - s).
\end{multline}
In Sec. \ref{chQuLight:sec:measureableConstruction} we used the example of $\varphi^{(\sigma)}$ as a real-valued Gaussian with mean zero and standard deviation $\sigma$.  In this case, it is easy to show that
\begin{equation}\label{appQuWZT:eq:autoCorrelationRescaled}
  \varphi^{(\sigma)} \star \varphi^{(\sigma)}\ (t - s) = \frac{1}{\sigma}\, \varphi^{(1)} \star \varphi^{(1)}\ \big((t - s)/\sigma\big)
\end{equation}
where $\varphi^{(1)}$ is a mean zero Gaussian with \emph{unit} variance.  In fact this is a general property often used in distribution theory where if $\int_\R dt\, \varphi^{(1)}(t) = 1$ then
\begin{equation}
  \lim_{\sigma \rightarrow 0}\frac{1}{\sigma} \varphi^{(1)}(t/\sigma) = \delta(t).
\end{equation}
In this replacement Eq. (\ref{appQuWZT:eq:autoCorrelationRescaled}) is also satisfied.  Using this relation in the right-hand side of Eq. (\ref{appQuWZT:eq:secondOrderCausalCommutator}),
\begin{multline}
    \lim_{\lambda \rightarrow 0} \int_0^\tau dt \int_0^t ds\, c_{ij}(\lambda, t-s) = \delta_{ij}\, \lim_{\sigma \rightarrow 0} \int_0^\tau dt \int_0^t ds\, \frac{1}{\sigma}\, \varphi^{(1)} \star \varphi^{(1)}\ \big((t - s)/\sigma \big).
\end{multline}
This limit is easily evaluated by making the change of variables $\bar{t} \define (t - s)/\sigma$,
\begin{multline}\label{appQuWZT:eq:secondOrderCausalCommutatorLimit}
    \delta_{ij}\, \lim_{\sigma \rightarrow 0} \int_0^\tau dt \int^{t/\sigma}_0 d\bar{t}\, \varphi^{(1)} \star \varphi^{(1)}\ (\bar{t}) = \delta_{ij}\, \tau \int_0^{\infty} d\bar{t}\, \varphi^{(1)} \star \varphi^{(1)}\ (\bar{t}).
\end{multline}
If $\varphi^{(1)}$ is a normalized real-valued distribution then $\varphi^{(1)} \star \varphi^{(1)}\ (\bar{t}) =  \varphi^{(1)} \star \varphi^{(1)}\ (-\bar{t})$ and so
\begin{equation}
\begin{split}
  \int_0^{\infty} d\bar{t}\, \varphi^{(1)} \star \varphi^{(1)}\ (\bar{t}) &= \frac{1}{2}\int_{-\infty}^\infty d\bar{t}\, \varphi^{(1)} \star \varphi^{(1)}\ (\bar{t}) \\
  &= \frac{1}{2} \left(\int_{-\infty}^\infty d\bar{t}\, \varphi^{(1)}(\bar{t}) \right)^2\\
  &= \frac{1}{2}.
\end{split}
\end{equation}
The change of variables in Eq. (\ref{appQuWZT:eq:secondOrderCausalCommutatorLimit}) is the delta correlation limit but because of the time-ordered integration, we obtain the factor of $\half$.  However in the general case, $\varphi^{(1)}$ need not be real-valued.  For our example involving quasi-monochromatic fields, it is \emph{sufficient} for $\varphi$ to be a real-valued function as it is simply a mathematical tool representing a limit on the rate of change of the arbitrary complex functions $\Vf \in \mathcal{L}^2(\R^+)\otimes\Cn{d}$.

In the general Wong-Zakai limit $c_{ij}(\lambda, t)$ is not assumed to be real, only that the criteria of Eq. (\ref{appQuWZT:eq:quWongZakaiLimitCriteria}) are satisfied.  From Eq. (\ref{appQuWZT:eq:goughwhitenoise}) we have that
\begin{equation}
\lim_{\lambda \rightarrow 0}  \int_{-\infty}^{\infty} dt\  c_{ij}(\lambda, t) = \delta_{ij},
\end{equation}
meaning that we can define the potentially complex constants
\begin{equation}
 \begin{split}
    \kappa_{ij} &\define \lim_{\lambda \rightarrow 0} \int_0^\infty dt\, c_{ji}(\lambda, t)\\
    \kappa^\ast_{ij} &\define \lim_{\lambda \rightarrow 0} \int_0^\infty dt\, c_{ji}^\ast(\lambda, t)  = \lim_{\lambda \rightarrow 0} \int_{-\infty}^0 dt\, c_{ji}(\lambda, t)
 \end{split}
\end{equation}
where we used the fact that $c_{ji}^\ast(\lambda, t) = c_{ji}(\lambda, -t)$.
Combining these two results means
\begin{equation}\label{appQuWZT:eq:kappaConstraint}
    \delta_{ij} = \kappa_{ij} + \kappa^\ast_{ij}.
\end{equation}
In the case of a real-valued smoother $\kappa_{ij} = \half \delta_{ij}$,  but the general case allows for a complex coefficient.  In complex analysis there is a general relation that
\begin{equation}
  \int_0^{\infty} dt\, e^{- i \omega t} = \pi\, \delta(\omega) - i\, \mathcal{P.V.} \left[\tfrac{1}{\omega}\right]
\end{equation}
where $\mathcal{P.V.}$ denotes taking the Cauchy principal value.  Expressing $\varphi^{(1)} \star \varphi^{(1)}\ (\bar{t})$ in the frequency domain means that
\begin{equation}
  \begin{split}
    \kappa_{ij} &= \delta_{ij}\,\int_0^{\infty} d\bar{t}\, \int_{-\infty}^{\infty} d\omega\, \abs{\F{\varphi}^{(1)}(\omega)}^2 e^{-i\omega \bar{t}} \\
    &= \delta_{ij} \left(\pi\, \int_{-\infty}^{\infty} d\omega\, \abs{\F{\varphi}^{(1)}(\omega)}^2 \delta(\omega) -i \mathcal{P.V.} \int_{-\infty}^{\infty} d\omega\, \frac{\abs{\F{\varphi}^{(1)}(\omega)}^2}{\omega} \right)\\
    &= \delta_{ij} \left(\pi\, \abs{\F{\varphi}^{(1)}(0)}^2  -i \mathcal{P.V.} \int_{-\infty}^{\infty} d\omega\, \frac{\abs{\F{\varphi}^{(1)}(\omega)}^2}{\omega} \right).
  \end{split}
\end{equation}
The requirement of Eq. (\ref{appQuWZT:eq:kappaConstraint}) implies that $\pi\, \abs{\F{\varphi}^{(1)}(0)}^2 = \half$.  The principal value will be zero for a symmetric power distribution (\emph{i.e.} for real $\varphi^{(1)}$) but it is nonzero in general.  The remaining complex coefficient is what Gough refers to as a gauge freedom as it depends upon the nature of $\varphi^{(1)}$.  Here we will assume that $\varphi^{(1)}$ is real-valued and so $\kappa_{ij} = \kappa_{ij}^* = \half \delta_{ij}$.

\section{The Limiting Propagator}
We just considered one part of the second order term in the time-ordered exponential and how through normal ordering it develops a nonzero correction.  As this is just one part of the total time-ordered exponential, more complicated expressions are generated, involving iterated integrals of the form
\begin{equation}\label{appQuWZT:eq:iteratedIntegral}
  \int_0^t dt_n \dots \int_0^{t_{2}} dt_1\ c_{ij}(\lambda,t_n - t_{n-3})\,\dots c_{k l}(\lambda, t_5 - t_1).
\end{equation}
Depending upon the relative order of the times, these integrals may or may not converge to zero as $\lambda \rightarrow 0$.  It turns out that the only terms that are nonzero have time consecutive integrals, meaning that for each contraction we must have $c_{ij}(\lambda,\tau_n)$ be evaluated at $\tau_n = t_n - t_{n-1}$ \citep{gough_quantum_2005}.  In the above example, $\tau_n = t_n - t_{n-3}$ and $\tau_5 = t_5 - t_1$ are not time consecutive intervals and so Eq. (\ref{appQuWZT:eq:iteratedIntegral}) converges to zero.  Upon identifying the class of nonzero integrals, it is possible to then re-sum the expansion, which in general generates a Neumann series \cite{gough_quantum_2006}.

Once the $\lambda \rightarrow 0$ limit is taken, the normally ordered propagator can be identified and its equivalent It\={o} form QSDE written down.  A general QSDE for the unitary propagator $U_t$ can be written as
\begin{equation}\label{appQuWZT:eq:dUtG}
  d U_t = \left( G_{ij} d\Lambda^{ij}_{t} + G_{i0} dA^{i\, \dag}_t + G_{0j} dA^j_t + G_{00} dt\right) U_t.
\end{equation}
The coefficients $G_{\alpha\beta}$ are constrained to insure unitarity.  Appendix \ref{appQSDEs:sec:dUt} discusses these constraints at length and shows how they can be re-expressed in terms of the operators  $S$, $L$ and $H$.  Eq. (\ref{appQSDEs:eq:GtoSLHconversion}) gives this conversion.  We note here that $H$ is the part of the unitary that is completely uncoupled to the field operators and so we will have the correspondence $H = E_{00}$.

In terms of the system operators $E_{\alpha \beta}$ ($\alpha, \beta = 0, 1, \dots, d$) defining $H_{int}(\lambda, t)$ in Eq. (\ref{chQuLight:eq:QuWongZakaiHamiltonian}), the general limiting coefficients are
\begin{equation}\label{appQuWZT:eq:LimitdUCoeffs}
    G_{\alpha \beta} = -i E_{\alpha\beta} - E_{\alpha i} \left(\frac{1}{\ident + i\, \mathbbm{K}\mathbbm{E}}\right)_{ij}\, \kappa_{jk} E_{k \beta},
\end{equation}
where $i$ and $j$ start from 1 and we introduced the following notation.  In the most general case, we have a $d\times d$ matrix of constants $\mathbbm{K} = \kappa_{ij}$, as well as a matrix of operators $\mathbbm{E} = E_{ij}$.  A Neumann series is the operator-valued generalization of a geometric series, so that for an operator $T$, 
\begin{equation}
\sum_{n = 0}^\infty T^n = (1 - T)^{-1}
\end{equation}
is well defined whenever $1 - T$ is invertible.  The time consecutive contractions ultimately generate a Neumann series where $T$ is the matrix of operators $-i \mathbbm{K}\mathbbm{E} = -i \kappa_{ij} E_{jk}$.  The limiting coefficient then involves the $i,j$ component of the operator/matrix inverse $1/(\ident + i \mathbbm{K}\mathbbm{E})$.

\bibliography{Dissertation}

\end{document}